\documentclass{lmcs} 
\pdfoutput=1
\usepackage[utf8]{inputenc}

\usepackage{lastpage}
\lmcsdoi{21}{2}{16}
\lmcsheading{}{\pageref{LastPage}}{}{}%
{Feb.~02,~2024}{May~30,~2025}{}


\usepackage{hyperref}
\theoremstyle{plain} 

\usepackage{url}
\usepackage{xspace} 
\usepackage{hyperref}
\usepackage{cleveref}
\usepackage{mathabx} 
\usepackage{tikz}
\usetikzlibrary{automata}
\usetikzlibrary{arrows}
\usetikzlibrary{backgrounds}
\usetikzlibrary{positioning}
\usetikzlibrary{fit}
\usetikzlibrary{calc}
\usetikzlibrary{matrix}
\usetikzlibrary{patterns}
\usetikzlibrary{shapes.geometric}
\usetikzlibrary{decorations.pathreplacing}
\usetikzlibrary{decorations.pathmorphing}
\usepackage{mathbbol} 

\usepackage{subcaption}

\newcommand{\set}[1]{\{ #1 \}}

\newcommand{\pair}[2]{(#1,#2)}
\newcommand{\triple}[3]{(#1,#2,#3)}
\newcommand{\tuple}[2]{(#1,\ldots,#2)}

\newcommand{\Nat}{\ensuremath{\mathbb{N}}}
\newcommand{\Rat}{\ensuremath{\mathbb{Q}}}

\newcommand{\Zed}{\ensuremath{\mathbb{Z}}}
\newcommand{\true}{\top}

\newcommand{\powerset}[1]{\mathcal{P}(#1)}

\newcommand{\card}[1]{\mbox{\textup{card}}(#1)}

\newcommand{\adatadomain}{\ensuremath{\mathbb{D}}}
\newcommand{\rabinacc}{F} 

\newcommand {\mynext}{\mathsf{X}}
\newcommand {\until}{\mathsf{U}}
\newcommand {\release}{\mathsf{R}}
\newcommand {\sometimes}{\mathsf{F}}
\newcommand {\always}{\mathsf{G}}

\newcommand {\forallpaths} {\mathsf{A}}
\newcommand {\existspath} {\mathsf{E}}

\newcommand{\aset}{X}
\newcommand{\asetbis}{Y}

\newcommand{\avarprop}{p}

\newcommand{\aformula}{\phi} 
\newcommand{\aformulabis}{\psi} 
\newcommand{\aformulater}{\varphi} 
\newcommand{\subf}[1]{\text{sub}(#1)}
\newcommand{\parsubf}[2]{\text{sub}_{#1}(#2)}

\newcommand{\aautomaton}{{\mathbb A}}
\newcommand{\aautomatonbis}{{\mathbb B}}

\newcommand {\length}[1] {\ensuremath{|#1|}}

\newcommand{\dom}{dom} 

\newcommand{\egdef}{\stackrel{\mbox{\begin{tiny}def\end{tiny}}}{=}} 
\newcommand{\equivdef}{\stackrel{\mbox{\begin{tiny}def\end{tiny}}}{\equivaut}} 
\newcommand{\equivaut}{\;\Leftrightarrow\;}

\newcommand{\amap}{\mathfrak{f}}
\newcommand{\amapbis}{\mathfrak{g}}

\newcommand{\step}[1]{\xrightarrow{\!\!#1\!\!}}

\newcommand {\pspace} {\textsc{PSpace}\xspace}

\newcommand {\np} {\textsc{NP}\xspace}
\newcommand {\conp} {co\textsc{NP}\xspace}

\newcommand {\exptime} {\textsc{ExpTime}\xspace}

\newcommand {\ptime} {\textsc{PTime}\xspace}
\newcommand{\twoexptime}{\textsc{2ExpTime}\xspace}

\newcommand{\aalphabet}{\Sigma}     
\newcommand{\aword}{w}
\newcommand{\awordbis}{v}

\newcommand{\locations}{Q}
\newcommand{\alocation}{q}

\newcommand{\aletter}{\ensuremath{\mathtt{a}}}
\newcommand{\aletterbis}{\ensuremath{\mathtt{b}}}
\newcommand{\aletterter}{\ensuremath{\mathtt{c}}}

\newcommand{\variables}{\VAR}
\newcommand{\avariable}{\mathtt{x}}
\newcommand{\avariablebis}{\mathtt{y}}
\newcommand{\avariableter}{\mathtt{z}}

\newcommand{\atranslation}{\mathfrak{t}}
\newcommand{\defstyle}[1]{{\emph{#1}}}

\newcommand{\cut}[1]{}
\newcommand{\interval}[2]{[#1,#2]}

\newcommand{\adatum}{\ensuremath{\mathbb{d}}}

\newcommand{\Logic}[1]{\textup{#1}\xspace}
\newcommand{\CTL}{\Logic{CTL}}

\newcommand{\CTLStar}{\Logic{CTL$^*$}}
\newcommand{\LTL}{\Logic{LTL}}

\newcommand{\apath}{\pi}
\newcommand{\arun}{\rho}

\newcommand{\atree}{\mathbb{t}}

\newcommand{\asymtree}{t}
\newcommand{\anode}{\mathbf{n}}
\newcommand{\anodebis}{\mathbf{m}}

\newcommand{\abranch}{\mathcal{B}}
\newcommand{\alang}{{\rm L}}

\newcommand{\avaluation}{\mathfrak{v}}

\newcommand{\aterm}{\mathtt{t}}

\makeatletter
\def\mylab#1#2{%
  {\textbf{#1}}%
  \begingroup%
    \def\@currentlabel{\textbf{#1}}%
    \phantomsection\label{#2}%
  \endgroup%
}
\makeatother

\newcommand{\MSO}{\Logic{MSO}}

\newcommand{\transitions}{\delta}
\newcommand{\aatomiccons}{\theta}
\newcommand{\acons}{\Theta}

\newcommand{\apermutation}{\sigma}

\newcommand{\atransition}{\mathtt{T}}

\newcommand{\size}[1]{\mathtt{size}(#1)}

\newcommand{\vect}[1]{\boldsymbol{#1}}

\newcommand{\prefix}{<_{\footnotesize{\textup{pre}}}}

\newcommand{\apathformula}{\Phi}
\newcommand{\apathformulabis}{\Psi}
\newcommand{\myterms}[2]{{\rm T}^{#1}_{#2}}
\newcommand{\VAR}{V}

\newcommand{\adks}{\mathcal{K}}
\newcommand{\arelation}{\mathcal{R}}
\newcommand{\worlds}{\mathcal{W}}
\newcommand{\aworld}{\ensuremath{\mathit{w}}}
\newcommand{\satproblem}[1]{{\rm SAT(}#1{\rm)}}

\newcommand{\domain}[1]{\mathtt{dom}(#1)}

\newcommand{\init}{\text{in}}
\renewcommand{\degree}{D}

\newcommand{\aatm}{\mathcal{M}}
\newcommand{\aks}{\mathcal{K}}

\newcommand{\fd}[1]{\text{fd}(#1)}
\newcommand{\treeconstraints}[1]{{\rm TreeCons}(#1)}
\newcommand{\sattypes}[1]{{\rm STypes}(#1,\adatum_1,\adatum_\alpha)} 

\newcommand{\aconfiguration}{C}

\newcommand{\permut}[1]{\mathbf{S}_{#1}}
\newcommand{\kripke}{\mathcal{K}}

\newcommand{\pathquantifier}{\mathcal{Q}}
\newcommand{\safra}{\mathbf{s}}
\newcommand{\saflab}{\text{Lab}}
\newcommand{\acc}{\text{acc}}

\newcommand{\maxconstraintsize}[1]{\mathtt{MCS}(#1)}

\newcommand{\Configs}{\text{Configs}}

\newcommand{\DVAR}[3]{\mathtt{T}(#1,#2,#3)}
\newcommand{\advar}{\mathtt{xd}}
\newcommand {\slen}[1] {\ensuremath{\mathrm{slen}(#1)}}
\newcommand {\dslen}[1] {\ensuremath{\mathrm{slen}^{\hspace{-0.03in} d}(#1)}}
\newcommand {\rslen}[1] {\ensuremath{\mathrm{slen}^{\hspace{-0.03in} dr}(#1)}}
\newcommand {\uslen}[1] {\ensuremath{\mathrm{slen}^{\hspace{-0.03in} \downtouparrow}(#1)}}

\newcommand{\dlogic}{\mathcal{A}\mathcal{L}\mathcal{C}\mathcal{F}^{\mathcal{P}}(\Zed)}
\newcommand{\arbitraryletter}{\dag}

\newcommand{\newbigstar}{\bigstar^{\hspace{-0.02in}\mbox{\tiny{C}}}}

\newcommand{\newGt}{G_{\asymtree}^{\mbox{\tiny{C}}}}
\newcommand{\domnewGt}{V_{\asymtree}}
\newcommand{\domnewGtvar}{V^{var}_{\asymtree}}

\newcommand{\rp}{rp} 

\newcommand{\ancautomaton}{A}
\newcommand{\locautomaton}{\ancautomaton_{\mbox{\tiny cons($\aautomaton$)}}} 
\newcommand{\starautomaton}{\ancautomaton_{\newbigstar}}

\newcommand{\boundsattypes}{((\adatum_{\alpha} - \adatum_1)+3)^{2 \beta} 3^{2 \beta^2}}

\begin{document}

\title[Constraint Automata on Infinite Data Trees]{
  Constraint Automata on Infinite Data Trees: \texorpdfstring{\\}{} 
       From $\CTL(\Zed)$/$\CTLStar(\Zed)$  To Decision Procedures}
\thanks{Karin Quaas is supported by the Deutsche Forschungsgemeinschaft (DFG), project 504343613.}	

\author[S.~Demri]{St\'ephane Demri\lmcsorcid{0000-0002-3493-2610}}[a]
\author[K.~Quaas]{Karin Quaas}[b]

\address{Universit{\'e} Paris-Saclay, ENS Paris-Saclay, CNRS, Laboratoire M{\'e}thodes Formelles, 91190,
  Gif-sur-Yvette, France}	

\address{Universit\"at Leipzig, Fakult\"at f\"ur Mathematik und Informatik, Germany}	
\email{quaas@informatik.uni-leipzig.de}  






\begin{abstract}
We introduce the class of tree constraint automata with data values in $\Zed$ equipped with the less than relation  and equality predicates to constants, 
and we show that its nonemptiness problem is in \exptime. Using an automata-based approach, we establish that the satisfiability problem for $\CTL(\Zed)$ (\CTL with constraints in $\Zed$) is \exptime-complete, and the satisfiability problem for $\CTLStar(\Zed)$ is \twoexptime-com\-ple\-te (only decidability was known so far). By-product results with other concrete domains and other logics are also briefly discussed. 
\end{abstract}

\maketitle

\section{Introduction}
\label{section-introduction}
In this paper,  we study the satisfiability problem for the branching-time temporal logics $\CTL(\Zed)$ and $\CTLStar(\Zed)$, extending the classical temporal logics $\CTL$ and $\CTLStar$ in that atomic formulae express constraints about the relational structure $(\Zed, <, =, (=_\adatum)_{\adatum\in\Zed})$. Formulae in these logics are interpreted over Kripke structures that are annotated with values in $\Zed$,
see for instance the tree $\atree$ in  Figure~\ref{figure-symbolic-tree}
(page~\pageref{figure-symbolic-tree}).  A typical $\CTLStar(\Zed)$ formula is the expression $\forallpaths\always\sometimes(\avariable<\mynext\avariable)$
stating that on all paths infinitely often the value of the variable $\avariable$ at the current position is strictly smaller than the value of $\avariable$ at the next position. Formalisms defined over relational structures, also known as \emph{concrete domains}, are considered in many works,  including works on temporal logics~\cite{Groote&Mateescu98,Carapelle15,Mayr&Totzke16,Lechneretal18,Figueira&Majumdar&Praveen20,Conduracheetal21,Felli&Montali&Winkler22,Faella&Parlato24},
description logics~\cite{Lutz02,Lutz03,Lutz04,CarapelleTurhan16,Labai21,Baader&Rydval22,Alrabbaaetal23,Borgwardt&DeBortoli&Koopmann24}, 
and automata~\cite{Gascon09,Segoufin&Torunczyk11,Kartzow&Weidner15,Weidner16,Torunczyk&Zeume22,Peteler&Quaas22}. 
Combining reasoning in your favourite logic with reasoning in a relevant concrete domain reveals to be essential for numerous applications,
for instance for reasoning about ontologies, see e.g.~\cite{Lutz04,Labai&Ortiz&Simkus20}, or data-aware systems,
see e.g.~\cite{Deutsch&Hull&Vianu14,Felli&Montali&Winkler22bis}.
A brief survey can be found in~\cite{Demri&Quaas21}.  

Decidability results for concrete domains handled in~\cite{Lutz&Milicic07,Gascon09,Baader&Rydval22,Baader&DeBortoli24} exclude the ubiquitous
concrete domain $\pair{\Zed}{<,=,(=_{\adatum})_{\adatum \in \Zed}}$. By contrast,  decidability results for logics with concrete domain $\Zed$ require dedicated proof techniques, see e.g.~\cite{Bozzelli&Gascon06,Demri&DSouza07,Segoufin&Torunczyk11,Labai&Ortiz&Simkus20,Bhaskar&Praveen23bis}.
In particular, {\em fragments} of $\CTLStar(\Zed)$ are shown decidable
in~\cite{Bozzelli&Gascon06} using integral relational automata from~\cite{Cerans94}, and the satisfiability problem for the existential and the  universal fragment of \CTLStar with gap-order constraints (more general
than the ones in this paper) can be solved in \pspace~\cite[Theorem 14]{Bozzelli&Pinchinat14}. Another important breakthrough came with the decidability of $\CTLStar(\Zed)$~\cite[Theorem 32]{Carapelle&Kartzow&Lohrey16} (see also~\cite{Carapelle&Kartzow&Lohrey13}) by designing a reduction to a decidable second-order logic, whose formulae are made of  Boolean combinations of formulae from \MSO and from Weak MSO+U~\cite{Bojanczyk&Torunczyk12}, where U is the unbounding second-order quantifier, see e.g.~\cite{Bojanczyk04,Bojanczyk&Colcombet06} (in Weak MSO, second-order
quantification is over finite sets). This is all the more
remarkable as the decidability result is part of a powerful general approach~\cite{Carapelle&Kartzow&Lohrey16}, but no sharp complexity upper bound
can be inferred. More recently, the condition $C_{\Zed}$~\cite{Demri&DSouza07} to approximate the set of satisfiable symbolic models of a given $\LTL(\Zed)$ formula (in a problematic way, not necessarily an $\omega$-regular language)
is extended to the branching case in~\cite{Labai&Ortiz&Simkus20} leading
to the \exptime-membership of the concept satisfiability problem w.r.t. general TBoxes for the description logic $\dlogic$. 
However, no elementary complexity upper bounds for the satisfiability problem for $\CTL(\Zed)$ nor $\CTLStar(\Zed)$ were known since their decidability was
  established in~\cite{Carapelle&Kartzow&Lohrey13,Carapelle15,Carapelle&Kartzow&Lohrey16}. 
 
In this paper, we prove that the satisfiability problem for $\CTL(\Zed)$ is \exptime-complete, and the satisfiability problem for $\CTLStar(\Zed)$ is \twoexptime-complete.  We pursue the \emph{automata-based approach} for solving decision problems for temporal logics, following seminal works for temporal logics, see e.g.~\cite{Vardi&Wolper86,Vardi&Wolper94,Kupferman&Vardi&Wolper00}. 
This popular approach consists of reducing logical problems (satisfiability, model-checking) to automata-based decision problems while taking advantage of existing results and decision procedures from automata theory, see e.g.~\cite{Vardi&Wilke08}. In the presence of concrete domains, one can distinguish two approaches. The first one consists in designing constraint automata (see e.g.~\cite{Revesz02}) accepting
directly structures with data values, 
see e.g.~\cite{Segoufin&Torunczyk11,Kartzow&Weidner15,Peteler&Quaas22}
and also in~\cite{Figueira12} related to  automata accepting finite data trees.
The translation from logics with concrete domains to constraint automata often smoothly follows the plain case with temporal logics, see e.g.~\cite{Vardi&Wolper94} and the difficulty rests on the design of decision procedures for checking nonemptiness of constraint automata (but this is done only once). The second approach consists in reducing the satisfiability problem into the nonemptiness problem for automata handling {\em finite} alphabets,
see e.g.~\cite{Lutz01,Lutz04b,Gascon09,Labai&Ortiz&Simkus20}. In this case, the main effort is focused on the design of the translation (based on abstractions for tuples of data values) since the decision procedures for the target automata are usually already well-studied. 

It is well-known that decision procedures for \CTLStar are difficult to
design, and the combination with the concrete domain $\Zed$ is definitely challenging. Moreover, we aim at proposing a general framework:
we investigate a new class of {\em tree constraint automata}, understood as a target formalism in the pure tradition of the automata-based approach, and easy to reuse. The structures accepted by such tree constraint automata are {\em infinite}  trees in which nodes are labelled by a letter from a finite alphabet and a tuple in $\Zed^{\beta}$ for some $\beta \geq 1$ (this excludes the automata designed in~\cite{Figueira10,Figueira12} dedicated to \emph{finite} trees, and no predicate $<$ is involved). Decision problems for alternating automata over infinite alphabets are often undecidable, see e.g.~\cite{Neven&Schwentick&Vianu04,Lasota&Walukiewicz08,Demri&Lazic09,Iosif&Xu19},
and therefore we advocate the introduction of {\em nondeterministic} tree constraint automata without alternation. Our definition of tree constraint automata naturally extends the definition of constraint automata for words
(see e.g.~\cite{Cerans94,Revesz02,Segoufin&Torunczyk11,Kartzow&Weidner15,Peteler&Quaas22})
and as far as we know, the extension to infinite trees in the way done herein
has not been considered earlier in the literature. Note that our tree automaton model differs from the Presburger B\"uchi tree automata from~\cite{Seidl&Schwentick&Muscholl08,Bednarczyk&Fiuk22} for which, in the runs, arithmetical expressions are related  to constraints between the  numbers of children labelled by different locations. Herein, the arithmetical expressions state constraints between data values (not necessarily at the same node).

As a key result, we show that the nonemptiness problem for tree constraint automata over $\pair{\Zed}{<,=,(=_{\adatum})_{\adatum \in \Zed}}$
is \exptime-complete. In order to obtain the \exptime upper bound, we adapt results from~\cite{Labai&Ortiz&Simkus20,Labai21} (originally expressed in the context of interpretations for description logics) and we take advantage of several automata-based constructions for Rabin/Streett tree automata (see Lemma~\ref{lemma-intersection-rtca} and Lemma~\ref{lemma-intersection-automaton}). As a corollary, we establish that the satisfiability problem for $\CTL(\Zed)$ is \exptime-complete (Theorem~\ref{theorem-ctlz}), which is one of the main results of the paper. As a by-product, it also allows us to conclude that the concept satisfiability problem w.r.t. general TBoxes for the description logic $\dlogic$ is in \exptime, a result  known since~\cite{Labai&Ortiz&Simkus20}. The details can be found in~\cite[Section 5.2]{Demri&Quaas23}. By lack of space, we do not provide the details herein. 
  
Our main  contribution  is the  characterisation of the complexity for $\CTLStar(\Zed)$ satisfiability, which is an open problem evoked in~\cite[Section 9]{Carapelle&Kartzow&Lohrey16} and~\cite[Section 5]{Labai&Ortiz&Simkus20} (decidability was established ten years ago in~\cite{Carapelle&Kartzow&Lohrey13}). In general, our contributions stem from the cross-fertilisation of automata-based techniques for temporal logics and reasoning about (infinite) structures made of symbolic $\Zed$-constraints.
In Section~\ref{section-ctlstarz}, we show that the satisfiability problem for $\CTLStar(\Zed)$ is in \twoexptime by using Rabin tree constraint automata (introduced herein). We have to check that the essential steps for \CTLStar can be lifted to $\CTLStar(\Zed)$ to get the optimal upper bound
while guaranteeing that computationally we are in a position to provide an optimal upper bound. In Section~\ref{section-ctlstarz-special-form}, we establish a special form for $\CTLStar(\Zed)$ formulae from which tree constraint automata are defined, adapting the developments from~\cite{Emerson&Sistla84}. Moreover, determinisation of nondeterministic (B\"uchi) word constraint automata with Rabin word constraint automata is proved in Section~\ref{section-ctlstarz-determinisation-safra} following developments from~\cite[Chapter 1]{Safra89} but carefully adapted to the context of constraint automata. Observe that we can get away with nondeterminism, but
using alternation with registers/variables would more problematic because, in some way, this would require to know how to handle an unbounded number of data values due to alternation. 

\noindent
{\em This paper is a revised and completed version of the conference paper~\cite{Demri&Quaas23bis}. It can be seen also as a revised and more compact version of the arXiv report~\cite{Demri&Quaas23}. In order to limit the length of the body of this paper, several proofs are placed in the technical appendix
(see Appendices~\ref{appendix-first}--\ref{appendix-last}). 
}

\section{Temporal Logics with Numerical Domains}
\label{section-introduction-temporal-logics} 
\subsection{Constraints and Kripke Structures}
Let $\VAR = \set{\avariable, \avariablebis, \ldots}$ be a countably infinite set of variables. A \defstyle{term $\aterm$ over $\VAR$} is an expression of the form $\mynext^i \avariable$, where $\avariable \in \VAR$ and $\mynext^i$ is a (possibly empty) sequence of $i$ symbols `$\mynext$'. 
A term $\mynext^i \avariable$ should be understood as a variable (that needs to be interpreted) but, later on, we will see  that the prefix $\mynext^i$ will have a temporal interpretation. We write $\myterms{}{\VAR}$ to denote the set of all terms over $\VAR$. For all $i \in \Nat$, we write $\myterms{\leq i}{\VAR}$ 
to denote the subset of terms of the form $\mynext^j \avariable$, where $j\leq i$. For instance, $\myterms{\leq 0}{\VAR}=\VAR$ and $\myterms{\leq 1}{\VAR}=\VAR \cup \{\mynext \avariable \mid \avariable\in\VAR\}$.  
A \defstyle{valuation} $\avaluation: \myterms{}{\VAR} \to \Zed$ is a function that maps terms in $\myterms{}{\VAR}$ to elements in $\Zed$. 
To be precise, and a bit more general, 
we assume that the set of valuations is polymorphic: 
as a common feature it admits as argument a term and returns an integer. 
Quite often, valuations $\avaluation$ are of the form
$\set{\avariable_1, \ldots, \avariable_{\beta}} \to \Zed$ when we are only interested in the values for the variables in $\set{\avariable_1, \ldots, \avariable_{\beta}}$. 
Possibly, additional arguments are considered to provide
some context to the interpretation, such as a world or a path in a Kripke structure; details will follow but should not lead to any confusion. 

We consider the concrete domain $\pair{\Zed}{<,=,(=_{\adatum})_{\adatum \in \Zed}}$ (also written $\Zed$), where $=_{\adatum}$ is a unary predicate stating the equality with the constant $\adatum$, $<$ is the usual order on $\Zed$, and $=$ denotes  equality.  An \defstyle{atomic constraint $\aatomiccons$ over $\myterms{}{\VAR}$} is an expression of one of the forms
below:
\[
\aterm < \aterm' \ \ \ \ \
\aterm = \aterm' \ \ \ \ \
=_{\adatum}(\aterm) \ \mbox{(also written $\aterm = \adatum$)},
\]
where $\adatum \in \Zed$ and $\aterm, \aterm' \in \myterms{}{\VAR}$. 
A \defstyle{constraint} $\acons$ is defined as a Boolean combination of 
atomic constraints. Constraints are interpreted on valuations $\avaluation: \myterms{}{\VAR} \to \Zed$: a valuation $\avaluation$ \defstyle{satisfies}  $\aatomiccons$, written  $\avaluation \models \aatomiccons$ iff 
the interpretation of the terms in $\aatomiccons$ makes  $\aatomiccons$ true
in $\Zed$ in the usual way. The Boolean connectives are interpreted as usual. 
A constraint $\acons$  is \defstyle{satisfiable}
iff there is a valuation
$\avaluation: \myterms{}{\VAR} \to \Zed$ such that $\avaluation \models \acons$. 
Similarly, 
a constraint $\acons_1$ \defstyle{entails} a constraint $\acons_2$ (written
$\acons_1 \models \acons_2$)
iff for all valuations $\avaluation$, we have
$\avaluation \models \acons_1$ implies $\avaluation \models \acons_2$. 
The satisfiability problem
restricted to finite conjunctions of atomic constraints
can be solved in \ptime (see e.g.~\cite[Lemma 5.5]{Cerans94}) and entailment is in \conp.
The \ptime upper bound can be refined to a cubic bound by performing a linear
amount of calls to Bellman-Ford algorithm~\cite{Cormen&Leiserson&Rivest&Stein09}
that computes shortest paths in weighted directed graphs and detects negative cycles,
see also~\cite{Dechter&Meiri&Pearl91,Candeagoetal16}. 

\paragraph{Kripke structures.}
In order to define logics with the concrete domain $\Zed$, 
the semantical structures for such logics
are enriched with valuations that interpret
the variables by elements in $\Zed$. 
A \defstyle{$\Zed$-decorated Kripke structure} (or \defstyle{Kripke structure} for short) $\adks$ is a
triple $\triple{\worlds}{\arelation}{\avaluation}$, 
where $\worlds$ is a non-empty set of \defstyle{worlds},
$\arelation \subseteq \worlds \times \worlds$ is the accessibility relation
and $\avaluation: \worlds \times \variables \rightarrow \Zed$ is a valuation.
A Kripke structure $\aks$ is \defstyle{total} if 
for all $\aworld \in \worlds$ there is $\aworld' \in \worlds$ such that~$\pair{\aworld}{\aworld'} \in \arelation$. 
Given a Kripke structure $\aks =  \triple{\worlds}{\arelation}{\avaluation}$ and a world $w \in \worlds$, an \defstyle{infinite path}
$\apath$ from $\aworld$ is an $\omega$-sequence $\aworld_0,\aworld_1 \ldots \aworld_n, \ldots$ such that 
$\aworld_0 = \aworld$ and  for all
$i \in \Nat$, we have $\pair{w_i}{w_{i+1}} \in \arelation$. Finite paths are defined accordingly. 

\paragraph{Labelled trees.}
Along the paper, the expression $\interval{n}{m}$ with $n,m \in \Zed$
denotes the set $\set{k \in \Zed \mid n \leq k \leq m}$. 
Given $\degree \geq 1$, a \defstyle{labelled tree} of degree $\degree$ is a map
$\atree: \dom(\atree) \to \aalphabet$, where $\aalphabet$ is some (potentially infinite) alphabet and $\dom(\atree)$ is an infinite subset of  
$\interval{0}{\degree-1}^*$ such that $\anode\in \dom(\atree)$, and $\anode\cdot i\in\dom(\atree)$ for all $0\leq i <j$ whenever $\anode \cdot  j\in\dom(\atree)$ for some $\anode\in \interval{0}{\degree-1}^*$ and $j\in \interval{0}{\degree-1}$. 
The elements of $\dom(\atree)$ are called \defstyle{nodes}.
The empty word $\varepsilon$ is the \defstyle{root node} of $\atree$. 
For every
$\anode\in \dom(\atree)$, the elements
$\anode\cdot i \in \dom(\atree)$ with $i\in \interval{0}{\degree-1}$ are called the
\defstyle{children nodes of $\anode$}, and $\anode$ is called
the \defstyle{parent node of $\anode \cdot i$}.  
Nodes with the same parent node are called \defstyle{sibling nodes}; sibling nodes are implicitly ordered, though this feature is seldom used
in the document.
We say that the tree $\atree$ is a \defstyle{full $\degree$-ary tree} if every
node $\anode$ has exactly $\degree$
children $\anode\cdot 0, \dots, \anode\cdot (\degree-1)$ (equivalently, $\dom(\atree) = \interval{0}{\degree-1}^*$).
Given a tree $\atree$ and a node $\anode$ in $\dom(\atree)$, an infinite \defstyle{path}
in $\atree$ starting from
$\anode$ is an infinite sequence $\anode\cdot j_1 \cdot j_2 \cdot j_3 \dots$, where
$j_i\in\interval{0}{\degree-1}$ and $\anode \cdot j_1 \dots j_i\in \dom(\atree)$ for all $i\geq 1$.

A \defstyle{tree Kripke structure $\aks$} is a Kripke structure
$\triple{\worlds}{\arelation}{\avaluation}$ such that $\pair{\worlds}{\arelation}$ is a tree
(not necessarily a full $\degree$-ary tree).
Tree Kripke structures $\triple{\worlds}{\arelation}{\avaluation}$ such that
$\pair{\worlds}{\arelation}$ is isomorphic to the full $\degree$-ary tree are represented by maps of the form $\atree: \interval{0}{\degree-1}^* \to \Zed^{\beta}$.
This assumes that we only care about the values of the variables $\avariable_1, \ldots, \avariable_{\beta}\in\VAR$, 
and $\atree(\anode) = \tuple{\adatum_1}{\adatum_{\beta}}$ encodes $\avaluation(\anode, \avariable_i) = \adatum_i$  
for all $i \in \interval{1}{\beta}$. 

\subsection{The Logic $\CTLStar(\Zed)$} 
\label{section-ctlstar}
We introduce the logic $\CTLStar(\Zed)$, which extends the
branching-time temporal logic \CTLStar from~\cite{Emerson&Halpern86}
with constraints over $\Zed$. 
\defstyle{State formulae} $\aformula$ and \defstyle{path formulae}  $\apathformula$ of $\CTLStar(\Zed)$ 
are defined by
\[
  \aformula := \neg \aformula \mid \aformula \wedge \aformula \mid
\existspath \apathformula  
\ \ \ \ \ \ 
\apathformula  := \aformula \mid 
\aterm = \adatum \mid 
\aterm_1 = \aterm_2 \mid 
\aterm_1 < \aterm_2 \mid 
\neg \apathformula \mid \apathformula \wedge \apathformula \mid
\mynext \apathformula  \mid \apathformula \until \apathformula, 
\]
where $\aterm, \aterm_1, \aterm_2 \in \myterms{}{\VAR}$.
State formulae respectively path formulae are interpreted on worlds, respectively on infinite paths of a Kripke structure. 
Let $\adks = \pair{\worlds,\arelation}{\avaluation}$  is 
a total Kripke structure, and $\aworld \in\worlds$. 
We define the satisfaction relation (omitting the clauses for Boolean connectives) for state formulae by
\begin{itemize}
\item $\adks, \aworld \models \existspath \apathformula$ $\equivdef$ 
      there is an infinite path $\apath$ from $\aworld$ such that
       $\adks, \apath \models \apathformula$. 
       \end{itemize}
       Let $\apath=\aworld_0,\aworld_1,\dots$ be an infinite
path of $\adks$. Let us define $\avaluation(\apath,\mynext^j \avariable) \egdef 
\avaluation(\aworld_j,\avariable)$, for all terms of the form $\mynext^j \avariable$.
Hence, the term $\mynext^j \avariable$ refers to the value of the variable $\avariable$ exactly
$j$ steps ahead
along a path. 
For every  $n$, $\apath[n,+\infty)$ is the suffix of $\apath$ truncated by the $n$ first worlds. We define the satisfaction relation for path formula by
\begin{itemize}
\item $\adks, \apath \models \aterm = \adatum$ $\equivdef$
  $\avaluation(\apath,\aterm) = \adatum$;
\item       $\adks, \apath \models \aterm_1 \ \sim \ \aterm_2$  $\equivdef$
      $\avaluation(\apath,\aterm_1) \ \sim \ \avaluation(\apath,\aterm_2)$ for all $\sim \in \set{<,=}$;      
\item $\adks, \apath  \models \apathformula \until \apathformulabis$ $\equivdef$
there is $j \geq 0$ such that $\adks, \apath[j,+\infty) \models \apathformulabis$ and for
  all $j' \in \interval{0}{j-1}$, we have $\adks, \apath[j',+\infty)  \models \apathformula$;
\item    $\adks, \apath \models \mynext \apathformula$ $\equivdef$ 
  $\adks, \apath[1,+\infty) \models \apathformula$.
\end{itemize}

The \defstyle{size of a formula} is understood as its number of symbols with integers encoded in binary. 
As usual, we also use 
disjunction $\vee$, 
the universal path quantifier $\forallpaths$, defined by 
$\forallpaths \apathformula \egdef \neg \existspath \neg \apathformula$,  
and the standard
temporal connectives $\release$ and $\always$, defined by 
$\apathformula_1\release\apathformula_2 \egdef \neg(\neg\apathformula_1\until \, \neg \apathformula_2)$ and $\always \apathformula \egdef \existspath (\avariable<\avariable) \, \release \, \apathformula$. 
Propositional variables $\avarprop$ can easily be encoded with an atomic formula
$\existspath (\avariable_{\avarprop} = 0)$. 
A formula in $\CTLStar(\Zed)$ is in \defstyle{simple form} if it is in negation normal form (using $\vee$, $\forallpaths$, and $\release$ as primitives) and all terms occurring in the formula are from $\myterms{\leq 1}{\VAR}$. 
We define two fragments of $\CTLStar(\Zed)$:  formulae in the logic $\CTL(\Zed)$ are state formulae of the form 
\[
\aformula := \existspath \ \acons \mid \forallpaths \ \acons \mid
\neg \aformula \mid \aformula \wedge \aformula
\mid \aformula \vee \aformula \mid 
\existspath \mynext \aformula \mid
\existspath \aformula \until \aformula \mid
\existspath \aformula \release \aformula \mid
\forallpaths \mynext \aformula \mid
\forallpaths \aformula \until \aformula \mid
\forallpaths \aformula \release \aformula,
\]
where $\acons$ is a constraint.
By way of example, $\forallpaths \mynext(\existspath(\mynext \avariable =3)
\vee \forallpaths(\neg (\mynext \avariable = 3)))$ is a $\CTL(\Zed)$ formula. 
Formulae in the logic 
$\LTL(\Zed)$ are defined from path formulae for
$\CTLStar(\Zed)$ according to
$\apathformula  := \acons
\mid \apathformula \wedge \apathformula 
\mid \apathformula \vee \apathformula \mid 
\mynext \apathformula  \mid \apathformula \until \apathformula \mid \apathformula \release \apathformula
$,
where $\acons$ is a constraint. Negation occurs only in constraints
since the \LTL logical connectives have their dual in $\LTL(\Zed)$. 
In contrast to $\CTLStar(\Zed)$ and $\CTL(\Zed)$, 
$\LTL(\Zed)$ formulae are evaluated over infinite paths of valuations
$\avaluation:\VAR\to\Zed$ (no branching involved).

The \defstyle{satisfiability problem for $\CTLStar(\Zed)$}, written
$\satproblem{\CTLStar(\Zed)}$, is defined as follows.
\begin{description}
\item[Input] A $\CTLStar(\Zed)$ state formula $\aformula$. 
\item[Question] Is there a total Kripke structure $\aks$ and a world $\aworld$ such that
  $\aks, \aworld \models \aformula$? 
\end{description}
The satisfiability problem $\satproblem{\CTL(\Zed)}$ for $\CTL(\Zed)$  is defined analogously; for $\LTL(\Zed)$, $\satproblem{\LTL(\Zed)}$ is the problem to decide whether there exists an infinite sequence of valuations $\avaluation:\VAR\to\Zed$ such that $\avaluation\models\apathformula$ for a given $\LTL(\Zed)$ formula $\apathformula$.

Decidability, and more precisely, $\pspace$-completeness of $\satproblem{\LTL(\Zed)}$ is shown
in~\cite{Demri&Gascon08}. 
For some strict fragments of $\CTLStar(\Zed)$, decidability is shown in~\cite{Bozzelli&Gascon06,Bozzelli&Pinchinat14}.
It is only recently  in~\cite{Carapelle&Kartzow&Lohrey13,Carapelle15,Carapelle&Kartzow&Lohrey16}, that decidability is established for the full logic using a translation into a decidable second-order logic:
\begin{propC}[\cite{Carapelle&Kartzow&Lohrey13,Carapelle15,Carapelle&Kartzow&Lohrey16}]
\label{proposition-ckl16}
$\satproblem{\CTLStar(\Zed)}$ is decidable. 
\end{propC}
The proof
in~\cite{Carapelle&Kartzow&Lohrey13,Carapelle15,Carapelle&Kartzow&Lohrey16}
does not provide a complexity upper bound
as the target second-order 
logic admits an automata-based decision procedure
with open complexity~\cite{Bojanczyk&Torunczyk12,Bojanczyk04,Bojanczyk&Colcombet06}.
Moreover, the target logic uses the standard weak monadic theory of
one successor ($\Logic{WS1S}$) with a non-elementary complexity, which disqualifies
this approach to obtain a direct optimal complexity upper bound. 

Let us shortly explain why the satisfiability problem is challenging. 
First of all, 
observe that $\CTLStar(\Zed)$  has atomic formulae in which integer values at the current and successor states are compared.
This prevents us from using a simple translation from  $\CTLStar(\Zed)$ to \CTLStar with new propositions. Models of $\CTLStar(\Zed)$
formulae can be viewed as an infinite network of constraints on $\Zed$; 
even if a formula contains only a finite set of constants, a model may contain an infinite set of values, as it is the case for, e.g.,
the formula $\existspath \always (\avariable < \mynext \avariable)$. Hence a direct Boolean abstraction does not work;
by contrast, when variables cannot be compared at different positions, Boolean abstraction can be an option,
see recent developments in~\cite{Rodriguez&Sanchez23,Rodriguez&Sanchez24}. 
On the other hand, 
$\CTLStar(\Zed)$ has no freeze quantifier and no data variable quantification, and hence no way to directly compare values at
unbounded distance (this can only be done 
by propagating local constraints), unlike e.g. the formalisms
in~\cite{Deckeretal14,Song&Wu16,Bartek&Lelyk17,Abriola&Figueira&Figueira17}.
Hence, the lower  bounds from~\cite{Jurdzinski&Lazic11} cannot apply either.
Related work about the model-checking problem can be found in Section~\ref{section-rw-temporal-logics}.

In this paper, we prove the precise worst-case computational complexity of the problems $\satproblem{\CTLStar(\Zed)}$
and $\satproblem{\CTL(\Zed)}$, respectively. 
We follow the automata-based approach, that is, we translate formulae in our logics into equivalent
automata -- B\"uchi tree constraint automata for  $\CTL(\Zed)$, and Rabin tree constraint automata for
$\CTLStar(\Zed)$ -- so that we can reduce the satisfiability problem for the logics to the nonemptiness
problem for the corresponding automata. 

\section{Tree Constraint Automata}
\label{section-automata}
In this section, we introduce the class of tree constraint automata that accept
sets of infinite trees of the form $\atree:\interval{0}{\degree-1}^* \to (\aalphabet\times \Zed^\beta)$
for some finite alphabet $\aalphabet$ and some $\beta \geq 1$. The transition relation of such automata
puts constraints between the $\beta$ integer values at a node and the integer values at its children nodes. 
The automaton is equipped with an acceptance condition (B\"uchi, Rabin, Streett) that will be defined later on.  
The forthcoming definition is specific to the concrete domain $\Zed$, but it can be easily adapted to other concrete domains.
  Formally, a \defstyle{tree constraint automaton} (TCA, for short)  is a tuple
  $\aautomaton=(\locations,\aalphabet,\degree,\beta,\locations_\init,\delta,F)$, where
  \begin{itemize}
\item $\locations$ is a finite set of locations; $\aalphabet$ is a finite alphabet,
\item $\degree \geq 1$ is the branching degree of the trees processed by $\aautomaton$, or the degree of $\aautomaton$ for short, 
\item $\beta\geq 1$ is the number of variables, 
\item $\locations_\init\subseteq \locations$ is the set of initial locations, 
\item $\delta$ is a {\em finite} subset of $\locations \times \aalphabet \times (\treeconstraints{\beta} \times \locations)^\degree$, the
  transition relation.  
  Here, $\treeconstraints{\beta}$ denotes the Boolean combinations of atomic constraints  built over the terms $\avariable_1, \ldots, \avariable_{\beta}, \avariable'_1,\dots,\avariable'_\beta$, where $\avariable_i'$ stands for the term $\mynext \avariable_i$.  
  $\delta$ consists of tuples 
  $(\alocation,\aletter,(\acons_0,\alocation_0), \dots, (\acons_{\degree-1},\alocation_{\degree-1}))$,
  where $\alocation\in\locations$ is called the \defstyle{source location}, $\aletter\in\aalphabet$, $\alocation_0$,\dots, $\alocation_{\degree-1}\in \locations$,  and $\acons_0, \dots, \acons_{\degree-1}$ are constraints in $\treeconstraints{\beta}$.
\item  $F$ encodes the acceptance condition, defined below. 
\end{itemize}

\label{definition-run}
\paragraph{Runs.}
Let $\atree:\interval{0}{\degree-1}^* \to (\aalphabet\times \Zed^\beta)$ be an infinite full $\degree$-ary
tree over $\aalphabet\times \Zed^\beta$. 
A \defstyle{run} of $\aautomaton$ on $\atree$ is a mapping
$\arun:\interval{0}{\degree-1}^*\to \delta$ satisfying the following condition: 
for every $\anode \in \interval{0}{\degree-1}^*$ with
$\atree(\anode)=(\aletter,\vect{z})$ and $\atree(\anode\cdot  i)=(\aletter_i,\vect{z}_i)$
for all $0\leq i <\degree$, if $\arun(\anode)= (\alocation,\aletter, (\acons_0,\alocation_0), \dots, (\acons_{\degree -1}, \alocation_{\degree -1}))$, then for all $0\leq i <\degree$, we have
\begin{enumerate}[label=(\roman*)]
\item[(i)] the source location of $\arun(\anode\cdot i)$ is $\alocation_i$, and 
\item[(ii)]$\Zed \models \acons_i(\vect{z},\vect{z}_i)$, where $\Zed \models \acons_i(\vect{z},\vect{z}_i)$ is a shortcut for $[\vec{\avariable} \leftarrow \vect{z}, \vec{\avariable'} \leftarrow \vect{z}_i] \models
  \acons_i$ where $[\vec{\avariable} \leftarrow \vect{z}, \vec{\avariable'} \leftarrow \vect{z}_i]$ is a valuation $\avaluation$ on the variables
  $\set{\avariable_j, \avariable_j' \mid j \in \interval{1}{\beta}}$ with  $\avaluation(\avariable_j) = \vect{z}(j)$ and $\avaluation(\avariable_j') = \vect{z}_i(j)$ for all $j \in \interval{1}{\beta}$.
\end{enumerate}
We show an example of a run $\rho$ 
in Figure~\ref{figure-run}. 

\begin{figure}[t]
    \centering
    \begin{minipage}{0.32\textwidth}
        \centering
        \begin{tikzpicture}[->,>=stealth',shorten >=1pt,auto,node distance=4cm,thick,node/.style={circle,draw,scale=0.9}, roundnode/.style={circle, black, draw=black},]
\tikzset{every state/.style={minimum size=0pt}};
\node at (-1.5,0.2) {$\atree$}; 
\node[roundnode] (root) at (0,0) {}; 
\node[right=0.1mm of root] {$\aletter, (3,7)$};
\node[roundnode] (0) at (-1,-1) {}; 
\node[right=0.1mm of 0]{$\aletterbis, (0,0)$};
\node[roundnode] (1) at (1,-1) {}; 
\node[right=0.1mm of 1] {$\aletter, (2,7)$};
\node[roundnode] (10) at (0,-2) {}; 
\node[right=0.1mm of 10] {$\aletterbis, (0,0)$};
\node[roundnode] (11) at (2,-2) {}; 
\node[right=0.1mm of 11] {$\aletter, (1,7)$};
\node[below=0.05mm of 11] {$\vdots$}; 
\node[below=0.05mm of 10]  {$\vdots$}; 
\node[below=0.05mm of 0]  {$\vdots$}; 

\path [-] (root)  edge (0);
\path [-] (root)  edge (1);
\path [-] (1)  edge (10);
\path [-] (1)  edge (11);
 	\end{tikzpicture}
        \caption{A tree $\atree$}
        \label{figure-tree}
    \end{minipage}
    \hfill
    \begin{minipage}{0.60\textwidth}
        \centering
\begin{tikzpicture}[->,>=stealth',shorten >=1pt,auto,node distance=4cm,thick,node/.style={circle,draw,scale=0.9}, roundnode/.style={circle, black, draw=black},]
\tikzset{every state/.style={minimum size=0pt}};
\node at (-1.5,0.2) {$\arun$}; 
\node[roundnode] (root) at (0,0) {}; 
\node[right=0.1mm of root]{$\atransition_{\aletter}$};
\node[roundnode] (0) at (-0.8,-1) {}; 
\node[right=0.1mm of 0]{$\atransition_{\aletterbis}$};
\node[roundnode] (1) at (0.8,-1) {}; 
\node[right=0.1mm of 1]{$\atransition_{\aletter}$};
\node[roundnode] (10) at (0,-2) {}; 
\node[right=0.1mm of 10] {$\atransition_{\aletterbis}$};
\node[roundnode] (11) at (1.6,-2) {}; 
\node[right=0.1mm of 11]{$\atransition_{\aletter}$};
\node[below=0.05mm of 11] {$\vdots$}; 
\node[below=0.05mm of 10]  {$\vdots$}; 
\node[below=0.05mm of 0]  {$\vdots$}; 

\path [-] (root)  edge (0);
\path [-] (root)  edge (1);
\path [-] (1)  edge (10);
\path [-] (1)  edge (11);

\node at (5.3,-1.4) {\scalebox{.85}{\begin{tabular}{lll}
\multicolumn{3}{l}{$\atransition_{\aletter}=(\alocation, \aletter,(\acons'_0,\alocation),(\acons'_1,\alocation))$} \\
\hspace{5mm}$\acons'_0$  & $=$ & $(\avariable'_1=0)$ \\
\hspace{5mm}$\acons'_1$ & $=$ & $(\avariable'_1<\avariable_1<\avariable_2=\avariable'_2)$ \\
&&\\ 
\multicolumn{3}{l}{$\atransition_{\aletterbis}=(\alocation, \aletterbis,(\acons'_2,\alocation),(\acons'_2,\alocation))$} \\
\hspace{5mm}$\acons'_2$ & $=$ & $(\avariable_1=\avariable_2=0)$ \\
\end{tabular}}};
 	\end{tikzpicture} 
        \caption{A run $\arun$ of some TCA on $\atree$}
        \label{figure-run}
    \end{minipage}
\end{figure}
A run $\rho$ is \defstyle{initialized} if the source location of $\arun(\varepsilon)$ is in $\locations_\init$. 
Given a path $\apath=j_1 \cdot j_2 \cdot j_3 \dots $ in $\rho$ starting from  $\varepsilon$, we define 
$\inf(\arun,\apath)$ to be the set of locations that appear infinitely often as the source locations of the transitions in $\arun(\varepsilon) \arun(j_1) \arun(j_1 \cdot  j_2) \arun(j_1 \cdot j_2 \cdot j_3) \dots$.  

\paragraph{Acceptance Conditions.}
Analogously to classical $\omega$-tree automata, 
TCA are equipped with an acceptance condition. 
We will mainly study TCA with the B\"uchi acceptance condition: a TCA $\aautomaton=(\locations,\aalphabet,\degree,\beta,\locations_\init,\delta,F)$ is a \defstyle{B\"uchi TCA} if $F\subseteq \locations$ and a run $\arun$ is  \defstyle{accepting} if for all paths $\apath$ in $\arun$ starting from $\varepsilon$, we have $\inf(\arun,\apath) \cap  F\neq \emptyset$. 
We write $\alang(\aautomaton)$ to denote the set of trees $\atree$ for which there exists some initialized and accepting run of $\aautomaton$ on $\atree$.   
In the paper, whenever we do not further specify the acceptance condition of a TCA, we assume that the TCA is equipped with the B\"uchi acceptance condition.  

For dealing with the logics defined in the previous sections, 
we will also use TCA with generalised B\"uchi acceptance, Rabin acceptance, and Streett acceptance conditions. 
In a \defstyle{generalised B\"uchi TCA},  $\rabinacc$ is a set $\{F_1,\dots,F_k\} \subseteq \powerset{\locations}$ of states, and a run $\rho$ is accepting 
if for all
paths $\apath$ in $\arun$ starting from $\varepsilon$, for all $F_i \in \rabinacc$, we have 
$\inf(\arun,\apath) \cap  F_i\neq \emptyset$.
Unsurprisingly, as for generalised B\"uchi automata,  
for every generalised B\"uchi TCA $\aautomaton=(\locations,\aalphabet,\degree,\beta,\locations_\init,\delta,\rabinacc)$ with
$\rabinacc = \set{F_1, \ldots, F_{k}}$ there exists a TCA $\aautomaton'=(\locations',\aalphabet,\degree,\beta,\locations_\init',\delta',F')$ such that $\alang(\aautomaton)=\alang(\aautomaton')$ and the size of $\aautomaton'$ is quadratic in the size of
$\aautomaton$:
\begin{itemize}
\item $\locations' \egdef \interval{0}{k} \times \locations$, $\locations_{\init}' \egdef \set{0} \times \locations_{\init}$ and $F' \egdef \set{0} \times \locations$.
\item  Given $i \in \interval{0}{k}$ and $\alocation \in \locations$, we write
  $\text{nxt}(i,\alocation)$ to denote the copy number in $\interval{0}{k}$ such that
  $\text{nxt}(0,\alocation)= 1$ for all $\alocation \in \locations$,
  $\text{nxt}(i, \alocation) = i$ if $\alocation \not \in F_i$ and
  $\text{nxt}(i, \alocation) = i+1 \! \mod  (k+1)$ if $\alocation \in F_i$.
  The transition relation $\delta'$ is defined by
  $(\pair{i}{\alocation},\aletter,(\acons_0,\pair{i_0}{\alocation_0}), \dots, (\acons_{\degree-1},\pair{i_{\degree-1}}{\alocation_{\degree-1}})) \in \delta'$
  $\equivdef$ there is $(\alocation,\aletter,(\acons_0,\alocation_0), \dots, (\acons_{\degree-1},\alocation_{\degree-1})) \in \delta$ and
  for all $j \in \interval{0}{\degree-1}$, $i_j = \text{nxt}(i, \alocation_j)$.    
  \end{itemize}
Therefore, in the sequel, using generalised B\"uchi TCA instead of TCA has no consequence on worst-case complexity results. 

A TCA $\aautomaton$ is a \defstyle{Rabin TCA} if $F$ is a set of pairs of the form $\pair{L}{U}$, where $L,U\subseteq \locations$, and a run $\arun$ is accepting if 
for all paths $\apath$ in $\arun$ starting from $\varepsilon$, there is some pair  $\pair{L}{U} \in \rabinacc$ such that $\inf(\arun,\apath) \cap  L \neq \emptyset$ and $\inf(\arun,\apath) \cap U = \emptyset$. Note that
every TCA with set $F$ of accepting locations can be encoded as a Rabin TCA with a single Rabin pair $(F, \emptyset)$. Hence B\"uchi TCA can be seen as a special case of Rabin TCA. 
Finally, $\aautomaton$ is a \defstyle{Streett TCA} if $\rabinacc$ has the same form as for Rabin TCA, and a run is accepting 
if for all pairs $\pair{L}{U} \in \rabinacc$
(also called \defstyle{completemented pairs} in the literature) 
if $\inf(\arun,\apath) \cap  L \neq \emptyset$, then $\inf(\arun,\apath) \cap  U \neq \emptyset$.

\paragraph{Nonemptiness Problem.}
In this paper, we study the \defstyle{nonemptiness problem for TCA}, which asks whether a given TCA $\aautomaton$ satisfies  $\alang(\aautomaton) \neq \emptyset$. For determining the computational complexity, we need to consider the size of $\aautomaton$, which depends on several parameters.  Note that $\treeconstraints{\beta}$ is infinite, so that, in contrast to plain B\"uchi tree automata~\cite{Vardi&Wolper86}, the number of transitions in a TCA is {\em a priori} unbounded.  In particular, this means that $\card{\delta}$ is a priori unbounded, even if $\locations$, $\aalphabet$ and $\degree$ are fixed.
The maximal size of a constraint occurring in transitions is unbounded, too. 
We write $\maxconstraintsize{\aautomaton}$ to denote the maximal size of a constraint occurring in $\aautomaton$ (with binary encoding of the integers).

\paragraph{Closure Properties and Expressiveness.}
In order to conclude this section, let us drop a few lines about the expressive power of TCA, even though this is not really in the intended scope of this paper. Obviously, the languages accepted by TCA
(for any kind of acceptance condition) are closed under unions. Closure under complementations of tree languages accepted by TCA is to the best of our knowledge open.
For closure under intersection, 
we present the following lemma, which is also key 
for obtaining precise complexity bounds in Section~\ref{section-ctlstarz}. 

\renewcommand{\aautomatonbis}{{\mathbb A}}
\begin{lem}
    \label{lemma-intersection-rtca}
    Let $(\aautomatonbis_k)_{1 \leq k \leq n}$ be a family of
     Rabin TCA s.t. 
    $\aautomatonbis_k = \triple{\locations_k,\aalphabet,\degree,\beta}{\locations_{k,\init},\delta_k}{\rabinacc_k}$,
    $\card{\rabinacc_k} = N_k$ and $N = \underset{k}{\Pi} \ N_k$.
    There is a Rabin TCA $\aautomatonbis$ such that
    $\alang(\aautomatonbis) = \bigcap_{k} \alang(\aautomatonbis_k)$ and
    verifying the conditions below.
    \begin{itemize}
    \item The number of Rabin pairs is equal to $N$.
    \item $\maxconstraintsize{\aautomatonbis} \leq n+ \maxconstraintsize{\aautomatonbis_1} + \cdots +
      \maxconstraintsize{\aautomatonbis_n}$. 
    \item The number of locations is bounded by
      $
      \big( \underset{k}{\Pi} \ \card{\locations_k} \big) \times (2n)^{N}      
      $.
    \item The number of transitions is bounded by
      $\underset{k}{\Pi} \ \card{\delta_k}$.
    \end{itemize} 
\end{lem}
\renewcommand{\aautomatonbis}{{\mathbb B}}

The proof can be found in Appendix~\ref{appendix-proof-lemma-intersection-rtca}.
The proof contains a construction for the intersection of an arbitrary number of Rabin {\em tree constraint} automata mainly based on ideas from the proof of~\cite[Theorem 1]{Boker18} on Rabin {\em word} automata over {\em finite} alphabets. Our contribution is related to the extension to trees (in particular to handle the acceptance conditions),
to the use of constraints and to master the combinatorial explosion of the product automaton so that this is fine to get the final \twoexptime upper bound for \satproblem{$\CTLStar(\Zed)$}.  
 By the way,
the first part of the proof of the forthcoming Lemma~\ref{lemma-intersection-automaton} uses a particular
case of Lemma~\ref{lemma-intersection-rtca} with $n = 2$ Rabin TCA and no use of constraints. 

Regarding the expressive power of TCA,
there is no TCA that accepts the set of all full $\degree$-trees
over $\aalphabet \times \Zed$ such that all the data values are distinct, which is easier to capture with registers and alternation.
Indeed, {\em ad absurdum} suppose there is a TCA $\aautomaton$ such that
$\alang(\aautomaton)$ is the set of full binary trees such that all data values are distinct, say $\beta = 1$.
Let $\atree$ be the tree such that all data values are strictly greater than $\adatum_{\alpha}$ (greatest
constant occurring in $\aautomaton$, if any; otherwise zero),
the root is labelled by $\adatum_{\alpha}+1$ and then the data values
are defined in a breadth-first fashion by incrementing by one the data value. Hence,
the ``left'' child of the root has value  $\adatum_{\alpha}+2$,
the ``right'' child of the root has value  $\adatum_{\alpha}+3$, etc.
By definition, $\atree \in \alang(\aautomaton)$.
Let $\atree'$ be the tree obtained from $\atree$ by replacing the right subtree
of $\atree$'s root by its left subtree. In particular,
both the ``left'' and the ``right'' children of the root of $\atree'$ have data
value  $\adatum_{\alpha}+2$.
However, note that the same constraints are satisfied between a node and its parent node
in $\atree$ and in $\atree'$.
Consequently, $\atree' \in \alang(\aautomaton)$, which leads to a contradiction.

\section{Complexity of the Nonemptiness Problem for TCA}
\label{section-complexity-nonemptiness}
This section is dedicated to prove the
\exptime-completeness of the nonemptiness problem for B\"uchi 
TCA (Theorem~\ref{theorem-exptime-tca}) and Rabin TCA (Theorem~\ref{theorem-exptime-rtca}). 
We make a distinction between B\"uchi TCA and Rabin TCA because the complexity bounds differ slightly, see
Lemma~\ref{lemma-exptime-tca} and Lemma~\ref{lemma-intersection-automaton-rtca}. 
We mainly focus on the $\exptime$-membership, which will be used for establishing upper bounds 
for $\satproblem{\CTL(\Zed)}$ and $\satproblem{\CTLStar(\Zed)}$. 

But first, 
let us drop a few words on the proof of $\exptime$-hardness. 
The proof is 
by reduction from the acceptance problem for alternating Turing machines
running in polynomial space, see e.g.~\cite[Corollary 3.6]{Chandra&Kozen&Stockmeyer81}. 
The proof is presented for B\"uchi TCA in Appendix~\ref{appendix-exptime-hardness}. 
The key idea follows standard patterns: 
we use the variables of the TCA to store the content (that is a letter from a finite alphabet) of each of the polynomial cells of the alternating Turing machine. 
A bit more detailed, the TCA uses one variable for each tape cell. 
The tree structure of the runs of the alternating Turing machine can be enoded into trees that can be accepted by TCA. 
\exptime-hardness for Rabin TCA follows immediately as B\"uchi automata are a special case of Rabin TCA.

The proof of the \exptime-membership  of the nonemptiness problem for TCA is divided into two parts. 
In order to determine whether $\alang(\aautomaton)$ is nonempty for a given TCA $\aautomaton$, 
we first reduce the existence of some tree $\atree\in \alang(\aautomaton)$ to the existence of some \emph{regular symbolic tree} that is \emph{satisfiable}, that is, it admits a concrete model (Sections~\ref{section-symbolic-trees}
and~\ref{section-introduction-starproperty}). Second, we characterise the complexity of determining the
existence of such satisfiable regular symbolic trees
(Section~\ref{section-exptime-upper-bound}). 
The result for Rabin TCA is presented in Section~\ref{section-TCA-Rabin}. 

From now on, we assume a fixed TCA $\aautomaton=(\locations,\aalphabet,\degree,\beta,\locations_\init,\delta,F)$
with the constants $\adatum_1, \ldots, \adatum_{\alpha}$
occurring in $\aautomaton$ such that
$\adatum_1 <  \cdots <  \adatum_{\alpha}$ (we assume there is  at least one constant).

\subsection{Symbolic Trees}
\label{section-symbolic-trees}
A \defstyle{type} over the variables $\avariableter_1, \ldots,
\avariableter_n$ is an expression of the form
\begin{center}
$
(\bigwedge_{i} \acons_i^{\rm CST})
\wedge
(\bigwedge_{i < j} \avariableter_i \sim_{i,j}  \avariableter_j), \ \mbox{where} 
$
\end{center}
\begin{itemize}
\item for all $i \in \interval{1}{n}$,  $\acons_i^{\rm CST}$
  is equal to either $\avariableter_i < \adatum_1$, or $\avariableter_i > \adatum_{\alpha}$
  or $\avariableter_i = \adatum$ for some $\adatum \in \interval{\adatum_1}{\adatum_{\alpha}}$.
  This definition goes a bit beyond the constraint language in $\Zed$ (because of  expressions of the form $\avariableter_i < \adatum_1$ and
  $\avariableter_i > \adatum_{\alpha}$), but this is harmless in the sequel. 
  What really matters in a type is the way the variables are compared to each other
  and to the constants.
\item $\sim_{i,j} \in \set{>,=,<}$ for all $i < j$.   
\end{itemize}
Below, $\avariable = \avariablebis$ and $\avariablebis = \avariable$ are understood as identical,
as well as $\avariable < \avariablebis$ and $\avariablebis > \avariable$.
Checking the satisfiability of a type can be done in polynomial-time, based on a standard
cycle detection, see e.g.~\cite[Lemma 5.5]{Cerans94}. 
The set of \defstyle{satisfiable types} built over the terms $\avariable_1, \ldots, \avariable_{\beta},
\avariable_1', \ldots, \avariable_{\beta}'$ is written $\sattypes{\beta}$ ($n$ above is equal here to $2 \beta$).
By way of example, we provide below a satisfiable type in $\sattypes{2}$ with $\alpha = 1$: 
\[
\avariable_1 > \adatum_1 \wedge \avariable_2 = \adatum_1 \wedge
\avariable_1' = \adatum_1 \wedge \avariable_2' < \adatum_1 \wedge
\avariable_1 > \avariable_2 \wedge \avariable_1 > \avariable_1' \wedge
\avariable_1 > \avariable_2' \wedge
\avariable_2 = \avariable_1' \wedge \avariable_2 > \avariable_2'
\wedge \avariable_2' < \avariable_1'
\]

\begin{lem}
  \label{lemma-bound-sattypes}
$\card{\sattypes{\beta}} \leq \boundsattypes$. 
\end{lem}

Indeed, in a type
$
(\bigwedge_{i} \acons_i^{\rm CST})
\wedge
(\bigwedge_{i < j} \avariableter_i \sim_{i,j}  \avariableter_j)
$, 
each $\acons_i^{\rm CST}$ can take at most $(\adatum_{\alpha} - \adatum_{1}) + 3$ values
and the generalized conjunction has $(2 \beta -1) \beta \leq 2 \beta^2$ conjuncts, and each conjunct
can take at most three possible values. 

The \defstyle{restriction} of the type $\acons$ to some set of variables $\aset
\subseteq \set{\avariable_i, \avariable_{i}' \mid i \in \interval{1}{\beta}}$ is made of all the
conjuncts in which only variables in $\aset$ occur. The type $\acons_1$ restricted to
$\set{\avariable_{i}' \mid i \in \interval{1}{\beta}}$ \defstyle{agrees} with
the type $\acons_2$ restricted to
$\set{\avariable_{i} \mid i \in \interval{1}{\beta}}$ iff $\acons_1$ and $\acons_2$ are logically
equivalent modulo the renaming for which $\avariable_i$ and $\avariable'_i$ are substituted, for all
$i \in \interval{1}{\beta}$. 
For instance, in Figure \ref{figure-symbolic-tree}, $\acons$ restricted to $\{\avariable'_1,\avariable'_2\}$ agrees with $\acons_0$ restricted to $\{\avariable_1,\avariable_2\}$. 
The main properties of satisfiable types we use below are stated in the next lemma.

 \begin{lem} \label{lemma-types} \ 
  \begin{description}
   \item[(I)]
    Let $\vect{z}, \vect{z'} \in \Zed^{\beta}$. There is a unique
     satisfiable type
    $\acons \in \sattypes{\beta}$
  such that $\Zed \models \acons(\vect{z}, \vect{z'})$.
\item[(II)] For every constraint $\acons$ built over the terms
  $\avariable_1, \ldots, \avariable_{\beta},\avariable_1', \ldots, \avariable_{\beta}'$
  and the constants $\adatum_1, \ldots, \adatum_{\alpha}$ there is a disjunction
  $\acons_1 \vee \cdots \vee \acons_{\gamma}$ logically equivalent to $\acons$
  and each $\acons_i$ belongs to $\sattypes{\beta}$ (empty disjunction stands for $\perp$).
\item[(III)] For all $\acons \neq \acons' \in \sattypes{\beta}$,
      the constraint $\acons \wedge \acons'$ is not satisfiable. 
  \end{description}
   \end{lem}
The proof of Lemma~\ref{lemma-types} is by an easy verification and
its statement justifies the term `type' used in this context.

\paragraph{Abstraction with types.}
A \defstyle{symbolic  tree} $\asymtree$ is a map $\asymtree: \interval{0}{\degree-1}^* \rightarrow
\aalphabet \times \sattypes{\beta}$. Symbolic trees are intended
to be abstractions of trees labelled with concrete values in $\Zed$, defined as follows. 

Given a tree $\atree: \interval{0}{\degree-1}^* \rightarrow
\aalphabet \times \Zed^{\beta}$, its \defstyle{abstraction} is the symbolic tree $\asymtree_{\atree}: \interval{0}{\degree-1}^* \rightarrow
\aalphabet \times \sattypes{\beta}$ such that 
for all $\anode \cdot i \in \interval{0}{\degree-1}^*$ with $\atree(\anode) = \pair{\aletter}{\vect{z}}$ and
$\atree(\anode \cdot i) = \pair{\aletter_i}{\vect{z}_i}$,
$\asymtree_{\atree}(\anode \cdot i) \egdef \pair{\aletter_i}{\acons_i}$ for the
unique $\acons_i \in \sattypes{\beta}$ 
such that $\Zed \models \acons_i(\vect{z},\vect{z}_i)$. 
Note that the primed variables in $\acons_i$ refer to the $\beta$ values at the node $\anode\cdot i$, whereas the unprimed ones refer to
the $\beta$ values at the parent node $\anode$.
For the root node $\varepsilon$ with $\atree(\varepsilon) = \pair{\aletter}{\vect{z}}$, which has no parent node, 
we fix some arbitrary $\vect{0} \in \Zed^{\beta}$, and set $\asymtree_{\atree}(\varepsilon) \egdef \pair{\aletter}{\acons}$
for the
unique $\acons \in \sattypes{\beta}$ such that $\Zed \models \acons(\vect{0},\vect{z})$. 
A symbolic tree $\asymtree$ is \defstyle{satisfiable} $\equivdef$ there is $\atree: \interval{0}{\degree-1}^* \rightarrow
\aalphabet \times \Zed^{\beta}$ such that $\asymtree_{\atree} = \asymtree$. 
We say that $\atree$ \defstyle{witnesses the satisfaction of $\asymtree$}, also written
$\atree \models \asymtree$. 
A symbolic tree $\asymtree$ is \defstyle{regular} if its set of 
subtrees is finite. 

\begin{figure}[t]
    \centering
    \begin{tikzpicture}[->,>=stealth',shorten >=1pt,auto,node distance=4cm,thick,node/.style={circle,draw,scale=0.9}, roundnode/.style={circle, black, draw=black},]
\tikzset{every state/.style={minimum size=0pt}};
\node at (-1.5,0.2) {$\asymtree_\atree$}; 
\node[roundnode] (root) at (0,0) {}; 
\node[right=0.1mm of root] {$\aletter, \acons$};
\node[roundnode] (0) at (-1,-1) {}; 
\node[right=0.1mm of 0] {$\aletterbis, \acons_0$};
\node[roundnode] (1) at (1,-1) {}; 
\node[right=0.1mm of 1] {$\aletter, \acons_1$};
\node[roundnode] (10) at (0,-2) {}; 
\node[right=0.1mm of 10] {$\aletterbis, \acons_0$};
\node[roundnode] (11) at (2,-2) {}; 
\node[right=0.1mm of 11] {$\aletter, \acons_1$};
\node[below=0.05mm of 11] {$\vdots$}; 
\node[below=0.05mm of 10]  {$\vdots$}; 
\node[below=0.05mm of 0]  {$\vdots$}; 

\path [-] (root)  edge (0);
\path [-] (root)  edge (1);
\path [-] (1)  edge (10);
\path [-] (1)  edge (11);

\node at (6.5,-.5) {\scalebox{.85}{\begin{tabular}{lll}
$\acons$  & $\egdef$ & $0=\avariable_1=\avariable_2<\underline{\avariable'_1<\avariable'_2}$ \\
$\acons_0$ & $\egdef$ & $0=\avariable'_1=\avariable'_2<\underline{\avariable_1<\avariable_2}$ \\
$\acons_1$ & $\egdef$ & $0<\avariable'_1<\underline{\avariable_1<\avariable_2}=\avariable'_2$ 
\end{tabular}}}; 
 	\end{tikzpicture} 
    \caption{The symbolic tree $\asymtree_\atree$ for $\atree$ from Figure~\ref{figure-tree}}
    \label{figure-symbolic-tree}
\end{figure}

\paragraph{$\aautomaton$-consistency.}
In our quest to decide whether  $\alang(\aautomaton) \neq \emptyset$, 
we are interested in symbolic trees that satisfy certain properties that we subsume under the name
\defstyle{$\aautomaton$-consistent}. 
A symbolic tree $\asymtree: \interval{0}{\degree-1}^* \rightarrow \aalphabet \times \sattypes{\beta}$ is
\defstyle{$\aautomaton$-consistent} if the following conditions are satisfied:  
\begin{itemize}
\item $\asymtree$ is \defstyle{locally consistent}: for every node $\anode$, the type $\acons$ labelling
 $\anode$ restricted to $\avariable'_1,\dots,\avariable'_\beta$ agrees with all types $\acons_i$ labelling
 its children nodes $\anode\cdot i$ restricted to $\avariable_1,\dots,\avariable_\beta$, and
\item there exists an initialized and accepting tree $\arun:\interval{0}{\degree-1}^*\to \delta$ that satisfies condition (i) of the definition of runs (see page~\pageref{definition-run}),
  and for all $\anode\in\interval{0}{\degree-1}^*$ with $\asymtree(\anode)=(\aletter,\acons)$, $\asymtree(\anode\cdot i)=(\aletter_i,\acons_i)$ for all $i\in\interval{0}{\degree-1}$, and $\arun(\anode)=(\alocation,\aletter,(\acons'_0,\alocation_0)\dots(\acons'_{\degree-1},\alocation_{\degree-1}))$, 
 we have $\acons_i\models\acons'_i$ for all $i\in\interval{0}{\degree-1}$. 
\end{itemize}

\begin{exa}
In Figure~\ref{figure-symbolic-tree}, 
we show the abstraction $\asymtree_\atree$ of the tree $\atree$ from Figure \ref{figure-tree}. We assume that $\adatum_1=\adatum_\alpha=0$
is the only constant; consequently, $\asymtree_\atree$ uses constraints in $\sattypes{\beta}$ that are built
with variables $\avariable_1,\avariable_2$, their primed variants $\avariable'_1,\avariable'_2$, and the
constant $\adatum_1$. We underline constraints to illustrate the property of local consistency. 
\end{exa}

\begin{lem} \label{lemma-A-consistency}
  Let $\atree \in \alang(\aautomaton)$. The symbolic tree $\asymtree_{\atree}$
  is $\aautomaton$-consistent.
\end{lem}
It is routine to show Lemma~\ref{lemma-A-consistency}. 
Next we show that the 
set of all $\aautomaton$-consistent symbolic trees is $\omega$-regular, that is, it can be accepted by a classical tree automaton without constraints. 
In the following, we use the standard letter $\ancautomaton$ to distinguish
automata \emph{without constraints} from TCA
(in the same way, we use $\asymtree$ to denote symbolic trees whereas trees with data values
are denoted by $\atree$ by default).
\begin{lem}
\label{lemma-consistency-BTA}
There exists  a B\"uchi tree automaton (without constraints) $\locautomaton$ such that
$\alang(\locautomaton)= \{\asymtree \mid \asymtree \text{ is an  $\aautomaton$-consistent symbolic tree}\}$, 
the number of locations is bounded by $\card{\sattypes{\beta}} \times \card{\locations}$, and
the transition relation can be decided in polynomial-time in
$\card{\aalphabet}+\card{\delta}+D+\beta+\maxconstraintsize{\aautomaton}$.
\end{lem}

See also Lemma~\ref{lemma-bound-sattypes} for an upper bound on $\card{\sattypes{\beta}}$ that is
exponential in $\beta$ for  fixed $\adatum_1, \adatum_{\alpha}$. 

\begin{proof}
Let $\locautomaton$ be the B\"uchi tree automaton
\[
\locautomaton \egdef \triple{\locations',\aalphabet \times \sattypes{\beta},\degree}{\locations'_\init,\delta'}{F'},
\] where the components are defined as follows.
\begin{itemize}
\item $\locations' \egdef \sattypes{\beta} \times \locations$,
\item $\locations'_\init \egdef \set{\acons \in \sattypes{\beta}
  \mid \mbox{there exists} \ \vect{z} \in \Zed^{\beta} \ \mbox{such that} \ \Zed \models \acons(\vect{0}, \vect{z})} \times \locations_\init$,
\item $F' \egdef \sattypes{\beta} \times F$.
\item $\triple{\pair{\acons}{\alocation}}{\pair{\aletter}{\acons}}{
  \pair{\acons_0}{\alocation_0}, \ldots, \pair{\acons_{\degree-1}}{\alocation_{\degree-1}}} \in \delta'$ 
  $\equivdef$ there exists a transition  
  $$\triple{\alocation}{\aletter}{\pair{\acons_{0}'}{\alocation_0}, \ldots, \pair{\acons_{\degree-1}'}{\alocation_{\degree-1}}} \in \delta$$
  such that 
  \begin{itemize}
  \item for all $i \in \interval{0}{\degree-1}$, $\acons_i \models \acons_i'$ (\ptime check because $\acons_i
    \in \sattypes{\beta}$),
  \item $\acons_0, \ldots, \acons_{\degree-1}$ agree on $\avariable_1, \ldots, \avariable_{\beta}$ (same parent node),
  \item $\acons$ restricted to $\avariable_1', \ldots, \avariable_{\beta}'$ agrees
    with $\acons_0$ restricted to $\avariable_1, \ldots, \avariable_{\beta}$. 
  \end{itemize}
\end{itemize}
It is not hard to check that $\alang(\locautomaton)= \{\asymtree \mid \asymtree \text{ is an $\aautomaton$-consistent symbolic tree}\}$. 
The transition relation $\delta'$ can be decided in polynomial-time
in $\card{\aalphabet}+\card{\delta}+D+\beta+\maxconstraintsize{\aautomaton}$ (all the values in this sum are less than the size of $\aautomaton$): in order to check whether $\triple{\pair{\acons}{\alocation}}{\pair{\aletter}{\acons}}{
  \pair{\acons_0}{\alocation_0}, \ldots, \pair{\acons_{\degree-1}}{\alocation_{\degree-1}}} \in \delta'$, 
   we might go through all the transitions in $\delta$
   and verify the satisfaction of the
three conditions above. For the first condition, we need to check $\degree$ times whether $\acons_i \models \acons_i'$, which requires polynomial-time in
$\beta+\maxconstraintsize{\aautomaton}$. A similar time-complexity
is required to check the satisfaction of the two other conditions. 
Moreover,  
$\card{\locations'} =
\card{\sattypes{\beta}} \times \card{\locations}$ and $\card{\delta'}
\leq \card{\locations'}^{\degree+1} \times \card{\aalphabet} \times \card{\sattypes{\beta}}$.
\end{proof}

The result below is a variant of many similar results relating
symbolic models and concrete  models in logics for concrete domains,
see e.g.~\cite[Corollary 4.1]{Demri&DSouza07},~\cite[Lemma 3.4]{Gascon09},~\cite[Theorem 25]{CarapelleTurhan16}
and~\cite[Theorem 11]{Labai&Ortiz&Simkus20}. 
\begin{lem}
\label{lemma-satisfiable-symbolic-tree}
$\alang(\aautomaton) \neq \emptyset$ iff there is a \emph{satisfiable} symbolic  tree 
in $\alang(\locautomaton)$.
\end{lem}
\begin{proof}
``only if'':
  Suppose  $\alang(\aautomaton) \neq \emptyset$.  
  Then there exists some $\atree: \interval{0}{\degree-1}^* \to \aalphabet \times \Zed^{\beta}$
  in $\alang(\aautomaton)$. 
  The abstraction $\asymtree_{\atree}$ of $\atree$ is clearly satisfiable. 
  By Lemma~\ref{lemma-A-consistency}, 
  $\asymtree_{\atree}$ is $\aautomaton$-consistent. 
  By Lemma~\ref{lemma-consistency-BTA}, $\asymtree_{\atree}\in \alang(\locautomaton)$.

  ``if'': Suppose there exists some symbolic tree  $\asymtree \in \alang(\locautomaton)$
  such that $\asymtree$ is satisfiable. 
  Satisfiability of $\asymtree$ entails the existence of
  $\atree: \interval{0}{\degree-1}^* \rightarrow
  \aalphabet \times \Zed^{\beta}$ such that $\asymtree_\atree=\asymtree$, that is,  for all $\anode \cdot i \in \interval{0}{\degree-1}^*$ with
  $\atree(\anode) = \pair{\aletter}{\vect{z}}$,
  $\atree(\anode \cdot i) = \pair{\aletter_i}{\vect{z}_i}$
  and $\asymtree(\anode \cdot i) = \pair{\aletter_i}{\acons_i}$,
  we have  $\Zed \models \acons_i(\vect{z},\vect{z}_i)$.
  Moreover, if $\atree(\varepsilon) = \pair{\aletter}{\vect{z}}$,
  and $\asymtree(\varepsilon) = \pair{\aletter}{\acons}$, then 
  $\Zed \models \acons(\vect{0},\vect{z})$.
  By $\asymtree \in \alang(\locautomaton)$ and Lemma~\ref{lemma-consistency-BTA}, $\asymtree$ is $\aautomaton$-consistent. 
Hence there exists an initialized accepting tree $\arun: \interval{0}{\degree-1}^* \to\delta$ 
 that satisfies condition (i) of the definition of runs, and for all $\anode \in
\interval{0}{\degree-1}^*$
with $\asymtree(\anode) = \pair{\aletter}{\acons}$,
$\asymtree(\anode \cdot 0) = \pair{\aletter_0}{\acons_0}$, \ldots, $\asymtree(\anode \cdot (\degree-1)) =
\pair{\aletter_{\degree-1}}{\acons_{\degree-1}}$
and $\arun(\anode) = \triple{\alocation}{\aletter}{\pair{\acons_{0}'}{\alocation_0}, \ldots, \pair{\acons_{\degree-1}'}{\alocation_{\degree-1}}}$, we have $\acons_i \models \acons_i'$ 
for all $i \in \interval{0}{\degree-1}$ (by definition of $\delta'$ in $\locautomaton$). 
Let us briefly verify that $\arun$ is a run of $\aautomaton$ on $\atree$, and therefore $\alang(\aautomaton) \neq \emptyset$.
For this, it suffices to prove that condition (ii) of the definition of runs is satisfied. 
Let $\anode\in \interval{0}{\degree-1}^*$ with 
$\atree(\anode) = (\aletter,\vect{z})$ and $\atree(\anode \cdot i)=(\aletter_i,\vect{z}_i)$ for all $0\leq i<\degree$, 
and let  $\arun(\anode)= \triple{\alocation}{\aletter}{\pair{\acons_{0}'}{\alocation_0},
      \ldots, \pair{\acons_{\degree-1}'}{\alocation_{\degree-1}}}$. 
       By the transition relation of $\locautomaton$, for all $i \in \interval{0}{\degree-1}$, 
      $\acons_i \models \acons_i'$ with $\asymtree(\anode \cdot i) = \pair{\aletter_i}{\acons_i}$. 
      But, as proved above, we also have $\Zed \models \acons_i(\vect{z},\vect{z}_i)$ and
      therefore
      $\Zed \models \acons_i'(\anode \cdot i)(\vect{z},\vect{z}_i)$. 
\end{proof}

By virtue of Lemma \ref{lemma-satisfiable-symbolic-tree}, 
the nonemptiness problem for TCA can be reduced to deciding whether $\alang(\locautomaton)$ contains a
satisfiable symbolic tree $\asymtree$. 
However, not every 
symbolic tree is satisfiable, as the following example shows. 
\begin{exa} \label{example-symtree-not-sat}
Assume that every node along the rightmost branch in the symbolic tree $\asymtree_\atree$ in
Figure~\ref{figure-symbolic-tree} is labelled with $(\aletter,\acons_1)$. 
Then $\asymtree_\atree$ is not satisfiable: in order to satisfy the constraint $\avariable'_1<\avariable_1$,
the value of $\avariable_1$ must finally  become  smaller than $\adatum_1$, violating the constraint
$\adatum_1<\avariable_1$. 
\end{exa}
Thus, the most important property is to check whether $\alang(\locautomaton)$ contains some
\emph{satisfiable} symbolic tree. This is the subject of the next two subsections.

\subsection{Satisfiability for Regular Locally Consistent Symbolic Trees}
\label{section-introduction-starproperty}
In this subsection, 
we focus on deciding when $\alang(\locautomaton)$ contains a \emph{satisfiable} symbolic tree, 
while evaluating the computational complexity to check its existence.
For this, we define a property on symbolic trees, 
denoted by $(\newbigstar)$, 
such that if a symbolic tree $\asymtree$  satisfies $(\newbigstar)$ and is regular, 
then $\asymtree$ is satisfiable. 
We will prove that there exists a Rabin tree automaton that accepts all symbolic trees that satisfy $(\newbigstar)$. 
As a result, we will be able to reduce the nonemptiness problem to the nonemptiness problem for
the product automaton of this automaton and $\locautomaton$. 
The definition of $(\newbigstar)$ is based on an infinite tree-like graph inferred by a given 
symbolic tree $\asymtree$, denoted by $\newGt$.   
Similar symbolic structures are introduced in~\cite{Lutz01,Demri&DSouza07,Carapelle&Kartzow&Lohrey13,Labai&Ortiz&Simkus20}.

So let us start with defining $\newGt$.
Let $\asymtree:\interval{0}{\degree-1}^*\to \aalphabet\times\sattypes{\beta}$ be  a symbolic tree.
The graph $\newGt$ 
is equal to the structure
\[
\newGt = (\domnewGt,\step{=},\step{<},U_{<\adatum_1},
(U_{\adatum})_{\adatum\in [\adatum_1,\adatum_\alpha]},U_{>\adatum_\alpha}),
\]
where
\begin{itemize}
\item $\domnewGt=\interval{0}{\degree-1}^*\times \DVAR{\beta}{\adatum_1}{\adatum_{\alpha}}$ with
  $\DVAR{\beta}{\adatum_1}{\adatum_{\alpha}} \egdef \{\avariable_1,\dots,\avariable_\beta\}\cup\{\adatum_1,\adatum_\alpha\}$, 
\item $\step{=}$ and $\step{<}$ are two binary relations over $\domnewGt$, and
\item $\{U_{<\adatum_1},U_{\adatum_1},U_{\adatum_1+1},\dots,U_{\adatum_\alpha},U_{>\adatum_\alpha}\}$ is a partition of $\domnewGt$. 
\end{itemize}
Elements in $\{\avariable_1,\dots,\avariable_\beta\}\cup\{\adatum_1,\adatum_\alpha\}$ are denoted by $\advar, \advar_1, \advar_2, \ldots$ (variables or constants)
and $\domnewGtvar \egdef \interval{0}{\degree-1}^*\times \{\avariable_1,\dots,\avariable_\beta\}$. 

Two nodes $\anode, \anode' \in \interval{0}{\degree-1}^*$ in $\asymtree$ are \defstyle{neighbours}
$\equivdef$ either $\anode = \anode'$, or $\anode = \anode' \cdot j$ or $\anode' = \anode \cdot j$ for
some $j \in \interval{0}{\degree-1}$ (siblings are not neighbours). 
Two elements $\pair{\anode}{\advar}$ and $\pair{\anode'}{\advar'}$
in $\interval{0}{\degree-1}^* \times \DVAR{\beta}{\adatum_1}{\adatum_{\alpha}}$ are
\defstyle{neighbours} $\equivdef$ $\anode$ and $\anode'$ are neighbours.
By construction, the edges in $\newGt$ are possible only between  neighbour elements. 

The rationale behind the construction of $\newGt$ is to reflect the constraints between parent and children nodes as well as the constraints regarding constants, in such a way that, if
$\atree$ witnesses the satisfaction of $\asymtree$, 
then, e.g., $\atree(\anode)(\advar) < \atree(\anode')(\advar')$ if $(\anode,\advar)\step{<}(\anode',\advar')$, and 
$\atree(\anode)(\advar)=\adatum_1$ if $(\anode,\advar)\in U_{\adatum_1}$
(see also Lemma~\ref{lemma-correctness-newgt}).  
Here are all conditions for building $\newGt$. 
\begin{description}
\item[(VAR)] For all $\pair{\anode}{\avariable_i}, \pair{\anode'}{\avariable_{i'}}
      \in \domnewGtvar$, for all $\sim \in \set{<,=}$, 
      $\pair{\anode}{\avariable_i} \step{\sim} \pair{\anode'}{\avariable_{i'}}$ 
      iff one of the conditions below holds:
  \begin{itemize}
  \item either $\anode' = \anode \cdot j$ and $\avariable_i \sim \avariable_{i'}'$  in $\acons$ with $\asymtree(\anode') = \pair{\aletter}{\acons}$,
  \item or     $\anode = \anode'$ and $\avariable_i' \sim \avariable_{i'}'$ in $\acons$ with
  $\asymtree(\anode') = \pair{\aletter}{\acons}$, 
  \item or     $\anode = \anode' \cdot j$ and $\avariable_i' \sim \avariable_{i'}$ 
    in $\acons$ with $\asymtree(\anode) = \pair{\aletter}{\acons}$.
  \end{itemize}
\item[(P1)]  For all $\adatum \in \interval{\adatum_1}{\adatum_{\alpha}}$ and 
        $\pair{\anode}{\avariable_j} \in \domnewGtvar$,
  $\pair{\anode}{\avariable_j} \in U_{\adatum}$ 
  iff $\avariable_j' = \adatum$ in
  $\acons$ with $\asymtree(\anode) = \pair{\aletter}{\acons}$.
\item[(P2)]   For all $\pair{\anode}{\avariable_j} \in \domnewGtvar$,
  $\pair{\anode}{\avariable_j} \in U_{< \adatum_1}$ 
  iff
  $\avariable_j' < \adatum_1$ in
  $\acons$ with $\asymtree(\anode) = \pair{\aletter}{\acons}$.
\item[(P3)]   For all $\pair{\anode}{\avariable_j}
                 \in \domnewGtvar$,
                 $\pair{\anode}{\avariable_j} \in U_{> \adatum_{\alpha}}$
                  iff
  $\avariable_j' > \adatum_{\alpha}$ in $\acons$ with
  $\asymtree(\anode) = \pair{\aletter}{\acons}$.

\item[(P4)]    For all $\anode \in \interval{0}{\degree-1}^*$, $\pair{\anode}{\adatum_1} \in U_{\adatum_1}$ and 
  $\pair{\anode}{\adatum_{\alpha}} \in U_{\adatum_{\alpha}}$.
  
\item[(CONS)] This condition is about elements of $\domnewGt$ labelled by constants and how the
  edge labels reflect the relationships between the constants. Formally, 
  for all 
  $\pair{\anode}{\advar}, \pair{\anode'}{\advar'} \in
   (\domnewGt \setminus \domnewGtvar)$
  such that $\anode$ and $\anode'$ are neighbours, 
  for all $\adatum^{\dag}, \adatum^{\dag \dag}$ in `$< \adatum_1$', $\adatum_{1}$, \ldots,
  $\adatum_{\alpha}$, `$> \adatum_{\alpha}$' 
  such that $\pair{\anode}{\advar} \in U_{\adatum^{\dag}}$ and $\pair{\anode'}{\advar'} \in U_{\adatum^{\dag \dag}}$,
              for all $\sim \in \set{<,=}$, 
              there is an edge $\pair{\anode}{\advar} \step{\sim} \pair{\anode'}{\advar'}$ 
              iff
              \begin{itemize}
              \item either $\adatum^{\dag}, \adatum^{\dag \dag} \in \interval{\adatum_1}{\adatum_{\alpha}}$
                and $\adatum^{\dag} \sim \adatum^{\dag \dag}$,
              \item  or $\adatum^{\dag} =$ `$< \adatum_1$', $\adatum^{\dag \dag} \neq$ `$< \adatum_1$' and
                $\sim$ is equal to $<$,
              \item
                or $\adatum^{\dag} \neq$ `$> \adatum_{\alpha}$', $\adatum^{\dag \dag} =$ `$> \adatum_{\alpha}$'
                and $\sim$ is equal to $<$.
              \end{itemize}
              Here,  $\adatum^{\dag}$ and $\adatum^{\dag \dag}$ are subscripts that are not 
              associated to data values when it takes the values `$< \adatum_1$' or `$> \adatum_{\alpha}$'. 
\end{description}

Below, we also write $\pair{\anode}{\advar} \step{>} \pair{\anode'}{\advar'}$ instead
of $\pair{\anode'}{\advar'} \step{<} \pair{\anode}{\advar}$.
Observe that in the condition (CONS), it is also possible to remove the requirement that
$\anode$ and $\anode'$ are neighbours. This would still be fine but we added this condition
so that all over $\newGt$, an edge exists between two elements only if the underlying nodes are neighbours
(but other edges could be easily inferred). 
The forthcoming condition $(\newbigstar)$ shall be defined
on such labelled graphs, see also the symbolic trees accepted
by construction of $\starautomaton$ in Lemma~\ref{lemma-automaton-star}.
The construction of $\newGt$ is done from any symbolic tree $\asymtree$, even if it is not locally consistent;
local consistency is taken care of by $\locautomaton$.

\begin{exa}
In Figure~\ref{figure-graph}, we illustrate the definition of the graph ${G_{\asymtree_\atree}^{\mbox{\tiny{C}}}}$
for the   symbolic tree $\asymtree_\atree$ in Figure \ref{figure-symbolic-tree}. 
The edges labelled with $=$ or $<$ reflect the constraints (we omit edges if they can be inferred
from the other edges). For instance, 
$(1,\avariable_1)\step{<}(\varepsilon,\avariable_1)$ corresponds to the constraint
$\avariable'_1<\avariable_1$ (by application of (VAR), third item). 
Grey nodes are in $U_{\adatum_1}$, all other nodes are in $U_{>\adatum_1}$ (no nodes in $U_{<\adatum_1}$).

\begin{figure} 
\scalebox{0.9}{
  \begin{tikzpicture}[->,>=stealth',shorten >=1pt,auto,node distance=4cm,thick,node/.style={circle,draw,scale=0.9}, roundnode/.style={circle, black, draw=black},]
\tikzset{every state/.style={minimum size=0pt},
dotted_block/.style={draw=black!20!white, line width=1pt, dotted, inner sep=3mm, minimum width=5.4cm, rectangle, rounded corners, minimum height=8mm},
boxnode/.style={rectangle, draw=black,minimum width=10mm,rounded corners=0.2cm}
} 
;
\node[boxnode,fill=black!10] (root1) at (0,0) {$\pair{\varepsilon}{\adatum_1}$}; 
\node[boxnode] (root2) at (1.8,0) {$\pair{\varepsilon}{\avariable_1}$}; 
\node[boxnode] (root3) at (3.6,0) {$\pair{\varepsilon}{\avariable_2}$}; 
\node [dotted_block] at (1.8,0) {};

\node[boxnode,fill=black!10,xshift=-4cm,yshift=-1.5cm] (0_1) at (0,0) {$\pair{0}{\adatum_1}$}; 
\node[boxnode,fill=black!10,xshift=-4cm,yshift=-1.5cm] (0_2) at (1.8,0) {$\pair{0}{\avariable_1}$}; 
\node[boxnode,fill=black!10,xshift=-4cm,yshift=-1.5cm] (0_3) at (3.6,0) {$\pair{0}{\avariable_2}$}; 
\node [dotted_block,xshift=-4.4cm,yshift=-1.5cm] at (2.2,0) {};

\node[boxnode,fill=black!10,xshift=3cm,yshift=-1.5cm] (1_1) at (0,0) {$\pair{1}{\adatum_1}$}; 
\node[boxnode,xshift=3cm,yshift=-1.5cm] (1_2) at (1.8,0) {$\pair{1}{\avariable_1}$}; 
\node[boxnode,xshift=3cm,yshift=-1.5cm] (1_3) at (3.6,0) {$\pair{1}{\avariable_2}$}; 
\node [dotted_block,xshift=2.6cm,yshift=-1.5cm] at (2.2,0) {};

\node[boxnode,fill=black!10,xshift=-1cm,yshift=-3cm] (10_1) at (0,0) {$\pair{10}{\adatum_1}$}; 
\node[boxnode,fill=black!10,xshift=-1cm,yshift=-3cm] (10_2) at (1.8,0) {$\pair{10}{\avariable_1}$}; 
\node[boxnode,fill=black!10,xshift=-1cm,yshift=-3cm] (10_3) at (3.6,0) {$\pair{10}{\avariable_2}$}; 
\node [dotted_block,xshift=-1.4cm,yshift=-3cm] at (2.2,0) {};

\node[boxnode,fill=black!10,xshift=5cm,yshift=-3cm] (11_1) at (0,0) {$\pair{11}{\adatum_1}$}; 
\node[boxnode,xshift=5cm,yshift=-3cm] (11_2) at (1.8,0) {$\pair{11}{\avariable_1}$}; 
\node[boxnode,xshift=5cm,yshift=-3cm] (11_3) at (3.6,0) {$\pair{11}{\avariable_2}$}; 
\node [dotted_block,xshift=4.6cm,yshift=-3cm] at (2.2,0) {};

\node at (1, -3.8) {$\vdots$}; 
\node at (7, -3.8) {$\vdots$};

\path [->] (root1)  edge[bend left=40] node[above] {\scriptsize{$<$}}  (root2);
\path [->] (root2)  edge[bend left=40] node[above] {\scriptsize{$<$}}  (root3);

\path [-] (0_1)  edge[bend left=40] node[above] {\scriptsize{$=$}}  (0_2);
\path [-] (0_2)  edge[bend left=40] node[above] {\scriptsize{$=$}}  (0_3);

\path [-] (root1)  edge[bend right=10] node[above,sloped] {\scriptsize{$=$}}  (0_3);

\path [<-] (root2)  edge[bend right=10] node[above,sloped] {\scriptsize{$>$}}  (1_2);
\path [-] (root1)  edge[bend right=10] node[above,sloped] {\scriptsize{$=$}}  (1_1);
\path [-] (1_3)  edge[bend right=10] node[above,sloped] {\scriptsize{$=$}}  (root3);

\path [->] (1_1)  edge[bend right=40] node[below] {\scriptsize{$<$}}  (1_2);
\path [->] (1_2)  edge[bend right=40] node[below] {\scriptsize{$<$}}  (1_3);

\path [-] (10_3)  edge[bend right=10] node[above, sloped] {\scriptsize{$=$}}  (1_1);
\path [-] (11_3)  edge[bend right=10] node[above,sloped] {\scriptsize{$=$}}  (1_3);

\path [-] (10_1)  edge[bend left=40] node[above] {\scriptsize{$=$}}  (10_2);
\path [-] (10_2)  edge[bend left=40] node[above] {\scriptsize{$=$}}  (10_3);

\path [->] (11_1)  edge[bend right=40] node[below] {\scriptsize{$<$}}  (11_2);
\path [->] (11_2)  edge[bend right=40] node[below] {\scriptsize{$<$}}  (11_3);

\path [-] (1_1)  edge[bend right=10] node[below,sloped] {\scriptsize{$=$}}  (11_1);
\path [<-] (1_2)  edge[bend right=10] node[below,sloped] {\scriptsize{$>$}}  (11_2);

 	\end{tikzpicture} 
}
\caption{The labelled graph ${G_{\asymtree_\atree}^{\mbox{\tiny{C}}}}$ for the
  symbolic tree $\asymtree_\atree$ from Figure~\ref{figure-symbolic-tree}
  (page~\pageref{figure-symbolic-tree}).} 
\label{figure-graph}
\end{figure}
\end{exa}

The rationale behind the construction of $\newGt$ is best illustrated  with the next lemma. 
\begin{lem}\label{lemma-correctness-newgt}
  Let $\asymtree$ be a symbolic tree and $\atree:\interval{0}{\degree-1}^*\to\Zed^\beta$ be such that
 $\atree$ witnesses the satisfaction of $\asymtree$. 
Then, if $\pair{\anode}{\advar} \step{\sim}\pair{\anode'}{\advar'}$ is an edge in  $\newGt$ for some $\sim \in \set{<,=}$, then  $\atree(\anode)(\advar) \sim \atree(\anode')(\advar')$. 
\end{lem}
By convention, $\atree(\anode)(\adatum_1) \egdef \adatum_1$ and
$\atree(\anode)(\adatum_{\alpha}) \egdef \adatum_{\alpha}$.
Moreover, whenever $\atree(\anode) = \tuple{z_1}{z_{\beta}}$, 
we write $\atree(\anode)(\avariable_i)$ to denote $z_i$. 
\begin{proof}
Suppose $\pair{\anode}{\advar} \step{=}\pair{\anode'}{\advar'}$ is an edge in  $\newGt$. 
Let $\asymtree(\anode)=(\cdot, \acons_\anode)$ and  $\asymtree(\anode')=(\cdot,\acons_{\anode'})$. 
We use $\acons'_\anode$ to denote the restriction of $\acons_\anode$ to $\{\avariable_1',\dots,\avariable_\beta'\}$; similarly for $\acons'_{\anode'}$. 
By (VAR) and (CONS), there are four possible cases. 
\begin{itemize}
\item Suppose $\advar=\avariable_i$, $\advar'=\avariable_{i'}$ for some $1\leq i,i'\leq \beta$. 
\begin{itemize}
\item Suppose $\anode'=\anode \cdot j$ for some $j\in\interval{0}{\degree-1}$, and
$\avariable_i = \avariable'_{i'}\in \acons_{\anode'}$. 
By $\Zed \models \acons_{\anode'}\pair{\atree(\anode)}{\atree(\anode')}$
($\atree$ witnesses the satisfaction of $\asymtree$) we can conclude
$\atree(\anode)(\avariable_i) = \atree(\anode')(\avariable_{i'})$. 
\item Suppose $\anode'=\anode$ and $\avariable'_i = \avariable'_{i'}\in \acons_{\anode'}$.
 By $\Zed \models \acons'_{\anode'}(\atree(\anode'))$
 ($\atree$ witnesses the satisfaction of $\asymtree$) we can conclude
 $\atree(\anode)(\avariable_i) = \atree(\anode')(\avariable_{i'})$.
\item Suppose $\anode=\anode' \cdot j$ for some $j\in\interval{0}{\degree-1}$, and $\avariable'_i = \avariable_{i'}\in \acons_{\anode}$. 
By $\Zed \models \acons_{\anode}\pair{\atree(\anode')}{\atree(\anode)}$
($\atree$ witnesses the satisfaction of $\asymtree$, and we omit this precision
in the sequel) we can
conclude $\atree(\anode)(\avariable_i) = \atree(\anode')(\avariable_{i'})$. 
\end{itemize}
\item Suppose $\advar=\avariable_i$ for some $1\leq i\leq \beta$ and $\advar'=\adatum_1$. 
By (P4), 
$(\anode',\adatum_1)\in U_{\adatum_1}$. 
By (CONS), $(\anode,\avariable_i)\in U_{\adatum_1}$. 
By (P1), $\avariable'_i=\adatum_1\in \acons_\anode$. 
By $\Zed \models \acons'_\anode(\atree(\anode))$ we can conclude that 
$\atree(\anode)(\avariable_i)=\adatum_1$. 
Hence $\atree(\anode)(\avariable_i)=\atree(\anode')(\adatum_1)$. 
\item  Suppose $\advar=\adatum_1$ and $\advar'=\avariable_i$ for some $1\leq i\leq \beta$. The proof is symmetric to the proof for the previous case. 
\item Suppose $\advar=\adatum_\alpha$ or $\advar'=\adatum_\alpha$. The proof is very similar to the proof for the previous two cases. 
\end{itemize}

Now suppose $\pair{\anode}{\advar} \step{<}\pair{\anode'}{\advar'}$ is an edge in $\newGt$.
\begin{itemize}
\item Suppose $\advar=\avariable_i$, $\advar'=\avariable_{i'}$ for some $1\leq i,i'\leq \beta$ by (VAR).  The proof is analogous to the proof for the corresponding case for $\pair{\anode}{\advar} \step{=}\pair{\anode'}{\advar'}$. 
\item Suppose $\advar'=\adatum_1$. Then $\advar=\avariable_i$ for some $1\leq i\leq \beta$ and $(\anode,\avariable_i)\in U_{< \adatum_1}$ by (CONS). Then $\avariable'_i<\adatum_1\in\acons_\anode$ by (P2).
By $\Zed \models \acons'_\anode( \atree(\anode))$ we can conclude
$\atree(\anode)(\avariable_i)<\adatum_1$, and hence 
$\atree(\anode)(\avariable_i)<\atree(\anode')(\adatum_1)$. 
\item Suppose $\advar=\adatum_1$. Then -- by (CONS) -- we have $(\anode',\advar')\in U_{\adatum^{\dag}}$ for some
  $\adatum^{\dag} \in \interval{\adatum_1+1}{\adatum_{\alpha}} \cup \set{\mbox{`$> \adatum_{\alpha}$'}}$.
  \begin{itemize}
  \item Suppose $\adatum^{\dag} \in \interval{\adatum_1+1}{\adatum_{\alpha}}$
        and $\advar'=\avariable_i$ for some $i\in\interval{1}{\beta}$. By (P1), we have   $\avariable'_i=\adatum^{\dag} \in \acons_{\anode'}$.
By $\Zed \models \acons'_{\anode'}(\atree(\anode))$ we can
conclude $\atree(\anode)(\avariable_i)=\adatum^{\dag}$. 
Hence $\atree(\anode)(\adatum_1)<\atree(\anode')(\avariable_i)$. 
\item Suppose $\adatum^{\dag}= \adatum_{\alpha}$ and $\advar'=\adatum_\alpha$. Then obviously $\atree(\anode)(\adatum_1)<\atree(\anode')(\adatum_\alpha)$. 
\item Suppose $\adatum^{\dag} =$ `$> \adatum_{\alpha}$'. Then $\advar'=\avariable_i$ for some $i\in\interval{1}{\beta}$, and
  $\avariable'_i>\adatum_\alpha\in \acons_{\anode'}$ by
(P3). 
By $\Zed \models\acons'_{\anode'}(\atree(\anode'))$ we can conclude that 
$\atree(\anode')(\avariable_i)>\adatum_\alpha$, so that 
$\atree(\anode)(\adatum_1)<\atree(\anode')(\avariable_i)$. 
\end{itemize}
\item Suppose $\advar=\adatum_\alpha$. Then $\advar'=\avariable_i$ for some $i\in\interval{1}{\beta}$ and $(\anode',\avariable_i)\in
  U_{> \adatum_\alpha}$. Then $\avariable'_i>\adatum_\alpha\in\acons'_{\anode'}$. By $\Zed \models\acons'_{\anode'}(\atree(\anode'))$,
we can conclude that 
$\atree(\anode')(\avariable_i)>\adatum_\alpha$. 
Hence $\atree(\anode)(\adatum_\alpha)<\atree(\anode')(\avariable_i)$.
\item Suppose $\advar'=\adatum_\alpha$. Then $(\anode,\advar)\in U_{\adatum^{\dag}}$ for some
      $\adatum^{\dag} \in \set{\mbox{`$< \adatum_{1}$'}} \cup \interval{\adatum_1}{\adatum_{\alpha}-1}$.
      The proof is very similar to the last but one case. \qedhere
\end{itemize}
\end{proof}

A map $p: \Nat \rightarrow \domnewGt$ is a \defstyle{path map}
in $\newGt$ 
iff for all $i \in \Nat$, either $p(i) \step{=} p(i+1)$ or $p(i) \step{<} p(i+1)$ in $\newGt$.
Similarly, $\rp: \Nat \rightarrow \domnewGt$ is a \defstyle{reverse path map}
in $\newGt$ 
iff for all $i \in \Nat$, either $\rp(i+1) \step{=} \rp(i)$ or $\rp(i+1) \step{<} \rp(i)$.
These definitions make sense because between two elements there is at most one labelled edge
(or maybe two but with the same equality sign). 
A  path map $p$ (resp. reverse path map $\rp$) is \defstyle{strict} 
iff
$\set{i \in \Nat \mid p(i) \step{<} p(i+1)}$ (resp. $\set{i \in \Nat \mid \rp(i+1) \step{<} \rp(i)}$) is infinite.  
An \defstyle{infinite branch} $\abranch$ is an element of $\interval{0}{\degree-1}^{\omega}$. We write
$\abranch[i,j]$ with $i \leq j$ to denote the subsequence $\abranch(i) \cdot \mathcal{B}(i+1) \cdots\abranch(j)$. 
Given $\pair{\anode}{\advar}\in \domnewGt$, a path map $p$ \defstyle{from $\pair{\anode}{\advar}$ along $\abranch$}
is such that $p(0) = \pair{\anode}{\advar}$ and for all $i \geq 0$,
$p(i)$ is of the form $\pair{\anode \cdot \abranch[0,i]}{\cdot}$.
A reverse path map $\rp$ from $\pair{\anode}{\advar}$ along $\abranch$ admits a similar definition.

We say that a
symbolic tree $\asymtree$ satisfies \textcolor{black}{$(\newbigstar)$} if in $\newGt$ there are {\em no} elements $\pair{\anode}{\advar}$, $\pair{\anode}{\advar'}$ 
  (same node $\anode$ from $\interval{0}{\degree-1}^*$) and
  no infinite branch $\abranch$ such that
  \begin{enumerate}
  \item there exists a path map $p$ from $\pair{\anode}{\advar}$ along $\abranch$,
  \item there exists a reverse path map $\rp$ from $\pair{\anode}{\advar'}$ along $\abranch$,
  \item $p$ or $\rp$ is strict, and 
  \item for all $i \in \Nat$, $p(i) \step{<} \rp(i)$. 
  \end{enumerate}
The following proposition states a key property: 
non-satisfaction of a regular locally consistent symbolic tree
can be witnessed along a \emph{single} branch by violation of $(\newbigstar)$, following the remarkable result established
in~\cite[Lemma 22]{Labai&Ortiz&Simkus20}.
The approach developed in Section~\ref{section-starproperty} takes advantage of the techniques
to prove~\cite[Lemma 22]{Labai&Ortiz&Simkus20}. 
\begin{prop} \label{proposition-star-oplus}
    For every regular locally consistent symbolic tree $\asymtree$,
    $\asymtree$ satisfies $(\newbigstar)$ iff $\asymtree$ is satisfiable.
\end{prop}
A proof can be found in~Section~\ref{section-starproperty}.
A comparison between the condition $(\newbigstar)$ and the condition $(\bigstar)$ from~\cite{Labai&Ortiz&Simkus20,Labai21}
is postponed to Section~\ref{section-rw-comparison}.
We recall 
that there are \emph{nonregular} locally consistent symbolic trees $\asymtree$  such that $\newGt$ satisfies
$(\newbigstar)$ 
but $\asymtree$ is not satisfiable (see e.g. the proof of~\cite[Corollary 6.5]{Demri&DSouza07}, and~\cite{Labai&Ortiz&Simkus20});
indeed, satisfiability of symbolic trees is not an $\omega$-regular property.
Observe also that $\asymtree$ is satisfiable (not necessarily regular) entails that
$\newGt$ satisfies
$(\newbigstar)$ (by Lemma~\ref{lemma-correctness-newgt}). 

\begin{exa} 
\label{example-bigstar}
Consider the non-satisfiable symbolic tree $\asymtree_\atree$ from Example \ref{example-symtree-not-sat}, for which we depict ${G_{\asymtree_\atree}^{\mbox{\tiny{C}}}}$ in Figure \ref{figure-graph}. Note that $\asymtree_\atree$
does not satisfy $(\newbigstar)$: for the infinite branch $\abranch=1^\omega$, 
there exists a path map $p$ from $(\varepsilon,\adatum_1)$ along $\abranch$, 
there exists a strict reverse path map $\rp$ from $(\varepsilon, \avariable_1)$ along $\abranch$, 
and for all $i\in\Nat$ we have $p(i)\step{<}\rp(i)$. 
\end{exa}
The next result states that $(\newbigstar)$ is $\omega$-regular (Lemma~\ref{lemma-automaton-star}):
there is a Rabin tree automaton $\starautomaton$
that captures it, so that satisfiability of symbolic trees can be overapproximated advantageously, see
Lemma~\ref{lemma-intersection-reduction}.
Observe that local consistency could be easily defined in the conditions for  $\newGt$.
One reason for keeping
separate $\locautomaton$ dealing with local consistency  and the automaton for dealing with $\newGt$
(forthcoming $\starautomaton$)
is that for
certain concrete domains (including $\pair{\Rat}{<}$), local consistency is sufficient to guarantee the existence of a satisfiable
model, so that the construction of $\starautomaton$ is not even necessary. We hence prefer this modular approach.

\begin{lem}
\label{lemma-automaton-star}
There is a Rabin tree automaton $\starautomaton$ such that
$\alang(\starautomaton) = \{ \asymtree \mid \asymtree \text{ sat. } (\newbigstar)\}$, the number of Rabin pairs is bounded above  by $8(\beta+2)^2 + 3$, the number of locations is exponential in $\beta$, 
 the transition relation can be decided in polynomial-time in
\[
 \max (\lceil log(|\adatum_1|) \rceil ,
    \lceil log(|\adatum_{\alpha}|) \rceil) + \beta + \card{\aalphabet} + D .
\]
\end{lem}
As a consequence, the transition relation for $\starautomaton$
has  $\card{\aalphabet \times \sattypes{\beta}} \times (2^{\mathcal{O}(P^{\star}(\beta))})^{(\degree+1)}$
transitions for some polynomial $P^{\star}$,
which is in $\card{\aalphabet} \times
\boundsattypes \times 2^{\mathcal{O}(P^{\star}(\beta) \times (\degree+1))}$.

\begin{proof} 
The proof of Lemma~\ref{lemma-automaton-star} is structured as follows.
\begin{enumerate}
\item  We construct a B\"uchi word automaton  $\ancautomaton_B$ accepting the complement of $(\newbigstar)$ for $D=1$. 
\item 
We determinize $\ancautomaton_B$ and obtain a deterministic Rabin word automaton $\ancautomaton_{B \to R}$ 
such that $\alang(\ancautomaton_B) = \alang(\ancautomaton_{B \to R})$ (using the classical determinisation construction from~\cite[Theorem 1.1]{Safra89}). 
\item From $\ancautomaton_{B \to R}$, we obtain a deterministic Street word automaton $\ancautomaton_S$ accepting the complement of
$\alang(\ancautomaton_{B \to R})$; it accepts words that satisfy $(\newbigstar)$ for $D=1$.
Here, we use the well-known fact that the negation of a Rabin acceptance condition is a Street acceptance condition.
\item From $\ancautomaton_S$, we construct a deterministic Rabin word automaton $\ancautomaton_{R}$ such that $\alang(\ancautomaton_S) = \alang(\ancautomaton_{R})$ (using~\cite[Lemma 1.2]{Safra89}).
  Observe that both $\ancautomaton_{B \to R}$ and $\ancautomaton_{R}$ are Rabin word automata
  but
  $\ancautomaton_{B \to R}$ handles the complement
  of the condition $(\newbigstar)$ whereas $\ancautomaton_{R}$ handles $(\newbigstar)$ itself. 
\item We construct a Rabin tree automaton $\starautomaton$, by "letting run  $\ancautomaton_{R}$" along
every branch of a run of $\starautomaton$, which is  possible thanks to the determinism of $\ancautomaton_{R}$. Since $(\newbigstar)$  states
a property on every branch of the trees, we are done.
\end{enumerate}
So let us explain each of these steps. 

\noindent(1) We construct  a B\"uchi word automaton accepting the complement of $(\newbigstar)$ for $D=1$, similarly to what is done in~\cite[Section 6]{Demri&DSouza07} and~\cite[Section 3.4]{Labai&Ortiz&Simkus20}.  
\begin{figure}
\scalebox{1}{
  \begin{tikzpicture}[->,>=stealth',shorten >=1pt,auto,node distance=4cm,thick,node/.style={circle,draw,scale=0.9}, roundnode/.style={circle, black, draw=black},]
\tikzset{every state/.style={minimum size=0pt},
dotted_block/.style={draw=black!20!white, line width=1pt, dotted, inner sep=3mm, minimum width=5.4cm, rectangle, rounded corners, minimum height=8mm},
boxnode/.style={rectangle, draw=black,minimum width=10mm,rounded corners=0.2cm}
} 
;
\node[boxnode] (10) at (0,0) {$\pair{0}{\adatum_1}$};
\node[boxnode] (x0) at (0,1.5) {$\pair{0}{\avariable_1}$}; 
\node[boxnode] (11) at (2,0) {$\pair{1}{\adatum_1}$};
\node[boxnode] (x1) at (2,1.5) {$\pair{1}{\avariable_1}$}; 
\node[boxnode] (12) at (4,0) {$\pair{2}{\adatum_1}$};
\node[boxnode] (x2) at (4,1.5) {$\pair{2}{\avariable_1}$}; 
\node[boxnode] (13) at (6,0) {$\pair{3}{\adatum_1}$};
\node[boxnode] (x3) at (6,1.5) {$\pair{3}{\avariable_1}$};
\node[boxnode] (14) at (8,0) {$\pair{4}{\adatum_1}$};
\node[boxnode] (x4) at (8,1.5) {$\pair{4}{\avariable_1}$};

\node at (9,0) {$\dots$}; 
\node at (9,1.5) {$\dots$}; 

\path [-] (10)  edge node[above] {\scriptsize{$=$}}  (11);
\path [-] (11)  edge node[above] {\scriptsize{$=$}}  (12);
\path [-] (12)  edge node[above] {\scriptsize{$=$}}  (13);
\path [-] (13)  edge node[above] {\scriptsize{$=$}}  (14);

\path [<-] (x0)  edge node[above] {\scriptsize{$>$}}  (x1);
\path [-] (x1)  edge node[above] {\scriptsize{$=$}}  (x2);
\path [<-] (x2)  edge node[above] {\scriptsize{$>$}}  (x3);
\path [<-] (x3)  edge node[above] {\scriptsize{$>$}}  (x4);

\path [->] (10)  edge node[right] {\scriptsize{$<$}}  (x0);
\path [->] (11)  edge node[right] {\scriptsize{$<$}}  (x1);
\path [->] (12)  edge node[right] {\scriptsize{$<$}}  (x2);
\path [->] (13)  edge node[right] {\scriptsize{$<$}}  (x3);
\path [->] (14)  edge node[right] {\scriptsize{$<$}}  (x4);

\node at (1,-.7) {$\underbrace{\hphantom{hahahahahahaha}}_{\Theta}$};  
\node at (5,-.7) {$\underbrace{\hphantom{hahahahahahaha}}_{\Theta}$}; 
\node at (3,-1) {$\underbrace{\hphantom{hahahahahahaha}}_{\Theta'}$};  
\node at (7,-1) {$\underbrace{\hphantom{hahahahahahaha}}_{\Theta}$}; 
\node at (4.5,-1.8) {$\alocation_\init \xrightarrow{\acons}(\adatum_1,\avariable_1,\rp,<)\xrightarrow{\acons'}(\adatum_1,\avariable_1,\rp,=)\xrightarrow{\acons}(\adatum_1,\avariable_1,\rp,<)\xrightarrow{\acons}(\adatum_1,\avariable_1,\rp,<)\dots$}; 
 	\end{tikzpicture} 
}
\caption{A symbolic word $\aword$ representing an infinite branch along which a path map from $(0,\adatum_1)$ and a strict reverse path map from $(0,\avariable_1)$  do not satisfy the condition $(\newbigstar)$; below, a run of $\ancautomaton_B$ on $\aword$.} 
\label{figure_A_B}
\end{figure}
Let us first explain the idea of the construction. 
If a word $\aword$ over $\aalphabet \times \sattypes{\beta}$ does not satisfy $(\newbigstar)$, 
then the graph $G_{\aword}^{\mbox{\tiny{C}}}$ contains a node $\anode$ such that, for some $\advar_1$ and $\advar_2$, there exists a path map from $(\anode, \advar_1)$ and  there exists a reverse path map from $(\anode,\advar_2)$ that do not satisfy the condition $(\newbigstar)$, 
cf. Figure \ref{figure_A_B}, where there exists a path map from $(0,\adatum_1)$ and a reverse path map from $(0,\avariable_1)$.   
The automaton nondeterministically guesses such nodes, and checks whether it initializes a violation of $(\newbigstar)$. For this, the automaton 
remembers in its locations 
\begin{itemize}
\item the current variable/constant $\advar_1$ of the path map $p$,
\item the current variable/constant $\advar_2$ of the reverse path map $\rp$,
\item whether the path map $p$ or the reverse path map $\rp$ is strict ($p$ or $\rp$), and
\item whether in the strict (reverse) path map, it has just seen $<$ or $=$. 
\end{itemize}
We call a $4$-tuple $(\advar_1,\advar_2, d, \mathop{\bowtie})$ a \defstyle{local thread}, 
if $\advar_1,\advar_2\in  \DVAR{\beta}{\adatum_1}{\adatum_{\alpha}}$, 
$d\in\{p,\rp\}$ and $\mathop{\bowtie}\in \set{<,=}$, with the intended meaning as explained above. 
The transition relation of the automaton is 
   defined to guarantee that a sequence of local threads forms an infinite branch along which a path map and a reverse path map do not
   satisfy the condition  $(\newbigstar)$. 
   Formally, let  $\ancautomaton_{B} =
   \triple{\locations_{B},\aalphabet \times \sattypes{\beta}}{\locations_{B,\init}}{\delta_{B},F_B}$ 
   be the B\"uchi word automaton defined as follows.   
   \begin{itemize}
   \item $\locations_B \egdef \{\alocation_\init\} \cup \, ( \DVAR{\beta}{\adatum_1}{\adatum_{\alpha}}^2 \times \{p,\rp\}\times \set{<,=})$;
     $\locations_{B,\init} = \{\alocation_\init\}$.
   \item $\delta_B$ is the union of the following sets. 
  \begin{itemize}
  \item  $\{\triple{\alocation_\init}{\pair{\aletter}{\acons}}{\alocation_{\init}} \mid \pair{\aletter}{\acons}\in\aalphabet \times \sattypes{\beta}\}$.
  \item  $\{\triple{\alocation_\init}{\pair{\aletter}{\acons}}{(\advar_1,\advar_2,d,\mathop{\bowtie})}\mid \pair{\aletter}{\acons}\in\aalphabet \times \sattypes{\beta} \text{ such that }
    \acons\models \advar_1' < \advar_2', d\in\{p,\rp\},  \mathop{\bowtie}\in\{<,=\}\}$ (initialization of a violating thread).

     \item $\{ \triple{(\advar_1,\advar_2,p,\mathop{\bowtie})}{\pair{\aletter}{\acons}}{(\advar_3,\advar_4,p,\mathop{\bowtie'})} \mid
       \pair{\aletter}{\acons}\in\aalphabet \times \sattypes{\beta}$ such that
       $\acons\models
     (\advar'_3 < \advar'_4) \, \wedge \, (\advar_1 \ \mathop{\bowtie'} \ \advar'_3) \,
     \wedge \, ((\advar'_4 = \advar_2) \vee (\advar'_4<\advar_2)), \mathop{\bowtie}, \mathop{\bowtie'} \in\{<,=\}\}$.

      \item $\{ \triple{(\advar_1,\advar_2,\rp,\mathop{\bowtie})}{\pair{\aletter}{\acons}}{(\advar_3,\advar_4,\rp,\mathop{\bowtie'})} \mid
      \pair{\aletter}{\acons}\in\aalphabet \times \sattypes{\beta}$ such that $\acons\models
      (\advar'_3 < \advar'_4) \, \wedge \, (\advar'_4 \ \mathop{\bowtie'} \ \advar_2) \,  \wedge \, ((\advar_1 = \advar'_3) \vee
      (\advar_1<\advar'_3)), \mathop{\bowtie}, \mathop{\bowtie'} \in\{<,=\}\}$.
      
  \end{itemize}
  \item $F_B \egdef \set{(\advar_1,\advar_2, d, <) \in \locations_B \mid \advar_1,\advar_2\in \DVAR{\beta}{\adatum_1}{\adatum_{\alpha}}, d\in\{p,\rp\}}$.
      \end{itemize}
It is not hard to prove that 
      for all
   $\aword: \Nat \to \aalphabet \times \sattypes{\beta}$,
   $\aword \in \alang(\ancautomaton_{B})$ iff $\aword$ does not satisfy the condition $(\newbigstar)$. 
   Note that the number of locations in $\ancautomaton_B$ is bounded above by $4(\beta+2)^2+1$.  

   \ 

\noindent(2) Using~\cite[Theorem 1.1]{Safra89}, from $\ancautomaton_B$ 
we obtain a deterministic Rabin word automaton 
$\ancautomaton_{B \to R}=\triple{\locations_{B \to R},\aalphabet \times \sattypes{\beta}}{\locations_{B \to R, \init}}{\delta_{B \to R},\rabinacc_{B \to R}}$ such that 
$\alang(\ancautomaton_{B \to R})=\alang(\ancautomaton_B)$. 
The cardinality of $\locations_{B \to R}$ is in  $2^{\mathcal{O}(\card{\locations_B} \ log (\card{\locations_B}))}$, i.e. exponential in $\beta$, 
and the number of acceptance pairs in $\rabinacc_{B \to R}$ is equal to $2 \times \card{\locations_B}$, i.e. equal to
$8(\beta+2)^2+2$ (see also Section~\ref{section-ctlstarz-determinisation-safra} where
the proof of~\cite[Theorem 1.1]{Safra89} is generalised to B\"uchi word constraint automata). 
Without any loss of generality, we can assume that
$\ancautomaton_{B \to R}$ is also complete. 
As
a consequence, for any $\aword \in (\aalphabet \times \sattypes{\beta})^{\omega}$,
there is a unique (not necessarily
accepting) run $\arun_{\aword}$ on $\aword$.

\ 

\noindent(3) 
We write $\ancautomaton_S=\triple{\locations_{S},\aalphabet \times \sattypes{\beta}}{\locations_{S,\init}}{\delta_{S},\rabinacc_S}$ to denote the Streett automaton accepting the complement language
of $\alang(\ancautomaton_{B \to R})$. All the components of $\ancautomaton_{S}$ are  those
from $\ancautomaton_{B \to R}$  but $\rabinacc_S$ in $\ancautomaton_{S}$ is interpreted as a Streett condition
(recall that the negation of a Rabin condition is a Streett condition). 
Consequently, the deterministic Streett word automaton $\ancautomaton_S$ is syntactically equal to
$\ancautomaton_{B \to R}$ and  we have
$\aword\in \alang(\ancautomaton_S)$ iff $\aword$ satisfies the condition $(\newbigstar)$.

\

\noindent(4) Using~\cite[Lemma 1.2]{Safra89}, from $\ancautomaton_{S}$ 
we define the \emph{deterministic} Rabin word automaton $\ancautomaton_{R} =
   \triple{\locations_{R},\aalphabet \times \sattypes{\beta}}{\locations_{R,\init}}{\delta_{R},\rabinacc_R}$ such that
   $\alang(\ancautomaton_{R}) = \alang(\ancautomaton_{S})$. 
   Before defining its components, we define a few notions.
   First of all, we assume that $\ancautomaton_{S}$ has an additional $(8(\beta+2)^2+3)$$^{\rm th}$ Streett pair, namely
   $\pair{\locations_{S}}{\locations_{S}}$, which is technically helpful to design $\ancautomaton_{R}$
   but does not change the language $\alang(\ancautomaton_{S})$.  
   We write $\permut{(8(\beta+2)^2+3)}$ to denote the set of permutations on $\interval{1}{8(\beta+2)^2+3}$, where 
   a permutation $\apermutation$ is a bijection $\apermutation: \interval{1}{8(\beta+2)^2+3} \to
   \interval{1}{8(\beta+2)^2+3}$. 
   It is well-known that
   $\card{\permut{(8(\beta+2)^2+3)}} = (8(\beta+2)^2+3)!$. 
   Let us define the maps
   $\amapbis_1, \amapbis_2: \locations_{S} \times \permut{(8(\beta+2)^2+3)}  \to \interval{1}{8(\beta+2)^2+3}$ and
   $\amapbis_3: \locations_{S} \times \permut{(8(\beta+2)^2+3)}  \to \permut{(8(\beta+2)^2+3)}$
   that are instrumental to define the forthcoming transition relation $\delta_R$.
   Below, $\alocation \in \locations_S$ and $\apermutation \in \permut{(8(\beta+2)^2+3)}$.
   \begin{itemize}

   \item We set 
     $
     \amapbis_1(\alocation, \apermutation) \egdef \min \set{i \in \interval{1}{8(\beta+2)^2+3} \mid \alocation \in
       U_{\apermutation(i)}}
     $ where
     $U_{\apermutation(i)}$ is from the set of pairs $\rabinacc_S$. A minimal value always exists thanks to the addition of
     the new Streett pair $\pair{\locations_S}{\locations_S}$.

   \item The value $\amapbis_2(\alocation, \apermutation)$ is defined 
     using the $L_i$ sets from the set of pairs $\rabinacc_S$:
\[
     \amapbis_2(\alocation, \apermutation) \egdef \min \set{i \in \interval{1}{8(\beta+2)^2+3} \mid \alocation \in
       L_{\apermutation(i)}}.
\]

   \item 
     The permutation $\amapbis_3(\alocation, \apermutation)$
     is obtained from $\apermutation$
     by \emph{moving} $\apermutation(\amapbis_1(\alocation, \apermutation))$ to the
     rightmost position, so that $\apermutation(\amapbis_1(\alocation, \apermutation))$ is at position $8(\beta+2)^2+3$ now.
     Formally
     \[
     \amapbis_3(\alocation, \apermutation)(i) \egdef \begin{cases}
\apermutation(i) & \text{if } 1\leq i < \amapbis_1(\alocation, \apermutation) \\
\apermutation(i+1) & \text{if } \amapbis_1(\alocation, \apermutation)\leq i\leq 8(\beta+2)^2+2 \\
\apermutation(\amapbis_1(\alocation, \apermutation)) & \text{if } i=8(\beta+2)^2 +3
     \end{cases}
     \]
   \end{itemize}
   By way of example, if $\apermutation(1) = 2$, $\apermutation(2) = 4$, $\apermutation(3) = 1$, $\apermutation(4) = 3$,
   $\cdots$
   and, $3$ is the minimal value $i$ such that $\alocation \in U_{\apermutation(i)}$ (so $\alocation \not \in U_2 \cup U_4$),
   then $\amapbis_3(\alocation, \apermutation)(1) = 2$, $\amapbis_3(\alocation, \apermutation)(2) = 4$,
   $\amapbis_3(\alocation, \apermutation)(3) = 3$ $\cdots$ and
   $\amapbis_3(\alocation, \apermutation)(8(\beta+2)^2+3) = 3$.
 
   Let us now define the components of $\ancautomaton_{R}$. 

 \begin{itemize}

 \item $\locations_{R} \egdef \locations_{S}  \times \permut{(8(\beta+2)^2+3)} \times \interval{1}{8(\beta+2)^2+3}^2$.
  
 \item The transition relation $\delta_R$ is defined as follows:
   $\triple{\alocation'}{\apermutation'}{e',f'} \in \delta_R(\triple{\alocation}{\apermutation}{e,f},
   \aletter)$ iff $\alocation' \in \delta_S(\alocation, \aletter)$, 
   $\apermutation' = \amapbis_3(\alocation',\apermutation)$, $e' = \amapbis_1(\alocation', \apermutation)$
   and $f' = \amapbis_2(\alocation', \apermutation)$.
   Since $\amapbis_1(\alocation', \apermutation)$, $\amapbis_2(\alocation', \apermutation)$ and
   $\amapbis_3(\alocation',\apermutation)$ can be computed in polynomial-time in $\beta$
   and $\rabinacc_S$ has $8 (\beta+2)^2+3$ pairs,
   we get that $\delta_R$ can be decided in polynomial-time in
   $\max (\lceil log(|\adatum_1|) \rceil ,
   \lceil log(|\adatum_{\alpha}|) \rceil) + \beta + \card{\aalphabet}$.
   Indeed, this amounts to determine the complexity of deciding $\delta_S$.
   However, this question amounts to determine the complexity of
   deciding $\delta_B$ and $\delta_{B \to R}$. Both relations can be decided in
   in polynomial-time in $\max (\lceil log(|\adatum_1|) \rceil ,
   \lceil log(|\adatum_{\alpha}|) \rceil) + \beta + \card{\aalphabet}$.
   
 \item $\locations_{R,\init}$ has a unique initial location
   $
   \triple{\alocation_\init}{id}{8(\beta+2)^2+3,8(\beta+2)^2+3}
   $, with $\alocation_\init$ being the only initial location in $\ancautomaton_S$ and $id$ being
   the identity permutation. 
   Since $\delta_S$ is a function and $\locations_{R,\init}$ is a singleton,
   we can conclude that $\ancautomaton_R$ is deterministic too. 
 \item $\rabinacc_R$ is made of Rabin pairs $\pair{L'_i}{U'_i}$ with $i \in \interval{1}{8(\beta+2)^2+3}$ such that
   \[
   L'_i \egdef \set{\triple{\alocation}{\apermutation}{e,f} \in \locations_R \mid
     e = i} \ \ \ \mbox{and} \ \ \ 
   U'_i \egdef \set{\triple{\alocation}{\apermutation}{e,f} \in \locations_R \mid
     f <  i}.
   \]
 \end{itemize}
We recall that the Rabin condition $\rabinacc_R$ can be read as follows:
     along any accepting run,  there is $i \in \interval{1}{8(\beta+2)^2+3}$ such that 
     some location in $L_i'$ occurs infinitely often and all the locations in $U_i'$ occurs finitely.

 By~\cite[Lemma 1.2]{Safra89}, we have $\alang(\ancautomaton_{R}) = \alang(\ancautomaton_{S})$.
 By way of example, let us briefly explain why $\alang(\ancautomaton_S)
 \subseteq \alang(\ancautomaton_R)$.
 Given an accepting run $\arun$ of $\ancautomaton_S$, we define $\aset_{\arun} \subseteq
 \interval{1}{8(\beta+2)^2+3}$ such that $k \in \aset_{\arun}$ iff
 the set $U_{\apermutation(k)}$ is visited infinitely often.
 Hence, for all $k \in (\interval{1}{8(\beta+2)^2+3} \setminus \aset_{\arun})$,
 $U_{\apermutation(k)}$ and $L_{\apermutation(k)}$ are visited finitely along
 $\arun$. Moreover, from some position $I \in \Nat$ in $\arun$, elements of
 $(\interval{1}{8(\beta+2)^2+3} \setminus \aset_{\arun})$
 always occupy the leftmost position in the permutation $\apermutation$
 and none of its values change its place.
 Let $j = 8(\beta+2)^2+3 - \card{\aset_{\arun}}$. From the position $I$,
 we have $f \geq j$ and one can show that $e = j$ infinitely
 often along $\arun$. Since the run visits $U_{\apermutation(j)}$
 infinitely often, whenever $U_{\apermutation(j)}$ is visited, $e$
 shall take the value $j$ (minimal value among $\aset_{\arun}$).
 Hence, the run $\arun'$ of $\ancautomaton_R$ obtained
 from $\arun$ by completing deterministically the three last components
 satisfies the Rabin pair $\pair{L'_j}{U'_j}$. Consequently,
 the word accepted by $\arun$ is also accepted by $\arun'$.
 
 \
 
\noindent(5) Finally, the last (easy) stage consists in building the Rabin tree automaton
   $$\starautomaton \egdef
\triple{\locations_{\newbigstar},\aalphabet \times \sattypes{\beta},\degree}{\locations_{\newbigstar,\init}}{\delta_{\newbigstar},
  \rabinacc_{\newbigstar}}
$$ as follows.
\begin{itemize}
\item $\locations_{\newbigstar} \egdef \locations_R$, $\locations_{\newbigstar, \init} \egdef  \locations_{R,\init}$ and
      $\rabinacc_{\newbigstar} \egdef \rabinacc_R$. 
\item For all  $\alocation, \alocation_0, \ldots, \alocation_{\degree-1} \in \locations_{\newbigstar}$ and
  $\pair{\aletter}{\acons} \in \aalphabet \times \sattypes{\beta}$, we have 
  \[
  \triple{\alocation}{\pair{\aletter}{\acons}}{ \alocation_0, \ldots, \alocation_{\degree-1}} \in \delta_{\newbigstar}
  \ \equivdef \ 
  \alocation_0 = \cdots = \alocation_{\degree-1} \  \mbox{and}\  \alocation_0 \in
  \delta_R(\alocation,\pair{\aletter}{\acons}).
  \]
  Since $\ancautomaton_R$ is deterministic, $\delta_R(\alocation,\pair{\aletter}{\acons})$
  contains at most one location. 
\end{itemize}
We have $\asymtree \in \alang(\starautomaton)$ iff
all the branches of $\asymtree$ are in $\alang(\ancautomaton_{R})$, which means precisely
that $\asymtree$ satisfies the condition $(\newbigstar)$. 
Further, $\starautomaton$ satisfies all the size conditions in Lemma~\ref{lemma-automaton-star}.
Indeed,
\begin{itemize}
\item $\rabinacc_{\newbigstar}$ has exactly $8(\beta+2)^2+3$ Rabin pairs,
\item the number of locations in $\locations_{\newbigstar}$ is in
  \[
  2^{\mathcal{O}((4(\beta+2)^2+1) \cdot log (4(\beta+2)^{2}+1))} \times  (8(\beta+2)^2+3)! \times (8(\beta+2)^2+3)^2,
  \]
    i.e. exponential in $\beta$,
  \item as explained earlier, $\delta_S$ can be decided in polynomial-time in
    $\max (\lceil log(|\adatum_1|) \rceil ,
    \lceil log(|\adatum_{\alpha}|) \rceil) + \beta$ and therefore $\delta_R$  can be decided in polynomial-time in
\[
 \max (\lceil log(|\adatum_1|) \rceil ,
    \lceil log(|\adatum_{\alpha}|) \rceil) + \beta + \card{\aalphabet} + D. \qedhere
\] 
\end{itemize}
\end{proof}

Summarizing the developments so far,  we can conclude this subsection as follows:  
\begin{lem}
\label{lemma-intersection-reduction}
$\alang(\aautomaton)\neq\emptyset$ iff $\alang(\locautomaton)  \cap \alang(\starautomaton) \neq \emptyset$.
\end{lem}
\begin{proof}
One can show that the statements below are equivalent.
\begin{description}
\item[(I)] $\alang(\aautomaton) \neq \emptyset$.
\item[(II)] There is a symbolic
tree $\asymtree: \interval{0}{\degree-1}^* \rightarrow \aalphabet \times \sattypes{\beta}$
in $\alang(\locautomaton)$
  that is satisfiable.
\item[(III)]
      There is a symbolic tree
      $\asymtree: \interval{0}{\degree-1}^* \rightarrow \aalphabet \times \sattypes{\beta}$
      in $\alang(\locautomaton)  \cap \alang(\starautomaton)$.
\end{description}
The equivalence between (I) and (II) is by Lemma~\ref{lemma-satisfiable-symbolic-tree}.
That condition (II) implies (III) follows from the fact that for every satisfiable
 symbolic tree $\asymtree$ (not necessarily regular), $\asymtree$ satisfies the condition
$(\newbigstar)$ (Lemma~\ref{lemma-correctness-newgt}),
$\alang(\locautomaton)$
contains all the satisfiable symbolic  trees
(Lemma~\ref{lemma-A-consistency} and Lemma~\ref{lemma-consistency-BTA})
and $\alang(\starautomaton)$ is equal to the set of 
symbolic trees satisfying the condition $(\newbigstar)$.
Hence, $\asymtree \in \alang(\locautomaton)  \cap \alang(\starautomaton)$. 
That condition (III) implies (II) follows from the fact that if
$\alang(\locautomaton)  \cap \alang(\starautomaton)$
is non-empty, then $\alang(\locautomaton)  \cap \alang(\starautomaton)$
is regular and therefore contains
a regular $\aautomaton$-consistent symbolic tree $\asymtree$ (see e.g.~\cite{Rabin69}
and~\cite[Section 6.3]{Thomas90} for the existence of regular trees)
and by Proposition~\ref{proposition-star-oplus},  $\asymtree$ is satisfiable. 
By Lemma~\ref{lemma-satisfiable-symbolic-tree},
since $\asymtree$ is satisfiable and
$\asymtree \in \alang(\locautomaton)$, we get $\alang(\aautomaton)\neq\emptyset$.
\end{proof} 

\subsection{\exptime Upper Bound for TCAs}
\label{section-exptime-upper-bound}
Lemma~\ref{lemma-intersection-reduction} justifies why 
deciding the nonemptiness of $\alang(\locautomaton)  \cap \alang(\starautomaton)$ is crucial.
Fortunately, regular tree languages are closed under intersection. However, assuming that a  Rabin tree automaton
$\ancautomaton$ satisfies $\alang(\ancautomaton) = \alang(\locautomaton) \cap \alang(\starautomaton)$, we need to guarantee
that the construction of $\ancautomaton$ does not lead to any complexity blow up. This is the purpose of
Lemma~\ref{lemma-intersection-automaton} below.
An exponential blow-up may have drastic consequences on the forthcoming complexity analysis.
In the proof of Lemma~\ref{lemma-intersection-automaton}, we propose a construction 
that is not polynomial but it only performs an exponential blow-up on the number of locations, which shall be
fine for our purpose.

\begin{lem}
\label{lemma-intersection-automaton}
There is a Rabin tree automaton $\ancautomaton$ such that
\begin{description}
\item[(I)] $\alang(\ancautomaton) = \alang(\locautomaton) \cap \alang(\starautomaton)$,
\item[(II)] the number of Rabin pairs is polynomial in $\beta$, 
\item[(III)] the number of locations is in
  $
  (\adatum_{\alpha} - \adatum_1)^{2\beta}  \times \card{\locations} \times 2^{\mathcal{O}(P(\beta))}
  $ for some polynomial $P$,
\item[(IV)] the cardinality of the transition relation is in 
$((\adatum_{\alpha} - \adatum_1)^{2\beta} \times \card{\locations} \times 2^{\mathcal{O}(P'(\beta))})^{\degree +2}
  \times \card{\aalphabet}$ for some polynomial $P'$,
 \item[(V)] the transition relation can be decided in polynomial-time in 
    $\card{\delta}
    + \beta + \card{\aalphabet} + D + \maxconstraintsize{\aautomaton}
    $.
\end{description}
\end{lem} 
\begin{proof} The proof is divided into two parts. In Part (I), we present a construction for
  the intersection of Rabin {\em tree} automata mainly based on ideas from the proof of~\cite[Theorem 1]{Boker18} on Rabin {\em word} automata
  but for trees (the developments in the proof of~\cite[Lemma 3.13]{Labai21} are not satisfactory, to our opinion).
  The construction is a particular case of the one
  for the proof of Lemma~\ref{lemma-intersection-rtca}, but here 
  with only two input Rabin tree automata and no constraints. 
In Part (II), we apply the general construction from Part (I) to $\locautomaton$ and $\starautomaton$ and perform
a quantitative analysis.

(I) For $i=1,2$, let $\ancautomaton_i =
\triple{\locations_i,\aalphabet,\degree}{\locations_{i,\init},\delta_i}{\rabinacc_i}$
with $\rabinacc_i = (L_i^j,U_i^j)_{j \in \interval{1}{N_i}}$, that is $N_i$ Rabin pairs, be a Rabin tree automaton.
Let us build a Rabin tree automaton
$\ancautomaton = \triple{\locations,\aalphabet,\degree}{\locations_{\init},\delta}{\rabinacc}$
such that $\alang(\ancautomaton) = \alang(\ancautomaton_1) \cap \alang(\ancautomaton_2)$.

\begin{itemize}

\item $\locations \egdef \locations_1 \times \locations_2 \times \interval{0}{3}^{\interval{1}{N_1} \times \interval{1}{N_2}}$.
  The elements in $\locations$ are of the form $\triple{\alocation_1}{\alocation_2}{\amap}$ with
  $\amap: \interval{1}{N_1} \times \interval{1}{N_2} \to \interval{0}{3}$. 

\item The tuple $\triple{\triple{\alocation_1}{\alocation_2}{\amap}}{\aletter}{
  \triple{\alocation_1^{0}}{\alocation_2^{0}}{\amap^{0}}, \ldots,\triple{\alocation_1^{\degree-1}}{\alocation_2^{\degree-1}}{\amap^{\degree-1}}}$
  belongs to $\delta$ iff the conditions below hold.
  \begin{enumerate}

   \item For $i=1,2$, we have $\triple{\alocation_i}{\aletter}{\alocation_i^{0}, \ldots,\alocation_i^{\degree-1}} \in \delta_i$.
    The two first components in elements from $\locations$ behave as in $\ancautomaton_1$ and $\ancautomaton_2$, respectively.

  \item For all $\pair{i}{j} \in \interval{1}{N_1} \times \interval{1}{N_2}$, the following conditions hold.

    \begin{enumerate}

    \item If $\amap(i,j)$ is odd, then for all $k \in \interval{0}{\degree-1}$, we have
      $\amap^{k}(i,j) = (\amap(i,j)+1)\!\mod 4$. Odd values in $\interval{0}{3}$ are unstable and are replaced
      at the next step by the successor value (modulo $4$).

    \item For all $k \in \interval{0}{\degree-1}$, if $\amap(i,j) = 0$ and $\alocation_1^k \in
      L_1^i$, then
      $\amap^{k}(i,j) = 1$. Hence, when the
      $\pair{i}{j}$$^{\rm th}$ component of $\amap$ is equal to $0$, it waits
      to visit a state in the  set $L_1^i$ to move to 1.

    \item For all $k \in \interval{0}{\degree-1}$, if $\amap(i,j) = 0$ and
      $\alocation_1^k \not \in L_1^i$, then
      $\amap^{k}(i,j) = 0$ (not yet the right moment to modify the
      $\pair{i}{j}$$^{\rm th}$ component).  

     \item For all $k \in \interval{0}{\degree-1}$, if $\amap(i,j) = 2$ and $\alocation_2^k \in L_2^j$, then
      $\amap^{k}(i,j) = 3$. Hence, when the $\pair{i}{j}$$^{\rm th}$ component of $\amap$ is equal to $2$, it waits
      to visit a state in the set $L_2^j$ to move to 3.

    \item For all $k \in \interval{0}{\degree-1}$, if $\amap(i,j) = 2$ and $\alocation_2^k \not \in L_2^j$, then
      $\amap^{k}(i,j) = 2$
      (not yet the right moment to modify the  $\pair{i}{j}$$^{\rm th}$ component).  

    \end{enumerate} 
    \end{enumerate}

    At this stage, it is worth noting that each $\amap^k$ for $k \in \interval{0}{\degree-1}$ takes a unique value, i.e.
    the update of the third component in $\locations$ is done deterministically.

    The transition relation $\delta$ can be decided in the sum of the time-complexity to decide
    $\delta_1$ and $\delta_2$ respectively, plus polynomial-time in $N_1 \times N_2$. 

  \item $\locations_{\init} \egdef \locations_{1,\init} \times \locations_{2,\init} \times \set{\amap_0}$, where $\amap_0$ is
    the unique map that takes always the value zero.

  \item The set of Rabin pairs in $\rabinacc$ contains exactly the pairs $\pair{L}{U}$ for which there is
    $\pair{i}{j} \in \interval{1}{N_1} \times \interval{1}{N_2}$ such that
    \[
    U \egdef \big(U^i_1 \times \locations_2 \cup \locations_1 \times U^j_2\big)  \times \interval{0}{3}^{\interval{1}{N_1} \times \interval{1}{N_2}} \ \ \ \ 
    L \egdef \locations_1 \times \locations_2 \times \set{\amap \mid \amap(i,j) = 1}
    \]
    Because the odd values are unstable, if a location in $L$ is visited infinitely along a branch of a run for $\ancautomaton$
    (and therefore a location in $L^i_1$ is visited infinitely often on the first component), then
    a location in $L^j_2$ is also visited infinitely often on the second component. Indeed, to revisit the value 1 on the $\pair{i}{j}$$^{\rm th}$ component
    one needs to visit first the value $3$, which witnesses that a location in $L^j_2$ has been found.

    If along a branch of a run for $\ancautomaton$ the triples in $U$ are visited finitely, then
    a location in $U^i_1$ is visited finitely on the first component and a location in $U^j_2$ is visited
    finitely on the second component.

    Consequently, $\rabinacc$ contains at most $N_1 \times N_2$ pairs.

\end{itemize}

\noindent 
We claim that $\alang(\ancautomaton) = \alang(\ancautomaton_1) \cap \alang(\ancautomaton_2)$. We omit the proof here
as the  proof of Lemma~\ref{lemma-intersection-rtca} generalises it.

\noindent 
(II) Let us analyse the size of the components in $\locautomaton$ and $\starautomaton$,
which provides bounds for $\ancautomaton$ such that $\alang(\ancautomaton) = \alang(\starautomaton) \cap \alang(\locautomaton)$
following the above construction.
The automaton $\locautomaton$ can be viewed
as a Rabin tree automaton with a single pair, typically $\pair{F'}{\emptyset}$,  where $F'$ is the
set of accepting states of the B\"uchi tree automaton $\locautomaton$. 

\begin{itemize}
\item $\locautomaton$ has a single Rabin pair,
  $\starautomaton$ has a number of Rabin pairs bounded by $8(\beta+2)^2+3$, so
  $\ancautomaton$ has a number of Rabin pairs bounded  by $8(\beta+2)^2+3$.
\item The locations in $\locautomaton$ are from $\sattypes{\beta} \times \locations$,
  the number of locations in $\starautomaton$ is in $\mathcal{O}(2^{P^{\dag}(\beta)})$
  for some polynomial $P^{\dag}$ 
  (see Lemma~\ref{lemma-automaton-star}).
  Therefore, based on the above construction for intersection, the number of locations in $\ancautomaton$ is in
  \[
  \card{\sattypes{\beta} \times \locations} \times 2^{\mathcal{O}(P^{\dag}(\beta))} \times 4^{8(\beta+2)^2+3},
  \]
  which is in $(\adatum_{\alpha} - \adatum_1)^{2\beta}  \times \card{\locations} \times 2^{\mathcal{O}(P(\beta))}$
  for some polynomial $P(\cdot)$. 
\item The transition relation for $\starautomaton$ can be decided in
      polynomial-time in $$\max (\lceil log(|\adatum_1|) \rceil ,
      \lceil log(|\adatum_{\alpha}|) \rceil) + \beta + \card{\aalphabet} + D $$ (by Lemma~\ref{lemma-automaton-star}),
      the transition relation for $\locautomaton$
      can be decided in polynomial-time in 
      $\card{\delta}
      + \beta + \card{\aalphabet} + D + \maxconstraintsize{\aautomaton}
      $.
      Note also that $\card{\delta^{\dag}}$ (where $\delta^{\dag}$ is the transition relation of $\ancautomaton$)
      is in
      \[
      (\card{\sattypes{\beta}} \times \card{\locations} \times 2^{\mathcal{O}(P(\beta))})^{\degree +1}
      \times \card{\aalphabet \times \sattypes{\beta}},
      \]
      since  the finite alphabet of $\ancautomaton$
      is $\aalphabet \times \sattypes{\beta}$.
        Consequently, $\card{\delta^{\dag}}$  is in
        \[
      \big( (\adatum_{\alpha} - \adatum_1)^{2\beta} \times \card{\locations} \times 2^{\mathcal{O}(P'(\beta))}\big)^{\degree +2}
      \times \card{\aalphabet},
      \]
      for some polynomial $P'$. 
    Moreover, the product of the number of Rabin pairs is polynomial in $\beta$.
    Therefore, the transition relation can be decided in polynomial-time in
    $\big(\card{\delta} 
    + \beta + \card{\aalphabet} + D +  \maxconstraintsize{\aautomaton}\big)
    $. \qedhere
\end{itemize}
\end{proof}
Nonemptiness of Rabin tree automata is polynomial in the cardinality of the transition relation
and exponential in the number of Rabin  pairs, see e.g.~\cite[Theorem 4.1]{Emerson&Jutla00}.
More precisely, it can be solved in time
$(\card{\delta} \times \gamma \times N)^{\mathcal{O}(N)}$ (by scrutiny of the proof of~\cite[Theorem 4.1]{Emerson&Jutla00}, page 144)
where $N$ is the number of Rabin pairs, $\delta$ is the transition relation and $\gamma$
is the time to decide $\delta$ (this depends on how the locations and the transitions are encoded).
For instance, in our case, $\gamma$ may depend on
  parameters related to $\aautomaton$ and in Lemma~\ref{lemma-exptime-tca} below, $\gamma$ takes
  the value
  $\card{\delta}  + \beta + \card{\aalphabet} + D + \maxconstraintsize{\aautomaton}$
  (by Lemma~\ref{lemma-intersection-automaton}).
  When a Rabin tree automaton is provided in extension, $\card{\delta} \times \gamma$ is polynomial in its size
  and usually $\gamma$ is omitted. 
    Hence the following result.
\begin{lem} \label{lemma-exptime-tca}
  The nonemptiness problem for TCA can be solved
  in time in
 \[
  R_1\big(\card{\locations} \times \card{\delta}
  \times \maxconstraintsize{\aautomaton} \times
  \card{\aalphabet} \times R_2(\beta) \big)^{\mathcal{O}(R_2(\beta) \times R_3(\degree))},
  \]
  for some polynomials $R_1$, $R_2$ and $R_3$.
\end{lem}
\begin{proof}(Sketch) In the above expression $(\card{\delta} \times \gamma \times N)^{\mathcal{O}(N)}$, let us
  see how this is instantiated for $\ancautomaton$ from Lemma~\ref{lemma-intersection-automaton}.
  \begin{itemize}
  \item $N$ is in $R'(\beta)$ for some polynomial $R'$ (Lemma~\ref{lemma-intersection-automaton}(II)).
  \item $\card{\delta}$ is in $((\adatum_{\alpha} - \adatum_1)^{2\beta} \times \card{\locations} \times 2^{\mathcal{O}(P'(\beta))})^{\degree +2}
    \times \card{\aalphabet}$ for some polynomial $P'$ (Lemma~\ref{lemma-intersection-automaton}(IV)).
  \item $\gamma$ is in $R''(\card{\delta}
    + \beta + \card{\aalphabet} + D + \maxconstraintsize{\aautomaton})$
    for some polynomial $R''$ (Lemma~\ref{lemma-intersection-automaton}(V)). 
  \end{itemize}
  This allows to get the bound from the statement. It is essential in the calculation  that
  the exponent is only polynomial in the size of the input constraint automaton $\aautomaton$, which is the case
  as the exponent is polynomial in $\beta+\degree$. 
\end{proof}

Assuming that the size of the
TCA $\aautomaton=(\locations,\aalphabet,\degree,\beta,\locations_\init,\delta,F)$,
written $\size{\aautomaton}$, is polynomial in
$
\card{\locations} + \card{\delta} + D + \beta + \maxconstraintsize{\aautomaton}
$
(which makes sense for a reasonably succinct encoding), from the computation of the bound
in Lemma~\ref{lemma-exptime-tca}, the nonemptiness of $\alang(\aautomaton)$ can be checked in
time $R(\size{\aautomaton})^{\mathcal{O}(R'(\beta+D))}$ for some polynomials $R$ and $R'$. 
The \exptime upper bound of the nonemptiness problem for TCA is now a
consequence of the above complexity expression
and using the fact
 that the transitions
in the product Rabin tree automaton between  $\locautomaton$ and $\starautomaton$ can be decided
in polynomial-time.
\begin{thm}
\label{theorem-exptime-tca}
Nonemptiness problem for tree constraint automata is \exptime-complete.
\end{thm}
We have seen that the \exptime-hardness holds as soon as $D = 2$. The case $D=1$ differs slightly. 

\begin{thm}
  \label{theorem-pspace-tca-d-equal-one}
  For the fixed degree $D = 1$, the nonemptiness problem for word constraint automata is \pspace-complete.
\end{thm}
\begin{proof}(Sketch) \pspace-hardness is obtained similarly to what is done
  in Appendix~\ref{appendix-exptime-hardness} 
  by reduction from the halting problem for deterministic
  Turing machines running in polynomial space.
  Actually, \pspace-hardness is  a corollary of the \exptime-hardness proof
  from Appendix~\ref{appendix-exptime-hardness}, since the construction also works
  for non-deterministic Turing machines and then the runs are sequences instead of trees. 
  Concerning the \pspace-membership, we need to check the nonemptiness of
  $\alang(\ancautomaton_R) \cap \alang(\locautomaton)$ with
  the Rabin word automaton $\ancautomaton_R$ from the proof of Lemma~\ref{lemma-automaton-star}
  and $\locautomaton$ is already a B\"uchi automaton.
  Not only $\ancautomaton_R$ can be transformed into an equivalent B\"uchi automaton $\ancautomaton$
  with a polynomial increase of the number of locations (because the number of Rabin pairs is bounded
  by $8(\beta+2)^2+3$), but
  the nonemptiness of the product B\"uchi automaton between $\ancautomaton$
  and $\locautomaton$ can be performed in \pspace. Indeed, the number
  of locations of the product is only exponential  in the size of
  the input word constraint automaton and the transition relation can be also decided in polynomial
  space.
\end{proof}

The Rabin word automaton $\ancautomaton_{R}$ captures therefore
the condition $C_{\Zed}$ from~\cite[Section 6]{Demri&DSouza07}  (see also
the condition $\mathcal{C}$ in~\cite[Definition 2]{Demri&Gascon08}) and can be turned
in polynomial-time into a nondeterministic B\"uchi automata, leading to the
\pspace upper bound for $\LTL(\Zed)$ (the linear-time temporal logic 
with constraints from the concrete domain $\Zed$,
see Section~\ref{section-ctlstar} and Section~\ref{section-ctlstarz-special-form}),
 proposing therefore an alternative proof
to~\cite[Theorem 1]{Demri&Gascon08} and~\cite[Theorem 16]{Segoufin&Torunczyk11}
for the concrete domain $\Zed$.

\subsection{Rabin Tree Constraint Automata}
\label{section-TCA-Rabin}
In this section, we show that the nonemptiness problem for Rabin tree constraint automata is also 
in \exptime. This will be key to characterize the complexity of $\satproblem{\CTLStar(\Zed)}$. 
Given an Rabin TCA $\aautomaton=(\locations,\aalphabet,\degree,\beta,\locations_\init,\delta,\rabinacc)$, 
the definition of symbolic trees respecting $\aautomaton$ is updated so that its uses
the acceptance condition $\rabinacc$.
From the Rabin TCA $\aautomaton$, we can define a Rabin tree automaton $\locautomaton'$ (instead
of a B\"uchi tree automaton with a B\"uchi TCA) such that the acceptance $\rabinacc'$ is equal
to
\[
\set{
  \pair{ \sattypes{\beta} \times L
  }{
    \sattypes{\beta} \times U
  }
  \mid
  \pair{L}{U} \in \rabinacc
  }.
\]
Similarly to Lemma~\ref{lemma-satisfiable-symbolic-tree}, one can show that 
  $\alang(\aautomaton) \neq \emptyset$ iff there is a symbolic tree $\asymtree \in \alang(\locautomaton')$
that is satisfiable.
Moreover, we can take advantage of $\starautomaton$ (the same as in the proof of Lemma~\ref{lemma-automaton-star})
so that $\alang(\aautomaton) \neq \emptyset$ iff $\alang(\locautomaton')  \cap \alang(\starautomaton)$
is non-empty (same arguments as for the proof of Lemma~\ref{lemma-intersection-reduction}).
It remains to determine how much it costs to test nonemptiness of
$\alang(\locautomaton')  \cap \alang(\starautomaton)$. We might expect
a complexity jump compared to the case with B\"uchi TCA (but this is not the case), 
because the nonemptiness problem for B\"uchi tree automata is in \ptime~\cite[Theorem 2.2]{Vardi&Wolper86}
whereas it is \np-complete for Rabin tree automata~\cite[Theorem 4.10]{Emerson&Jutla00}.
Here is the result that provides quantitative analysis about components of $\aautomaton$,
which is a variant of Lemma~\ref{lemma-intersection-automaton}, more particularly by
considering the value $\card{\rabinacc}$ in the analysis. For TCA with B\"uchi acceptance condition,
the value for $\card{\rabinacc}$ is equal to one.

\begin{lem}
\label{lemma-intersection-automaton-rtca}
There is a Rabin tree automaton $\ancautomaton$ such that
\begin{description}
\item[(I)] $\alang(\ancautomaton) = \alang(\locautomaton') \cap \alang(\starautomaton)$,
\item[(II)] the number of Rabin pairs is polynomial in $\beta + \card{\rabinacc}$,  where $\card{\rabinacc}$ is the number
    of Rabin pairs in $\aautomaton$, 
\item[(III)] the number of locations is in
  $
  (\adatum_{\alpha} - \adatum_1)^{2\beta}  \times \card{\locations} \times 2^{\mathcal{O}(P(\beta +\card{\rabinacc}))}
  $ for some polynomial $P$,
\item[(IV)] the cardinality of the transition relation is in 
  \[
  ((\adatum_{\alpha} - \adatum_1)^{2\beta} \times \card{\locations} \times 2^{\mathcal{O}(P'(\beta +\card{\rabinacc}))})^{\degree +2}
  \times \card{\aalphabet}
  \] for some polynomial $P'$,
 \item[(V)] the transition relation can be decided in polynomial-time in 
    $\card{\delta}
    + \beta + \card{\aalphabet} + D + \maxconstraintsize{\aautomaton}
    $.
\end{description}
\end{lem}
The proof of Lemma~\ref{lemma-intersection-automaton-rtca} is similar to the proof
of Lemma~\ref{lemma-intersection-automaton}. Moreover, as for Lemma~\ref{lemma-exptime-tca},
we can conclude that the nonemptiness problem for Rabin tree constraint automata  can be solved
in time in
\[
\hspace*{-0.15in} 
  R_1\big(\card{\locations} \times \card{\delta}
  \times \maxconstraintsize{\aautomaton} \times
  \card{\aalphabet} \times R_2(\beta + \card{\rabinacc}) \big)^{\mathcal{O}(R_2(\beta + \card{\rabinacc}) \times R_3(\degree))}
\]
  for some polynomials $R_1$, $R_2$ and $R_3$. 
Theorem~\ref{theorem-exptime-rtca} is one of the main results of the paper. 
  
\begin{thm}
\label{theorem-exptime-rtca}
  The nonemptiness problem for  Rabin tree constraint automata is \exptime-complete.
  \end{thm}

\section{Tree Constraint Automata for $\CTL(\Zed)$}
\label{section-complexity-ctlz}

Below, we  harvest the first result from what is achieved in the previous section:
$\satproblem{\CTL(\Zed)}$ is in \exptime. 
We follow the automata-based approach and -- after proving a refined version of the tree model
property for $\CTL(\Zed)$ -- the key step is to translate $\CTL(\Zed)$ formulae into
equivalent TCA  (Theorem~\ref{theorem-ctlz-to-tca}).
As usual, the tree model property means that we can restrict ourselves to
tree Kripke structures to determine the satisfiability status of $\CTL(\Zed)$
formulae. 
In preparation, we first show how  $\CTL(\Zed)$ formulae can be put into simple form
(Proposition~\ref{proposition-simple-form}),  and that 
the tree model property for $\CTL(\Zed)$ can follow a strict discipline
(Proposition~\ref{proposition-tree-model-for-Z}). 

A $\CTL(\Zed)$ formula is in \defstyle{simple form}
iff it is in negation normal form 
and terms are restricted to
those in $\myterms{\leq 1}{\VAR}$. 
Preprocessing  $\CTL(\Zed)$ formulae to obtain simple formulae is computationally harmless,
but will simplify the translation of  $\CTL(\Zed)$ formulae to tree constraint automata. 

\begin{prop}
\label{proposition-simple-form}
  For every $\CTL(\Zed)$ formula $\aformula$, one can construct in polynomial-time
  in the size of $\aformula$ 
  a $\CTL(\Zed)$ formula $\aformula'$ in simple form such that
  $\aformula$ is satisfiable iff  $\aformula'$ is satisfiable. 
\end{prop}
The proof of Proposition~\ref{proposition-simple-form} can be found
in Appendix~\ref{appendix-proof-proposition-simple-form} and it is made of two standard arguments.
First, it establishes
a tree model property (in the standard way using unfoldings of Kripke models). 
Then it uses a  renaming technique (see e.g.~\cite{Scott62}) to flatten the constraints
that introduces additional variables for values $i$ steps ahead.

From now on, we assume that the $\CTL(\Zed)$ formulae are in simple form.
Given a $\CTL(\Zed)$ formula $\aformula$ in simple form, we write $\subf{\aformula}$ to denote the smallest set
such that
\begin{itemize}
\item $\aformula \in \subf{\aformula}$; $\subf{\aformula}$ is closed under subformulae,
\item for all $\pathquantifier \in \set{\existspath, \forallpaths}$ and
  $\mathsf{Op} \in \set{\until, \release}$,
  if $\pathquantifier \ \aformula_1 \ \mathsf{Op} \ \aformula_2 \in \subf{\aformula}$, then
  $\pathquantifier \mynext \ \pathquantifier \ \aformula_1 \ \mathsf{Op} \ \aformula_2 \in \subf{\aformula}$. 
  \end{itemize}
The cardinality of $\subf{\aformula}$ is at most twice the
number of subformulae of $\aformula$.
Given $\aset \subseteq \subf{\aformula}$, we say that $\aset$
is \defstyle{propositionally consistent}  iff the conditions below hold.
\begin{itemize}
\item If $\aformula_1 \vee \aformula_2 \in \aset$, then
      $\set{\aformula_1, \aformula_2} \cap \aset \neq \emptyset$. 
\item If $\aformula_1 \wedge \aformula_2 \in \aset$, then
  $\set{\aformula_1, \aformula_2} \subseteq \aset$.
\item If $\existspath \aformula_1 \until \aformula_2 \in \aset$, then
  $\aformula_2 \in \aset$ or
  $\set{\aformula_1, \existspath \mynext \existspath \aformula_1 \until \aformula_2} \subseteq \aset$.
  \item If $\forallpaths \aformula_1 \until \aformula_2 \in \aset$, then
  $\aformula_2 \in \aset$ or
    $\set{\aformula_1, \forallpaths \mynext \forallpaths \aformula_1 \until \aformula_2} \subseteq \aset$.
  \item If $\existspath \aformula_1 \release \aformula_2 \in \aset$, then
  $\aformula_2 \in \aset$ and
    $\set{\aformula_1, \existspath \mynext \existspath \aformula_1 \release \aformula_2} \cap \aset \neq \emptyset$.
    \item If $\forallpaths \aformula_1 \release \aformula_2 \in \aset$, then
  $\aformula_2 \in \aset$ and
    $\set{\aformula_1, \forallpaths \mynext \forallpaths \aformula_1 \release \aformula_2} \cap \aset \neq \emptyset$.
\end{itemize}

We write $\parsubf{\existspath \mynext}{\aformula}$ to denote the set of formulae in $\subf{\aformula}$
of the form $\existspath \mynext \aformulabis$. Similarly, we write
$\parsubf{\existspath \until}{\aformula}$ (resp. $\parsubf{\forallpaths \until}{\aformula}$)
to denote the set of formulae in $\subf{\aformula}$
of the form $\existspath \aformulabis_1 \until \aformulabis_2$ (resp.  $\forallpaths \aformulabis_1 \until \aformulabis_2$).
Finally, we  write  $\parsubf{\existspath}{\aformula}$ to denote the set of formulae of
the form $\existspath \ \acons$ in $\subf{\aformula}$.

Let $\aformula$ be a $\CTL(\Zed)$  formula in simple form built over the variables $\avariable_1, \ldots, \avariable_{\beta}$ for some $\beta \geq 1$,
and set $\degree = \card{\parsubf{\existspath \mynext}{\aformula}} + \card{\parsubf{\existspath}{\aformula}}$. 
A \defstyle{direction map $\iota$ for $\aformula$} is a bijection
\[
\iota: (\parsubf{\existspath \mynext}{\aformula} \cup \parsubf{\existspath}{\aformula})
\rightarrow \interval{1}{\degree}.
\]
We say that a tree model
$\atree: \interval{0}{\degree}^* \rightarrow \Zed^{\beta}$ of $\aformula$ (i.e. $\atree, \varepsilon \models \aformula$)
\defstyle{obeys a direction map $\iota$}
if for all nodes $\anode \in \interval{0}{\degree}^*$, the following three conditions hold. 
\begin{enumerate}
\item For every $\existspath\mynext \aformula_1\in \parsubf{\existspath \mynext}{\aformula}$, if $\atree,\anode\models\existspath\mynext\aformula_1$, then 
$\atree, \anode\cdot j \models \aformula_1$ with $j=\iota(\existspath\mynext\aformula_1)$. 
\item For every $\existspath \aformula_1\until \aformula_2\in\parsubf{\existspath \until}{\aformula}$, if
  $\atree,\anode\models\existspath \aformula_1\until\aformula_2$,
then there exists some $k\geq 0$ such that $\atree,\anode\cdot j^k\models\aformula_2$ and $\atree, \anode\cdot j^i\models\aformula_1$ for all $0\leq i<k$, where
 $j=\iota(\existspath\mynext\existspath\phi_1\until\phi_2)$. In other words, the path $\anode, \anode \cdot j, \anode\cdot j^2 \dots \anode\cdot j^k$
satisfies $\aformula_1\until \aformula_2$.   
\item For every $\existspath \ \acons \in \parsubf{\existspath}{\aformula}$, 
if $\atree,\anode\models\existspath \ \acons$, 
then $\Zed\models \acons(\atree(\anode),\atree(\anode \cdot j))$ with $\iota(\existspath \ \acons)=j$. 
\end{enumerate}
Again, here, $\Zed \models \acons(\vect{z},\vect{z'})$ with $\vect{z}, \vect{z'}$ is a shortcut for
$[\vec{\avariable} \leftarrow \vect{z}, \vec{\avariable'} \leftarrow \vect{z'}] \models
  \acons$ where $[\vec{\avariable} \leftarrow \vect{z}, \vec{\avariable'} \leftarrow \vect{z'}]$ is a valuation $\avaluation$ on the variables
  $\set{\avariable_j, \avariable_j' \mid j \in \interval{1}{\beta}}$ with 
  $\avaluation(\avariable_j) = \vect{z}(j)$ and $\avaluation(\avariable_j') = \vect{z'}(j)$ for all $j \in \interval{1}{\beta}$.
  A tree model $\atree$ obeying a direction map $\iota$ follows a discipline to verify
  the satisfaction of formulae involving existential quantifications over paths, namely dedicated directions
  are reserved for specific formulae. 

\begin{prop}
  \label{proposition-tree-model-for-Z}
  Let $\aformula$ be a $\CTL(\Zed)$-formula in simple form and
  $\iota$ be a direction map for $\aformula$. Then, 
  $\aformula$  is satisfiable iff 
  $\aformula$ has a tree model with branching width  equal to $\card{\parsubf{\existspath \mynext}{\aformula}} + \card{\parsubf{\existspath}{\aformula}} + 1$ and 
  that obeys $\iota$. 
\end{prop}

The proof of Proposition~\ref{proposition-tree-model-for-Z} can be
found in Appendix~\ref{appendix-proof-proposition-tree-model-for-Z}.
It amounts to structure the paths from a state satisfying $\aformula$ in a tree-like fashion
while obeying the direction map $\iota$. 
Let us now turn to the key step and explain how to obtain a TCA $\aautomatonbis_{\aformula}$ from a $\CTL(\Zed)$ formula $\aformula$ in simple form. 
Observe that we take very good care of the quantitative properties of $\aautomatonbis_{\aformula}$: 
in general, the  size of $\aautomatonbis_{\aformula}$ is exponential in the size of the input formula $\aformula$. 
In Section \ref{section-complexity-nonemptiness}, we proved that the nonemptiness problem for TCA is \exptime-complete. 
In order to get \exptime-completeness of  $\satproblem{\CTL(\Zed)}$ and therefore avoiding
the {\em double} exponential blow-up, we  
refine the analysis of the size of  $\aautomatonbis_{\aformula}$ by looking at the size
of its different components, revealing which components are responsible for the exponential blow-up.

\begin{thm}
\label{theorem-ctlz-to-tca}
  Let $\aformula$ be a $\CTL(\Zed)$ formula in simple form. There exists a TCA $\aautomatonbis_{\aformula}$
  such that $\aformula$ is satisfiable iff $\alang(\aautomatonbis_{\aformula}) \neq \emptyset$, and satisfying
  the properties below.
  \begin{description}
  \item[(I)] The degree $\degree$  and the number of variables $\beta$ is bounded above by $\size{\aformula}$.
  \item[(II)] The alphabet $\aalphabet$ is a singleton. 
  \item[(III)] The number of locations is bounded by $(\degree \times 2^{\size{\aformula}}) \times (\size{\aformula} + 1)$.
  \item[(IV)] The number of transitions is in $2^{\mathcal{O}(P(\size{\aformula}))}$ for some polynomial $P(\cdot)$.
  \item[(V)] The maximal size of a constraint in transitions is quadratic in $\size{\aformula}$.
  \end{description}
\end{thm}
\begin{proof}
Let $\aformula$ be a $\CTL(\Zed)$ formula in simple form
built over $\avariable_1, \ldots, \avariable_{\beta},\mynext \avariable_1, \ldots, \mynext \avariable_{\beta}$.
Let $\degree = \card{\parsubf{\existspath \mynext}{\aformula}} + \card{\parsubf{\existspath}{\aformula}}$ and 
 $\iota:\parsubf{\existspath \mynext}{\aformula} \cup \parsubf{\existspath}{\aformula}\to \interval{1}{\degree}$ be a direction map. 
 We start building a \emph{generalised B\"uchi} TCA (see Section \ref{section-automata})  
 $\aautomaton_{\aformula} = \triple{\locations,\aalphabet,\degree+1,\beta}{\locations_\init,\delta}{\rabinacc}$ such that $\aformula$ is satisfiable iff $\alang(\aautomaton_{\aformula}) \neq \emptyset$.
The automaton $\aautomaton_{\aformula}$ accepts infinite trees of the form 
 $\atree: \interval{0}{\degree}^* \rightarrow \aalphabet \times \Zed^{\beta}$. 
Let us define $\aautomaton_{\aformula}$ formally.
\begin{itemize}
\item $\aalphabet \egdef \set{\arbitraryletter}$ ($\arbitraryletter$ is an arbitrary letter).
\item $\locations$ is the subset of $\interval{0}{\degree} \times \powerset{\subf{\aformula}}$
  such that $\pair{i}{\aset}$ belongs to $\locations$ only if $\aset$ is propositionally consistent.
  The first argument records the direction (from which the nodes are reached), as indicated by $\iota$.  
\item $\locations_\init \egdef \set{\pair{0}{\aset} \in \locations \mid \aformula \in \aset}$.
\item The transition relation $\delta$ is made of tuples of the form
  \[
  \triple{\pair{i}{\aset}}{\arbitraryletter}{\pair{\acons_0}{\pair{0}{\aset_0}},\ldots,\pair{\acons_\degree}{\pair{\degree}{\aset_\degree}}}
  \]
  verifying the conditions below.
\begin{enumerate}
  \item For all $\existspath \mynext \aformulabis \in \aset$, we have $\aformulabis \in \aset_{\iota(\existspath \mynext \aformulabis)}$.
  \item For all $\forallpaths \mynext \aformulabis \in \aset$ and $j \in \interval{0}{\degree}$,
        we have $\aformulabis \in \aset_{j}$.
  \item
  For all $j \in \interval{0}{\degree}$,
  if there is $\existspath  \ \acons \in \aset$ such that $\iota(\existspath  \ \acons) =j$, then
    \[
    \acons_j \egdef (\bigwedge_{\forallpaths \ \acons' \in \aset} \acons') \wedge \acons
    \ \ \ 
    \mbox{otherwise,}
    \ \ \ 
    \acons_j  \egdef \bigwedge_{\forallpaths \ \acons' \in \aset} \acons'.
    \]
    \end{enumerate}
 
\item $\rabinacc$ is made of two types of sets, those parameterised by some element in $\parsubf{\existspath \until}{\aformula}$
  and those parameterised by some element in $\parsubf{\forallpaths \until}{\aformula}$.
  Recall that $\rabinacc$ is a generalised B\"uchi condition made of subsets of $\locations$:
  for all branches of the accepted trees, 
  for all sets $F$ in $\rabinacc$,  some location in $F$ must occur infinitely often. 
  The subformulae whose
    outermost connective is either $\existspath \release$ or $\forallpaths \release$ do not impose additional acceptance conditions.
    For each $\existspath \aformulabis_1 \until \aformulabis_2 \in \parsubf{\existspath \until}{\aformula}$, the set
    $F_{\existspath \aformulabis_1 \until \aformulabis_2 }$ defined below belongs to $\rabinacc$:
    \[
    F_{\existspath \aformulabis_1 \until \aformulabis_2 } \egdef
    \set{\pair{i}{\aset} \in \locations \mid
      i \neq \iota(\existspath \mynext \existspath \aformulabis_1 \until \aformulabis_2 ) \
      \mbox{or} \
      \aformulabis_2 \in \aset \
      \mbox{or} \
      \existspath \aformulabis_1 \until \aformulabis_2 \not \in \aset
    }.
    \]
    Hence, if $ i = \iota(\existspath \mynext \existspath \aformulabis_1 \until \aformulabis_2 )$, then
    along a branch $\anode \cdot i^{\omega}$, the satisfaction of  $\aformulabis_1 \until \aformulabis_2$
    cannot be postponed forever. This is a
    standard encoding to translate \LTL formulae into
    B\"uchi automata, see e.g.~\cite{Vardi&Wolper94}. 
    Moreover, for each $\forallpaths \aformulabis_1 \until \aformulabis_2 \in \parsubf{\forallpaths \until}{\aformula}$, the set
    $F_{\forallpaths \aformulabis_1 \until \aformulabis_2 }$ belongs to $\rabinacc$:
    \[
    F_{\forallpaths \aformulabis_1 \until \aformulabis_2 } \egdef
    \set{\pair{i}{\aset} \in \locations \mid
      \aformulabis_2 \in \aset \
      \mbox{or} \
      \forallpaths \aformulabis_1 \until \aformulabis_2 \not \in \aset
    }.
    \]
    By contrast to the behaviours induced by the sets $F_{\existspath \aformulabis_1 \until \aformulabis_2 }$'s, on any branch from a node $\anode$ satisfying
    $\forallpaths \aformulabis_1 \until \aformulabis_2$,
    the satisfaction of  $\aformulabis_1 \until \aformulabis_2$
    cannot be postponed foreover.
\end{itemize}
The correctness of $\aautomaton_{\aformula}$ is stated below.
\begin{lem}
  \label{lemma-correctness-ctlz-aut-to-form}
$\aformula$ is satisfiable iff  $\alang(\aautomaton_{\aformula}) \neq \emptyset$.
\end{lem}
The proof of Lemma~\ref{lemma-correctness-ctlz-aut-to-form}
can be found in Appendix~\ref{appendix-proof-lemma-correctness-ctlz-aut-to-form}.
It follows a standard pattern but we need to handle the constraints on data values, 
and we take advantage of the direction map $\iota$ in order to guide the satisfaction
of the qualitative constraints related to the temporal connective $\until$. 

Finally, let $\aautomatonbis_\aformula$ be the TCA obtained from $\aautomaton_\aformula$ as explained in Section~\ref{section-automata},
and recall that the size of $\aautomatonbis_\aformula$ is polynomial in the size of $\aautomaton_\aformula$. It is now easy to check that all quantitative properties given in the claim of the theorem are satisfied.
\end{proof}

The construction of the generalised B\"uchi TCA $\aautomaton_{\aformula}$ is mainly inspired
from~\cite[page 702]{Vardi&Wilke08} for \CTL formulae, 
but a few essential differences need to be pointed out here.
Obviously, our construction handles constraints, which is expected since $\CTL(\Zed)$ extends \CTL by adding
constraints between data values.
More importantly,~\cite[Section 5.2]{Vardi&Wilke08} uses another tree automaton model and
$F_{\existspath \aformulabis_1 \until \aformulabis_2 }$
and $F_{\forallpaths \aformulabis_1 \until \aformulabis_2 }$ herein differ from~\cite[page 702]{Vardi&Wilke08}
by the  use of the direction map that we believe necessary here.

\begin{figure}
  \begin{center}
  \scalebox{1}{
  \begin{tikzpicture}[node distance=0.8cm,level 1/.style={sibling distance=4cm},
                                      level 2/.style={sibling distance=2.7cm},
                                      level 3/.style={sibling distance=1.5cm},
                                      level 4/.style={sibling distance=1.5cm},
                                      level distance=1cm]
  \node (t1) {\textcolor{red}{$\existspath \always  \ \existspath \sometimes \
              \existspath (\avariable = 0)$}}
  child {node {\textcolor{red}{$\existspath \sometimes \
              \existspath (\avariable = 0)$}}
       child {node {\textcolor{red}{$\existspath \sometimes \ 
              \existspath (\avariable = 0)$}}
              child { node {$\vdots$}}
              child {node {$\existspath \sometimes \ 
                     \existspath (\avariable = 0)$}
                    child { node {$\vdots$}}
                    child {node {$\existspath \sometimes \ 
                     \existspath (\avariable = 0)$}
                     child { node {$\vdots$}}
                           child {node {\fbox{$\existspath (\avariable = 0)$}}
                           child { node {$\vdots$}}
                           child { node {$\vdots$}}
                           }
                     }
                     }
             }
       child {node {$\existspath \sometimes \ 
              \existspath (\avariable = 0)$}
              child { node {$\vdots$}}
              child {node {\fbox{$\existspath (\avariable = 0)$}}
                    child { node {$\vdots$}}
                    child { node {$\vdots$}}
                    }
             }}
  child {node {\fbox{$\existspath (\avariable = 0)$}}
         child { node {$\vdots$}}
         child { node {$\vdots$}}
         }

  ;       
\end{tikzpicture}
}
\end{center}
\caption{Tree-like Kripke structure for $
\existspath \always  \ \existspath  \sometimes \ \existspath
(\avariable = 0)$ 
  }
\label{figure-ctl}
\end{figure}

\begin{exa}
Herein, we illustrate the use of direction maps.
First, observe that direction maps are helpful for the satisfaction of $\existspath \mynext$-formulae in the definition
of $\aautomaton_{\aformula}$ (see the proof of Theorem~\ref{theorem-ctlz-to-tca})
and each set $F_{\existspath \aformulabis_1 \until \aformulabis_2 }$  of accepting locations defined by
    $\set{\pair{i}{\aset} \in \locations \mid
      i \neq \iota(\existspath \mynext \existspath \aformulabis_1 \until \aformulabis_2 ) \
      \mbox{or} \
      \aformulabis_2 \in \aset \
      \mbox{or} \
      \existspath \aformulabis_1 \until \aformulabis_2 \not \in \aset
    }$ involves the direction map $\iota$.

Let us consider the formula $\aformula = \existspath \always  \ \existspath \sometimes \ \existspath (\avariable = 0)$
that states the existence of a path $\apath$ along which at every
position $i$ there is a path $\apath_i$ leading to a position $k_i$
such that $\avariable$ is equal to
zero. Note that $\existspath \sometimes \ \existspath (\avariable = 0)$ is a formula of the form
$\existspath \aformulabis_1 \until \aformulabis_2$ and therefore  would generate
a set $F_{\existspath \sometimes \ \existspath (\avariable = 0)}$ of accepting locations.
Let us explain why it is not adequate to define $F_{\existspath \sometimes \ \existspath (\avariable = 0)}$
by $\set{\pair{i}{\aset} \in \locations \mid
      \existspath (\avariable = 0) \in \aset \
      \mbox{or} \
      \existspath \sometimes \ \existspath (\avariable = 0) \not \in \aset
}$, i.e. without the use of the direction map $\iota$.

In Figure~\ref{figure-ctl}, we present a tree-like infinite Kripke structure
satisfying $\aformula$ on its root node.  Each node is labelled by a formula that holds
on it, witnessing the satisfaction of $\aformula$ at the root. Framed nodes are the only ones satisfying
the subformula $\existspath (\avariable = 0)$. 
For the satisfaction of $\aformula$ at the root, the above mentioned path
$\apath$ can be the leftmost branch, and each $\apath_i$ is the rightmost branch from
$\apath(i)$ with $k_i = i+1$. Observe that for all nodes along $\apath$ (in red in a coloured version of the document), 
the formula $\existspath (\avariable = 0)$ does not hold and the formula $\existspath \sometimes \ \existspath (\avariable = 0)$
holds. However, infinitely often on the path $\apath$, a node on $\apath$ is the leftmost child of its parent node
(and therefore it is not the rightmost child of its parent node). This is fine as soon as
the satisfaction of $\existspath \sometimes \ \existspath (\avariable = 0)$ is witnessed along
a branch (which is not the leftmost branch, actually the rightmost branch in this example). 

Hence, considering the structure in Figure~\ref{figure-ctl} as a run for $\aautomaton_{\aformula}$
in which instead of labelling the nodes by pairs $\pair{i}{\aset}$, we have just represented a distinguished
subformula, we can check that this leads to an accepting run assuming that
the direction map $\iota$ is such that $\iota(\existspath \mynext  \ \existspath \sometimes \ 
\existspath (\avariable = 0))$
 defined as the greatest direction, usually represented as the rightmost direction. 
Along the leftmost branch, no  element in $\set{\pair{i}{\aset} \in \locations \mid
      \existspath (\avariable = 0) \in \aset \
      \mbox{or} \
      \existspath \sometimes \ \existspath (\avariable = 0) \not \in \aset
}$ is visited infinitely often, whereas it does for  $F_{\existspath \sometimes \ \existspath (\avariable = 0)}$
involving the direction map $\iota$. Indeed, all the nodes along the leftmost branch have indices
different from $\iota(\existspath \mynext  \ \existspath \sometimes \ 
\existspath (\avariable = 0))$. 
\end{exa}
The following theorem is one of our main results. 
\begin{thm}
    \label{theorem-ctlz}
  The satisfiability problem for $\CTL(\Zed)$ is \exptime-complete.
\end{thm}
\begin{proof}
\exptime-hardness is inherited from \CTL. 
For \exptime-membership, let $\aformulabis$ be a $\CTL(\Zed)$ formula. 
First, we construct in polynomial-time from $\aformulabis$ a formula $\aformula$ 
in \emph{simple form} such that $\aformulabis$ is satisfiable iff $\aformula$ is satisfiable (see Proposition~\ref{proposition-simple-form}). 
Second, using Theorem~\ref{theorem-ctlz-to-tca}, we construct from  $\aformula$ the TCA
  $\aautomatonbis_{\aformula}=(\locations,\aalphabet,\degree,\beta,\locations_\init,\delta,F)$ such that $\aformula$ is satisfiable iff $\alang(\aautomatonbis_{\aformula}) \neq \emptyset$ 
  and $\aautomatonbis_{\aformula}$ satisfying the following quantitative properties: 
  \begin{itemize} 
\item the degree $\degree$ and the number of variables $\beta$ are bounded by $\size{\aformula}$,
\item $\card{\locations}$ is bounded by $(\degree \times 2^{\size{\aformula}}) \times (\size{\aformula} + 1)$, 
\item$ \card{\delta}$ is in $2^{\mathcal{O}(P(\size{\aformula}))}$ for some polynomial $P(\cdot)$, 
\item $\card{\aalphabet}=1$, and 
      $\maxconstraintsize{\aautomatonbis_{\aformula}}$ is quadratic in $\size{\aformula}$. 
\end{itemize}
By Lemma~\ref{lemma-exptime-tca},  the nonemptiness problem for TCA can be solved
in time
\[
  R_1\big(\card{\locations} \times \card{\delta}
  \times \maxconstraintsize{\aautomatonbis_{\aformula}} \times
  \card{\aalphabet} \times R_2(\beta)\big)^{\mathcal{O}(R_2(\beta) \times R_3(\degree))}. 
  \]
Since the transition relations of the automata
$\ancautomaton_{\mbox{\tiny cons($\aautomatonbis_{\aformula}$)}}$
and $\starautomaton$ can be built
  in polynomial-time, we get that
  nonemptiness of $\alang(\aautomatonbis_{\aformula})$ can be solved in exponential-time.
\end{proof}

\paragraph{Results for other domains.}
Let $\Nat$ be the concrete domain $\pair{\Nat}{<,=,(=_{\adatum})_{\adatum \in \Nat}}$ for which we can
also
show that 
nonemptiness of TCA with constraints interpreted
on $\Nat$ has the same complexity as for TCA with constraints interpreted on $\Zed$.
Indeed, the nonemptiness problem for TCA on $\Nat$ can be easily reduced in polynomial-time
to the nonemptiness problem for TCA on $\Zed$.
Let $\CTL(\Nat)$ be the variant
of $\CTL(\Zed)$ with constraints interpreted on $\Nat$.
As a corollary, \satproblem{$\CTL(\Nat)$} is \exptime-complete.
With the concrete domain $\pair{\Rat}{<,=,(=_{\adatum})_{\adatum \in \Rat}}$,
all the trees in $\alang(\locautomaton)$ are satisfiable
(no need to intersect $\locautomaton$ with a hypothetical $\starautomaton$, 
see e.g.~\cite{Lutz01,Balbiani&Condotta02,Demri&DSouza07,Gascon09}), and therefore
$\satproblem{\CTL(\Rat)}$ is 
in \exptime too.
TCA can  be also used to show that
the concept satisfiability  w.r.t. general TBoxes for the
description logic $\dlogic$ is in \exptime~\cite{Labai&Ortiz&Simkus20,Labai21},
see the details in~\cite[Section 5.2]{Demri&Quaas23}. 

\section{Complexity of the Satisfiability Problem for the Logic $\CTLStar(\Zed)$}
\label{section-ctlstarz}
In this section, we show that $\satproblem{\CTLStar(\Zed)}$ can be solved in \twoexptime by using
Rabin TCA (see e.g. Section~\ref{section-TCA-Rabin}). 
We follow the automata-based approach for \CTLStar, see e.g.~\cite{Emerson&Sistla84,Emerson&Jutla00}. 
Besides checking that 
the essential steps for \CTLStar can be lifted to $\CTLStar(\Zed)$, we also prove that computationally
we are in a position to provide an optimal complexity upper bound.

\subsection{$\CTLStar(\Zed)$ formulae in special form}
\label{section-ctlstarz-special-form}
We start by establishing a special form for $\CTLStar(\Zed)$ formulae from which Rabin
TCA will be defined, following ideas from~\cite{Emerson&Sistla84} for \CTLStar. 
A  $\CTLStar(\Zed)$  state formula $\aformula$ is in \defstyle{special form} if
it has the form below
\begin{equation}
  \label{equation-special-form}
  \tag{SF} 
\existspath \ (\avariable_1 = 0) \ \wedge \ 
\big(
\bigwedge_{i \in \interval{1}{\degree-1}}
\forallpaths \always \existspath \ \apathformula_i
\big)
\ \wedge \
\big(
\bigwedge_{j \in \interval{1}{\degree'}}
\forallpaths \ \apathformula_j'
\big),
\end{equation}
where the $\apathformula_i$'s and the $\apathformula_j'$'s
are $\LTL(\Zed)$ formulae in simple form (see Section~\ref{section-introduction-temporal-logics}),
for some $\degree \geq 1$, $\degree' \geq 0$. 
We prove below that we can restrict ourselves to $\CTLStar(\Zed)$ state formulae in special form.
Indeed, formulae in simple form are in negation normal form and its terms are only from
$\myterms{\leq 1}{\VAR}$. Formulae in special form are not only in simple form but also the subformulae
are normalised as described in~(\ref{equation-special-form}). 
In order to restrict ourselves to
formulae in special form, we first adapt the proof of  Proposition~\ref{proposition-simple-form}
to $\CTLStar(\Zed)$, leading to Proposition~\ref{proposition-simple-form-ctlstarz}
and this is to handle only simple forms. Then, Proposition~\ref{proposition-ctlstarz-special-form}
is dedicated to formulae in special form. 
The proof of Proposition~\ref{proposition-simple-form-ctlstarz} below  
can be found
in Appendix~\ref{appendix-proof-proposition-simple-form-ctlstarz}.

\begin{prop}
\label{proposition-simple-form-ctlstarz}
  For every $\CTLStar(\Zed)$ state formula $\aformula$, one can construct in polynomial-time
  in the size of $\aformula$ 
  a $\CTLStar(\Zed)$ formula $\aformula'$ in simple form such that
  $\aformula$ is satisfiable iff  $\aformula'$ is satisfiable. 
\end{prop}
Hence the size of $\aformula'$ is polynomial in the size of $\aformula$. 
Using this, we establish the property that allows us to restrict ourselves
to $\CTLStar(\Zed)$ state formulae in special form only.  
\begin{prop}
\label{proposition-ctlstarz-special-form}
For every $\CTLStar(\Zed)$ state formula $\aformula$,
one can construct in polynomial time in the size of $\aformula$ a $\CTLStar(\Zed)$ formula
$\aformula'$ in special form such that $\aformula$ is satisfiable
iff 
$\aformula'$ is satisfiable.
\end{prop}
So $\aformula'$ is also of polynomial size in the size of $\aformula$.
The proof of Proposition~\ref{proposition-ctlstarz-special-form}
can be found in Appendix~\ref{appendix-proof-proposition-ctlstarz-special-form}.
It is similar to the proof for \CTLStar except that we need to handle constraints.

Let us now state a tree model property of formulae in special form, with a strict discipline on the witness paths. 
Given a tree $\atree: \interval{0}{\degree-1} \to \Zed^{\beta}$ with $\atree
\models \forallpaths \always \existspath \ \apathformula$, we say that
$\atree$ \defstyle{satisfies $\forallpaths \always \existspath \ \apathformula$
  via the direction $i$},
 for some $i \in \interval{1}{\degree-1}$, iff 
 for all nodes $\anode\in\interval{0}{\degree-1}^*$, 
we have 
$\aword \models\apathformula$, where $\aword:\Nat\to\Zed^\beta$ is defined by
$\aword(0) \egdef \atree(\anode)$ and $\aword(j) \egdef
\atree(\anode \cdot i \cdot 0^{j-1})$ for all $j\geq 1$. 
Proposition~\ref{proposition-tmp-ctlstarz-direction} below
is a counterpart of~\cite[Theorem 3.2]{Emerson&Sistla84} but for
$\CTLStar(\Zed)$ instead of \CTLStar, see also the variant~\cite[Lemma 3.3]{Gascon09}. 

\begin{prop}
\label{proposition-tmp-ctlstarz-direction}
Let $\aformula$ be a $\CTLStar(\Zed)$ state formula in special form
(with the notations of equation~\eqref{equation-special-form} on
page~\pageref{equation-special-form}) built over
the variables $\avariable_1,\dots,\avariable_\beta$.
  The formula $\aformula$ is satisfiable iff there is a tree $\atree:\interval{0}{\degree-1}
  \rightarrow \Zed^{\beta}$ such that $\atree, \varepsilon \models \aformula$
  and for each $i \in \interval{1}{\degree-1}$, $\atree$
  satisfies $\forallpaths \always \existspath \ \apathformula_i$
  via the direction $i$, that is, if $\atree, \anode \models \existspath \ \apathformula_i$, then $ \apathformula_i$
  is satisfied on the path  $\anode \cdot i \cdot 0^\omega$. 
\end{prop}
\begin{proof} Assume that $\aformula$ has the form below:
  $$
\existspath \ (\avariable_1 = 0) \ \wedge \ 
\big(
\bigwedge_{i \in \interval{1}{\degree-1}}
\forallpaths \always \existspath \ \apathformula_i
\big)
\ \wedge \
\big(
\bigwedge_{j \in \interval{1}{\degree'}}
\forallpaths \ \apathformula_j'
\big).
$$

The direction from right to left is trivial. 
So let us prove the direction from left to right: suppose that
$\aformula$ is satisfiable.
Let $\kripke=\triple{\worlds}{\arelation}{\avaluation}$ be a  Kripke structure,
$\aworld_\init\in\worlds$ be a world in $\kripke$ such that $\kripke,\aworld_\init\models\aformula$.
Since $\aformula$ contains only the variables in $\avariable_1,\dots,\avariable_\beta$, the map
$\avaluation$ can be restricted to the variables among $\avariable_1,\dots,\avariable_\beta$.
Furthermore,  below, we can represent $\avaluation$ as a map $\worlds \to \Zed^{\beta}$ such that
$\avaluation(\aworld)(i)$ for some $i \in \interval{1}{\beta}$ is understood as
the value of the variable $\avariable_i$ on $\aworld$. 
We construct a tree $\atree: \interval{0}{\degree-1}^* \to \Zed^\beta$ such that
$\atree \models \aformula$
  and $\atree$ satisfies $\forallpaths \always \existspath \ \apathformula_i$
  via $i$, for each $i \in \interval{1}{\degree-1}$. 

We introduce an auxiliary map $g:\interval{0}{\degree-1}^*\to\worlds$ such that 
$g(\varepsilon)\egdef\aworld_\init$, 
$\atree(\varepsilon)\egdef \avaluation(g(\varepsilon))$. 
More generally, we require that for all nodes $\anode\in\interval{0}{\degree-1}^*$,
we have $\atree(\anode)=\avaluation(g(\anode))$, 
and if $\anode=j_1\cdot  \dots \cdot j_k$,
then 
there exists a finite path $g(\varepsilon) g(j_1) g(j_1j_2) \dots g(\anode)$ in $\kripke$. 
Note that this is satisfied by $\varepsilon$, too.
The definition of $g$ is performed by picking the smallest node $\anode \cdot j\in \interval{0}{\degree-1}^*$ with respect to the lexicographical ordering such that $g(\anode)$ is defined and $g(\anode\cdot j)$ is undefined. 
So let $\anode \cdot j$ be the smallest node $\anode \cdot j\in \interval{0}{\degree-1}^*$ such that $g(\anode)$ is defined, 
$g(\anode\cdot j)$ is undefined, 
and 
if $\anode=j_1\cdot  \dots \cdot j_k$,
then 
there exists a finite path $g(\varepsilon) g(j_1) g(j_1j_2) \dots g(\anode)$ in $\kripke$.

If $j = 0$, then since $\kripke$ is total, there is an infinite path
$\apath = \aworld_0 \aworld_1 \aworld_2 \cdots$
  starting from  $g(\anode)$. For all $k \geq 1$, we set $g(\anode \cdot 0^k) \egdef \aworld_k$
  and $\atree(\anode \cdot 0^k) \egdef \avaluation(\aworld_k)$.

Otherwise ($j \neq 0$), since $g(\anode)$ is a world on a path starting from $\aworld_\init$, 
we obtain $\kripke, g(\anode) \models \existspath \apathformula_j$ by assumption.  
So there exists some infinite path $\apath=\aworld_0\aworld_1\aworld_2\dots$ starting from $g(\anode)$ such that $\kripke,\pi\models\apathformula_j$. 
Define $g(\anode \cdot j \cdot 0^k)\egdef \aworld_{k+1}$, and 
$\atree(\anode\cdot j\cdot 0^k) \egdef \avaluation(\aworld_{k+1})$ for all $k\geq 0$. Note that this implies that $\atree(\anode)\atree(\anode\cdot j) \atree(\anode\cdot j \cdot 0) \atree(\anode\cdot j\cdot 0^2)\dots $
satisfies $\apathformula_j$ and therefore $\atree, \anode \models \existspath \ \apathformula_j$.
By construction, we get $\atree \models \forallpaths \always \existspath \ \apathformula_j$
for all $j \in \interval{1}{\degree-1}$. 
Moreover, by construction of $\atree$,
for all infinite paths $j_1 j_2 \cdots \in \interval{0}{\degree-1}^{\omega}$,
$g(\varepsilon) g(j_1) g(j_1 j_2) \cdots$ is an infinite path from $g(\varepsilon)$,
consequently $\atree \models \existspath \ (\avariable_1 = 0) \wedge
\bigwedge_{j \in \interval{1}{\degree'}}
\forallpaths \ \apathformula_j'$ too.
Hence, $\atree \models \aformula$
  and $\atree$ satisfies $\forallpaths \always \existspath \ \apathformula_i$
  via the direction $i$, for each $i \in \interval{1}{\degree-1}$. 
\end{proof}
Proposition~\ref{proposition-tmp-ctlstarz-direction} justifies our restriction to infinite
trees and to TCA in the rest of this section.  
Next, we will show how to translate  $\CTLStar(\Zed)$ formulae in special form as in (SF) into TCA accepting their corresponding infinite tree models. 
We will start by giving the construction of \emph{word} constraint automata for the underlying $\LTL(\Zed)$ formulas.

\subsection{Constructing Automata for $\LTL(\Zed)$ Formulas}
\label{section-automata-from-ltl-formulas}
In this section, we 
translate  $\LTL(\Zed)$ formulae in  simple form into equivalent \emph{word} constraint  automata
(TCA of degree $1$). 
Adapting the standard automata-based approach for \LTL~\cite{Vardi&Wolper94}, we can show the following proposition. 
 
\begin{prop}
  \label{proposition-ltlz}
  Let $\apathformula$ be an $\LTL(\Zed)$ formula in simple form.
  There is a word constraint automaton $\aautomaton_\apathformula$ such that
  $\alang(\aautomaton_\apathformula) = 
  \set{\aword: \Nat \to \Zed^{\beta} \mid \aword \models \apathformula}
  $,
  and  the following conditions hold.
  \begin{description}
  \item[(I)] The number of locations in $\aautomaton_\apathformula$ is bounded by $\size{\apathformula} \times 2^{2 \times
    \size{\apathformula}}$. 
  \item[(II)] The cardinality of $\delta$ in $\aautomaton_\apathformula$ is in
    $2^{\mathcal{O}(P(\size{\apathformula}))}$ for some polynomial $P(\cdot)$.
  \item[(III)] The maximal size of a constraint in $\aautomaton_\apathformula$ is quadratic
    in $\size{\apathformula}$.
  \end{description}
\end{prop}
In the proof of Proposition~\ref{proposition-ltlz} (see Appendix~\ref{appendix-proof-proposition-ltlz}), the construction
of $\aautomaton_\apathformula$ is similar to the construction from $\CTL(\Zed)$ formulae
by imposing $\degree =1$ and by disqualifying the notion of direction map because there is a
single direction.
In Proposition~\ref{proposition-ltlz} and later in several places, we consider that the
constraint automata accept trees of the form $\interval{0}{\degree-1}^* \to \Zed^{\beta}$ for some
$\degree, \beta \geq 1$ (i.e., by dropping the finite alphabet), assuming implicitly a singleton
alphabet. 
We remark that the automaton $\aautomaton_\apathformula$ may not be deterministic; 
for instance, for  $\apathformula$ equal to $\sometimes \always  \ (\avariable=\avariable')$ there
does not exist any \emph{deterministic} word constraint automaton with B\"uchi acceptance condition such that $\aword \models\apathformula$
iff $\aword \in \alang(\aautomaton_\apathformula)$. 
However, for handling $\CTLStar(\Zed)$ formulae of the form $\forallpaths \  \apathformula$, 
we need the automaton accepting the word models of $\apathformula$ to be deterministic. 
The forthcoming Theorem~\ref{theorem-size-safra-automaton} dedicated to the  determinisation of word constraint automata  (in Section \ref{section-ctlstarz-determinisation-safra}) 
and Proposition~\ref{proposition-ltlz} imply the following result. 
\begin{cor}
  \label{corollary-ltlz-rabin}
  Let $\apathformula$ be an $\LTL(\Zed)$ formula in simple form
  built over the variables
  $\avariable_1, \ldots, \avariable_{\beta}$ and the constants $\adatum_1, \ldots, \adatum_{\alpha}$.
  There exists a deterministic
  Rabin word constraint automaton $\aautomaton_\apathformula$  such that
  $\alang(\aautomaton_\apathformula) = 
  \set{\aword: \Nat \to \Zed^{\beta} \mid \aword \models \apathformula}
  $,
  and the following conditions hold.
  \begin{description} 
  \item[(I)] The number of locations in $\aautomaton_\apathformula$ is bounded by
    $2^{2^{\mathcal{O}(P^{\dag}(\size{\apathformula}))}}$ for some polynomial $P^{\dag}(\cdot)$.
  \item[(II)] The number of Rabin pairs is bounded by $2 \times
     \size{\apathformula} \times 2^{2 \times \size{\apathformula}}$.
   \item[(III)] The cardinality of $\delta$ in $\aautomaton_\apathformula$ is bounded
     by $\card{\sattypes{\beta}} \times 2^{2^{\mathcal{O}(P^{\dag}(\size{\apathformula}))+1}}$.
   \item[(IV)] $\maxconstraintsize{\aautomaton_\apathformula}$ 
     is cubic in $\beta + \max (\lceil log(|\adatum_1|) \rceil ,
    \lceil log(|\adatum_{\alpha}|) \rceil)$, i.e. polynomial in $\size{\apathformula}$. 
  \end{description}
\end{cor}

\subsection{Constructing Automata for $\CTLStar(\Zed)$ Formulas}
\label{section-final-steps}
We are now ready to give the construction of the Rabin TCA accepting precisely the
tree models of a $\CTLStar(\Zed)$ formula in special form. 
For the rest of this section, suppose $\aformula$ is a $\CTLStar(\Zed)$ formula over the variables $\avariable_1,\dots,\avariable_\beta$ and of the form 
$$
\existspath \ (\avariable_1 = 0) \ \wedge \ 
\big(
\bigwedge_{i \in \interval{1}{\degree-1}}
\forallpaths \always \existspath \ \apathformula_i
\big)
\ \wedge \
\big(
\bigwedge_{j \in \interval{1}{\degree'}}
\forallpaths \ \apathformula_j'
\big), 
$$
where the $\apathformula_i$'s and the $\apathformula_j'$'s
are $\LTL(\Zed)$ formulae in simple form. 
Our approach is modular, that is, we construct for each conjunct a corresponding (Rabin) TCA of degree $\degree$
with $\beta$ variables. 
Calling Lemma \ref{lemma-intersection-rtca} then yields the desired Rabin TCA.  
 Let us remark that the alphabet $\aalphabet$ does not play any role here, and we set $\aalphabet$ to the dummy singleton alphabet $\{\arbitraryletter\}$ for the rest of this section.
 
\paragraph{Handling $\existspath (\avariable_1 = 0)$.}
Define the TCA $\aautomaton_0=(\locations,\aalphabet,\degree,\beta,\locations_{\init},\delta,F)$  where
\begin{itemize}
\item $\locations=\{\alocation,\alocation'\}, \locations_{\init}=\{\alocation_0\}, F=\{\alocation'\}$, 
\item $\delta=\{(\alocation, \arbitraryletter, (\avariable_1=0,\alocation'), (\top,\alocation'),\dots,(\top,\alocation')),(\alocation',\arbitraryletter, (\top,\alocation'), \dots,\dots,(\top,\alocation'))\}$.
\end{itemize} 
Clearly, 
 $\alang(\aautomaton_0)=\{\atree:\interval{0}{\degree-1}^*\to\Zed^\beta \mid \atree\models\existspath \ (\avariable_1=0) \}$. 

\paragraph{Handling $\forallpaths \always \existspath \ \apathformula_i$.}
Let us construct, for every $i\in\interval{0}{\degree-1}$, 
a (B\"uchi) TCA $\aautomaton_{i}$ of degree $\degree$ such that 
$\alang(\aautomaton_{i})=\set{\atree: \interval{0}{\degree-1}^* \to \Zed^{\beta} \mid
 \atree\models\forallpaths \always \existspath \ \apathformula_i \mbox{ and $\atree$ satisfies $\forallpaths \always \existspath \ \apathformula_i$ via direction $i$}}$. 
The idea is to construct $\aautomaton_i$ so that it starts off the word constraint  automaton
dedicated to the $\LTL(\Zed)$ formula $\apathformula_i$
at each node $\anode$ of the tree and runs it down the designated path $\anode \cdot i \cdot 0^\omega$ to check whether $\apathformula_i$ actually holds along this path. 
Let $\aautomaton= \triple{\locations}{\aalphabet, \beta}{
  \locations_{\init}, \delta,F}$ be the B\"uchi word constraint automaton 
  such that
  $\alang(\aautomaton)=  
  \set{\aword: \Nat \to \Zed^{\beta} \mid \aword \models \apathformula_i}
  $ (see Proposition~\ref{proposition-ltlz}). 
  Let us define 
  $\aautomaton_i=\triple{\locations'}{\aalphabet,\degree}{\beta, \locations'_\init,\delta',
    F'}$, where 
 \begin{itemize}
 \item $\locations' \egdef \interval{0}{\degree-1} \times (\locations\cup\{\bot\})$,
   where $\bot\not\in \locations$, 
 \item $\locations'_\init \egdef \set{ \pair{0}{\bot} }$,
       $F' \egdef \set{\pair{j}{\alocation} \mid j \neq 0 \mbox{ or } \alocation \in F}$,
 \item The transition relation $\delta'$ is made of tuples of the form
   \[
   ((j,\alocation), \arbitraryletter, (\acons_0,(0, \alocation_0)), \dots, (\acons_{\degree-1},(\degree-1,
   \alocation_{\degree-1})))
   \]
   verifying the conditions below. 
   \begin{enumerate}
   \item $(\alocation_\init,\arbitraryletter,\acons_{i},\alocation_{i})\in\delta$ 
     for some $\alocation_\init \in \locations_\init$
     (starting off  $\aautomaton$ in the $i$$^{\rm th}$ child).
   \item $\alocation_0 = \bot$ and $\acons_0 = \top$ if $\alocation = \bot$
     and 
     $(\alocation,\arbitraryletter,\acons_{0},\alocation_{0})\in \delta$ if $\alocation \in \locations$
     (continuing a run from $\alocation$ of $\aautomaton$ in the $0$$^{\rm th}$ child).

  \item $\alocation_j = \bot$ and $\acons_j= \top$ for all $j \in (\interval{0}{\degree-1} \setminus
    \set{0,i})$.
  \end{enumerate}
 \end{itemize}
 
 \begin{lem}
 \label{lemma-automaton-AGEPhi}
 $\alang(\aautomaton_i) =
 \set{\atree: \interval{0}{\degree-1}^* \to \Zed^{\beta} \mid \atree \models
   \forallpaths \always \existspath \ \apathformula_i \
  \mbox{and $\atree$ satisfies $\forallpaths \always \existspath \ \apathformula_i$ via $i$}},
  $ and 
   \begin{description} 
  \item[(I)] the number of locations is bounded exponential in $\size{\aformula}$, and
  \item[(II)] $\maxconstraintsize{\aautomaton_i}$ is bounded polynomial in $\size{\aformula}$. 
  \end{description}
\end{lem}

The proof of Lemma~\ref{lemma-automaton-AGEPhi} can be found in
Appendix~\ref{appendix-proof-lemma-automaton-AGEPhi}.
Recall that a B\"uchi TCA can be seen as a Rabin TCA with a single Rabin pair.

\paragraph{Handling $\forallpaths \ \apathformula_j'$.}
Let us construct, for every $j\in\interval{0}{\degree'}$, 
a Rabin TCA $\aautomaton'_{j}$ of degree $\degree$ such that 
$\alang(\aautomaton'_{j})=\set{\atree: \interval{0}{\degree-1}^* \to \Zed^{\beta} \mid
 \atree\models\forallpaths \ \apathformula'_j}$, and the number of Rabin acceptance pairs is exponential in $\size{\apathformula'_j}$.   
Let $\apathformula'_j$ be an $\LTL(\Zed)$ formula in simple form built over the variables
$\avariable_1,\dots,\avariable_\beta$ and the constants
$\adatum_1, \ldots, \adatum_{\alpha}$.
By Corollary~\ref{corollary-ltlz-rabin}, there exists a   
deterministic Rabin word constraint automaton
$\aautomaton= \triple{\locations}{\aalphabet, \beta}{
  \locations_{\init}, \delta,\rabinacc}$
  such that
  $\alang(\aautomaton) = 
  \set{\aword: \Nat \to \Zed^{\beta} \mid \aword \models \apathformula_j'}
  $.
  Define the Rabin 
  TCA 
  $\aautomaton'_j=\triple{\locations'}{\aalphabet,\degree}{\beta, \locations'_\init,\delta',
    \mathcal{F'}}$, where
 \begin{itemize}
 \item $\locations' \egdef \locations$; $\locations'_\init \egdef \locations_\init$;
       $\mathcal{F'} \egdef \rabinacc$. 
 \item
   $\delta'$ is made of tuples of the form
   $(\alocation,\aletter, \pair{\acons_0}{\alocation_0},\dots,
   \pair{\acons_{\degree-1}}{\alocation_{\degree-1}})
   $,
 where $(\alocation,\aletter,\acons_i,\alocation_i)\in\delta$ for all $0\leq i<\degree$.
 \end{itemize}

 \begin{lem}
 \label{lemma-specific-formula-APhi}
   $\alang(\aautomaton'_j) = \set{\atree: \interval{0}{\degree-1}^* \to \Zed^{\beta} \mid \atree \models \forallpaths \ \apathformula'_j}$, and   
   \begin{description} 
  \item[(I)] the number of locations is double exponential in $\size{\aformula}$, 
  \item[(II)] the number of Rabin acceptance pairs is exponential in $\size{\aformula}$, 
  \item[(III)] $\maxconstraintsize{\aautomaton'_j}$ is bounded polynomial in $\size{\aformula}$. 
  \end{description}
 \end{lem}
 The proof is by an easy verification thanks to the determinism of
 $\aautomaton$, which is essential here. Typically, Proposition~\ref{proposition-ltlz}
 is not sufficient because determinism of the word automaton is required.
 For example, assuming that all the branches of $\atree: \interval{0}{\degree-1}^* \to \Zed^{\beta}$
 satisfy the formula $\apathformula'_j$, if $\alang(\aautomaton)$ is equal to all words satisfying
 $\apathformula'_j$ and $\aautomaton$ were nondeterministic, then we cannot guarantee that each nondeterministic step
 in  $\aautomaton$
 can lead to acceptance for all possible future branches.
 As a matter of fact, the determinisation 
 construction presented in Section~\ref{section-ctlstarz-determinisation-safra}
is key to design Rabin TCA for formulae of the form $\forallpaths \ \apathformula'_j$, which is why we put so much efforts in developing the corresponding material. 
Moreover, Lemma~\ref{lemma-specific-formula-APhi}(I) is a direct consequence
of Corollary~\ref{corollary-ltlz-rabin}(I),
Lemma~\ref{lemma-specific-formula-APhi}(II) is a direct consequence
of Corollary~\ref{corollary-ltlz-rabin}(II)
and Lemma~\ref{lemma-specific-formula-APhi}(III) is a consequence
of Corollary~\ref{corollary-ltlz-rabin}(IV).
Indeed, 
 $\maxconstraintsize{\aautomaton}$ is polynomial in
  $\size{\apathformula_j'}$ (and therefore
  polynomial in $\size{\aformula}$), 
  the constraints in $\aautomaton_j'$ are those from $\aautomaton$
  and therefore $\maxconstraintsize{\aautomaton_j'}$ is polynomial in
  $\size{\aformula}$.

 \subsection{The final step}
 We are now ready to give the final Rabin TCA $\aautomaton$ such that $\alang(\aautomaton)$ is precisely the set of tree models of $\aformula$.  
 By Lemma~\ref{lemma-intersection-rtca}, 
 there exists a Rabin TCA 
 $\aautomaton = \triple{\locations,\aalphabet,\degree,\beta}{\locations_{\init},\delta}{\rabinacc}$ such that
 \[
  \alang(\aautomaton) = \alang(\aautomaton_0)
  \bigcap_{i \in \interval{1}{\degree-1}} \alang(\aautomaton_i)
  \bigcap_{j \in \interval{1}{\degree'}} \alang(\aautomaton_j'),
  \] 
In Section~\ref{section-TCA-Rabin}, we have seen that nonemptiness of the language $\alang(\aautomaton)$ for $\aautomaton$ 
can be solved in time in 
\[
\hspace*{-0.2in} 
  R_1\big(\card{\locations} \times \card{\delta}
  \times \maxconstraintsize{\aautomaton} \times
  \card{\aalphabet} \times R_2(\beta + \card{\rabinacc}))^{\mathcal{O}(R_2(\beta + \card{\rabinacc}) \times R_3(\degree))}
 \]
  Let us thus evaluate the size of the respective components for $\aautomaton$. 
  First of all, observe that the involved B\"uchi TCA can be seen as Rabin TCA with a single Rabin acceptance pair. 
   Further, $\beta, \degree+\degree' \leq \size{\aformula}$, and $\aalphabet$ is a singleton.

  \begin{itemize}
  \item By Lemmas~\ref{lemma-specific-formula-APhi} and~\ref{lemma-intersection-rtca} 
    and $\degree \leq \size{\aformula}$,
    the number of Rabin pairs in $\aautomaton$ is exponential in $\size{\aformula}$. 
  \item By Lemmas~\ref{lemma-automaton-AGEPhi} and ~\ref{lemma-specific-formula-APhi}, the number of locations in each involved automaton is at most double-exponential in $\size{\aformula}$; and 
    the number of Rabin pairs is only exponential in $\size{\aformula}$. By Lemma~\ref{lemma-intersection-rtca} and $\degree+\degree' \leq \size{\aformula}$,
    the number of locations in $\aautomaton$ is double-exponential in $\size{\aformula}$.
    A similar analysis can be performed for the number of transitions, leading to
    a double-exponential number of transitions.

  \item The maximal size of a constraint appearing in a transition from
    any involved automaton 
    is polynomial
    in $\size{\aformula}$. The maximal size of a constraint in the product
    automaton $\aautomaton$ is therefore polynomial in $\aformula$ too ($\degree+\degree'
    \leq \size{\aformula}$). 
  \end{itemize}
  Putting all results together, the nonemptiness of $\alang(\aautomaton)$ can be
  checked in double-exponential time in $\size{\aformula}$, leading to
  Theorem~\ref{theorem-ctlstarz} below. 
  It answers open questions
  from~\cite{Bozzelli&Gascon06,Carapelle&Kartzow&Lohrey16,CarapelleTurhan16,Labai&Ortiz&Simkus20}, 
  and it is the main result of the paper.

\begin{thm}
  \label{theorem-ctlstarz}
 $\satproblem{\CTLStar(\Zed)}$ is \twoexptime-complete.
\end{thm}

\twoexptime-hardness is inherited from \satproblem{\CTLStar}~\cite[Theorem 5.2]{Vardi&Stockmeyer85}. 
As a corollary, $\satproblem{\CTLStar(\Nat)}$ is also \twoexptime-complete.
Indeed, $\satproblem{\CTLStar(\Nat)}$ can be easily reduced to $\satproblem{\CTLStar(\Zed)}$
in polynomial-time (for instance, one could introduce a fresh variable $\avariable$
enforced to be equal to zero --use of $\forallpaths \always (\avariable = 0)$-- and further enforce that
all the original variables $\avariable_i$ are such that $\avariable_i \geq \avariable$). 
Furthermore, assuming that $\prefix$ is the prefix relation on the set of finite strings $\set{0,1}^*$,
we can use the reduction  from~\cite[Section 4.2]{Demri&Deters16} to get the following result.

\begin{cor} \label{corollary-strings-with-prefix}
  $\satproblem{\CTLStar(\set{0,1}^*, \prefix)}$ is \twoexptime-complete.
\end{cor} 

Actually, the \twoexptime-membership also holds for $\satproblem{\CTLStar(\adatadomain^*, \prefix)}$,
where $\adatadomain$ is an infinite domain, see e.g.~\cite[Section 4.2]{Demri&Deters16}.

As observed earlier, when the concrete domain is $\pair{\Rat}{<,=,(=_{\adatum})_{\adatum \in \Rat}}$,
all the trees in $\alang(\locautomaton)$ are satisfiable
(no need to intersect $\locautomaton$ with a hypothetical $\starautomaton$), and therefore 
$\satproblem{\CTLStar(\Rat)}$ is also in \twoexptime, which is a  result already known from~\cite[Theorem 4.3]{Gascon09}.

\section{Determinisation of B\"uchi Word Constraint Automata}
\label{section-ctlstarz-determinisation-safra}
In this section, 
we present a construction to transform nondeterministic  B\"uchi word constraint
automata to equivalent deterministic Rabin word constraint automata (Theorem \ref{theorem-size-safra-automaton} below, used to establish Corollary~\ref{corollary-ltlz-rabin} in Section~\ref{section-final-steps}). 
The construction is based on Safra's well known construction for
B\"uchi automata~\cite[Theorem 1.1]{Safra89}. 
In our adaptation, we give 
a special attention to the cardinality of the
transition relation and to the size of the constraints in transitions, as these two
parameters are, {\em a priori}, unbounded in constraint automata, but they are essential
to perform the  complexity analysis in Section~\ref{section-final-steps}.

A constraint word automaton $\aautomaton$ is \defstyle{deterministic} whenever
for all locations $\alocation$, letters $\aletter$ and pairs of valuations
$(\vect{z},\vect{z}')\in\Zed^{2\beta}$, there exists in $\aautomaton$ \emph{at most} one transition
$(\alocation,\aletter,\acons,\alocation')$ such that $\Zed\models\acons(\vect{z},\vect{z}')$.

We start with defining the key notion, namely \defstyle{Safra trees}. 
A \defstyle{Safra tree over $\locations$} is a finite tree $\safra$, satisfying the following conditions.
Note that we use the symbol '$\safra$' for trees because later on, such trees shall be locations (a.k.a. \underline{s}tates) of the forthcoming
(deterministic) Rabin word constraint automaton. 
\begin{enumerate}
\item \label{st1}
  $\safra$ is ordered, that is, if a node in the tree has children nodes, then there is a first child node, a second child node etc. 
  In other words, given two sibling nodes, it is uniquely determined which of the two nodes is younger than
  the other (the rightmost sibling is understood as the youngest one). 

\item \label{st2} Every node in $\safra$ has a unique \defstyle{name} from
  the interval $\interval{1}{2 \cdot \card{\locations}}$,
      no two nodes have the same name. 

    \item \label{st3}
      Every node has a \defstyle{label} from $\powerset{\locations} \setminus \set{\emptyset}$. 
We use $\saflab(\safra,J) \subseteq \locations$ to denote the set of labels of a node with name $J$ in $\safra$. 

\item \label{st4}
  The label of a node is a proper superset of the union of the labels of its children nodes.

\item \label{st5}
  Two nodes with the same parent node have disjoint labels.

\item \label{st6}
  Every node is either \defstyle{marked} or \defstyle{unmarked}. 
\end{enumerate}

The proof of the following lemma can be found in
Appendix~\ref{appendix-proof-lemma-safra-node-number}. 
\begin{lem}
\label{lemma-safra-node-number}
A Safra tree over  $\locations$ has at most $\card{\locations}$ nodes. 
\end{lem}

Let us prove that for every (nondeterministic) B\"uchi  word constraint  automaton one can construct a deterministic  Rabin word constraint  automaton
accepting the same language. 
The following result generalizes~\cite[Theorem 1.1]{Safra89} to constraint automata. 
\begin{thm}
\label{theorem-size-safra-automaton}
Let $\aautomaton=(\locations, \aalphabet,\beta, \locations_\init, \delta, F)$
be a B\"uchi word constraint automaton
  involving the constants
  $\adatum_1, \ldots, \adatum_{\alpha}$.
  There is a deterministic Rabin word constraint automaton
  $\aautomaton'=(\locations', \aalphabet, \beta, \locations'_\init, \delta', \rabinacc)$
  such that $\alang(\aautomaton) = \alang(\aautomaton')$ verifying the following quantitative
  properties.
  \begin{description}
  \item[(I)] $\card{\locations'}$ is exponential in $\card{\locations}$
    and the number of Rabin pairs in $\aautomaton'$ is bounded
    by $2 \cdot \card{\locations}$ (same bounds as in~\cite[Theorem 1.1]{Safra89}).
  \item[(II)] The constraints in the transitions are from $\sattypes{\beta}$, are of
    size cubic in $\beta + \max (\lceil log(|\adatum_1|) \rceil ,
    \lceil log(|\adatum_{\alpha}|) \rceil)$ and
    $
    \card{\delta'} \leq \card{\locations'}^2 \times \card{\aalphabet} \times
    \boundsattypes 
    $.
  \end{description}
\end{thm}
\begin{proof}We define the det. Rabin word constraint automaton
$\aautomaton'= (\locations', \aalphabet,\beta, \locations'_\init, \delta', \mathcal{F'})$ as follows.
Essential differences with the construction to prove~\cite[Theorem 1.1]{Safra89} can be found in the definition
of $\delta'$ below, condition (3), as constraints need to be taken into account. 
\begin{itemize}
\item $\locations'$ is the set of all Safra trees over $\locations$.
\item $\locations'_\init \egdef \{\safra_{\locations_\init}\}$, where
  $\safra_{\locations_\init}$ is the Safra tree with a single unmarked node $\anode$ with name $1$ and
  $\saflab(\safra_{\locations_\init},1) \egdef \locations_\init$.
 \item The finite transition relation  $\delta'\subseteq (\locations'\times \aalphabet
   \times\treeconstraints{\beta} \times \locations')$ is defined as follows. 
  Recall that $\sattypes{\beta}$ denotes the set of all satisfiable complete constraints
  over $\avariable_1,\dots,\avariable_\beta,\avariable'_1,\dots,\avariable'_\beta$. 
The transition relation $\delta'$ is defined over this strict subset of $\treeconstraints{\beta}$. 
  Let $\safra$ be a Safra tree, $\aletter\in \aalphabet$, and
  $\acons \in \sattypes{\beta}$. 
We set $(\safra, \aletter, \acons, \safra')\in \delta'$, where
$\safra'$ is obtained from $\safra$, $\aletter$ and $\acons$ by applying the following steps.
\begin{enumerate}
\item Unmark all nodes in $\safra$. Let us use $\safra^{(1)}$ to denote the resulting Safra tree.
\item For every node in $\safra^{(1)}$ with label $P\subseteq \locations$ such that $P\cap F \neq\emptyset$,
  create a new youngest child node with label $P\cap F$. 
  The name of this node is the smallest number in $\interval{1}{2 \cdot \card{\locations}}$
  that is not assigned to any of the other nodes in $\safra^{(1)}$ yet.  
Let us use $\safra^{(2)}$ to denote the resulting Safra tree.
\item Apply the powerset construction to every node in $\safra^{(2)}$,  that is, for every
  node with label $P\subseteq \locations$, replace $P$ by 
  \[\bigcup_{\alocation\in P} \{\alocation'\in \locations \mid \ \mbox{there exists} \
  (\alocation,\aletter,\acons',\alocation')\in
  \delta \text{ such that } \acons \models \acons' \}.
  \]
  Let us use $\safra^{(3)}$ to denote the resulting tree. Note that $\acons \models \acons'$
  can be checked in polynomial-time because $\acons \in \sattypes{\beta}$. 
The tree $\safra^{(3)}$ may not satisfy the condition (5) of Safra trees, but this is only provisionally. 
\item (Horizontal Merge) For every two nodes in $\safra^{(3)}$ with the same parent node and such that $\alocation\in \locations$ is contained in the labels of both nodes, remove $\alocation$ from the labels of the younger node and
  {\em all its descendants}.  
Let us use $\safra^{(4)}$ to denote the resulting Safra tree.
\item Remove all nodes with empty label, {\em except the root node}, yielding  $\safra^{(5)}$.
\item (Vertical Merge) For every node in $\safra^{(5)}$ whose label equals the union of the labels of its children nodes (if there are any),
  remove {\em all descendants of this node} and mark it. 
The resulting Safra tree is $\safra'$. 
\end{enumerate}
\item The set of Rabin pairs $\mathcal{F'}$ is equal to
      $\mathcal{F'}=\{(L_1,U_1), \dots, (L_{2 \cdot \card{\locations}},U_{2 \cdot \card{\locations}})\}$,
  where for all $1\leq J \leq 2 \cdot \card{\locations}$,
  \begin{itemize}
  \item $L_J = \{ \safra\in \locations' \mid \safra \text{ contains a node with name } J  \text{ marked}\}$, 
  \item $U_J = \{\safra\in \locations' \mid \safra \text{ does not contain a node with name } J\}$.
  \end{itemize}
\end{itemize}

It is worth observing that there is a slight abuse of notation here: $\sattypes{\beta}$ is not strictly
a subset of $\treeconstraints{\beta}$ because of the constraints of the form
$\avariable < \adatum_1$ and $\avariable > \adatum_{\alpha}$. However, this can be easily simulated
by adding two new variables $\avariable_{\adatum_1}$ and $\avariable_{\adatum_{\alpha}}$ and to perform
the following changes in the automata:
replace $\avariable < \adatum_1$ by $\avariable < \avariable_{\adatum_1}$,
replace  $\avariable > \adatum_{\alpha}$ by  $\avariable > \avariable_{\adatum_{\alpha}}$
and add to every constraint the conjunct $\avariable_{\adatum_1} = \adatum_1 \wedge
\avariable_{\adatum_{\alpha}} = \adatum_{\alpha}$. This only adds a constant to the maximal constraint
size as well as two to the number of variables. Therefore, this is harmless for all the
complexity results established in the paper. For the sake of readability, we keep below $\beta$
instead of $\beta+2$ for the number of variables, but the reader should keep in mind that
$\sattypes{\beta}$ are  constraints in $\aautomaton'$ at the cost of adding
two auxiliary variables and a constant-size constraint to every constraint.

Let us prove the correctness of the construction of $\aautomaton'$. 
For proving $\alang(\aautomaton)\subseteq \alang(\aautomaton')$, 
let $\aword =(\aletter_1,\vect{z}_1)(\aletter_2,\vect{z}_2)(\aletter_2,\vect{z}_2)\dots$ be an infinite word over $\aalphabet \times \Zed^\beta$ such that $\aword \in \alang(\aautomaton)$. 
Then there is some initialized B\"uchi accepting run 
$\arun:\Nat\to\transitions$ of $\aautomaton $ on $\aword$, say, of the form
\[
(\alocation_{1}, \aletter_1,\acons_1',\alocation_{2})(\alocation_{2}, \aletter_2,\acons_2',
\alocation_{3})(\alocation_{3}, \aletter_3,\acons_3',\alocation_{4})\dots
\]
satisfying
 $\alocation_1\in \locations_\init$, 
 $\Zed \models\acons'_i(\vect{z}_i,\vect{z}_{i+1})$, 
 and there exists some $\alocation\in F$ that appears infinitely often in $\arun$.
 Let us fix such an accepting location and denote it by $\alocation_\acc$. 
We are going to construct an initialized Rabin accepting run of $\aautomaton'$ on $\aword$. 
Set $\safra_1$ to be the Safra tree $\safra_{\locations_\init}$, that is, the Safra tree with a single node
with name $1$ and $\saflab(\safra_{\locations_\init},1)=\locations_\init$ -- the only initial location of $\aautomaton'$. 
We prove that for all $i\geq 1$, 
if $\safra_{i}$ has root node with name $1$ and $\alocation_{i}\in \saflab(\safra_{i},1)$, then 
there exists a unique constraint $\acons_i \in \sattypes{\beta}$ and a unique Safra tree
$\safra_{i+1}$  such that 
\begin{itemize}
\item $\Zed\models \acons_i(\vect{z}_{i},\vect{z}_{i+1})$,
\item $(\safra_{i}, \aletter_i, \acons_i,\safra_{i+1})\in\delta'$,
\item $\safra_{i+1}$ has root node with name $1$ and $\alocation_{i+1}\in \saflab(\safra_{i+1},1)$. 
\end{itemize}
So let $i\geq 1$. 
Let $\acons_i$ be the unique 
constraint in $\sattypes{\beta}$ such that 
$\Zed \models \acons_i(\vect{z}_{i},\vect{z}_{i+1})$ (unicity and existence guaranteed
by Lemma~\ref{lemma-types}). 
  By definition of $\delta'$, there exists a unique Safra tree
  $\safra_{i+1}$ such that $(\safra_{i},\aletter_i, \acons_i,\safra_{i+1})\in\delta'$. 
For proving the third condition, 
recall how $\safra_{i+1}$ is obtained from $\safra_{i}$: 
\begin{itemize}
\item By assumption, $\safra_{i}$ has root node with name $1$ and $\alocation_{i}\in \saflab(\safra_{i},1)$. 
\item By definition, $\alocation_{i}\in \saflab(\safra_i^{(2)},1)$.
\item Recall that $(\alocation_{i},\aletter_i,\acons_i',\alocation_{i+1})\in \delta$ and
$\Zed\models\acons_i'(\vect{z}_{i},\vect{z}_{i+1})$. 
Hence $\acons_i \models \acons_i'$ (by Lemma~\ref{lemma-types}) 
and $\alocation_{i+1}\in \set{\alocation'\in \locations \mid \text{there exists} \ (\alocation_{i}, \aletter_i,\acons', \alocation')\in
\delta, \ \acons_i \models \acons'}$. Then also $\alocation_{i+1}\in \saflab(\safra_i^{(3)}, 1)$. 
\item $\alocation_{i+1}\in \saflab(\safra_i^{(4)}, 1)$ because the root node $1$ has
no siblings. 
\item $\alocation_{i+1}\in \saflab(\safra_i^{(5)}, 1)$ as the label set of the root node is
nonempty. 
\item Finally  $\alocation_{i+1}\in \saflab(\safra_{i+1}, 1)$ because the root node cannot be removed. 
\end{itemize}
Since $\safra_1$'s root node has the name $1$ and $q_1\in\saflab(\safra_1,1)$, 
we just have proved that 
for all $i\geq 1$, 
$\safra_i$'s  root node has name $1$, $q_i\in \saflab(\safra_i,1)$, and there exist a unique constraint $\acons_i$ and
a unique Safra tree $\safra_{i+1}$ such that $(\safra_i,a_i,\acons_i,\safra_{i+1})\in\delta'$ and $\Zed\models\acons_i(\vect{z}_i,\vect{z}_{i+1})$.  
Hence $\arun'$  of the form 
$$(\safra_{1}, \aletter_1, \acons_1,\safra_{2})(\safra_{2}, \aletter_2, \acons_2,\safra_{3})(\safra_{3}, \aletter_3, \acons_3,\safra_{4})\dots$$
 initialized (unique) run 
of
$\aautomaton'$ on $\aword$.  
For proving that $\arun'$ is Rabin accepting, 
we use the following lemma, whose proof can be found in Appendix \ref{appendix-proof-lemma_safra_dir_one_rabin}.
\begin{lem}
\label{lemma_safra_dir_one_rabin}
For every position $i\geq 1$ and every name $1\leq J_i\leq 2 \cdot \card{\locations}$, 
if $\safra_k$ contains a node with name $J_i$ and
$\alocation_k\in \saflab(\safra_k, J_i)$ for all $k\geq i$, 
then 
\begin{enumerate}
\item there exist infinitely many $k\geq i$ such that the node with name $J_i$ is marked in $\safra_k$; or
\item there exists some position $i'\geq i$ and some name $1\leq J_{i'} \leq 2 \cdot \card{\locations}$
  with $J_i\neq J_{i'}$ such that for all $k\geq i'$, the node with name $J_i$ has a child node with name $J_{i'}$ and $\alocation_k\in \saflab(\safra_k, J_{i'})$. 
\end{enumerate}
\end{lem}
Finally, recall that for every $k\geq 0$, the Safra tree $\safra_k$ contains a node with name $1$ and $q_k\in \saflab(\safra_k, 1)$. 
Set $i=1$ and $J=1$. 
We distinguish the following two cases. 
\begin{description}
\item[Case 1] There exist infinitely many $k\geq i$ such that $\safra_k$ contains node with name $J$ marked, that is, $\safra_k\in L_J$. 
This implies that for only finitely many $k\geq 0$ the Safra tree $\safra_k$ does not contain the node with name $J$, that
is, $\safra_k\in U_J$ (indeed, by assumption, for all $k \geq 1$, there is in $\safra_k$ a node with name $J$). 
Hence, the run $\rho'$ is Rabin accepting. 
\item[Case 2] Otherwise, by Lemma \ref{lemma_safra_dir_one_rabin}, 
  there exists some $i'\geq i$ and some $1\leq J'\leq 2 \cdot \card{\locations}$
  with $J\neq J'$ such that for all $k\geq i'$, the node with name $J$ has a child with name $J'$, and $q_k\in \saflab(\safra_k,J')$. We can now repeat the same case distinction, this time for $i=i'$ and $J=J'$. 
\end{description}
Note that after at most $\card{\locations}$ steps, {\bf Case 1} must necessarily be true, as by
Lemma~\ref{lemma-safra-node-number}, every Safra tree has at most $\card{\locations}$ nodes.

For proving $\alang(\aautomaton')\subseteq \alang(\aautomaton)$,
 let $\aword =(\aletter_1,\vect{z}_1)(\aletter_2,\vect{z}_2)(\aletter_2,\vect{z}_2)\dots$ be an
  infinite word over $\aalphabet\times \Zed^\beta$ such that $\aword\in \alang(\aautomaton')$. 
Then there exists some (unique) initialized Rabin accepting run $\arun':\Nat\to\delta'$ of
$\aautomaton'$ on $\aword$, say 
$$(\safra_1, \aletter_1, \acons_1,\safra_2) (\safra_2,\aletter_2,\acons_2,\safra_3) (\safra_3,\aletter_3,\acons,\safra_4) \dots.$$ 
That is, 
$\safra_1=\safra_{\locations_\init}$, 
$\Zed\models\acons_i(\vect{z}_i,\vect{z}_{i+1})$ for all $i\geq 1$, and 
there exists some $1\leq J\leq 2 \cdot \card{\locations}$ such that
$\safra_i\in L_J$ for infinitely many $i\geq 0$ and $\safra_i\in U_J$ for only finitely many $i\geq 0$. 
Fix such a name $1\leq J\leq 2 \cdot \card{\locations}$. 
Then, there must exist some position $i\geq 0$ such that $\safra_k$ contains a
node with name $J$, for all $k\geq i$. 
We let $i_0$ be the minimal such position, and let $i_0< i_1 <i_2 <i_3 \dots$ be the infinitely many positions greater than $i_0$ such that the node with name $J$ is marked in $\safra_{i_k}$, for all $k\geq 1$. 
 
Let us define, for every  $j\geq 1$ and for every $\alocation\in\saflab(\safra_{i_j}, J)$, the set 
$\textup{Acc}(\alocation,j)$  by 
\begin{align*}
\textup{Acc}(\alocation,j) \egdef \{\alocation'\in F \mid \text{there exists} \  i_{j-1} \leq k < i_j  \text{ such that } \alocation'\in \saflab(\safra_k,J)  \\
 \text{ and }  (\alocation_k,\aletter_k,\acons'_k,\alocation_{k+1})\dots(\alocation_{i_j-1},\aletter_{i_j-1},\acons'_{i_j-1},\alocation_{i_j})  \\
\ \text{with} \ \alocation_k=\alocation', \alocation_{i_j}=\alocation \  \text{is a finite run of } \aautomaton \ \text{on } (\aletter_{k},\vect{z}_k)\dots(\aletter_{i_j},\vect{z}_{i_j})  \}; 
\end{align*}
and for every $j\geq 2$, every $i_{j-1} \leq k < i_j $ and every $\alocation \in \saflab(\safra_{k},J)\cap F$, 
we define  
\begin{align*}
  \text{Pre}(\alocation,j,k)\egdef
  \{\alocation'\in \saflab(\safra_{i_{j-1}},J) \mid  \text{there exists some finite run  } 
  (\alocation_{i_{j-1}},\aletter_{i_{j-1}},\acons'_{i_{j-1}},\alocation_{i_{j-1}+1}) \\
  \dots(\alocation_{k-1},\aletter_{k-1},\acons_{k-1},\alocation_k) 
  \text{ with } \alocation' = \alocation_{i_{j-1}}, \alocation=\alocation_k  \text{ of } \aautomaton \text{ on } \\
  (a_{i_{j-1}},\vect{z}_{i_{j-1}})\dots (\aletter_{k},\vect{z}_{k}) \}.
\end{align*} 
In Appendix~\ref{appendix-proof-lemma-safra-correctness}, we prove that these sets $\textup{Acc}(\alocation,j)$
and $\text{Pre}(\alocation,j,k)$ are nonempty. 
Next we define an infinite tree from which we will derive an initialized  B\"uchi accepting run of $\aautomaton$ on $\aword$. 
Each node in this tree has a unique identifier in 
$\{\locations_\init\} \, \cup  \bigcup_{j\geq 1} I_{\acc,j} \cup \bigcup_{j\geq 1} I_{\text{mrkd},j}$, where
\begin{itemize}
\item $I_{\acc,j} \egdef \{(\alocation, j) \mid \alocation\in\textup{Acc}(\alocation',j) \text{ for some } \alocation'\in \saflab(\safra_{i_j}, J)\}$,
\item $I_{\text{mrkd},j}\egdef \{(\alocation,\alocation',j) \mid \alocation\in\saflab(\safra_{i_j},J), \alocation'\in\textup{Acc}(\alocation,j)\}$.
\end{itemize}
The root node of the tree has identifier $\locations_\init$, 
and the root node is parent of a  node $\anode$ iff the identifier of the $\anode$ is $(\alocation_1,1)$ for some $(\alocation_1,1)\in I_{\acc,1}$. 
For all $j\geq 1$, 
node with identifier $(\alocation_j,j)\in I_{\acc,j}$ is parent of every node $\anode$ iff the  identifier of $\anode$ is 
$(\alocation'_j,\alocation_j,j)$ for some $(\alocation'_j,\alocation_j,j)\in I_{\text{mrkd},j}$. 
Finally, 
for all $j\geq 1$, 
a node with identifier $(\alocation'_j,\alocation_j,j)\in I_{\text{mrkd},j}$ is parent of a node $\anode$ iff the identifier of $\anode$ is $(\alocation_{j+1},j+1)$ for some $(\alocation_{j+1},j+1) \in I_{\acc,j+1}$ such that
$\alocation'_j\in \text{Pre}(\alocation_{j+1},j+1,k)$ for some $i_{j}\leq k<i_{j+1}$.
Since the sets $\textup{Acc}(\alocation,j)$ and $\text{Pre}(\alocation,j,k)$
are never empty, 
every node has at least one child, and every node has at most $\card{\locations}^2$ child nodes. 
In particular, this infinite tree is finite branching. 
By K\"onig's Lemma, there must be some infinite path from the root node of the tree, let us say, of the form $\anode_0, \anode_1, \anode_2\dots$ 
By construction, we have
\begin{itemize}
\item $\anode_0 = \locations_\init$, 
\item $\anode_1 = (p_1, 1)$ for some $p_1\in \saflab(\safra_{k_1},J)\cap F$, $i_0\leq k_1<i_1$, 
\item $\anode_2 = (p_2,p_1, 1)$ for some $p_2\in \saflab(\safra_{i_1}, J)$ such that $p_1\in \textup{Acc}(p_2,1)$,
\item $\anode_3 = (p_3, 2)$ for some $p_3\in\saflab(\safra_{k_2}, J)\cap F$, $i_1\leq k_2<i_2$, such that $p_2\in \textup{Pre}(p_3,2,k_2)$, 
\item $\anode_4 = (p_4,p_3,2)$ for some $p_4\in \saflab(\safra_{i_2}, J)$ such that $p_3\in \textup{Acc}(p_4,2)$,
\item etc. 
\end{itemize}
Let us argue that this yields an initialized B\"uchi-accepting run of $\aautomaton$ on $\aword$. 
One can show that
there exists some initialized finite run from $\alocation_\init\in\locations_\init$  to $p_1$ of $\aautomaton$ on
$(\aletter_1,\vect{z}_1) \dots (\aletter_{k_1},\vect{z}_{k_1})$ (as a consequence of
Lemma~\ref{lemma_safra_powerset_path} in Appendix \ref{appendix-proof-lemma-safra-correctness}). 
From $p_1\in \textup{Acc}(p_2,1)$, we obtain  a run from $p_1$ to $p_2$ of $\aautomaton$ on $(\aletter_{k_1},\vect{z}_{k_1})\dots(\aletter_{i_1},\vect{z}_{i_1})$, etc. 
Putting the pieces together, we obtain an infinite initialized run of $\aautomaton$ on $\aword$. 
Since $p_1, p_3, \dots$ are in $F$, the run is B\"uchi accepting. 

We recall that
$\card{\sattypes{\beta}} \leq \boundsattypes$. 
To define $\aautomaton'$, in the constraints,
we need to know how the variables are compared
to the constants, but not necessarily to determine the equality with a value when
the variable is strictly between $\adatum_i$ and $\adatum_{i+1}$ for some $i$. 
So the above bound can be certainly improved but it is handful to use
the set $\sattypes{\beta}$ already defined in this document. 
\end{proof}

As observed by one anonymous reviewer, an alternative approach for determinisation consists in
viewing word constraint automata as B\"uchi automata (or Rabin automata) over types (made of constraints)
and then determinize it using Safra's construction as a blackbox.
Most probably, this approach would work  by introducing
the alphabet of satisfiable types; the semantics of Rabin word automata
would then take care of which infinite sequences of satisfiable types
is satisfiable. Instead, in this section, we worked directly with constraint word
automata and we lifted arguments from Safra's thesis, apart from our
handling of constraints.
On the downside, this may not be optimal size-wise, but 
we prefer to follow our approach to work directly with constraint automata
and to control most of the technical developments involved in the work.

\section{Proving the correctness of the condition $(\texorpdfstring{\newbigstar}{*\textasciicircum C})$}
\label{section-starproperty}
This section is dedicated to the proof of Proposition~\ref{proposition-star-oplus} (see Section~\ref{section-introduction-starproperty}): for every regular locally consistent symbolic tree $\asymtree$, 
    $\asymtree$ satisfies $(\newbigstar)$ iff $\asymtree$ is satisfiable.
This is all the more important as the condition $(\newbigstar)$ is central in our paper.
Below, we assume that $\asymtree: \interval{0}{\degree-1}^* \to \aalphabet \times \sattypes{\beta}$
is a locally consistent symbolic tree. 

Given a finite path $\apath = \pair{\anode_0}{\advar_0} \step{\sim_1} \pair{\anode_1}{\advar_1}
\cdots \step{\sim_n} \pair{\anode_n}{\advar_n}$ in $\newGt$,
its \defstyle{strict length}, written $\slen{\apath}$, is the number of edges labelled by `$<$' in $\apath$, i.e.
$\card{\set{i \in \interval{1}{n} \mid \ \sim_i \ \mbox{equal to} \ <}}$.
Given two nodes $\pair{\anode}{\advar}$ and $\pair{\anode'}{\advar'}$, the \defstyle{strict length}
from $\pair{\anode}{\advar}$ to $\pair{\anode'}{\advar'}$, written
$\slen{\pair{\anode}{\advar},\pair{\anode'}{\advar'}}$, is the supremum of
all the strict lengths of paths from $\pair{\anode}{\advar}$ to $\pair{\anode'}{\advar'}$. 
Though the strict length of any finite path is always finite, $\slen{\pair{\anode}{\advar},\pair{\anode'}{\advar'}}$
may be infinite.
Given $\pair{\anode}{\avariable} \in U_{< \adatum_1}$, its \defstyle{strict length}, written $\slen{\anode,\avariable}$,
is defined as $\slen{\pair{\anode}{\avariable},\pair{\anode}{\adatum_1}}$.
Given $\pair{\anode}{\avariable} \in U_{> \adatum_{\alpha}}$, its \defstyle{strict length}, written
$\slen{\anode,\avariable}$,
is defined as $\slen{\pair{\anode}{\adatum_{\alpha}},\pair{\anode}{\avariable}}$.

\subsection{A simple characterisation for satisfiability}
\label{section-starproperty-characterisation}
First, we establish a few auxiliary results about $\newGt$ that are helpful
in the sequel and that take advantage of the
fact that every type in $\sattypes{\beta}$ is satisfiable.

\begin{lem} \label{lemma-shortcircuit-equality-lessthan}
  Let $\apath = \pair{\anode_0}{\advar_0} \step{\sim_1} \cdots \step{\sim_n} \pair{\anode_n}{\advar_n}$ be a path
  in $\newGt$ such that $\anode_0$ and $\anode_n$ are neighbours.
  If $\sim_1 = \cdots = \sim_n$ is equal to `$=$',
  then
  $\pair{\anode_0}{\advar_0} \step{=} \pair{\anode_n}{\advar_n}$,
  otherwise $\pair{\anode_0}{\advar_0} \step{<} \pair{\anode_n}{\advar_n}$.
\end{lem}
See the proof in Appendix~\ref{appendix-proof-lemma-shortcircuit-equality-lessthan}. 
As a corollary, we get the following lemma. 
\begin{lem} \label{lemma-nocycle}
Let $\apath = \pair{\anode_0}{\advar_0} \step{\sim_1} \cdots \step{\sim_n} \pair{\anode_n}{\advar_n}$ be a path
  in $\newGt$ such that $\pair{\anode_0}{\advar_0} = \pair{\anode_n}{\advar_n}$.
  All $\sim_i$'s are equal to $=$. 
\end{lem}
Below, we provide a simple characterisation for a locally consistent symbolic tree to be satisfiable. Note that
this is different from~\cite[Lemma 5.18]{Labai21} because we have \emph{constant elements}
of the form $\pair{\anode}{\adatum_1}$ and $\pair{\anode}{\adatum_{\alpha}}$ in $\newGt$. 
\begin{lem} \label{lemma-characterisation-satisfiability-goplus}
  Let $\asymtree: \interval{0}{\degree-1}^* \to \aalphabet \times \sattypes{\beta}$ be a locally consistent symbolic  tree. The statements below are
  equivalent.
  \begin{description}
  \item[(I)] $\asymtree$ is satisfiable. 
  \item[(II)] For all $\pair{\anode}{\avariable}$ in $U_{< \adatum_1} \cup U_{> \adatum_{\alpha}}$ in $\newGt$,
              $\slen{\anode,\avariable} < \omega$. 
  \end{description}
\end{lem}

The proof can be found in Appendix~\ref{appendix-proof-lemma-characterisation-satisfiability-goplus}. 
We provide an elementary self-contained proof that is similar to the proof of~\cite[Lemma 7.1]{Demri&DSouza07}. 
We remark that our proof does not use~\cite[Lemma 34]{Carapelle&Kartzow&Lohrey16} as was done in the proof of~\cite[Lemma 5.18]{Labai21}. Observe that $G_{\asymtree}$ from~\cite{Labai21} does not contain constant nodes, which is problematic to apply~\cite[Lemma 34]{Carapelle&Kartzow&Lohrey16}.
Our proof of Lemma~\ref{lemma-characterisation-satisfiability-goplus} is direct, although it would be possible to use~\cite[Lemma 34]{Carapelle&Kartzow&Lohrey16} on our labelled graph $\newGt$.

\subsection{Final steps in the proof of Proposition~\ref{proposition-star-oplus}}
\label{section-starproperty-finalsteps} 
When $\asymtree$ is regular and $\newGt$ has an element with an infinite
strict length, the proof of  Proposition~\ref{proposition-star-oplus} firstly consists in
showing that the existence of paths with infinitely increasing strict lengths can be further constrained
so that the paths are without detours and with a strict discipline to visit the tree structure
underlying $\newGt$. That is why, below, we introduce several restrictions on paths
followed by forthcoming Lemma~\ref{lemma-equivalences-slen} stating the main properties
we aim for. Actually, this approach is borrowed from the proof sketch for~\cite[Lemma 22]{Labai&Ortiz&Simkus20} as well as
from the detailed developments in Labai's PhD thesis~\cite[Section 5.2]{Labai21}.
Hence, the main ideas in the developments
below are due to~\cite{Labai&Ortiz&Simkus20,Labai21}, sometimes adapted and completed to meet our needs (for instance,
we introduce a simple and explicit taxonomy on paths). 

A path $\apath = \pair{\anode_0}{\advar_0} \step{\sim_1} \cdots \step{\sim_n} \pair{\anode_n}{\advar_n}$ is \defstyle{direct}
iff (1)--(3) below hold: 
\begin{enumerate}
\item for all $i \in \interval{1}{n-1}$, $\anode_i \neq \anode_{i+1}$
      (if $n = 1$, then we authorise $\anode_0 = \anode_1$),
\item for all $j \neq i$, $\pair{\anode_i}{\advar_i} \neq
  \pair{\anode_j}{\advar_j}$ (no element is visited twice, see Lemma~\ref{lemma-nocycle}),
\item for all $i < j$, if $\anode_i = \anode_j$ and $n > 1$, then
  $<$ belongs to $\set{\sim_{i+1}, \ldots, \sim_j}$
  (revisiting a node implies some progress in the strict length).
\end{enumerate}
We write $\dslen{\pair{\anode}{\advar},\pair{\anode'}{\advar'}}$ to denote the strict length between
$\pair{\anode}{\advar}$ and $\pair{\anode'}{\advar'}$ based on {\em direct} paths only. Obviously,
$$
\dslen{\pair{\anode}{\advar},\pair{\anode'}{\advar'}} \leq \slen{\pair{\anode}{\advar},\pair{\anode'}{\advar'}} .
$$
Similarly, a path $\apath = \pair{\anode_0}{\advar_0} \step{\sim_1} \cdots \step{\sim_n} \pair{\anode_n}{\advar_n}$
is \defstyle{rooted} iff all the $\anode_i$'s are either descendants of the initial node $\anode_0$,
or equal to it (equivalently, the unique parent node of $\anode_0$, if any, is never visited).
We write $\rslen{\pair{\anode}{\advar},\pair{\anode}{\advar'}}$ (same node $\anode$ for both elements) to denote
the strict length between $\pair{\anode}{\advar}$ and $\pair{\anode}{\advar'}$
based on  {\em direct and rooted} paths only. 
This definition extends to $\rslen{\anode, \advar}$ with $\pair{\anode}{\advar} \in
(U_{< \adatum_1} \cup U_{> \adatum_{\alpha}})$. Obviously,
$$
\rslen{\pair{\anode}{\advar},\pair{\anode}{\advar'}} \leq
\slen{\pair{\anode}{\advar},\pair{\anode}{\advar'}}, \ \ \ \ \ 
\rslen{\anode, \advar} \leq \slen{\anode, \advar} .
$$
The last restriction we consider on direct and rooted paths is to be made of a unique descending part
followed by a unique ascending part. The terms `descending' and `ascending' refer to the underlying tree
structure for $\interval{0}{\degree-1}^*$.
A path $\apath = \pair{\anode_0}{\advar_0} \step{\sim_1} \cdots \step{\sim_n} \pair{\anode_n}{\advar_n}$
is \defstyle{$\downtouparrow$-structured} iff there is $i \in \interval{0}{n}$ such that
$\anode_0 \prefix \anode_1 \prefix \cdots \prefix \anode_i$ and
$\anode_n \prefix \anode_{n-1} \prefix \cdots \prefix \anode_{i}$, where $\prefix$ denotes
the (strict) prefix relation (in this context, we get the parent-child relation in the tree $\interval{0}{\degree-1}^*$).
We write $\uslen{\pair{\anode}{\advar},\pair{\anode}{\advar'}}$ (same node $\anode$ for both elements) to denote
the strict length between $\pair{\anode}{\advar}$ and $\pair{\anode}{\advar'}$ based on  {\em direct, rooted
and $\downtouparrow$-structured} paths only.
This definition extends to $\uslen{\anode, \advar}$ with $\pair{\anode}{\advar} \in
(U_{< \adatum_1} \cup U_{> \adatum_{\alpha}})$. Again,
$\uslen{\pair{\anode}{\advar},\pair{\anode}{\advar'}} \leq \slen{\pair{\anode}{\advar},\pair{\anode}{\advar'}}$
and
$\uslen{\anode, \advar} \leq \slen{\anode, \advar}$.
In Figure~\ref{figure-Gct}, we illustrate  the different kinds of paths on a subgraph of $\newGt$
with $\beta =2$, $\alpha = 1$ and $\adatum_1 = 7$.
\begin{figure}
\begin{center}
\scalebox{.95}{
\begin{tikzpicture}[node distance=2cm,scale=.8]
\tikzset{every state/.style={minimum size=0pt},
dotted_block/.style={draw=black!20!white, line width=1pt, dotted, inner sep=3mm, minimum width=4.3cm, rectangle, rounded corners, minimum height=8mm},       
boxnode/.style={rectangle, draw=black,minimum width=10mm,rounded corners=0.2cm}};

\node[boxnode] (ep1) at (0,0) {$\pair{\varepsilon}{\avariable_1}$}; 
\node[boxnode] (ep2) at (1.8,0) {$\pair{\varepsilon}{\avariable_2}$}; 
\node[boxnode] (ep3) at (3.5,0) {$\pair{\varepsilon}{7}$};

\node [dotted_block] at (1.7,0) {};
\node [dotted_block,minimum width=4.3cm] at (-3.5,-1.5) {};
\node [dotted_block,minimum width=4.3cm] at (2.3,-1.5) {};
\node [dotted_block,minimum width=4.3cm] at (8.3,-1.5) {};
\node [dotted_block,minimum width=5.2cm] at (2.9,-3) {};
\node [dotted_block,minimum width=5.2cm] at (-3.7,-3) {};
\node [dotted_block,minimum width=5.2cm] at (9.5,-3) {};

\node[boxnode] (01) at (-5.2,-1.5) {$\pair{0}{\avariable_1}$};
\node[boxnode] (02) at (-3.4,-1.5) {$\pair{0}{\avariable_2}$};
\node[boxnode] (03) at (-1.7,-1.5) {$\pair{0}{7}$};

\node[boxnode] (11) at (0.6,-1.5) {$\pair{1}{\avariable_1}$};
\node[boxnode] (12) at (2.4,-1.5) {$\pair{1}{\avariable_2}$};
\node[boxnode] (13) at (4.1,-1.5) {$\pair{1}{7}$};

\node[boxnode] (21) at (6.6,-1.5) {$\pair{2}{\avariable_1}$};
\node[boxnode] (22) at (8.4,-1.5) {$\pair{2}{\avariable_2}$};
\node[boxnode] (23) at (10.1,-1.5) {$\pair{2}{7}$};

\node[boxnode] (011) at (-5.8,-3) {$\pair{0\!\cdot\!1}{\avariable_1}$};
\node[boxnode] (012) at (-3.6,-3) {$\pair{0\!\cdot\!1}{\avariable_2}$};
\node[boxnode] (013) at (-1.5,-3) {$\pair{0\!\cdot\!1}{7}$};

\node[boxnode] (111) at (0.8,-3) {$\pair{1\!\cdot\!1}{\avariable_1}$};
\node[boxnode] (112) at (3,-3) {$\pair{1\!\cdot\!1}{\avariable_2}$};
\node[boxnode] (113) at (5.1,-3) {$\pair{1\!\cdot\!1}{7}$};

\node[boxnode] (211) at (7.4,-3) {$\pair{2\!\cdot\!1}{\avariable_1}$};
\node[boxnode] (212) at (9.6,-3) {$\pair{2\!\cdot\!1}{\avariable_2}$};
\node[boxnode] (213) at (11.7,-3) {$\pair{2\!\cdot\!1}{7}$};

\path[->] (22) edge [bend right,left] node  {} (211);
\path[->] (211) edge [bend left,left] node  {} (212);
\node at (9,-2.2) {\tiny{\bf non-direct path}};

\path[->,line width=1.1pt] (02) edge [bend left,right] node[pos=0.7,left]  {$<$} (013);
\path[->,line width=1.1pt] (013) edge  node  {} (01);
\path[->,line width=1.1pt] (01) edge [bend left,left] node  {} (ep1);
\node at (-2,-0.7) {\tiny{\bf direct, non-rooted path}};

\path[->,line width=1.9pt] (ep1) edge  node  {} (13);
\path[->,line width=1.9pt] (13) edge  [bend left,right] node  {$=$} (113);
\path[->,line width=1.9pt] (113) edge [bend left,left] node[pos=0.8,right]  {$<$} (12);
\path[->,line width=1.9pt] (12) edge [bend right,left] node  {} (111);
\node at (5.5,-0.7) {\tiny{\bf direct, rooted, non $\downtouparrow$-structured path}};

\end{tikzpicture}
}
\end{center}
\caption{Different kinds of paths in $\newGt$}
\label{figure-Gct}
\end{figure}
The equivalence between (II)(1) and (II)(4) in Lemma~\ref{lemma-equivalences-slen} below,
is the main property to prove Proposition~\ref{proposition-star-oplus}.
The proof for the equivalence between (II)(1) and (II)(3)
(resp.  between (II)(3) and (II)(4)) is inspired from
the proof sketch of~\cite[Lemma 5.20]{Labai21}
(resp. from the proof  of~\cite[Lemma 5.27]{Labai21}).
In both cases, we provide several substantial adjustments. 
\begin{lem} \label{lemma-equivalences-slen}\ 
\begin{description}
\item[(I)] If $\slen{\pair{\anode}{\advar},\pair{\anode'}{\advar'}} < \omega$, then
  $$
  \frac{\slen{\pair{\anode}{\advar},\pair{\anode'}{\advar'}}}{\beta+2}
  \leq
  \dslen{\pair{\anode}{\advar},\pair{\anode'}{\advar'}}
  \leq 
  \slen{\pair{\anode}{\advar},\pair{\anode'}{\advar'}} 
  $$ 
 
\item[(II)] The statements below are equivalent.
  \begin{enumerate}
  \item There is $\pair{\anode}{\avariable}$ such that $\slen{\anode,\avariable} = \omega$.
  \item There is $\pair{\anode}{\avariable}$ such that $\dslen{\anode,\avariable} = \omega$.
  \item There is $\pair{\anode}{\avariable}$ such that $\rslen{\anode,\avariable} = \omega$.
  \item There is $\pair{\anode}{\avariable}$ such that $\uslen{\anode,\avariable} = \omega$. 
  \end{enumerate}
\end{description}
\end{lem}
\begin{proof}
  \newcommand{\maxvisits}{\gamma}
(I) Since direct paths are paths and the strict length between two elements is computed
  by taking the supremum, we get
  $$\dslen{\pair{\anode}{\advar},\pair{\anode'}{\advar'}}
  \leq 
  \slen{\pair{\anode}{\advar},\pair{\anode'}{\advar'}}.
  $$
  In order to show that
  $$
  \frac{\slen{\pair{\anode}{\advar},\pair{\anode'}{\advar'}}}{\beta+2}
  \leq
  \dslen{\pair{\anode}{\advar},\pair{\anode'}{\advar'}},
  $$
  we establish that for any path
  $\apath = \pair{\anode_0}{\advar_0} \step{\sim_1} \cdots \step{\sim_n} \pair{\anode_n}{\advar_n}$,
  there is a direct path $\apath'$ from $\pair{\anode_0}{\advar_0}$ to
  $\pair{\anode_n}{\advar_n}$ with $\slen{\apath'} \geq \frac{\slen{\apath}}{\beta+2}$.

  Given a path $\apath$ of the above form, we transform it a finite amount of times
  leading to a final path $\apath'$.
  If there are $i < j$ such that $\anode_i = \anode_{i+1} = 
  \cdots = \anode_j$ (and no way to extend the interval of indices $\interval{i}{j}$
  while satisfying the sequence
  of equalities),
  we replace
  the subpath $\pair{\anode_i}{\advar_i}  \cdots  \pair{\anode_j}{\advar_j}$ in $\apath$ by a shortcut.
  Note that $j-i \leq \beta+1$ because
  $\domnewGt =
  \interval{0}{\degree-1}^*
  \times \set{\avariable_1,\dots,\avariable_\beta\}\cup\{\adatum_1,\adatum_\alpha}$
  and $\beta+2 = \card{\set{\avariable_1,\dots,\avariable_\beta\}\cup\{\adatum_1,\adatum_\alpha}}$. 
  Several cases need to be distinguished.
  \begin{itemize}
  \item If $i > 0$, then $\pair{\anode_{i-1}}{\advar_{i-1}}$ and $\pair{\anode_j}{\advar_j}$
    are neighbours and therefore by Lemma~\ref{lemma-shortcircuit-equality-lessthan}, we can safely replace
    $\pair{\anode_{i-1}}{\advar_{i-1}}  \cdots  \pair{\anode_j}{\advar_j}$ by
    $\pair{\anode_{i-1}}{\advar_{i-1}} \step{\sim} \pair{\anode_j}{\advar_j}$ for some $\sim \in \set{=,<}$
    leading to a new value
    for $\apath$. The label $\sim$ on the edge depends whether $<$ occurs in the subpath
    from $\pair{\anode_{i-1}}{\advar_{i-1}}$ to $\pair{\anode_j}{\advar_j}$.
    Moreover, the strict length of the new path is decreased by at most $\beta+1$
    (and we do not meet again an element with the node $\anode_i$) and divided by at most $\beta+2$. 
  \item If $i = 0$ and $j < n$,
        then $\pair{\anode_{0}}{\advar_{0}}$ and $\pair{\anode_{j+1}}{\advar_{j+1}}$
    are neighbours and therefore by Lemma~\ref{lemma-shortcircuit-equality-lessthan}, we can safely replace
    $\pair{\anode_{0}}{\advar_{0}}  \cdots  \pair{\anode_{j+1}}{\advar_{j+1}}$ by
    $\pair{\anode_{0}}{\advar_{0}} \step{\sim} \pair{\anode_{j+1}}{\advar_{j+1}}$ for some $\sim \in \set{=,<}$
    leading to a new value
    for $\apath$. 
   Again, the strict length of the new path is decreased by at most $\beta+1$ and divided by at most $\beta+2$.
  \item Finally, if $i=0$ and $j = n$, then by Lemma~\ref{lemma-shortcircuit-equality-lessthan},
    we  replace
    $\pair{\anode_{0}}{\advar_{0}}  \cdots  \pair{\anode_{n}}{\advar_{n}}$ by
    $\pair{\anode_{0}}{\advar_{0}} \step{\sim} \pair{\anode_{n}}{\advar_{n}}$ for some $\sim \in \set{=,<}$
    leading to a new value
    for $\apath$. Again, the strict length of the new path is decreased by at most $\beta+1$
    and divided by at most $\beta+2$.
  \end{itemize}
  
  In the second step, we proceed as follows.
  If there are $i+1 < j$ such that $\pair{\anode_i}{\advar_i} = \pair{\anode_j}{\advar_j}$
  ($j-i > 1$ because otherwise we could apply the previous transformation),
  then by Lemma~\ref{lemma-nocycle}, the subpath
  $\pair{\anode_i}{\advar_i} \cdots \pair{\anode_j}{\advar_j}$ in $\apath$ contains
  no edge labelled by $<$ and therefore by Lemma~\ref{lemma-shortcircuit-equality-lessthan},
   we can safely replace
    $\pair{\anode_{i}}{\advar_{i}}  \cdots  \pair{\anode_j}{\advar_j}$ by
    $\pair{\anode_{i}}{\advar_{i}}$, 
    leading to a new value
    for $\apath$.
    Note that the strict length of the new path is unchanged.
    Similarly, if there are $i+1 < j$ such that $\anode_i = \anode_j$ and
    $\sim_{i+1} = \cdots = \sim_j$ is equal to '$=$', we can safely remove
    the subpath $\pair{\anode_i}{\advar_i} \cdots \pair{\anode_j}{\advar_j}$
    by Lemma~\ref{lemma-shortcircuit-equality-lessthan}
    along the lines of the previous transformations (easy details are omitted). 
    We proceed as many times as necessary until the path $\apath$ becomes direct.
 As a consequence, the strict length of the final direct path is at least equal to the
 strict length of the initial value for $\apath$ divided by $\beta+2$.
 Indeed, in the first round of transformations, a subpath with $\beta+2$ strict edges
 may be replaced by a path with only one strict edge.

  \noindent
  (II) Obviously, (4) implies (3), (3) implies (2) and (2) implies (1). It remains to prove
  that (1) implies (2), (2) implies (3) and (3) implies (4).
  By the proof of (I), we get that (1) implies (2).

  \noindent
  \fbox{Let us show that (2) implies (3).} Let
  $\pair{\anode}{\avariable}$ be an element of $\newGt$ such that
  $\dslen{\anode,\avariable} = \omega$.
  The element 
  $\pair{\anode}{\avariable}$ belongs to $U_{< \adatum_1} \cup U_{> \adatum_{\alpha}}$ in $\newGt$.
  Suppose that $\pair{\anode}{\avariable}$ is in $U_{< \adatum_1}$
  (we omit the other case as it admits a  similar analysis).
  By definition of $\dslen{\anode,\avariable}$, there is a family of direct paths
  $(\apath_{i})_{i \in \Nat}$ from $\pair{\anode}{\avariable}$ to $\pair{\anode}{\adatum_1}$
  such that $\slen{\apath_i} \geq i$.

  For each path in the family $(\apath_{i})_{i \in \Nat}$, below, we define its 
  \defstyle{maximal entrance signature}
  as an element of the set 
  $ES \egdef \interval{-1}{\degree-1} \times
  \DVAR{\beta}{\adatum_1}{\adatum_{\alpha}}^2$.
  Since the set 
  $ES$
  is finite,
  there is 
   $\triple{j}{\advar}{\advar'} \in ES$
  such that
  for infinitely many $i$, $\apath_i$ has 
  maximal entrance signature
  $\triple{j}{\advar}{\advar'}$, which
  shall allow us to conclude easily. Firstly, let us provide a few definitions. 

  Let $\apath = \pair{\anode_0}{\advar_0} \cdots \pair{\anode_m}{\advar_m}$ be a path
  in $(\apath_{i})_{i \in \Nat}$ of strict length at least $3(\beta + 1)$. 
  This means that $\pair{\anode}{\avariable} = \pair{\anode_0}{\advar_0}$
  and $\pair{\anode}{\adatum_1} = \pair{\anode_m}{\advar_m}$.
  Since $\apath$ is direct, there are at most $\beta+2$ positions in $\apath$ visiting
  an element on the node $\anode$ (called an \defstyle{$\anode$-element} later on).
  Such positions are written $i_0 < i_1 < \cdots < i_s$
  with $i_0 = 0$ and $i_s = m$. 
  For each $h \in \interval{0}{s-1}$, we write $\apath_h^{\dag}$ to denote the subpath
  of $\apath$ below:
  $$
  \apath_h^{\dag} \egdef \pair{\anode_{i_h+1}}{\advar_{i_h+1}} \cdots \pair{\anode_{i_{h+1}-1}}{\advar_{i_{h+1}-1}}.
  $$
  The
  \defstyle{entrance signature}
  of $\apath_h^{\dag}$, 
   written $es(\apath_h^{\dag})$,
  is the triple 
  $\triple{j}{\advar}{\advar'} \in ES$
  defined as follows:
  \begin{itemize}
  \item $\advar \egdef \advar_{i_h+1}$ (entrance term) and $\advar' \egdef \advar_{i_{h+1}-1}$ (exit term).
  \item If $\anode_{i_h+1} = \anode \cdot j'$ for some $j' \in \interval{0}{\degree-1}$,
    then $j \egdef j'$. Otherwise $j \egdef -1$. The value $j$ is a direction, where $-1$
    is intended to mean ``to the parent node''
    and $j \in \interval{0}{\degree-1}$ ``to the $j$$^{\rm th}$ child node''.
    Note also that in the degenerate case $i_{h}+1 = i_{h+1}-1$, the path
    $\apath_h^{\dag}$ is made of a single element and $\advar = \advar'$.  
  \end{itemize}
  A few useful properties are worth being stated.
  \begin{itemize}
  \item $\anode_{i_h+1} = \anode_{i_{h+1}-1}$, $\apath_h^{\dag}$ does not contain
    an $\anode$-element and $\apath_h^{\dag}$ is rooted. 
  \item By Lemma~\ref{lemma-shortcircuit-equality-lessthan},
    $\pair{\anode}{\avariable} \step{\sim} \pair{\anode_{i_h+1}}{\advar_{i_h+1}}$ 
    and 
    $\pair{\anode_{i_{h+1}-1}}{\advar_{i_{h+1}-1}} \step{\sim'} \pair{\anode}{\adatum_1}$
    for some $\sim,\sim' \in \set{<,=}$.
  \item For all $h \neq h'$,
    if
    $es(\apath_h^{\dag}) = \triple{j_1}{\advar_1}{\advar'_1}$,
    $es(\apath_{h'}^{\dag}) = \triple{j_2}{\advar_2}{\advar'_2}$
    and $j_1 = j_2$, then
    $\advar_1 \neq \advar_1'$ and $\advar_2 \neq \advar_2'$ ($\apath$ is direct).
    As a consequence, $s \leq \beta + 1$
    and there are no $\apath_h^{\dag}$ and
    $\apath_{h'}^{\dag}$ with $h \neq h'$ having the same
    entrance signature. 
  \end{itemize}
  As a conclusion, there is $h$ such that the strict length of $\apath_h^{\dag}$ is at least
  $$
  \left \lceil
  \frac{\slen{\apath} - 2(\beta+1)
  }{(\beta+1)
    }
  \right \rceil 
  $$
  When $\slen{\apath} \geq 3(\beta+1)$, the above value is at least one. 
  Indeed, we subtract $2(\beta+1)$ from $\slen{\apath}$
  to take into account the edges from some $\anode$-element to
  the elements $\pair{\anode_{i_h+1}}{\advar_{i_h+1}}$,
  and from the elements $\pair{\anode_{i_{h+1}-1}}{\advar_{i_{h+1}-1}}$
  to some $\anode$-element.
  The maximal entrance of $\apath$ is defined as an entrance signature $es(\apath_h^{\dag})$
  such that $\slen{\apath_h^{\dag}}$ is maximal (we can also fix an arbitrary linear
  ordering on
  $ES$
  in case maximality of the -- direct -- strict lengths is witnessed
  strictly more than once).
  As announced earlier, there is
  $\triple{j}{\advar}{\advar'} \in ES$
  such that
  for infinitely many $i$ (say belonging to the infinite set $\aset
  \subseteq [3(\beta+1),+\infty)$), $\apath_i$ has maximal type $\triple{j}{\advar}{\advar'}$.

  \begin{enumerate}[label=(\alph*)]
  \item If $j \in \interval{0}{\degree-1}$, we have $\rslen{\pair{\anode \cdot j}{\advar},
    \pair{\anode \cdot j}{\advar'}} = \omega$ because
    $$
    \rslen{\pair{\anode \cdot j}{\advar},
      \pair{\anode \cdot j}{\advar'}} \geq
      \sup_{i \in \aset}
      \left \lceil
  \frac{i - 2(\beta+1)
  }{(\beta+1)
    }
  \right \rceil 
  $$
  Since $\pair{\anode}{\avariable} \step{\sim} \pair{\anode \cdot j}{\advar}$
  and $\pair{\anode \cdot j}{\advar'} \step{\sim'} \pair{\anode}{\adatum_{1}}$
  for some $\sim, \sim' \in \set{<,=}$, we get that $\rslen{\anode, \avariable} = \omega$ too. 
\item If $j = -1$, then let $\anode'$ be the unique parent of $\anode$. 
           We have $\dslen{\pair{\anode'}{\advar},
             \pair{\anode'}{\advar'}} = \omega$
           (the maximal entrance signature is $\triple{j}{\advar}{\advar'}$)
           and
           $\pair{\anode'}{\advar'} \step{\sim} \pair{\anode}{\adatum_{1}}$ for
           some $\sim \in \set{<,=}$. Since $\anode$ and $\anode'$ are neighbour nodes,
           we also get that  $\pair{\anode'}{\advar'} \step{\sim} \pair{\anode'}{\adatum_{1}}$
           by construction of $\newGt$.
           Moreover, 
           $\advar$ is necessarily a variable. Indeed, 
           by Lemma~\ref{lemma-shortcircuit-equality-lessthan}, $\advar$ is distinct from $\adatum_1$
           and by satisfiability of the elements in $\sattypes{\beta}$, $\advar$ is distinct from $\adatum_{\alpha}$. 
           Consequently, $\dslen{\anode',\advar} = \omega$ and we can apply the above
           construction but $\length{\anode'} < \length{\anode}$.
           This means that at some point, we must meet  the case (a), which allows
           us eventually to identify some element $\pair{\anodebis}{\avariablebis}$ such that
           $\rslen{\anodebis,\avariablebis} = \omega$.
           Another way to proceed would be to provide a proof by induction on
           $\length{\anode}$.
           The base case $\length{\anode} = 0$  corresponds to  $\anode = \varepsilon$
           for which only case (a) can hold.
  \end{enumerate}

\noindent
\fbox{Let us show that (3) implies (4).}
 Let
  $\pair{\anode}{\avariable}$ be an element of $\newGt$ such that
  $\rslen{\anode,\avariable} = \omega$.
  The element 
  $\pair{\anode}{\avariable}$ belongs to $U_{< \adatum_1} \cup U_{> \adatum_{\alpha}}$ in $\newGt$.
  Suppose that $\pair{\anode}{\avariable}$ is in $U_{> \adatum_{\alpha}}$
  (we omit the other case as it admits a similar analysis).
  By definition of $\rslen{\anode,\avariable}$, there is a family of direct and rooted paths
  $(\apath_{i})_{i \in \Nat}$  from $\pair{\anode}{\adatum_{\alpha}}$ to $\pair{\anode}{\avariable}$ 
  such that $\slen{\apath_i} \geq i$.
  Below, we show that $\uslen{\anode,\avariable} = \omega$ (and therefore no need
  to witness unbounded strict length on another element as it may happen to prove that (2) implies (3)).
  To do so, for every $j$, we explain how to construct
  a directed, rooted and $\downtouparrow$-structured path
  from $\pair{\anode}{\adatum_{\alpha}}$ to $\pair{\anode}{\avariable}$ of strict length
  at least $j$. We use the auxiliary family $(L_{i,j})_{i,j \in \Nat}$ of integers defined
  recursively.
  \begin{itemize}
  \item For all $i \in \Nat$, $L_{i,0} \egdef i$ and for all $j \in \Nat$,
        $L_{i,j+1} \egdef \left \lceil \frac{L_{i,j} -  2 (\beta+1)}{(\beta+1)} \right \rceil$.  
  \end{itemize}
  Here are properties that can be easily shown. 
\begin{description}
\item[(PL1)] For all $i < i'$ and $j$, $L_{i,j} \leq L_{i',j}$.
\item[(PL2)] For all $i$ and $j < j'$,  $L_{i,j} \geq L_{i,j'}$. 
\end{description}
Moreover, partly based on (PL1), by induction on $j$, it is easy to show that for each
fixed $j$, 
  $\sup \set{L_{i,j} \mid i \in \Nat} = \omega$.
  Indeed, for the base case $\sup \set{L_{i,0} \mid i \in \Nat} =
  \sup \set{i \mid i \in \Nat} = \omega$. Suppose the property is true
  for $j$, i.e. $\sup \set{L_{i,j} \mid i \in \Nat} = \omega$.
  Substracting a constant and dividing by a constant preserves
  the limit behavior, consequently
  $\sup \set{\frac{L_{i,j}  -  2 (\beta+1)}{(\beta+1)} \mid i \in \Nat}
  = \omega$ too
  and therefore $\sup \set{L_{i,j+1} \mid i \in \Nat} = \omega$.
  As a consequence, for all $j \in \Nat$, there is $N_j$ such that
  $L_{N_j,j} \geq 1$. For instance, $N_0$ can take the value 1, $N_1$ the value
  $3 (\beta+1)$ and $N_2$ the value $3 (\beta+1)^2 + 2 (\beta+1)$.
  However, we shall require a bit more properties about the $N_j$'s
  using simple properties about $(L_{i,j})_{i,j \in \Nat}$.
  \begin{description}
  \item[(PL3)] For all $j$, $N_{j+1} >  N_{j} \geq 2$ and
    $\left \lceil \frac{N_{j+1} -  2 (\beta+1)}{(\beta+1)} \right \rceil \geq N_j$.
  \item[(PL4)] For all $M, j \in \Nat$,
    $M \geq N_j$ implies (if $M,j \geq 1$, then $M-1 \geq N_{j-1}$ and if $M,j \geq 2$, then
    $M-2 \geq N_{j-2}$).     
  \end{description}
  Thanks to (PL1), we can choose the $N_j$'s as large as possible, which allows us to establish
  (PL3). By contrast, the satisfaction of (PL4) is a consequence of (PL3).
  By way of example, $M, j \geq 1$ and $M \geq N_j$ imply $M \geq N_j \geq N_{j-1} + 1$ and
  therefore $M-1 \geq N_{j-1}$. 

  Let  $\apath = \pair{\anode_0}{\advar_0} \step{\sim_1} \cdots \step{\sim_m}
  \pair{\anode_m}{\advar_m}$ be a direct and rooted path from
  $\pair{\anode}{\adatum_{\alpha}}$ to $\pair{\anode}{\avariable}$
  of strict length at least $N_j$ for some $j \geq 1$.
  By (PL3), we have $\left \lceil \frac{N_{j} -  2 (\beta+1)}{(\beta+1)} \right
  \rceil \geq N_{j-1} \geq 2$.
  For instance, $\apath$ can take
  the value $\apath_{N_j}$ from the family $(\apath_{i})_{i \in \Nat}$.
  In the developments below, we build a path $\apath'$  
  from $\pair{\anode}{\adatum_{\alpha}}$ to $\pair{\anode}{\avariable}$
  that is of strict length at least $j$ and that is direct, rooted and
  $\downtouparrow$-structured. 
  To do so, we maintain three auxiliary paths while guaranteing the satisfaction of an invariant.
  \begin{itemize}
  \item $\apath_{des}$ is a descending and direct  path from
        $\pair{\anode}{\adatum_{\alpha}}$. Its initial value is $\pair{\anode}{\adatum_{\alpha}}$.
  \item $\apath_{asc}$ is an ascending and direct path to $\pair{\anode}{\avariable}$.
    Its initial value is $\pair{\anode}{\avariable}$.
  \item $\apath^{\dag}$ is a direct and rooted path from the last
    element of $\apath_{des}$ (say $\pair{\anode^{\dag}}{\advar_1^{\dag}}$) to the
    first element of $\apath_{asc}$ (say $\pair{\anode^{\dag}}{\advar_2^{\dag}}$) and is a
    (consecutive) subpath
    of the initial path $\apath$. The first and last elements belong therefore to the
    same node $\anode^{\dag}$.
    The initial value for $\apath^{\dag}$ is $\apath$. 
    By slight abuse, below we assume that $\apath^{\dag}$ can be also written
    $\pair{\anode_0}{\advar_0} \step{\sim_1} \cdots \step{\sim_m}
    \pair{\anode_m}{\advar_m}$.
  \item Let $C$ be equal to $\slen{\apath_{des}} + \slen{\apath_{asc}}$.
    If $C \geq j$, then there is a direct, rooted and $\downtouparrow$-structured path
    from  $\pair{\anode}{\adatum_{\alpha}}$ to $\pair{\anode}{\avariable}$
    passing via $\pair{\anode^{\dag}}{\advar_1^{\dag}}$. It uses the edges from
    $\apath_{des}$ and $\apath_{asc}$.
    Otherwise ($C < j$), we maintain $\slen{\apath^{\dag}} \geq N_{j - C}$.
    Consequently, if $\apath^{\dag}$ is made of elements from the same node, then
    $\slen{\apath^{\dag}} \leq 1$ and therefore $C \geq j$. 
    
  \end{itemize}
  We transform these three paths $\apath_{des}$, $\apath_{asc}$ and  $\apath^{\dag}$  with a process that terminates because
  the length of $\apath^{\dag}$ decreases strictly after  each step.
  Moreover, we shall verify that the invariant
  holds after each step (sometimes after a step, $\slen{\apath^{\dag}}$ and $C$ are unchanged).
  Initially, we have $\slen{\apath^{\dag}} \geq N_{j - C}$ because
    $C =0 < j$, $\apath^{\dag} = \apath$ and $\slen{\apath} \geq N_j$. 

  Let us define the
  \defstyle{maximal entrance signature}
  of $\apath^{\dag}$ similarly to what is
  done earlier. Here,
  $ES \egdef \interval{0}{\degree-1} \times
  \DVAR{\beta}{\adatum_1}{\adatum_{\alpha}}^2$
  (no value $-1$
  in the interval because $\apath^{\dag}$ is rooted). 
  Since $\apath^{\dag}$ is direct, there is at most $\beta+2$ positions in $\apath^{\dag}$ visiting
  an element on the node $\anode^{\dag}$.
  Such positions are written $i_0 < i_1 < \cdots < i_s$
  with $i_0 = 0$, $i_s = m$ (and $s \leq (\beta+1)$).
  For each $h \in \interval{0}{s-1}$, we write $\apath_h^{\dag \dag}$ to denote the subpath
  of $\apath^{\dag}$ below:
  $$
  \apath_h^{\dag \dag} \egdef
  \pair{\anode_{i_h+1}}{\advar_{i_h+1}} \cdots \pair{\anode_{i_{h+1}-1}}{\advar_{i_{h+1}-1}}.
  $$
  The
  \defstyle{entrance signature}
  of $\apath_h^{\dag \dag}$,
  written $es(\apath_h^{\dag})$,
  is the triple
  $\triple{j}{\advar}{\advar'} \in ES$ defined as follows:
  $\advar \egdef \advar_{i_h+1}$, $\advar' \egdef \advar_{i_{h+1}-1}$ and 
  $\anode_{i_h+1} = \anode \cdot j$ for some $j \in \interval{0}{\degree-1}$.
  A few useful properties are worth being stated.
  \begin{itemize}
  \item $\anode_{i_h+1} = \anode_{i_{h+1}-1}$, $\apath_h^{\dag \dag}$ does not contain
    an $\anode^{\dag}$-element and $\apath_h^{\dag \dag}$ is rooted. 
  \item By Lemma~\ref{lemma-shortcircuit-equality-lessthan},
    $\pair{\anode^{\dag}}{\advar_1^{\dag}} \step{\sim^{\dag}_{1,h}}
    \pair{\anode_{i_h+1}}{\advar_{i_q+1}}$ for some $\sim^{\dag}_{1,h}
    \in \set{<,=}$ and  \\ 
    $\pair{\anode_{i_{h+1}-1}}{\advar_{i_{h+1}-1}} \step{\sim^{\dag}_{2,h}} \pair{\anode^{\dag}}{\advar_2^{\dag}}$
    for some $\sim^{\dag}_{2,h} \in \set{<,=}$.
  \end{itemize}
  There is $h$ such that the strict length of $\apath_h^{\dag \dag}$ is at least
  $
  \lceil
  \frac{\slen{\apath^{\dag}} - 2 (\beta+1)
  }{(\beta+1)
    }
  \rceil 
  $.
  Indeed, $s \leq \beta +1$ and there are at most $\beta+1$ paths of the form
  $\apath_h^{\dag \dag}$ with at most $2(\beta+1)$ connecting edges. 
  
  \begin{itemize}
  \item If $s = 1$, then $\anode_1 = \anode_{m-1}$, there is a single $\apath_0^{\dag \dag}$ and
    the paths are updated as follows.
    \begin{itemize}
    \item $\apath_{des}$ becomes $\apath_{des} \step{\sim^{\dag}_{1,h}}
      \pair{\anode_1}{\advar_1}$,
    $\apath_{asc}$ becomes $\pair{\anode_{m-1}}{\advar_{m-1}}
      \step{\sim^{\dag}_{2,h}} \apath_{asc}$.
    \item $\apath^{\dag}$ takes the value  $\apath_0^{\dag \dag}$ if
      $\apath_0^{\dag \dag}$ is made of elements from distinct nodes.
    \end{itemize}
    Note that the strict length of $\apath^{\dag}$ decreases by at most two
   (can be zero if $\sim^{\dag}_{1,h}$ and $\sim^{\dag}_{2,h}$ are both the equality)
    and thanks to (PL4), the invariant is maintained.
    If $\apath_0^{\dag \dag}$ is made of elements from the  same node,
    $\slen{\apath_0^{\dag \dag}} \leq 1$ because $\apath$ is direct and
    $C \geq j$. The final direct, rooted and $\downtouparrow$-structured path
    is $\apath_{des} \step{\sim^{\dag}_{1,h}} \pair{\anode_1}{\advar_1} \step{\sim}
    \pair{\anode_{m-1}}{\advar_{m-1}}
    \step{\sim^{\dag}_{2,h}} \apath_{asc}$ and its strict length is at least $j$.
    
  \item Otherwise, i.e. $s > 1$. Pick $h \in \interval{0}{s-1}$ such that
    $$
  \slen{\apath_h^{\dag \dag}} \geq 
  \left \lceil
  \frac{\slen{\apath^{\dag}} - 2 (\beta+1)
  }{(\beta+1)
    }
  \right \rceil 
  $$
  Note that $\pair{\anode^{\dag}}{\advar_1^{\dag}} \step{\sim^{\dag}_{1,h}}
    \pair{\anode_{i_h+1}}{\advar_{i_h+1}}$ and
    $\pair{\anode_{i_{h+1}-1}}{\advar_{i_{h+1}-1}} \step{\sim^{\dag}_{2,h}}
    \pair{\anode^{\dag}}{\advar_2^{\dag}}$ and $< \in \set{\sim^{\dag}_{1,h},\sim^{\dag}_{2,h}}$
    because  $h < s-1$ or $0 < h$, and $\apath^{\dag}$ is direct.
   \end{itemize}
  The paths are updated as follows.
  $\apath_{des}$ becomes $\apath_{des} \step{\sim^{\dag}_{1,h}}
      \pair{\anode_{i_h+1}}{\advar_{i_h+1}}$, 
  $\apath_{asc}$ becomes $\pair{\anode_{i_{h+1}-1}}{\advar_{i_{h+1}-1}}
      \step{\sim^{\dag}_{2,h}} \apath_{asc}$ and
   $\apath^{\dag}$ becomes $\apath_h^{\dag \dag}$. 
    From $\slen{\apath^{\dag}} \geq N_{j - C}$, 
    we get
    $$ 
    \slen{\apath_h^{\dag \dag}} \geq 
  \left \lceil
  \frac{\slen{\apath^{\dag}} - 2 (\beta+1)
  }{(\beta+1)
    }
  \right \rceil
  \geq  \underbrace{N_{j - C-1} >  N_{j - C-2}}_{\mbox{by (PL3)}}
   $$
  by the invariant, (PL3) and (PL1) (assuming that $j - C-2 \geq 0$, otherwise we remove the appropriate expressions).
  Indeed, by the satisfaction of the invariant, we have $\slen{\apath^{\dag}} \geq N_{j-C}$ and
  \[
  \overbrace{
  \left \lceil
  \frac{\slen{\apath^{\dag}} - 2 (\beta+1)
  }{(\beta+1) 
    }
  \right \rceil
  \geq
  \left \lceil
  \frac{N_{j-C} - 2 (\beta+1)
  }{(\beta+1)
    }
  \right \rceil}^{\mbox{by (PL1)}}
  \geq
  N_{j - C-1}
  \]
  Therefore the invariant is maintained.
  The process terminates and at termination, we have seen that this implies
  that we have built a direct, rooted and $\downtouparrow$-structured path
  from $\pair{\anode}{\adatum_{\alpha}}$ to $\pair{\anode}{\avariable}$
  that is of strict length at least $j$.  
\end{proof}
Here is the final step to prove Proposition~\ref{proposition-star-oplus}.
Since
the violation of $(\newbigstar)$ is witnessed on a single branch, the proof
is analogous to the proof of~\cite[Lemma 6.2]{Demri&DSouza07} and reformulates
the final part of  the proof of~\cite[Lemma 5.16]{Labai21}. 

\begin{proof} (Proposition~\ref{proposition-star-oplus})
  Let $\asymtree$ be a regular locally consistent symbolic tree.
  We write $N$ to denote the number of distinct subtrees in $\asymtree$ ($N$ exists because $\asymtree$
  is regular). 
  It is easy to show that the satisfiability of $\asymtree$ implies that
  $\asymtree$ satisfies $(\newbigstar)$. Indeed, if
  $\asymtree$ does not satisfy $(\newbigstar)$, then the existence of the witness
  path map $p$ and the reverse path map $rp$
  forbids the possibility to interpret the variables so that $\asymtree$ is satisfiable (by Lemma~\ref{lemma-correctness-newgt}).
  The main part of the proof consists in showing that if $\asymtree$ is not satisfiable,
  then $\newGt$ does not satisfy $(\newbigstar)$.
  By Lemma~\ref{lemma-characterisation-satisfiability-goplus},
  there is $\pair{\anode}{\avariable}$ in $U_{< \adatum_1} \cup U_{> \adatum_{\alpha}}$ in $\newGt$,
  such that  $\slen{\anode,\avariable} = \omega$.
  By Lemma~\ref{lemma-equivalences-slen}, there is
  $\pair{\anode}{\avariable}$ in $U_{< \adatum_1} \cup U_{> \adatum_{\alpha}}$ in $\newGt$,
  such that  $\uslen{\anode,\avariable} = \omega$.
  Suppose that $\pair{\anode}{\avariable}$ belongs to $U_{> \adatum_{\alpha}}$ in $\newGt$
  (the other case is very similar, and is  omitted).

  Let $M = 2((\beta+2)^2 N + 1)$ and $\apath$ be a direct, rooted and $\downtouparrow$-structured
  path of strict length at least $M$
  from $\pair{\anode}{\adatum_{\alpha}}$ to $\pair{\anode}{\avariable}$.
  The path $\apath$ is of the form below
  $$
  \pair{\anode_0}{\advar_0} \step{\sim_1} \pair{\anode_1}{\advar_1} \cdots
  \step{\sim_\ell} \pair{\anode_\ell}{\advar_\ell} =
  \pair{\anode_\ell}{\advar_\ell'} \step{\sim_\ell'} \pair{\anode_{\ell-1}}{\advar_{\ell-1}'}
  \cdots \step{\sim_1'} \pair{\anode_0}{\advar_0'},
  $$
  where the following conditions hold for some $j_1 \cdots j_{\ell} \in \interval{0}{\degree-1}^*$:
  \begin{itemize}
  \item $\pair{\anode_0}{\advar_0}  = \pair{\anode}{\adatum_{\alpha}}$
    and $\pair{\anode_0}{\advar_0'} = \pair{\anode}{\avariable}$.
  \item For all $k \in \interval{1}{\ell}$, we have $\anode_k = \anode \cdot j_1 \cdots j_k$. 
  \end{itemize}
  Since $\slen{\apath} \geq 2((\beta+2)^2 N + 1)$, one of the paths among
  $\pair{\anode_0}{\advar_0} \cdots \pair{\anode_\ell}{\advar_\ell}$ and
  $\pair{\anode_\ell}{\advar_\ell'} \cdots \pair{\anode_0}{\advar_0'}$ has strict length
  at least $(\beta+2)^2 N + 1$.
  Below, we assume that the strict length
  of $\pair{\anode_0}{\advar_0} \cdots \pair{\anode_\ell}{\advar_\ell}$ is
  at least $(\beta+2)^2 N + 1$ (the other case is very similar, and is omitted herein).
  By the Pigeonhole Principle, there are $K < K' \in \interval{0}{\ell}$ such that
  the subtree rooted at $\anode_K$ is equal to the subtree rooted at $\anode_{K'}$,
  $\advar_K = \advar_{K'}$ and $\advar_K' = \advar_{K'}'$ and
  the (consecutive) subpath $\pair{\anode_K}{\advar_K} \cdots \pair{\anode_{K'}}{\advar_{K'}}$
  has strict length at least one.
  Let us provide a bit more details about this claim. 
  The strict length of the descending path $\apath^{\dag} = \pair{\anode_0}{\advar_0} \step{\sim_1} \pair{\anode_1}{\advar_1} \cdots
  \step{\sim_\ell} \pair{\anode_\ell}{\advar_\ell}$ is at least $L = (\beta+2)^2 N + 1$.
  Both $\pair{\anode_K}{\advar_K}$ and $\pair{\anode_{K'}}{\advar_{K'}}$ are claimed to belong
  to this descending path. Let us see why exactly.
  The path $\apath^{\dag}$ can be decomposed as a sequence of consecutive subpaths  such that
  $\apath^{\dag} = \cdots \apath^{\dag}_1 \cdots \apath^{\dag}_L \cdots$ where each $\apath^{\dag}_i$ contains a single strict edge (i.e., labelled by `$<$'), say its first edge.
  So, let us say that $\apath^{\dag}_i$ starts by the elements $\pair{\anode_{s_i}}{\advar_{s_i}}$.
  Hence, if $i < i'$, then there is a strict path from $\pair{\anode_{s_i}}{\advar_{s_i}}$ to $\pair{\anode_{s_{i'}}}{\advar_{s_{i'}}}$.
  Let us define  that the 
  subtree signature of $\pair{\anode_{s_i}}{\advar_{s_i}}$ 
  is a triple made of the subtree at
  $\anode_{s_i}$ from the regular tree (a finite amount of possibilities), $\advar_{s_i}$ and
  $\advar_{s_i}'$ on the return ascending path at the node $\anode_{s_i}$ (unique visit in the ascending part).
  By the Pigeonhole Principle, there are $K < K'$ such that $\pair{\anode_{s_{K}}}{\advar_{s_K}}$ and
  $\pair{\anode_{s_K'}}{\advar_{s_K'}}$ have the same 
  subtree signature. 
  
  Observe that by Lemma~\ref{lemma-shortcircuit-equality-lessthan},  for every $k \in \interval{K}{K'}$,
  we have $\pair{\anode_k}{\advar_k} \step{<} \pair{\anode_k}{\advar_k'}$.
  Moreover,  since $\asymtree_{\mid \anode_K} = \asymtree_{\mid \anode_{K'}}$ (same subtree), 
  for all $\pair{\anodebis_1}{\advar_1^{*}}, \pair{\anodebis_2}{\advar_2^{*}} \in \set{\anode_{K}, \ldots, \anode_{K'-1}} \times
  \DVAR{\beta}{\adatum_1}{\adatum_{\alpha}}$,  for all $i \in \Nat$
  and $\sim \in \set{<,=}$,
  $$
  \pair{\anodebis_1}{\advar_1^{*}} \step{\sim} \pair{\anodebis_2}{\advar_2^{*}} \ 
  \mbox{iff}
  $$
  $$
  \pair{\anodebis_1 \cdot (j_{K+1} \cdots j_{K'})^i (j_{K+1} \cdots j_{K+m_1})}{\advar_1^{*}}
  \step{\sim}
  \pair{\anodebis_2 \cdot (j_{K+1} \cdots j_{K'})^i (j_{K+1} \cdots j_{K+m_2})}{\advar_2^{*}},
  $$
  with $\anodebis_1 = \anode_K \cdot j_{K+1} \cdots j_{K+m_1}$,
  $\anodebis_2 = \anode_K \cdot j_{K+1} \cdots j_{K+m_2}$
  and by convention if $m= 0$, then $j_{K+1} \cdots j_{K+m} = \varepsilon$.

  Let us define $\pair{\anode}{\advar}$, $\pair{\anode}{\advar'}$,
  $\abranch$, $p$ and $rp$ that witness the violation of the condition $(\newbigstar)$
  (see Section~\ref{section-introduction-starproperty}).
  \begin{itemize}
  \item $\pair{\anode}{\advar} \egdef \pair{\anode_K}{\advar_K}$, 
    $\pair{\anode}{\advar'} \egdef \pair{\anode_K}{\advar_K'}$ and
    $\abranch \egdef (j_{K+1} \cdots j_{K'})^{\omega}$.
  \item for all $i \in \Nat$, $m \in \interval{0}{K'-K-1}$,
    \begin{itemize}
    \item $p(i (K'-K) + m ) \egdef
      \pair{\anode_{K} (j_{K+1} \cdots j_{K'})^i (j_{K+1} \cdots j_{K+m})}{\advar_{K+m}}$
    and,
  \item $rp(i (K'-K) + m ) \egdef
    \pair{\anode_{K} (j_{K+1} \cdots j_{K'})^i (j_{K+1} \cdots j_{K+m})}{\advar_{K+m}'}$.
    \end{itemize}
  
  \end{itemize}
  One can check that all the properties hold to violate $(\newbigstar)$,
  in particular, $p$ is strict and for all $i$, $p(i) \step{<} rp(i)$. Indeed, let us explain a bit more why the conditions
  for violating $(\newbigstar)$ are meet.
  \begin{itemize}
  \item The path map $p$ can be written as follows:
    {\footnotesize
    $$
      \pair{\anode_K}{\advar_K}  \step{\sim_{K+1}} \cdots  \step{\sim_{K'}}
      \pair{\anode_{K} (j_{K+1} \cdots j_{K'})}{\advar_{K}}
    \step{\sim_{K+1}} \pair{\anode_{K} (j_{K+1} \cdots j_{K'}) \cdot j_{K+1}}{\advar_{K+1}} \cdots \step{\sim_{K'}}
    $$
$$
    \pair{\anode_{K} \cdot (j_{K+1} \cdots j_{K'})^2}{\advar_{K}}
    \step{\sim_{K+1}} \pair{\anode_K \cdot (j_{K+1} \cdots j_{K'})^2 \cdot j_{K+1}}{\advar_{K+1}} \cdots \step{\sim_{K'}}
    \pair{\anode_{K} \cdot (j_{K+1} \cdots j_{K'})^3}{\advar_{K}}  \cdots
    $$
    }
    Note that this is a valid  path map because $\anode_K$ and $\anode_K'$ have the same subtree. 
    Moreover $p$ is strict, i.e. edges of the form $\step{<}$  occur infinitely often because the
    strict length of the path
    $\pair{\anode_K}{\advar_K}  \step{\sim_{K+1}} \cdots  \step{\sim_{K'}} \pair{\anode_{K'}}{\advar_{K'}}$ is at least one.
  \item Similarly, the reverse path map $rp$ can be written as follows:
    {\footnotesize
    $$
      \pair{\anode_K}{\advar_K'}  \step{\overline{\sim_{K+1}'}} \cdots  \step{\overline{\sim_{K'}'}}
      \pair{\anode_{K} (j_{K+1} \cdots j_{K'})}{\advar_{K}'}
      \step{\overline{\sim_{K+1}'}} \pair{\anode_K (j_{K+1} \cdots j_{K'})
        \cdot j_{K+1}}{\advar_{K+1}'} \cdots \step{\overline{\sim_{K'}'}}
    $$
$$
    \pair{\anode_{K} \cdot (j_{K+1} \cdots j_{K'})^2}{\advar_{K}'}
    \step{\overline{\sim_{K+1}'}} \pair{\anode_K (j_{K+1} \cdots j_{K'})^2 \cdot j_{K+1}}{\advar_{K+1}'} \cdots \step{\overline{\sim_{K'}'}}
    \pair{\anode_{K} \cdot (j_{K+1} \cdots j_{K'})^3}{\advar_{K}'}  \cdots
    $$
    }
    where $\overline{=}$ (``reverse of $=$'') is $=$ and $\overline{<}$ (``reverse of $<$'') is $>$.  Note that this is a valid  reverse path map ($\anode_K$ and $\anode_K'$
    have the same subtree). 
  \item Concerning the fourth condition for violating  $(\newbigstar)$, namely for all $i$, $p(i) \step{<} rp(i)$,
    we have already observed that by Lemma~\ref{lemma-shortcircuit-equality-lessthan},  for every $k \in \interval{K}{K'}$,
    we have $\pair{\anode_k}{\advar_k} \step{<} \pair{\anode_k}{\advar_k'}$.
    Since $\anode_K$ and $\anode_K'$
    have the same subtree, for all $i \in \Nat$, for all $k \in \interval{K}{K'}$,
    $\pair{\anode_K \cdot (j_{K+1} \cdots j_{K'})^i \cdot (j_{K+1} \cdots j_k)}{\advar_k} \step{<}
    \pair{\anode_K (j_{K+1} \cdots j_{K'})^i \cdot (j_{K+1} \cdots j_k)}{\advar_k'}$.\qedhere 
  \end{itemize}
\end{proof}

\section{Related Work}
\label{section-related-work}

This section is dedicated to related work and more specifically
to results about temporal logics with numerical values
(Section~\ref{section-rw-temporal-logics}) and some technical differences
with~\cite{Labai&Ortiz&Simkus20,Labai21} (Section~\ref{section-rw-comparison} and Section~\ref{section-rabin-pairs-comparison}).
It goes a bit further than what can be found in the rest of the document
but, of course, related work is also discussed in other sections.

\subsection{Model-checking problem}
\label{section-rw-temporal-logics}
A problem related to satisfiability is the \emph{model-checking problem} (bibliographical
references about satisfiability can be mainly found in
Sections~\ref{section-introduction}-\ref{section-introduction-temporal-logics}). 
Fragments of the model-checking problem involving a temporal logic similar to
$\CTLStar(\Zed)$ are investigated in~\cite{Cerans94,Bozzelli&Gascon06,Bozzelli&Pinchinat14,Felli&Montali&Winkler22bis}
(see also~\cite{Golleretal12,Cook&Khlaaf&Piterman15,Vester15,Amparore&Donatelli&Galla20}). 
However, 
model-checking problems with $\CTLStar(\Zed)$-like languages are easily undecidable,
see e.g.~\cite[Theorem 1]{Cerans94} and~\cite[Theorem 4.1]{Mayr&Totzke16} (more general constraints are
used in~\cite{Mayr&Totzke16} but the undecidability proof
uses only the constraints involved herein). 
The difference between model-checking and satisfiability is subtle and underlines that the decidability/complexity of
$\CTL(\Zed)$/$\CTLStar(\Zed)$ satisfiability is not immediate.
For instance, a reduction following the undecidability proof in~\cite{Mayr&Totzke13} fails:
tree models satisfying a similar formula as in~\cite{Mayr&Totzke13} can still "cheat" as one
cannot express that {\em it is not possible}
to move to a state detecting cheating.

\subsection{$(\texorpdfstring{\newbigstar}{*\textasciicircum C})$: variant of $(\texorpdfstring{\bigstar}{*})$.}
\label{section-rw-comparison}

Proposition~\ref{proposition-star-oplus} is a variant of~\cite[Lemma 22]{Labai&Ortiz&Simkus20}. 
Let us  explain the improvement of our developments compared to what is done in~\cite{Labai&Ortiz&Simkus20,Labai21}. 
The \emph{framified constraint graphs} defined in~\cite[Definition 14]{Labai&Ortiz&Simkus20} correspond to the graph $\newGt$ without
the elements in
$\interval{0}{\degree-1}^*\times\{\adatum_1,\adatum_\alpha\}$ and corresponding edges. 
However, Example~\ref{example-bigstar} illustrates the importance of taking into account these elements when deciding satisfiability (without $\adatum_1$, the graph would satisfy $(\newbigstar)$).
Actually,  Example~\ref{example-bigstar} invalidates $(\bigstar)$ as used in~\cite{Labai&Ortiz&Simkus20,Labai21}
because the constants are missing to apply properly results from~\cite{Carapelle&Kartzow&Lohrey16}.
The problematic part in~\cite{Labai&Ortiz&Simkus20,Labai21}
is due to the proof of~\cite[Lemma 5.18]{Labai21} whose main argument takes advantage of~\cite{Carapelle&Kartzow&Lohrey16} but without
the elements related to constant values (see also~\cite[Lemma 8]{Demri&Gascon08}).
Note also that the condition $(\bigstar)$ in~\cite[Section 3.3]{Labai&Ortiz&Simkus20} generalises
the condition $C_{\Zed}$ from~\cite[Section 6]{Demri&DSouza07}
(see also the condition $\mathcal{C}$ in~\cite[Definition 2]{Demri&Gascon08}
and a similar condition in~\cite[Section 2]{Exibard&Filiot&Khalimov21}).
A condition similar to $(\bigstar)$ 
is also introduced recently in~\cite{Bhaskar&Praveen23bis} (conference version in~\cite{Bhaskar&Praveen22}) to decide the single-sided realizability
problem based on $\LTL(\Zed,<,=)$ (without constant values).
Though the technical developments are presented quite differently and the settings and purposes are distinct,
one can establish interesting connections. For instance, there are correspondences between~\cite[Lemma 23]{Bhaskar&Praveen23bis} and
Lemma~\ref{lemma-characterisation-satisfiability-goplus}, between~\cite[Lemma 26]{Bhaskar&Praveen23bis} and
Lemma~\ref{lemma-automaton-star}, and between the construction of tree automata in~\cite[Section 7]{Bhaskar&Praveen23bis}
and what is done in the paper to build Rabin tree automata in order to analyze the computational complexity
of the decision problems. It remains open how to use the best of the two papers for further analysis, for instance how to
take advantage of our results on constraint automata?
It is also possible to find relationships between techniques developed
for lock-sharing systems in~\cite{Mascle&Muscholl&Walukiewicz23} and what we developed
herein. 

Besides, our proof of Lemma~\ref{lemma-automaton-star} also proposes a slight novelty compared to the construction
 in~\cite{Labai&Ortiz&Simkus20}:  we design $\starautomaton$ without
firstly constructing  {\em a tree automaton for the complement language}
(as done in~\cite[Section 3.4]{Labai&Ortiz&Simkus20})
and then
using 
results from~\cite{Muller&Schupp95} (elimination of alternation in tree automata).
Our new path  happens to be rewarding: not only we can better understand how to
express the condition $(\newbigstar)$, but also we 
control  the size parameters 
of $\starautomaton$ involved in our  complexity analysis. 
This may prove useful when implementing
the decision procedure
for solving the satisfiability problem for $\CTL(\Zed)$  (resp. for  $\CTLStar(\Zed)$). 

\subsection{Differences with~\cite{Labai&Ortiz&Simkus20} about the number of Rabin pairs}
\label{section-rabin-pairs-comparison}

To conclude this section about related works, we would like to draw the attention of the reader that 
Lemma~\ref{lemma-automaton-star} is  similar
to~\cite[Proposition 26]{Labai&Ortiz&Simkus20}
but there is an essential difference: the number of Rabin pairs in
Lemma~\ref{lemma-automaton-star} is not a constant
but a value depending on $\beta$, an outcome of our investigations.
Note also that the question on the number of Rabin pairs
discussed herein
is independent from
the question discussed above about having the elements in $\interval{0}{\degree-1}^*\times\{\adatum_1,\adatum_\alpha\}$ within $\newGt$. 
It is important to know the number of Rabin pairs in
$\starautomaton$ for our  complexity analysis
as checking nonemptiness of
 Rabin tree automata is {\em exponential} in the number of Rabin pairs~\cite[Theorem 4.1]{Emerson&Jutla00}. 
It is  polynomial in the cardinality of the transition relation
and exponential in the number of Rabin  pairs, see e.g.~\cite[Theorem 4.1]{Emerson&Jutla00}.
More precisely, for binary trees, it is in time $(m \times N)^{\mathcal{O}(N)}$, where $m$  is the number of locations
  and $N$ is the number of Rabin pairs, see the statement~\cite[Theorem 4.1]{Emerson&Jutla00}.
  However, this is not exactly what we need herein as the branching  degree $\degree \geq 1$ is arbitrary
  in our investigations.
  That is  why, for the proof of Lemma~\ref{lemma-exptime-tca} we used
  that nonemptiness can be checked in time $(\card{\delta} \times \gamma \times N)^{\mathcal{O}(N)}$ ($\degree$ is taken into account in $\card{\delta}$).
  This makes a difference with the argument used to establish~\cite[Proposition 27]{Labai&Ortiz&Simkus20}.

  As a conclusion, our analysis about the number of Rabin pairs,  and the complexity measure to test nonemptiness
  of Rabin tree automata differs slightly from what is used in~\cite{Labai&Ortiz&Simkus20}. This lead us to establish the complexity
  measures sometimes quite pedantically. However, this  was for the sake of guaranteing that, indeed, nonemptiness of $\alang(\locautomaton) \cap \alang(\starautomaton)$
  can be tested in exponential-time, as stated in Lemma~\ref{lemma-exptime-tca}.

\section{Concluding Remarks}\label{section-conclusion}
We developed an automata-based approach to solve
\satproblem{$\CTL(\Zed)$} and \satproblem{$\CTLStar(\Zed)$}, by introducing
tree constraint automata that
accept infinite data trees with data domain $\Zed$.
The nonemptiness problem for tree constraint automata with B\"uchi acceptance conditions (resp.
      with Rabin pairs) is \exptime-complete, see Theorem~\ref{theorem-exptime-tca} (resp. Theorem~\ref{theorem-exptime-rtca}). 
      The difficult part consists in proving the \exptime-membership
      for which we show how to substantially adapt
      the material in~\cite[Section 5.2]{Labai21} that guided us to design  the correctness proof of $(\newbigstar)$.
      The work~\cite{Labai&Ortiz&Simkus20} was indeed a great inspiration but we adjusted a few statements from there.
      We recall that $(\bigstar)$ in~\cite{Labai&Ortiz&Simkus20} is not fully correct
      (see Section~\ref{section-rw-comparison})
      as we need to add constants (leading to the variant condition $(\newbigstar)$). 
      Moreover, our construction of the automaton in Lemma~\ref{lemma-automaton-star} does depend on the number of variables
      unlike~\cite[Proposition 26]{Labai&Ortiz&Simkus20}
      (again, see Section~\ref{section-rw-comparison}).
      This is crucial for complexity, as it is related to
      the number of Rabin pairs. We also use~\cite{Emerson&Jutla00} more precisely than~\cite[p.621]{Labai&Ortiz&Simkus20} as  we handle non-binary trees.
In short, we introduced TCA for which we characterise complexity of the nonemptiness problem
(providing a few improvements to~\cite{Labai&Ortiz&Simkus20}). We left aside the question of the  expressiveness of TCA, which is
interesting but  out of the scope of this paper.

      This lead us to show that $\satproblem{\CTL(\Zed)}$ is \exptime-complete (Theorem~\ref{theorem-ctlz}),
and $\satproblem{\CTLStar(\Zed)}$ is \twoexptime-complete (Theorem~\ref{theorem-ctlstarz}).
  The only decidability proof for $\satproblem{\CTLStar(\Zed)}$ done so far,
  see~\cite[Theorem 32]{Carapelle&Kartzow&Lohrey16}, is  by reduction to a decidable
  second-order
  logic. 
  Our complexity characterisation for $\satproblem{\CTLStar(\Zed)}$ provides an  answer to
  several open problems related to $\CTLStar(\Zed)$ fragments,
  see e.g.~\cite{Bozzelli&Gascon06,Gascon09,Carapelle&Kartzow&Lohrey16,CarapelleTurhan16,Labai&Ortiz&Simkus20}.

We believe that our results on TCA
  can help to establish complexity results for other logics (see also Section~\ref{section-ctlstarz}
  about a domain for strings,~\cite[Section 4]{Exibard&Filiot&Khalimov22} to handle more concrete domains
  and a recent result for some description logic with a concrete domain made of finite
  strings~\cite{Demri&Quaas23ter}).

  \section*{Acknowledgement} We would like to warmly thank the anonymous reviewers for all their comments
  and suggestions that help us a lot to improve the quality of this document. Thanks also to the reviewers
  of the conference version of this work. 


\newcommand{\etalchar}[1]{$^{#1}$}

 \newpage
\appendix

\section{Proof for Section~\ref{section-automata}}
\label{appendix-first}

\subsection{Proof of Lemma~\ref{lemma-intersection-rtca}}
\label{appendix-proof-lemma-intersection-rtca}
\renewcommand{\aautomatonbis}{{\mathbb A}}
 \newcommand{\Fun}{\mathbf{F}}
\begin{proof}
  For each $k \in \interval{1}{n}$,
    let $\aautomatonbis_k = \triple{\locations_k,\aalphabet,\degree,\beta}{\locations_{k,\init},\delta_k}{\rabinacc_k}$
with $\rabinacc_k = (L_k^{\gamma},U_k^{\gamma})_{\gamma \in \interval{1}{N_k}}$  be a Rabin TCA.
We define a Rabin TCA
$\aautomatonbis = \triple{\locations,\aalphabet,\degree,\beta}{\locations_{\init},\delta}{\rabinacc}$
such that $\alang(\aautomatonbis) = \bigcap_{1\leq k\leq n} \alang(\aautomatonbis_k)$.
The Rabin TCA $\aautomatonbis$ is designed as a product automaton, where the locations
are of the form $\tuple{\alocation_1}{\alocation_n, \amap}$ with
$\tuple{\alocation_1}{\alocation_n} \in \locations_1 \times \cdots \times \locations_n$
and $\amap$ is a map $\amap: \interval{1}{N_1} \times \cdots \times\interval{1}{N_n} \to \interval{0}{2n-1}$ that can be viewed
as a finite memory related to the visit of locations in the sets $L_k^{\gamma}$'s.
We write $\Fun(N_1, \ldots,N_n)$ to denote the set of maps of the form  $\amap: \interval{1}{N_1} \times \cdots
\times\interval{1}{N_n} \to \interval{0}{2n-1}$. Let us explain how we use its elements to encode the conjunction of the Rabin acceptance conditions $\rabinacc_1, \ldots, \rabinacc_n$ into
a single Rabin acceptance condition in the product tree automaton, with a single exponential
blow-up, which happens to be of reasonable magnitude. 
Suppose that along a path of a run for $\aautomatonbis$ (yet to be defined), for the sequence of locations below 
$$
\tuple{\alocation_{1,0}}{\alocation_{n,0}, \amap^0}, \tuple{\alocation_{1,1}}{\alocation_{n,1}, \amap^1}, \tuple{\alocation_{1,2}}{\alocation_{n,2}, \amap^2}
\ldots, 
$$
we wish to guarantee that it satisfies the conjunction of the Rabin acceptance conditions $\rabinacc_1, \ldots, \rabinacc_n$.
Since the tuples $\tuple{\alocation_{1,i}}{\alocation_{n,i}}$'s are updated synchronously according to the transition relations
of the Rabin TCA $\aautomatonbis_k$'s, we have  to enforce that
for some $\tuple{\gamma_1}{\gamma_n} \in  \interval{1}{N_1} \times \cdots \times\interval{1}{N_n}$,
the following properties hold. 
\begin{enumerate}
\item[(a)] For all $k \in \interval{1}{n}$,
           $L_k^{\gamma_k}$ is visited infinitely often in $\alocation_{k,0}, \alocation_{k,1}, \alocation_{k,2}, \ldots$.
\item[(b)] For all $k \in \interval{1}{n}$,
           $U_k^{\gamma_k}$ is visited finitely  in $\alocation_{k,0}, \alocation_{k,1}, \alocation_{k,2}, \ldots$.
\end{enumerate}
In the product automaton $\aautomatonbis$, we shall associate a unique Rabin pair $\pair{L}{U}$ to each tuple
$\tuple{\gamma_1}{\gamma_n}$ in order to enforce the satisfaction of (a) and (b) above.
Since the $\locations_k$'s and $\Fun(N_1, \ldots,N_n)$ are finite sets, $U$ can be defined as
the  union of sets of the form
$\locations_1 \times \cdots \times \locations_{k-1} \times U^{\gamma_k}_k \times \locations_{k+1} \times
    \cdots \times \locations_{n}
    \times \Fun(N_1, \ldots,N_n)$, $k \in \interval{1}{n}$.
    In order to take care of (a), we use the maps from the set $\Fun(N_1, \ldots,N_n)$.
    The idea is to enforce that $L_1^{\gamma_1}$ is visited, then $L_2^{\gamma_2}$, then \ldots then $L_n^{\gamma_n}$, and this is repeated
    infinitely along the sequence of locations labelling the path of the run.
    For each map $\amap \in \Fun(N_1, \ldots,N_n)$, the intentions for the value $\amap(\gamma_1,\ldots,\gamma_n)$ are the 
    following ones. If $\amap(\gamma_1,\ldots,\gamma_n) = 2k-2$ for some $k \in \interval{1}{n}$, then along the path 
    we are waiting to meet a forthcoming location $\tuple{\alocation_{1,i}}{\alocation_{n,i}, \amap_i}$ such that
    $\alocation_{k,i}$ belongs to $L_k^{\gamma_k}$. Once it is done, $\amap_i(\gamma_1,\ldots,\gamma_n)$ takes the odd value $2k-1$. By construction, at the next step $\amap_{i+1}(\gamma_1,\ldots,\gamma_n) = 2k \ \bmod \ 2n$.
    So, $\amap_{i+1}(\gamma_1,\ldots,\gamma_n)$ is equal to $2k'$ for some $k' \in \interval{0}{n-1}$, and now along the path
    we are waiting to meet a forthcoming location $\tuple{\alocation_{1,i'}}{\alocation_{n,i'}, \amap_{i'}}$ such that
    now $\alocation_{k'+1,i'}$ belongs to $L_{k'+1}^{\gamma_{k'+1}}$. By requiring that $\amap_i(\gamma_1,\ldots,\gamma_n) = 1$ infinitely often
    along the sequence of locations labelling the path of the run, we shall guarantee the satisfaction of (a).
    As a consequence, $L$ simply contains all the locations $\tuple{\alocation_1}{\alocation_n, \amap}$ such that
    $\amap(\gamma_1,\ldots,\gamma_n) = 1$. 

    Before providing the formal definition for $\aautomatonbis$, we define below the function
     $$\mathcal{U}:  \Fun(N_1, \ldots,N_n) \times (\locations_1 \times \cdots \times \locations_n) \rightarrow
    \Fun(N_1, \ldots,N_n)$$ that is instrumental to update the maps in $\Fun(N_1, \ldots,N_n)$ along
    the paths of $\aautomatonbis$'s runs:  
     define $\mathcal{U}(\amap,\tuple{\alocation_1}{\alocation_n}) = \amap'$, where, for all 
      $\tuple{\gamma_1}{\gamma_n} \in \interval{1}{N_1} \times \cdots \times\interval{1}{N_n}$,
    $$\amap'\tuple{\gamma_1}{\gamma_n}\egdef
     \begin{cases} 
     2k \bmod 2n & \text{if } \amap\tuple{\gamma_1}{\gamma_n}=2k-1 \text{ for some } k\in\interval{1}{n}, \\
     2k - 1 & \text{if } \amap\tuple{\gamma_1}{\gamma_n}=2k-2 \text{ for some } k\in\interval{1}{n}, \text{ and  }q_k\in L_k^{\gamma_k}, \\
     2k - 2 & \text{if } \amap\tuple{\gamma_1}{\gamma_n}=2k-2 \text{ for some } k\in\interval{1}{n}, \text{ and  }q_k\not\in L_k^{\gamma_k}. 
     \end{cases}$$    
     Next, we state several simple properties. 
Let us consider the $\omega$-sequence below in $(\locations_1 \times \cdots \times \locations_n \times \Fun(N_1, \ldots,N_n))^{\omega}$
$$
\tuple{\alocation_{1,0}}{\alocation_{n,0}, \amap_0}, \tuple{\alocation_{1,1}}{\alocation_{n,1}, \amap_1}, \tuple{\alocation_{1,2}}{\alocation_{n,2}, \amap_2}
\ldots, 
$$
such that for all $i \in \Nat$, $\amap_{i+1} = \mathcal{U}(\amap_i, \tuple{\alocation_{1,i+1}}{\alocation_{n,i+1}})$
and $\amap_0$ is the unique map that takes always the value zero. 
Such a sequence may correspond to the source locations of the transitions along a path in a run of $\aautomatonbis$, see below.
One can easily show the truth of the following two claims.

\,

\noindent
{\bf Claim 1}
For all $\tuple{\gamma_1}{\gamma_n} \in \interval{1}{N_1} \times \cdots \times\interval{1}{N_n}$, the following three statements are equivalent.
\begin{enumerate}
\item For infinitely many $i\geq 0$, $\amap_i(\gamma_1,\ldots,\gamma_n) = 1$.
\item For all $k \in \interval{0}{2n-1}$,  for infinitely many $i\geq 0$, $\amap_i(\gamma_1,\ldots,\gamma_n) = k$.
\item For all $k \in \interval{1}{n}$, $L_k^{\gamma_k}$ is visited infinitely often in
   $\alocation_{k,0}, \alocation_{k,1}, \alocation_{k,2} \ldots$.
\end{enumerate}
\begin{proof}[Proof of Claim 1]
The direction from (2) to (1) is trivial. For the direction from (1) to (2), suppose $\amap_i(\gamma_1,\ldots,\gamma_n) = 1$ for infinitely many $i\geq 0$. 
Let $i$ be the first index such that $\amap_i(\gamma_1,\ldots,\gamma_n) = 1$. 
Recall that $\amap_0(\gamma_1,\dots,\gamma_n)=0=2k-2$ for $k=1$. 
By definition of $\mathcal{U}$, we have $\amap_j(\gamma_1,\dots,\gamma_n)=0$ for $0\leq j< i$, and $q_{1,i}\in L^{\gamma_1}_1$. 
But then, also by definition of $\mathcal{U}$, 
$\amap_{i+1}(\gamma_1,\dots,\gamma_n)=2\neq 1$.  Using similar arguments, one can show that the successive values for $\amap_0(\gamma_1,\ldots,\gamma_n), \amap_1(\gamma_1,\ldots,\gamma_n),\amap_2(\gamma_1,\ldots,\gamma_n)  \ldots
  $
  modulo stuttering 
  are $(0 \cdot 1  \cdots (2n-2) \cdot (2n-1))^{\omega}$.
  For the direction from (2) to (3), 
  the satisfaction of (2) implies that for all $k \in \interval{1}{n}$,  for infinitely many $i$, $\amap_i(\gamma_1,\ldots,\gamma_n) = 2k-1$.  Therefore  $L_k^{\gamma_k}$ is visited infinitely often. 
  Conversely, if (3) holds,
  then again, 
  the successive values for
  $\amap_0(\gamma_1,\ldots,\gamma_n), \amap_1(\gamma_1,\ldots,\gamma_n),\amap_2(\gamma_1,\ldots,\gamma_n)  \ldots$ modulo stuttering
  are $(0 \cdot 1  \cdots (2n-2) \cdot (2n-1))^{\omega}$ and therefore (2)  holds.\end{proof}

\noindent
{\bf Claim 2}
For all $\tuple{\gamma_1}{\gamma_n} \in \interval{1}{N_1} \times \cdots \times\interval{1}{N_n}$, the following two statements are equivalent.
\begin{enumerate}
\item For all $k \in \interval{1}{n}$,
$U_k^{\gamma_k}$ is visited only finitely in $\alocation_{k,0}, \alocation_{k,1}, \alocation_{k,2} \ldots$. 
\item $U$ is visited only finitely in
$
\tuple{\alocation_{1,0}}{\alocation_{n,0}, \amap_0}, \tuple{\alocation_{1,1}}{\alocation_{n,1}, \amap_1}, \tuple{\alocation_{1,2}}{\alocation_{n,2}, \amap_2}
\ldots.
$, where
$$
U =
\big(
    \bigcup_{k=1}^{n} \locations_1 \times \cdots \times \locations_{k-1} \times U^{\gamma_k}_k \times \locations_{k+1} \times
    \cdots \times \locations_{n}
    \big) \times \Fun(N_1, \ldots,N_n).
$$
\end{enumerate}
\begin{proof}[Proof of Claim 2]
Suppose that (1) holds. Since all $Q_k$'s are finite, it follows that $U$ cannot be  visited infinitely often, and hence (2) holds. The direction from (2) to (1) is trivial. 
\end{proof}

Below, we provide the formal definition for $\aautomatonbis$ according to the above informal description and
then we prove the correctness of the construction
based on the properties established above. 

\begin{itemize}

\item $\locations \egdef \locations_1 \times \cdots \times \locations_n \times \Fun(N_1, \ldots,N_n)$.

\item The tuple
  $\triple{\pair{\alocation_1, \ldots,\alocation_n}{\amap}}{\aletter}{
  \pair{\acons_0}{\pair{\alocation^1_0, \ldots,\alocation^n_0}{\amap^0}}, \ldots,
  \pair{\acons_{\degree-1}}{\pair{\alocation^1_{\degree-1}, \ldots,\alocation^n_{\degree-1}}{\amap^{\degree-1}}}}
  $
  belongs to $\delta$ iff the conditions below hold.
  \begin{enumerate}

  \item  For each $k \in \interval{1}{n}$, there is
    $\triple{\alocation_k}{\aletter}{\tuple{\pair{\acons^k_0}{\alocation^k_0}}{
      \pair{\acons^k_{\degree-1}}{\alocation^k_{\degree-1}}}} \in \delta_k$ and
    for each $j \in \interval{0}{\degree-1}$, 
    $$
    \acons_j \egdef \bigwedge_{k=1}^{n} \acons^k_j .
    $$
    Observe that the size of $\acons_j$ is bounded above by
    $\underset{k}{\sum} \ (1 + \maxconstraintsize{\aautomatonbis_k})$.

  \item For all $j \in \interval{0}{\degree-1}$,
    $\amap^j \egdef  \mathcal{U}(\amap, \tuple{\alocation^1_j}{\alocation^n_j})$
    (deterministic update).

    When $n =2$, the above conditions are identical to those in the proof of
    Lemma~\ref{lemma-intersection-automaton} and therefore the above developments generalise
    what is done in that proof. 
  \end{enumerate}

\item $\locations_{\init} \egdef \locations_{1,\init} \times \cdots \times \locations_{n,\init} \times
  \set{\amap_0}$, where $\amap_0$ is
  the unique map that takes always the value zero
  (this value is arbitrary and any value will do the job). 

  \item The set of Rabin pairs in $\rabinacc$ contains exactly the pairs $\pair{L}{U}$ for which there is
    $\tuple{\gamma_1}{\gamma_n} \in \interval{1}{N_1} \times \cdots \times \interval{1}{N_n}$ such that
    {\small 
    $$
    U \egdef
    \big(
    \bigcup_{k=1}^{n} \locations_1 \times \cdots \times \locations_{k-1} \times U^{\gamma_k}_k \times \locations_{k+1} \times
    \cdots \times \locations_{n}
    \big) \times
     \Fun(N_1, \ldots,N_n)
    $$
    $$
    L \egdef \big(\locations_1 \times \locations_2 \times \cdots \times \locations_n \big) \times
    \set{\amap \mid \amap(\gamma_1,\ldots,\gamma_n) = 1}.
    $$
    }Again, when $n =2$, the above construction is identical to
    the one  in the proof of
    Lemma~\ref{lemma-intersection-automaton}.
    Strictly speaking, $\pair{L}{U}$ above is indexed by  a tuple $\tuple{\gamma_1}{\gamma_n}$
    but we omit such decorations as it will not lead to any confusion. 
\end{itemize}
Finally, by construction $\rabinacc$ contains at most $N = \underset{k}{\Pi} \ N_k$ pairs
as required in the statement.

Let us prove that 
$\alang(\aautomatonbis) = \bigcap_{1\leq k\leq n} \alang(\aautomatonbis_k)$.
For the first inclusion,  assume that $\atree \in \bigcap_{1\leq k\leq n} \alang(\aautomatonbis_k)$.
For each $1\leq k\leq n$, 
there exists some initialized accepting run $\rho_k:\interval{0}{\degree-1}^* \to \delta_k$ of $\aautomatonbis_k$ on $\atree$. 
That is, for every node $\anode\in\interval{0}{\degree-1}^*$ with $\atree(\anode)=(\aletter,\vect{z})$ and 
$\atree(\anode\cdot i) = (\aletter_i, \vect{z}_i)$ for all $i\in\interval{0}{\degree-1}$, 
if $\arun_k(\anode)=(\alocation^k, \aletter,(\acons_0^k,\alocation_0^k), \dots, (\acons_{\degree-1}^k,\alocation_{\degree-1}^k)$, then, for all $i\in\interval{0}{\degree-1}$
\begin{enumerate}
\item $\arun_k(\anode\cdot i)$'s source location is $\alocation_i^k$, 
\item $\Zed\models\acons_i^k(\vect{z},\vect{z}_i)$. 
Further, since $\arun_k$ is initialized and accepting, we have 
\item $\arun_k(\varepsilon)$'s source location is in $\locations_{k,\init}$,  and
\item for all paths $\apath$ in $\arun$ starting from the root node $\varepsilon$, there exists
      some $(L,U)\in\rabinacc_k$ such that 
      $\inf(\arun_k,\apath)\cap L\neq\emptyset$ and $\inf(\arun_k,\apath)\cap U\ = \emptyset$. 
\end{enumerate}

From these runs $\arun_k$, 
we define 
the map $\arun: \interval{0}{\degree-1}^* \to \delta$ as follows: 
for all $\anode \in \interval{0}{\degree-1}^*$
with $\atree(\anode) = \pair{\aletter}{\vect{z}}$, define 
$$\arun(\anode)\egdef\triple{\pair{\alocation_1, \ldots,\alocation_n}{\amap}}{\aletter}{
  \pair{\acons_0}{\pair{\alocation^1_0, \ldots,\alocation^n_0}{\amap^0}}, \ldots,
  \pair{\acons_{\degree-1}}{\pair{\alocation^1_{\degree-1}, \ldots,\alocation^n_{\degree-1}}{\amap^{\degree-1}}}},
  $$
where
$
    \acons_j \egdef \bigwedge_{k=1}^{n} \acons_k^j
    $
    and
    $\amap^j \egdef  \mathcal{U}(\amap, \tuple{\alocation^1_j}{\alocation^n_j})$
for each $j \in \interval{0}{\degree-1}$. 
For this definition to make sense, 
we additionally require the source location of $\rho(\varepsilon)$ to be equal to 
$(\arun_1(\varepsilon), \ldots,
\arun_n(\varepsilon), \amap_0)$, where
$\amap_0$ is the function mapping everything to $0$. 
Note that the corresponding transition is indeed in $\delta$.  
Let us prove that $\arun$ is an initialized accepting run of $\aautomatonbis$ on $\atree$. 
Let $\anode\in\interval{0}{\degree-1}^*$ and $j\in\interval{0}{\degree-1}$. Suppose $\atree(\anode)=(\aletter,\vect{z})$, 
$\atree(\anode\cdot j) = (\aletter_j, \vect{z}_j)$, and 
$\arun(\anode)=\triple{\pair{\alocation_1, \ldots,\alocation_n}{\amap}}{\aletter}{
  \pair{\acons_0}{\pair{\alocation^1_0, \ldots,\alocation^n_0}{\amap^0}}, \ldots,
  \pair{\acons_{\degree-1}}{\pair{\alocation^1_{\degree-1}, \ldots,\alocation^n_{\degree-1}}{\amap^{\degree-1}}}}$.
\begin{itemize}
\item That $\arun(\anode\cdot j)$'s source location is $(\alocation^1_j,\dots,\alocation^n_j,\amap^j)$ follows from (1) above, and by the deterministic update of the function. 
\item That $\Zed\models\acons_j(\vect{z},\vect{z}_j)$ follows from (2) above.
\item That $\arun(\varepsilon)$'s source location $(\arun_1(\varepsilon), \ldots,
\arun_n(\varepsilon), \amap_0)$ is in $\locations_{\init}$ follows from the definition of $\locations_{\init}$ and (3) above. 
\item For proving that $\arun$ is accepting, let $\apath=j_1 j_2 \cdots \in \interval{0}{\degree-1}^{\omega}$ be an infinite path in $\arun$ starting from $\varepsilon$. 
Let $s=(\alocation_{1,0},\alocation_{2,0},\dots,\alocation_{n,0},\amap^0), 
(\alocation_{1,1},\alocation_{2,1},\dots,\alocation_{n,1},\amap^1), 
(\alocation_{1,2},\alocation_{2,2},\dots,\alocation_{n,2},\amap^2)\dots$ be the corresponding sequence of source locations.  
By (4) above, 
for each $k \in \interval{1}{n}$,
there is an index $\gamma_k$ such that some location in $L_k^{\gamma_k}$ occurs
infinitely often in 
$q_{k,0},q_{k,1},q_{k,2}\dots$.
and all the locations in  $U_k^{\gamma_k}$ occur only finitely often in
$q_{k,0},q_{k,1},q_{k,2}\dots$.
Let us pick $\pair{L}{U}$ in
$\rabinacc$ such that
{\small 
    $$
    U =
    \big(
    \bigcup_{k=1}^{n} \locations_1 \times \cdots \times \locations_{k-1} \times U^{\gamma_k}_k \times \locations_{k+1} \times
    \cdots \times \locations_{n}
    \big) \times
    \Fun(N_1, \ldots,N_n)
    $$
    $$
    L = \big(\locations_1 \times \locations_2 \times \cdots \times \locations_n \big) \times
    \set{\amap \mid \amap(\gamma_1,\ldots,\gamma_n) = 1}.
    $$
}By Claim 2, 
all the locations in $U$ occur only finitely often in
the sequence $s$. Hence $\inf(\arun, \apath)\cap U=\emptyset$. 
By Claim 1 and the definition of $L$ for the tuple $\tuple{\gamma_1}{\gamma_n}$,  there is some location in $L$ occurring  infinitely often in the sequence $s$. Hence 
$\inf(\arun, \apath)\cap L\neq \emptyset$, which finishes the proof of the first direction.
\end{itemize}

For the other inclusion, assume that $\atree \in \alang(\aautomatonbis)$. 
Then there exists some initialized accepting run $\arun: \interval{0}{\degree-1}^* \to \delta$. 
That is, 
for every node $\anode\in\interval{0}{\degree-1}^*$ with $\atree(\anode)=(\aletter,\vect{z})$ and 
$\atree(\anode\cdot i) = (\aletter_i, \vect{z}_i)$ for all $i\in\interval{0}{\degree-1}$, 
if
\[\arun(\anode)=((\alocation_1,\dots,\alocation_n,\amap), \aletter,(\acons_0,(\alocation_0^1,\dots,\alocation_0^n,\amap^0)), \dots, (\acons_{\degree-1},(\alocation_{\degree-1}^1,\dots,\alocation_{\degree-1}^n,\amap^{\degree-1}))
\]
then, for all $i\in\interval{0}{\degree-1}$
\begin{enumerate}
\item $\arun(\anode\cdot i)$'s source location is $(\alocation_i^1,\dots,\alocation_i^n,\amap^i)$, 
\item $\Zed\models\acons_i(\vect{z},\vect{z}_i)$. 

\item  By definition of $\delta$, we can also infer that 
    $\triple{\alocation_k}{\aletter}{\tuple{\pair{\acons^k_0}{\alocation^k_0}}{
      \pair{\acons^k_{\degree-1}}{\alocation^k_{\degree-1}}}} \in \delta_k$,
      and
    $\acons_i \egdef \bigwedge_{k=1}^{n} \acons_k^i$ for all $k \in \interval{1}{n}$. 
Further, since $\arun$ is initialized and accepting, we have 
\item $\arun(\varepsilon)$'s source location is in $\locations_{\init}$,  and
\item for all paths $\apath$ in $\arun$ starting from $\varepsilon$, there exists some $(L,U)\in\rabinacc$ such that 
$\inf(\arun,\apath)\cap L\neq\emptyset$ and $\inf(\arun,\apath)\cap U\ = \emptyset$.
\end{enumerate}
For all $k\in\interval{1}{n}$, 
we define the map $\arun_k: \interval{0}{\degree-1}^* \to \delta_k$ by 
    $$\arun_k(\anode) \egdef \triple{\alocation_k}{\aletter}{\tuple{\pair{\acons^k_0}{\alocation^k_0}}{
        \pair{\acons^k_{\degree-1}}{\alocation^k_{\degree-1}}}},$$ for each 
        $\anode\in\interval{0}{\degree-1}^*$. 
    Let us prove that $\arun_k$ is an initialized accepting run of $\aautomatonbis_k$ on $\atree$. 
    Let $\anode\in\interval{0}{\degree-1}^*$ and $j\in\interval{0}{\degree-1}$. Suppose $\atree(\anode)=(\aletter,\vect{z})$, 
$\atree(\anode\cdot j) = (\aletter_j, \vect{z}_j)$, and 
$\arun_k(\anode)=\triple{\alocation^k}{\aletter}{ 
  \pair{\acons^k_0}{\alocation^k_0}, \ldots,
  \pair{\acons^k_{\degree-1}}{\alocation^k_{\degree-1}}}$.
\begin{itemize}
\item That $\arun_k(\anode\cdot j)$'s source location is $\alocation^1_k$ follows from (1) above. 
\item That $\Zed\models\acons_j^k(\vect{z},\vect{z}_j)$ follows from (2) and (3) above ($\acons^k_j$ is a conjunct of $\acons_j$).
\item That $\arun_k(\varepsilon)$'s source location is in $\locations_{k,\init}$ follows from  (4) above and the definition of $\locations_{\init}$. 
\item For proving that $\arun_k$ is accepting, 
let $\apath= j_1 j_2 \cdots \in \interval{0}{\degree-1}^{\omega}$ be an infinite path of $\arun$ from $\varepsilon$. 
By (5) above, there exists some $(L,U)\in\rabinacc$ such that 
$\inf(\arun,\apath)\cap L\neq\emptyset$ and $\inf(\arun,\apath)\cap U\ = \emptyset$. 
By definition of $\rabinacc$, $(L,U)$ is defined for some tuple of indices $(\gamma_1,\dots,\gamma_n)\in \interval{1}{N_1} \times \cdots \times \interval{1}{N_n}$. 
This implies that, assuming that  $s=(\alocation_{1,0},\alocation_{2,0},\dots,\alocation_{n,0},\amap^0), 
(\alocation_{1,1},\alocation_{2,1},\dots,\alocation_{n,1},\amap^1), 
(\alocation_{1,2},\alocation_{2,2},\dots,\alocation_{n,2},\amap^2)\dots$ be the  sequence of source locations in $\aautomatonbis$ corresponding to $\pi$, 
that all elements in
    $$\big(
    \bigcup_{k=1}^{n} \locations_1 \times \cdots \times \locations_{k-1} \times U^{\gamma_k}_k \times \locations_{k+1} \times
    \cdots \times \locations_{n}
    \big) \times
    \Fun(N_1, \ldots,N_n)
    $$
    occur only finitely in $s$.
    Moreover, some location in $s$
    $$
    \big(\locations_1 \times \locations_2 \times \cdots \times \locations_n \big) \times
    \set{\amap \mid \amap(\gamma_1,\ldots,\gamma_n) = 1}
    $$
    occurs infinitely often in the sequence $s$. 
    By projecting on the locations in $\locations_k$, this means that
    all the elements in $U^{\gamma_k}_k$ occur only finitely in
    $q_{k,0},q_{k,1},q_{k,2}\dots$ 
    and, by 
    Claim 1, some location in $L_k^{\gamma_k}$ occurs
    infinitely often in 
    $q_{k,0},q_{k,1},q_{k,2}\dots$. 
    This concludes that for all $k \in \interval{1}{n}$,
    $\arun_k$ is an accepting run on $\atree$ and therefore
    $\atree \in \bigcap_{k} \alang(\aautomatonbis_k)$.  \qedhere
\end{itemize}  
  \end{proof}

\renewcommand{\aautomatonbis}{{\mathbb B}}
\section{Proof for Section~\ref{section-complexity-nonemptiness}}
\subsection{Proof of \exptime-hardness for the nonemptiness problem}
\label{appendix-exptime-hardness} 
\begin{proof}
 \exptime-hardness of the nonemptiness problem for TCA
can be shown by reduction from the acceptance problem for alternating Turing machines
running in polynomial space, see e.g.~\cite{Chandra&Kozen&Stockmeyer81}. 
Let us first give a formal definition of alternating Turing machines. 
An \defstyle{alternating Turing machine} (ATM) is a tuple
$\aatm = \triple{\locations,\aalphabet}{\delta,\alocation_0,\alocation_{\text{acc}},\alocation_{\text{rej}}}{g}$ defined as follows. 

\begin{itemize}
\item $\locations$ is the finite set of control states.
\item $\aalphabet$ is the finite tape alphabet including the blank symbol $\sharp$ and the left endmarker $\rhd$.
\item $\delta: \locations\setminus\{\alocation_{\text{acc}},\alocation_{\text{rej}}\} \times \aalphabet \rightarrow \powerset{\locations\setminus\{\alocation_0\} \times \aalphabet \times
  \set{\leftarrow, \rightarrow}}$ is the transition function and each $\delta(\alocation,\aletter)$ contains exactly two elements. 
\item $\alocation_0 \in \locations$ is the initial state, $\alocation_{\text{acc}}\in\locations$ is the accepting state, and 
$\alocation_{\text{rej}}\in\locations$ is the rejecting state. 
\item $g: \locations\setminus\{\alocation_{\text{acc}},\alocation_{\text{rej}}\} \rightarrow \set{\forall,\exists}$ specifies the type of a state 
  ($\forall$ for universal states, $\exists$ for existential states); without loss of generality, we assume $g(\alocation_0)=\forall$. 
  \end{itemize}
A \defstyle{configuration} is a finite word in $\aalphabet^* (\locations \times \aalphabet) \aalphabet^{*}$, 
and we write $\Configs(\aatm)$ to denote the set of all configurations of the ATM $\aatm$. 
As usual, a configuration $\aconfiguration$ of the form $\aword \pair{\alocation}{\aletter} \aword'$ encodes a word
$\aword \aletter \aword'$ on the tape, with control state $\alocation$ and the head is on the $\length{\aword \aletter}$$^{\rm th}$ tape cell
where $\length{\aword \aletter}$ denotes the length of the word $\aword \aletter$.
We also say that $\aconfiguration$ uses $\length{\aword \aletter \aword'}$ tape cells. 
We write $\vdash_{\aatm}$ to denote the derivation relation of the ATM $\aatm$,  defined as follows, and 
where $\aconfiguration = \aword \pair{\alocation}{\aletter} \aword'$.
\begin{itemize}
\item $\aconfiguration \vdash_{\aatm} \aconfiguration'$ with $\aword = \aword'' \aletterbis$, $\aconfiguration' = \aword'' \pair{\alocation'}{\aletterbis} \aletterter \aword'$
  and $(\alocation',\aletterter,\leftarrow) \in \delta(\alocation,\aletter)$.
\item $\aconfiguration \vdash_{\aatm} \aconfiguration'$ with $\aword' = \aletterbis \aword''$, $\aconfiguration' = \aword \aletterter \pair{\alocation'}{\aletterbis} 
  \aword''$
  and $(\alocation',\aletterter,\rightarrow) \in \delta(\alocation,\aletter)$.
\item $\aconfiguration \vdash_{\aatm} \aconfiguration'$ with $\aword' = \varepsilon$, $\aconfiguration' = \aword \aletterter \pair{\alocation'}{\sharp}$
  and $(\alocation',\aletterter,\rightarrow) \in \delta(\alocation,\aletter)$. 
 \end{itemize}
 
 Given a word $\aword$ in $(\aalphabet \setminus \set{\rhd, \sharp})^*$, its initial configuration $\aconfiguration_{\init}(\aword)$ is
$\pair{\alocation_0}{\rhd} \aword$.  
The configuration $\aconfiguration$ is \defstyle{accepting} if $\alocation = \alocation_{\text{acc}}$, and \defstyle{rejecting}
if $\alocation = \alocation_{\text{rej}}$. 
An \defstyle{accepting run} for $\aword$ is a finite tree $\atree: \domain{\atree} \rightarrow \Configs(\aatm)$
with at most two children per node (but in $\domain{\atree}$, some nodes may have a unique child)
satisfying the following conditions:
\begin{itemize}
\item We have $\atree(\varepsilon) = \aconfiguration_{\init}(\aword)$ and all the leaves are labelled by accepting configurations.
\item The tree $\atree$ does not contain any node labelled by a rejecting configuration.
\item For all $\anode \in \domain{\atree}$ with $\atree(\anode) = \aword \pair{\alocation}{\aletter} \aword'$ and $\alocation$ is universal,
  the node $\anode$ has two children $\anode \cdot 0$ and $\anode \cdot 1$ that are labelled respectively
  by 
  the  configurations $\aconfiguration'$
   such that $\aword \pair{\alocation}{\aletter} \aword' \vdash_{\aatm} \aconfiguration'$ (maybe the same for the two options).
\item For all $\anode \in \domain{\atree}$ with $\atree(\anode) = \aword \pair{\alocation}{\aletter} \aword'$ and $\alocation$ is existential,
  the node $\anode$ has one child $\anode \cdot 0$  that is labelled by some configuration
  $\aconfiguration'$ such that $\aconfiguration \vdash_{\aatm} \aconfiguration'$.
\end{itemize}
In that case, we say that $\aword$ is \defstyle{accepted by the ATM $\aatm$}. 
The ATM $\aatm$ is polynomially space-bounded if there is a polynomial $P$ such that
for all $\aword \in (\aalphabet \setminus \set{\rhd, \sharp})^*$, if $\aword$ has an accepting
run, then it has an accepting run such that all the configurations labelling the nodes use at most
$P(\length{\aword})$ tape cells. 
The problem of determining whether a polynomially space-bounded ATM $\aatm$ accepts
the input word $\aword$ is known to be \exptime-complete~\cite[Corollary 3.6]{Chandra&Kozen&Stockmeyer81}.

Let us now turn to the reduction proof.  
Let $\aatm = \triple{\locations,\aalphabet}{\delta,\alocation_0}{g}$
be a polynomially space-bounded ATM (with polynomial $P(n) \geq n$) and $\aword \in \aalphabet^*$ be an input
word.
We define a B\"uchi TCA
$\aautomaton_{\aatm, \aword} = \triple{\locations',\aalphabet',2,\beta}{\locations'_\init,\delta'}{F}$
such that $\aatm$ accepts $\aword$ iff $\alang(\aautomaton_{\aatm, \aword}) \neq \emptyset$.

 The idea is to design $\aautomaton_{\aatm, \aword}$ so that it accepts a representation of the accepting runs of $\aatm$ from $\aword$. Let us first explain how we can represent a configuration of $\aatm$. 
 Without any loss of generality, we assume that all the configurations  
 use exactly $P(\length{\aword})$ tape cells,
possibly by padding the suffix with copies of the letter $\sharp$. 
The content of each of the $P(\length{\aword})$ tape cells is encoded by the value of its own variable in the TCA; hence we set the number of variables $\beta = P(\length{\aword})$.  For the encoding,  let 
$\amap: \aalphabet \cup \set{\sharp,\rhd} \rightarrow \interval{1}{\card{\aalphabet}+2}$
be an arbitrary one-to-one map. 
The tape of length $P(\length{\aword})$ is encoded by the values of the variables
$\avariable_1, \ldots, \avariable_{P(\length{\aword})}$, i.e. the letter $\aletterbis$ 
on the $j$$^{\rm th}$ tape cell
is encoded by the constraint $\avariable_j = \amap(\aletterbis)$.  
Let us define a simple and natural correspondence between configurations
from $\aalphabet^* (\locations \times \aalphabet) \aalphabet^{*}$ using $\beta$ tape cells
and a subset of  $\locations \times \interval{1}{\beta} \times \interval{1}{\card{\aalphabet}}^{\beta}$.
Elements of the set $\locations \times \interval{1}{\beta} \times \interval{1}{\card{\aalphabet}}^{\beta}$ are called
\defstyle{numerical configurations}. 
We say that
$
\aword \pair{\alocation}{\aletter} \aword' \approx
\triple{\alocation'}{i}{\adatum_1, \ldots, \adatum_{\beta}}
$
iff $\alocation = \alocation'$, $i = \length{\aword \aletter}$
and for all $j \in \interval{1}{\beta}$, we have $\adatum_j = \amap((\aword \aletter \aword')(j))$.

Let us complete the definition of the TCA
$\aautomaton_{\aatm, \aword} = \triple{\locations',\aalphabet',2,\beta}{\locations'_\init,\delta'}{F'}$ by 

\begin{itemize}

\item $\aalphabet' = \set{\arbitraryletter}$ ($\aalphabet'$ plays no essential role). 

\item $\locations' \egdef \locations \times \interval{1}{\beta} \uplus \set{\alocation^{\star}}$.
  Each pair $\pair{\alocation}{i}$ encodes part of a configuration with control state $\alocation$ and the head
  is on the $i$$^{\rm th}$ tape cell. The location $\alocation^{\star}$ is a special accepting state. 

\item $\locations'_\init = \set{\pair{\alocation_0}{1}}$ and $F' \egdef  \set{\pair{\alocation_{\text{acc}}}{i} \mid
 i\in \interval{1}{\beta}} \cup \set{\alocation^{\star}}$.

\item It remains to define the transition relation $\delta'$.

  \begin{itemize}

  \item For all $\alocation \in \{\alocation_{\text{acc}},\alocation_{\text{rej}}\}$, 
    and all $i \in \interval{1}{\beta}$,
    the only transitions starting from $\pair{\alocation}{i}$ are 
    $
    \triple{\pair{\alocation}{i}}{\arbitraryletter}{\pair{\pair{\true}{\pair{\alocation}{i}}}{\pair{\true}{\pair{\alocation}{i}}}}
    $.

  \item
    The only transition starting from $\alocation^{\star}$ is
    $
    \triple{\alocation^{\star}}{\arbitraryletter}{\pair{\pair{\true}{\alocation^{\star}}}{\pair{\true}{\alocation^{\star}}}}
    $.

  \item  Given $i \in \interval{1}{\beta}$, $\aletter, \aletter' \in \aalphabet$, we write
    $\acons(i, \aletter, \aletter')$ to denote the constraint in $\treeconstraints{\beta}$ below:
    $$
    \big( \bigwedge_{j \in \interval{1}{\beta} \setminus \set{i}} \mynext \avariable_j = \avariable_j \big)
    \wedge \avariable_i = \amap(\aletter)
    \wedge \mynext \avariable_i = \amap(\aletter').
    $$
    For all $\alocation \in \locations$ s.t. $g(\alocation)=\exists$, for all $\aletter \in \aalphabet$ with
    $\delta(\alocation, \aletter)$ of the form $\set{\triple{\alocation_1}{\aletter_1}{d_1},\triple{\alocation_2}{\aletter_2}{d_2}}$,
    for all $u \in \set{1,2}$ and $i \in \interval{1}{\beta}$,  $\delta'$ contains the transition
    $$
    \triple{\pair{\alocation}{i}}{\arbitraryletter}{\pair{\pair{\acons(i, \aletter, \aletter_u)}{\pair{\alocation_u}{m(i,d_u)}}}{
        \pair{\true}{\alocation^{\star}}}}
    $$
    with the expression $m(i,d_u)$ just above defined from
     $m(i,\rightarrow) \egdef i+1$ and $m(i,\leftarrow) \egdef i-1$.
    
   \end{itemize}    

\item For all $\alocation \in \locations$ such that $g(\alocation)=\forall$  and $\alocation \neq \alocation_0$, $\aletter \in \aalphabet$ with
    $\delta(\alocation, \aletter) = \set{\triple{\alocation_1}{\aletter_1}{d_1},\triple{\alocation_2}{\aletter_2}{d_2}}$,
    for all $i \in \interval{1}{\beta}$,  $\delta'$ contains the transition 
    $$
    \triple{\pair{\alocation}{i}}{\arbitraryletter}{\pair{\pair{\acons(i, \aletter, \aletter_1)}{\pair{\alocation_1}{m(i,d_1)}}}{
        \pair{\acons(i, \aletter, \aletter_2)}{\pair{\alocation_2}{m(i,d_2}}}}.
    $$

  \item Assuming that $\aword = \aletterbis_1 \cdots \aletterbis_n$ and
        $\awordbis = \rhd \aword \sharp^{\beta'}$ with $\beta'
    = \beta - (n+1)$, for all
    $\aletter \in \aalphabet$ with
    $\delta(\alocation_0, \aletter) = \set{\triple{\alocation_1}{\aletter_1}{d_1},
      \triple{\alocation_2}{\aletter_2}{d_2}}$,
    $\delta'$ contains the transition 
    $$
    \triple{\pair{\alocation_0}{1}}{\arbitraryletter}{\pair{\pair{\acons'(1, \aletter, \aletter_1)}{\pair{\alocation_1}{m(1,d_1)}}}{
        \pair{\acons'(1, \aletter, \aletter_2)}{\pair{\alocation_2}{m(1,d_2}}}}, 
    $$
    where
    $
    \acons'(1, \aletter, \aletter_u) \egdef
    \acons(1, \aletter, \aletter_u) \wedge
    \bigwedge_{1 \leq j \leq \beta} \avariable_j = \amap(\awordbis(j))
    $ ($u \in \set{1,2}$).
    Observe that necessarily $d_1 = d_2 = \rightarrow$ because the head cannot go to the left of the first position. 
    The location $\alocation_0$ has a special status and the transitions from it encodes the input
    word $\aword$ (and that is why we never come back to it). 
\end{itemize}

Correctness of the construction is stated below. 

\begin{lem} \label{lemma-correctness-hardness-tca}
  $\aatm$ accepts $\aword$ iff $\alang(\aautomaton_{\aatm, \aword}) \neq \emptyset$.
\end{lem}
\begin{proof} Let us state a few simple properties that are used in the sequel (whose easy proofs
  are omitted herein).
\begin{description}

  \item[(I)] For every configuration $\aconfiguration$ using $\beta$ tape cells, there is a unique 
        numerical configuration $\triple{\alocation}{i}{\adatum_1, \ldots, \adatum_{\beta}}$ such that
        $\aconfiguration \approx \triple{\alocation}{i}{\adatum_1, \ldots, \adatum_{\beta}}$.

  \item[(II)] For every numerical configuration  $\triple{\alocation}{i}{\adatum_1, \ldots, \adatum_{\beta}}$,
        there is a unique configuration such that $\aconfiguration \approx \triple{\alocation}{i}{\adatum_1, \ldots, \adatum_{\beta}}$.

      \item[(III)] For every configuration $\aconfiguration$ with universal control state such that $\aconfiguration \vdash_{\aatm} \aconfiguration_1$ and
        $\aconfiguration \vdash_{\aatm} \aconfiguration_2$ using the transition function $\delta$,
        there are numerical configurations
        \begin{center}
        $\triple{\alocation}{i}{\adatum_1, \ldots, \adatum_{\beta}}$, $\triple{\alocation_1}{i_1}{\adatum_1^1, \ldots, \adatum_{\beta}^1}$ and
          $\triple{\alocation_2}{i_2}{\adatum_1^2, \ldots, \adatum_{\beta}^2}$,
        \end{center}
        and a transition 
        $\triple{\pair{\alocation}{i}}{\dag}{\pair{\pair{\acons_1}{\pair{\alocation_1}{i_1}}}{\pair{\acons_2}{\pair{\alocation_2}{i_2}}}}
        \in \delta'$ such that
        $$
        \aconfiguration \approx \triple{\alocation}{i}{\adatum_1, \ldots, \adatum_{\beta}}, \ \ \ \ 
        \aconfiguration_1 \approx \triple{\alocation_1}{i_1}{\adatum_1^1, \ldots, \adatum_{\beta}^1}, \ \ \ \
        \aconfiguration_2 \approx \triple{\alocation_2}{i_2}{\adatum_1^2, \ldots, \adatum_{\beta}^2}
        $$
$$
        \Zed \models \acons_1(\overline{\adatum},\overline{\adatum^1}), \ \ \ \
        \Zed \models \acons_2(\overline{\adatum},\overline{\adatum^2}).
        $$

  \item[(IV)] For every configuration $\aconfiguration$ with existential control state such that $\aconfiguration \vdash_{\aatm} \aconfiguration_1$,
    there are numerical configurations $\triple{\alocation}{i}{\adatum_1, \ldots, \adatum_{\beta}}$ and $\triple{\alocation_1}{i_1}{\adatum_1^1, \ldots, \adatum_{\beta}^1}$,
    and a transition  $\triple{\pair{\alocation}{i}}{\dag}{\pair{\pair{\acons_1}{\pair{\alocation_1}{i_1}}}{\pair{\top}{\alocation^{\star}}}}$ such that
    $$
    \aconfiguration \approx \triple{\alocation}{i}{\adatum_1, \ldots, \adatum_{\beta}}, \ \ \ \ 
    \aconfiguration_1 \approx \triple{\alocation_1}{i_1}{\adatum_1^1, \ldots, \adatum_{\beta}^1}, \ \ \ \
    \Zed \models \acons_1(\overline{\adatum},\overline{\adatum^1}).
    $$
    Conditions  (III) and (IV) are proved using condition (I) and the way $\delta'$ is defined.
    Moreover, this means that the relationships between nodes in an accepting run can be simulated by runs on trees for TCA.

  \item[(V)] For all numerical configurations $\triple{\alocation}{i}{\adatum_1, \ldots, \adatum_{\beta}}$, 
    $\triple{\alocation_1}{i_1}{\adatum_1^1, \ldots, \adatum_{\beta}^1}$ and \\ $\triple{\alocation_2}{i_2}{\adatum_1^2, \ldots, \adatum_{\beta}^2}$,
    and transitions  $\triple{\pair{\alocation}{i}}{\dag}{\pair{\pair{\acons_1}{\pair{\alocation_1}{i_1}}}{\pair{\acons_2}{\pair{\alocation_2}{i_2}}}}$ such that
    $\Zed \models \acons_1(\overline{\adatum},\overline{\adatum^1})$ and 
    $\Zed \models \acons_2(\overline{\adatum},\overline{\adatum^2})$, there are configurations $\aconfiguration$, $\aconfiguration_1$ and $\aconfiguration_2$ such that
    $$
        \aconfiguration \approx \triple{\alocation}{i}{\adatum_1, \ldots, \adatum_{\beta}}, \ \ \ \ 
        \aconfiguration_1 \approx \triple{\alocation_1}{i_1}{\adatum_1^1, \ldots, \adatum_{\beta}^1}, \ \ \ \
        \aconfiguration_2 \approx \triple{\alocation_2}{i_2}{\adatum_1^2, \ldots, \adatum_{\beta}^2},
        $$
$$
        \aconfiguration \vdash_{\aatm} \aconfiguration_1, \ \ \ \ 
        \aconfiguration \vdash_{\aatm} \aconfiguration_2.
        $$

  \item[(VI)] For all numerical configurations $\triple{\alocation}{i}{\adatum_1, \ldots, \adatum_{\beta}}$ and 
    $\triple{\alocation_1}{i_1}{\adatum_1^1, \ldots, \adatum_{\beta}^1}$ and transitions
    $$\triple{\pair{\alocation}{i}}{\dag}{\pair{\pair{\acons_1}{\pair{\alocation_1}{i_1}}}{\pair{\top}{\alocation^{\star}}}}$$
    such that $\Zed \models \acons_1(\overline{\adatum},\overline{\adatum^1})$, there are configurations $\aconfiguration$ and $\aconfiguration_1$ such that
    $$
    \aconfiguration \approx \triple{\alocation}{i}{\adatum_1, \ldots, \adatum_{\beta}}, \ \ \ \ 
    \aconfiguration_1 \approx \triple{\alocation_1}{i_1}{\adatum_1^1, \ldots, \adatum_{\beta}^1}, \ \ \ \
    \aconfiguration \vdash_{\aatm} \aconfiguration_1.
    $$
    Conditions (V) and (VI) are proved using condition (II) and the way $\delta'$ is defined.
    Moreover, this means that the relationships between nodes in runs on  trees for TCA can be simulated by accepting runs.

  \end{description}

  \noindent
  ($\Rightarrow$) Let $\atree: \domain{\atree} \rightarrow \Configs(\aatm)$ be an accepting run of $\aatm$ on the input word $\aword$.
  Let $\atree^{\star}: \interval{0}{1}^* \rightarrow \Zed^{\beta}$ be an infinite tree such that
  for all $\anode \in \domain{\atree}$ with $\atree(\anode)  \approx
  \triple{\alocation}{i}{\adatum_1, \ldots, \adatum_{\beta}}$, we have $\atree^{\star}(\anode) \egdef
  \tuple{\adatum_1}{\adatum_{\beta}}$. Note that $\atree^{\star}(\anode)$ has a unique value by Condition (I) when
  $\anode \in \domain{\atree}$ and we have omitted to represent the unique possible letter on each node.
  Let $\arun^{\star}: \interval{0}{1}^* \rightarrow \delta'$ be the run defined as follows. 
  \begin{itemize}
  \item For all $\anode \in \domain{\atree}$ with $\atree(\anode) \approx
    \triple{\alocation}{i}{\adatum_1, \ldots, \adatum_{\beta}}$, we have $\arun^{\star}(\anode) \egdef (\pair{\alocation}{i},\cdot, \cdot,\cdot)$, that is, a transition with source location $\pair{\alocation}{i}$.
    By way of example, assuming
    that $\delta(\alocation, \amap^{-1}(\adatum_i)) =
    \set{\triple{\alocation_1}{\aletter_1}{d_1},\triple{\alocation_2}{\aletter_2}{d_2}}$,
    then more precisely,
    $$
    \arun^{\star}(\anode) \egdef
    \triple{\pair{\alocation}{i}}{\arbitraryletter}{\pair{\pair{\acons(i, \amap^{-1}(\adatum_i),
          \aletter_1)}{\pair{\alocation_1}{m(i,d_1)}}}{
        \pair{\acons(i, \amap^{-1}(\adatum_i), \aletter_2)}{\pair{\alocation_2}{m(i,d_2}}}}.
    $$
    
\item For all $\anode \in (\interval{0}{1}^* \setminus \domain{\atree})$ such that there is no strict prefix $\anode'$
  of $\anode$ that is a leaf of $\domain{\atree}$, we have $\arun^{\star}(\anode) \egdef
  \triple{\alocation^{\star}}{\arbitraryletter}{\pair{\pair{\top}{\alocation^{\star}}}{\pair{\top}{\alocation^{\star}}}}$. 
  \item  For all $\anode \in (\interval{0}{1}^* \setminus \domain{\atree})$ such that there is a strict prefix $\anode'$
    of $\anode$ that is a leaf of $\domain{\atree}$ and $\atree(\anode') \approx \triple{\alocation}{i}{\adatum_1, \ldots, \adatum_{\beta}}$,
    we have $\arun^{\star}(\anode) \egdef
    \triple{\pair{\alocation}{i}}{\arbitraryletter}{\pair{\pair{\top}{\pair{\alocation}{i}}}{\pair{\top}{\pair{\alocation}{i}}}}$. 
  \end{itemize}
  One can easily show that $\arun^{\star}$ is an accepting run on  $\atree^{\star}$ by using the conditions (III) and (IV),
  as well as the definition of the set $F$ of accepting states in $\aautomaton_{\aatm, \aword}$.

  \noindent
  ($\Leftarrow$) Let $\atree: \interval{0}{1}^* \rightarrow \aalphabet' \times \Zed^{\beta}$ be an infinite tree
  and $\arun: \interval{0}{1}^* \rightarrow \delta'$ be an accepting run on $\atree$.
  By construction of $\aautomaton_{\aatm, \aword}$, along any infinite branch, once a location in $F$ is visited, it is visited
  forever along that branch, and moreover any infinite branch always visits such a location in $F$.
  Let $\aset$ be the finite subset of $\interval{0}{1}^*$ such that $\anode \in \aset$ iff
  either $\arun(\anode)$'s source locations is not in $F$, or the source location of $\arun(\anode)$ is in $F \setminus \set{\alocation^{\star}}$ and the source locations of all its ancestors do not belong to $F$.
  One can check that $\aset$ is a finite tree (use of K\"onig's Lemma here). 
  Let $\atree^{\star}: \aset \rightarrow \Configs(\aatm)$ be the map such that for all $\anode \in \aset$, we have
  $\atree^{\star}(\anode) = \aconfiguration$ for the unique configuration
  such that $\aconfiguration \approx \triple{\alocation}{i}{\adatum_1, \ldots, \adatum_{\beta}}$
  with the source location of $\arun(\anode)$ being $\pair{\alocation}{i}$ and $\atree(\anode) = \tuple{\adatum_1}{\adatum_{\beta}}$.
  One can easily show that $\atree^{\star}$ is an accepting run for $\aword$ by using the conditions (V) and (VI),
  as well as the definition of the set of accepting states in $\aatm$. 
\end{proof}
This concludes the \exptime-hardness proof for the nonemptiness problem for
tree constraint  automata. Observe that it is also easy to get \exptime-hardness
without using any constant. Indeed, $N$ distinct constants can be simulated by $N$ variables
whose values remain unchanged along the accepted tree. The initial constraints between these
variables at the root correspond to the constraints between the corresponding constants.
Furthermore, note that the ordering $<$ is not needed for the reduction (the constraints use
only the equality), the \exptime-hardness holds without constants and $<$. 
\end{proof} 

\section{Proofs for Section~\ref{section-complexity-ctlz}}
\subsection{Proof of Proposition~\ref{proposition-simple-form}}
\label{appendix-proof-proposition-simple-form}
\begin{proof} First, we establish that $\CTL(\Zed)$ has the tree model
  property using the standard unfolding technique for Kripke structures. Then, we show that the restriction
  to formulae in simple form is possible using the renaming technique (correctness is guaranteed by the
  tree model property).

  Let $\aks = \triple{\worlds}{\arelation}{\avaluation}$ be a  total Kripke structure and $w \in \worlds$.
  We write $\hat{\aks}_{w}$ to denote the (total) Kripke structure $\triple{\hat{\worlds}}{\hat{\arelation}}{\hat{\avaluation}}$ defined as follows
  (understood as the unfolding of $\aks$ from the world $\aworld$).
  \begin{itemize}
  \item $\hat{\worlds}$ is the set of finite paths in $\aks$ starting from the world $\aworld$.
  \item The relation $\hat{\arelation}$ contains all the pairs of the form $\pair{\apath}{\apath'}$
    such that $\apath$ is of the form $w_1 \cdots w_n$ and $\apath'$ is of the form $w_1 \cdots w_{n+1}$
    with $\pair{w_n}{w_{n+1}} \in \arelation$.
  \item For all paths $\apath = w_1 \cdots w_n$, for all variables
        $\avariable \in \VAR$, we have $\hat{\avaluation}(\apath, \avariable) \egdef \avaluation(w_n,\avariable)$.
  \end{itemize}
  Satisfaction of all the state formulae is preserved using classical arguments. More precisely,
  for all $\apath = w_1 \cdots w_n \in \hat{\worlds}$, for all state formulae $\aformula$ in $\CTL(\Zed)$, we have
  $\aks, w_n \models \aformula$ iff $\hat{\aks}_{w}, \apath \models \aformula$.
  The proof is by structural induction using the correspondence between paths in $\aks$ from some $\aworld'\in\worlds$ such that $\aworld'$ is reachable
  from $\aworld$, and paths in $\hat{\aks}_{\aworld}$. 
  As a conclusion, the logic $\CTL(\Zed)$ has the tree model property since the structures of the
  form $\hat{\aks}_{w}$ are trees.
  
   \begin{figure}
\begin{center}
  \begin{tabular}{ll}
  \scalebox{1}{
  \begin{tikzpicture}[->,>=stealth',shorten >=1pt,auto,node distance=4cm,thick,node/.style={circle,draw,scale=0.9}, roundnode/.style={circle, black, draw=black},]
\tikzset{every state/.style={minimum size=0pt},
dotted_block/.style={draw=black!20!white, line width=1pt, dash pattern=on 1pt off 2pt on 6pt off 2pt,
            inner sep=3mm, minimum width=5.4cm, rectangle, rounded corners}}, 
;
\node[roundnode] (eps) at (0,0) {}; 
\node[roundnode] (0) at (-.5,-.7) {}; 
\node[roundnode] (1) at (.5,-.7) {};
\node[roundnode] (11) at (1,-1.4) {};
\node[roundnode] (111) at (1.5,-2.1) {};

\node at (-.5, -1.1) {$\vdots$}; 
\node at (1.5, -2.5) {$\vdots$}; 

\path [-] (eps)  edge (0);
\path [-] (eps)  edge (1);
\path [-] (1)  edge (11);
\path [-] (11)  edge (111);

\node[gray!50] at (.8,0.4) {\scriptsize{$(\avariable_1,\avariable_2)$}};
\node at (.8,0) {\large{$(  2, 2 )$}};
\node at (1.4,-.7) {\large{$( 3, 4  )$}};
\node at (1.9,-1.4) {\large{$( 9,  7  )$}};
\node at (2.5,-2.1) {\large{$( 5,    8  )$}};

\node (a) at (1.2, -.7) {}; 
\node (b) at (2.7, -2.1) {}; 
\path [->, line width=.5mm, cyan, dotted] (a)  edge[bend left=90]  node [above] {$<$} (b); 

\node (c) at (0.8,0) {}; 
\node (d) at (1.2,0) {}; 
\path [-,orange]  (c) edge[loop below] node [right] {$=$}  (c); 

\node at (0,-3.4) {};

 	\end{tikzpicture} 
}
&  \hspace{1cm}\scalebox{0.8}{
   \begin{tikzpicture}[->,>=stealth',shorten >=1pt,auto,node distance=4cm,thick,node/.style={circle,draw,scale=0.9}, roundnode/.style={circle, black, draw=black},]
\tikzset{every state/.style={minimum size=0pt},
dotted_block/.style={draw=black!40, line width=1pt, dash pattern=on 1pt off 1pt,
            inner sep=3mm, minimum height=3cm, rectangle, rounded corners}}, 
;

\node[roundnode] (eps) at (0,0) {}; 
\node[roundnode] (0) at (-.5,-.7) {}; 
\node[roundnode] (1) at (.5,-.7) {};
\node[roundnode] (11) at (1,-1.4) {};
\node[roundnode] (111) at (1.5,-2.1) {};

\node at (-.5, -1.1) {$\vdots$}; 
\node at (1.5, -2.5) {$\vdots$}; 

\path [-] (eps)  edge (0);
\path [-] (eps)  edge (1);
\path [-] (1)  edge (11);
\path [-] (11)  edge (111);

\node[gray!50] at (2.8,0.4) {\scriptsize{$(\avariable^{-3}_1,\avariable^{-2}_1,\avariable^{-1}_1, \avariable^0_1, \avariable^{-3}_2,\avariable^{-2}_2,\avariable^{-1}_2, \avariable^0_2)$}};
\node at (2.8,0) {\large{$(\cdot \ , \ \cdot \  , \  \cdot \  ,  2 \ , \  \cdot \ , \  \cdot \  , \  \cdot \ , 2)$}}; 
\node at (3.4,-.7) {\large{$(\cdot \ , \  \cdot \ , 2 \  , 3 \ , \  \cdot \ , \  \cdot \ , 2 \  , \ 4 )$}};
\node at (3.94,-1.4) {\large{$(\cdot \ , 2 \ , 3 \ , \hspace{.4mm} 9 \  , \  \cdot \  , 2 \  , \  4 \ , 7)$}};
\node at (4.5,-2.1) {\large{$(2 \ , 3 \ ,  \hspace{.4mm} 9 \ ,  5 \ ,  2 \ , \  4 \  ,  7 \  , \  8)$}};

\node (a) at (2.8, -2.2) {}; 
\node (b) at (6.4, -2.2) {}; 
\path [->, line width=.5mm, cyan, dotted] (a)  edge[bend left=-50]  node [below] {$<$} (b); 

\node (c) at (2.2,-2.2) {}; 
\node (d) at (4.6,-2.2) {}; 
\path [-,orange]  (c) edge[bend left=-90]  node [below] {$=$}  (d); 
 	\end{tikzpicture} 
}
\end{tabular}
\end{center}
\caption{A tree over two variables $\avariable_1$ and $\avariable_2$ (left) and the corresponding tree over eight variables $\avariable^{-3}_1,\avariable^{-2}_1,\avariable^{-1}_1,\avariable^{0}_1,\avariable^{-3}_2,\avariable^{-2}_2,\avariable^{-1}_2,\avariable^{0}_2$ (right) for illustrating the translation of the formula $\existspath \mynext  \ \existspath ((\mynext \avariable_1 < \mynext \mynext \mynext \avariable_2) \wedge (\avariable_1 = \avariable_2))$ into its simple form. }
\label{figure-graph-symb-tree}
\end{figure}

  Now, we move to the construction of $\aformula'$ in simple form. Let us introduce below the natural notion of forward degree. 
  Given a term $\aterm = \mynext^i \avariable$ for some $i \in \Nat$,
  the \label{page-fd} \defstyle{forward degree} of $\aterm$, written $\fd{\aterm}$,
  is equal to $i$. Given a constraint $\acons$, we write $\fd{\acons}$ to denote the  \defstyle{forward degree} of $\acons$
  defined as the maximal forward degree of any term occurring in $\acons$. For instance,
  $\fd{(\mynext \avariable_1 < \mynext \mynext \mynext \avariable_2) \wedge   (\avariable_1 = \avariable_2)} = 3$.
  By extension, given a $\CTL(\Zed)$ formula $\aformula$, we write $\fd{\aformula}$ to denote
  the  \defstyle{forward degree} of $\aformula$ defined as the maximal forward degree of any constraint occurring in $\aformula$.
  For instance, $\fd{\existspath \mynext  \ \existspath ((\mynext \avariable_1 < \mynext \mynext \mynext \avariable_2) \wedge (\avariable_1 = \avariable_2))}
  = 3$ and $\fd{\existspath \ ((\mynext \mynext \avariable_1 < \mynext \mynext \mynext \mynext \avariable_2) \wedge (\mynext \avariable_1 = \mynext \avariable_2))}
  = 4$. 

  Let $\aformula$ be a $\CTL(\Zed)$ formula in negation normal form  such that $\fd{\aformula} = N \geq 2$ and the variables occurring in
  $\aformula$ are among $\avariable_1, \ldots, \avariable_{\beta}$.
  Below, we build a formula $\aformula'$  over 
  $\avariable_1^{-N}, \ldots, \avariable_1^{0}, \ldots, \avariable_{\beta}^{-N}, \ldots, \avariable_{\beta}^{0}$
  with $\fd{\aformula'} \leq 1$ such that $\aformula$ is satisfiable in a tree Kripke structure iff
  $\aformula'$ is satisfiable in a tree Kripke structure and $\aformula'$ can be computed in polynomial-time in the size of
  $\aformula$.
  The above-mentioned variables can be obviously renamed but the current naming is helpful to grasp the correctness of the whole enterprise.
  In short, the value for $\avariable_j^{-k}$ on a node  should be understood as the value for $\avariable_j$ exactly $k$ nodes behind along
  the branch leading to the node. Thanks to the tree structure, the variable $\avariable_j^{-k}$ takes a unique value.
  Note that we could not work with forward values (say with $\avariable_j^{+k}$ to keep
  the same syntactic rule for naming)
  because the structures are branching and therefore $k$ steps ahead
  does not lead necessarily to a unique world.
  
  Given  $\existspath \ \acons$ or $\forallpaths \ \acons$
  occurring in $\aformula$ with $\fd{\acons} = M$ (so $M \leq N$), we write $\mathtt{jump}(\acons,M)$ to denote
  the constraints built over $\avariable_1^{-N}, \ldots, \avariable_1^{0}, \ldots, \avariable
  _{\beta}^{-N}, \ldots, \avariable_{\beta}^{0}$
  such that
  \begin{itemize}
  \item $\mathtt{jump}(\cdot,M)$ is homomorphic for Boolean connectives,
  \item $\mathtt{jump}(\aterm_1 < \aterm_2, M) \egdef \mathtt{jump}(\aterm_1, M) < \mathtt{jump}(\aterm_2, M)$,
  \item $\mathtt{jump}(\aterm_1 = \aterm_2, M) \egdef \mathtt{jump}(\aterm_1, M) = \mathtt{jump}(\aterm_2, M)$,
  \item $\mathtt{jump}(\aterm = \adatum, M) \egdef \mathtt{jump}(\aterm, M) = \adatum$ with 
    $\mathtt{jump}(\mynext^i \avariable_j, M) \egdef \avariable_j^{i-M}$.
  The value $i-M$ is greater than $-N$ because $\fd{\acons} \leq N$. 
  \end{itemize}
  For instance, $\mathtt{jump}((\mynext \avariable_1 < \mynext \mynext \mynext \avariable_2) \wedge (\avariable_1 = \avariable_2),3) =
  (\avariable_1^{-2} < \avariable_2^{0}) \wedge (\avariable_1^{-3} = \avariable_2^{-3})$. 
  In short, $\mathtt{jump}(\acons,M)$ corresponds to the constraint equivalent to $\acons$ if we evaluate it exactly $M$ steps ahead.
  To do so, we therefore need to evoke in $\mathtt{jump}(\acons,M)$ variables capturing previous local values along the branch.
  Let $\atranslation$ be the translation map that is homomorphic for Boolean and temporal connectives such that
  \begin{itemize}
  \item $\atranslation(\existspath \ \acons) \egdef (\existspath \mynext)^M  \ \mathtt{jump}(\acons,M)$, 
  \item $\atranslation(\forallpaths \ \acons) \egdef (\forallpaths \mynext)^M \ \mathtt{jump}(\acons,M)$, 
  \end{itemize}
   where  $\fd{\acons} = M$.
   Let $\aformula'$ be defined as follows:
  $$
   \atranslation(\aformula) \wedge \forallpaths \always \ \forallpaths \big(\bigwedge_{j \in \interval{1}{\beta},
     \ k \in \interval{0}{N-1}} \avariable_j^{-k} = \mynext \avariable_j^{-k-1}\big).
  $$
  The second conjunct in the definition of $\aformula'$  corresponds to the renaming part.
  Recall that $\forallpaths \always$ is the \CTL temporal connective that enforces a property for all
  the reachable nodes from the root node (i.e. to all the worlds of the tree model). 
  Let us provide hints to understand why $\aformula$ is satisfiable in a tree Kripke structure iff
  $\aformula'$ is satisfiable in a tree Kripke structure.

  First, suppose that $\aks, \aworld \models \aformula$, where $\aks = \triple{\worlds}{\arelation}{\avaluation}$ is a tree Kripke
  structure with root $\aworld$.
  Let $\aks' = \triple{\worlds}{\arelation}{\avaluation'}$ be the Kripke structure that differs from
  $\aks$ only in the definition of the valuation.
  Given a node $\aworld' \in \worlds$ reachable from $\aworld$ via the branch $\aworld_0 \cdots \aworld_n$ with $\aworld_0 = \aworld$ and $\aworld_n = \aworld'$ for some $n \geq 0$,
  for all $k \in \interval{0}{N}$ and $j \in \interval{1}{\beta}$, we require
  $\avaluation'(w', \avariable_j^{-k}) \egdef \avaluation(w_{n-k}, \avariable_j)$ if $n-k \geq 0$,
  otherwise $\avaluation'(\aworld', \avariable_j^{-k}) \egdef 0$ (arbitrary value).
  Figure~\ref{figure-graph-symb-tree} illustrates how $\avaluation'$ is defined. 
  In order to verify that $\aks', \aworld \models \aformula'$, it boils down to check the three properties below (for which we omit the proofs at this stage).
  \begin{itemize}
  \item $\aks', \aworld \models \forallpaths \always \ \forallpaths \big(\bigwedge_{j \in \interval{1}{\beta}, \
    k \in \interval{0}{N-1}} \avariable_j^{-k} = \mynext \avariable_j^{-k-1}\big)$.
    This holds thanks to the definition of $\avaluation'$ and the tree structure of $\aks$.
    
  \item If $\existspath \ \acons$ occurs in $\aformula$ and $\aks, \aworld' \models \existspath \ \acons$, then
    $\aks', \aworld' \models (\existspath \mynext)^M \ \mathtt{jump}(\acons,M)$, where  $\fd{\acons} = M$.

  \item If $\forallpaths \ \acons$ occurs in $\aformula$ and $\aks, \aworld' \models \forallpaths \ \acons$, then
    $\aks', \aworld' \models (\forallpaths \mynext)^M \ \mathtt{jump}(\acons,M)$, where  $\fd{\acons} = M$.
    
  \end{itemize}

  For the other direction, suppose that $\aks, \aworld \models \aformula'$ and $\aks = \triple{\worlds}{\arelation}{\avaluation}$ is a tree Kripke
  structure with root $\aworld$ (recall the tree model property holds).
   Let $\aks' = \triple{\worlds}{\arelation}{\avaluation'}$ be the Kripke structure that differs from
   $\aks$ only in the definition of the valuation. More precisely, for all $\aworld' \in \worlds$
   and all $j \in \interval{1}{\beta}$, we have
   $\avaluation'(\aworld',\avariable_j) \egdef \avaluation(\aworld',\avariable_j^0)$.
   By structural induction, one can show that $\aks, \aworld \models \atranslation(\aformula)$
   implies $\aks', \aworld \models \aformula$.
\end{proof}

\subsection{Proof of Proposition~\ref{proposition-tree-model-for-Z}}
\label{appendix-proof-proposition-tree-model-for-Z}
\begin{proof}
The direction from right to left (``if'') is trivial. So let us prove the direction from left to right (``only if''): suppose that
$\aformula$ is a satisfiable $\CTL(\Zed)$ formula in simple form and built over the terms
$\avariable_1,\dots,\avariable_\beta, \mynext \avariable_1,\dots,\mynext \avariable_\beta$.
We recall that the terms $\mynext \avariable_1,\dots,\mynext \avariable_\beta$ are also sometimes denoted
by the primed variables $\avariable'_1,\dots,\avariable'_\beta$. 
Let $\iota: (\parsubf{\existspath \mynext}{\aformula} \cup \parsubf{\existspath}{\aformula})
\rightarrow \interval{1}{\degree}$ be a direction map for $\aformula$, 
where $\degree = \card{\parsubf{\existspath \mynext}{\aformula}} + \card{\parsubf{\existspath}{\aformula}}$.

Let $\adks=\triple{\worlds}{\arelation}{\avaluation}$ be a  total Kripke structure,
$\aworld_\init\in\worlds$ be a world in $\adks$ such that $\adks,\aworld_\init\models\aformula$.
Since $\aformula$ contains only the variables in $\avariable_1,\dots,\avariable_\beta$, the map
$\avaluation$ can be restricted to the variables among $\avariable_1,\dots,\avariable_\beta$.
Furthermore,  below, we can represent $\avaluation$ as a map $\worlds \to \Zed^{\beta}$ such that
$\avaluation(\aworld)(i)$ for some $i \in \interval{1}{\beta}$ is understood as
the value of the variable $\avariable_i$ on $\aworld$. 
Below, we construct a tree $\atree: \interval{0}{\degree}^* \to \Zed^\beta$ such that
$\atree$ obeys $\iota$ and $\atree, \varepsilon \models \aformula$.

We introduce an auxiliary map $g:  \interval{0}{\degree}^* \rightarrow \worlds$ such that
$g(\varepsilon) \egdef \aworld_{\init}$, $\atree(\varepsilon) \egdef \avaluation(g(\varepsilon))$ and
more generally, we require that for all $\anode \in \interval{0}{\degree}^*$, we have
$\atree(\anode) \egdef \avaluation(g(\anode))$. The definition of $g$ is performed by picking
the smallest
element $\anode \cdot j \in \interval{0}{\degree}^*$ with respect to the lexicographical ordering such that
$g(\anode)$ is defined and $g(\anode \cdot j)$ is undefined.
Let $\anode \cdot j$ be the smallest node such that $g(\anode)$ is defined and $g(\anode \cdot j)$ is undefined.
If $j=0$, then, since $\adks$ is total, there is an infinite path $\apath = \aworld_0 \aworld_1 \aworld_2 \cdots$ starting from  $g(\anode)$. For all $i \geq 1$, we set $g(\anode \cdot 0^i) \egdef \aworld_i$ and $\atree(\anode \cdot 0^i) \egdef \avaluation(\aworld_i)$. 
So let $j>0$ and $\aformula'$ be the unique subformula of $\aformula$ such that $\iota(\aformula')=j$. 
If $\adks, g(\anode) \not \models \aformula'$, then,  since $\adks$ is total, there is an infinite path $\apath = \aworld_0 \aworld_1 \aworld_2 \cdots$
  starting from  $g(\anode)$. We define $g(\anode \cdot j) \egdef \aworld_1$, $\atree(\anode \cdot j) \egdef \avaluation(\aworld_1)$ and
  for all $i \geq 1$, we set $g(\anode \cdot j \cdot 0^i) \egdef \aworld_{i+1}$
  and $\atree(\anode \cdot j \cdot 0^i) \egdef \avaluation(\aworld_{i+1})$. 
  If $\adks, g(\anode) \models \aformula'$, then we distinguish the following cases.   
\begin{description}
 \item[Case $\aformula'=\existspath \ \acons$] By $\adks, g(\anode) \models \existspath \ \acons$, 
there is an infinite path $\apath = \aworld_0 \aworld_1 \aworld_2 \cdots$
starting from  $g(\anode)$ and such that $\Zed \models \acons(\avaluation(\aworld_0),\avaluation(\aworld_1))$.
Recall that all the terms in $\acons$ are in $\myterms{\leq 1}{\VAR}$. 
As above, we define $g(\anode \cdot j) \egdef \aworld_1$, $\atree(\anode \cdot j) \egdef \avaluation(\aworld_1)$ and
  for all $i \geq 1$, we set $g(\anode \cdot j \cdot 0^i) \egdef \aworld_{i+1}$
  and $\atree(\anode \cdot j \cdot 0^i) \egdef \avaluation(\aworld_{i+1})$.
\item[Case $\aformula'=\existspath \mynext \existspath (\aformula_1 \until \aformula_2)$]

  By $\adks, g(\anode) \models \existspath \mynext \existspath (\aformula_1 \until \aformula_2)$,
   there is an infinite path $\apath = \aworld_0 \aworld_1 \aworld_2 \cdots$ starting from  $g(\anode)$
  and $k \geq 1$ such that $\adks, \aworld_k \models \aformula_2$ and for all $i \in \interval{1}{k-1}$
  we   have
  $\adks, \aworld_i \models \aformula_1$.
  For all $i \in \interval{1}{k}$, we set $g(\anode \cdot j^i) \egdef \aworld_i$ and $\atree(\anode \cdot j^i) \egdef \avaluation(\aworld_i)$,
  and
  $g(\anode \cdot j^k \cdot 0^i) \egdef \aworld_{k+i}$ and 
  $\atree(\anode \cdot j^k \cdot 0^i) \egdef \avaluation(\aworld_{k+i})$    
  for all $i \geq 1$. 
  \item[Case $\aformula'=\existspath \mynext \existspath (\aformula_1 \release \aformula_2)$]
 
  If 
  there is an infinite path $\apath = \aworld_0 \aworld_1 \aworld_2 \cdots$ starting from  $g(\anode)$
  such that 
  $\adks, \aworld_i \models \aformula_2$ for all $i\geq 1$,
  then for all $i > 0 $, $g(\anode \cdot j^i) \egdef \aworld_i$ and $\atree(\anode \cdot j^i) \egdef \avaluation(\aworld_i)$.
  Otherwise,  
  by $\adks, g(\anode) \models \existspath \mynext \existspath (\aformula_1 \release \aformula_2)$  there must exist an infinite path $\apath = \aworld_0 \aworld_1 \aworld_2 \cdots$ starting from  $g(\anode)$
  and $k \geq 1$ such that $\adks, \aworld_k \models \aformula_1 \wedge \aformula_2 $ and for all $i \in \interval{0}{k-1}$, we
  have 
  $\adks, \aworld_i \models  \aformula_2$, 
  then for all $i \in \interval{1}{k}$, we set $g(\anode \cdot j^i) \egdef \aworld_i$, $\atree(\anode \cdot j^i) \egdef \avaluation(\aworld_i)$,
  and
  $g(\anode \cdot j^k \cdot 0^i) \egdef \aworld_{k+i}$, 
  $\atree(\anode \cdot j^k \cdot 0^i) \egdef \avaluation(\aworld_{k+i})$
  for all $i \geq 1$. 

\item[Case $\aformula'=\existspath \mynext \aformulabis$ and $\aformulabis$ is neither an $\existspath \until$-formula
  nor an $\existspath \release$-formula] \ \\
  By $\adks, g(\anode) \models \existspath \mynext \aformulabis$, 
  there is an infinite path $\apath = \aworld_0 \aworld_1 \aworld_2 \cdots$ starting from  $g(\anode)$
  such that $\adks, \aworld_1 \models \aformulabis$.
  We set $g(\anode \cdot j) \egdef \aworld_1$, $\atree(\anode \cdot j) \egdef \avaluation(\aworld_1)$ and
  for all $i \geq 1$,
  $g(\anode \cdot j \cdot 0^i) \egdef \aworld_{i+1}$ and $\atree(\anode \cdot j \cdot 0^i) \egdef \avaluation(\aworld_{i+1})$.

\end{description}

It remains to show that for all $\anode \in \interval{0}{\degree}^*$ and for all $\aformula' \in \subf{\aformula}$,
if $\adks, g(\anode) \models \aformula'$, then $\atree, \anode \models \aformula'$, and $\atree$ obeys $\iota$, that is,
\begin{itemize}
\item if $\aformula'=\existspath\mynext\phi_1$ and $\atree,\anode\models\aformula'$, then $\atree,\anode\cdot j\models\phi_1$ with $j=\iota(\aformula')$, 
\item if $\aformula'=\existspath \aformula_1\until\aformula_2$ and $\atree,\anode\models\aformula'$, then there exists some $k\geq 0$ such that 
$\atree, \anode\cdot j^i \models\aformula_1$ for all $0\leq i<k$, and $\atree,\anode\cdot j^k\models\aformula_2$  with $j=\iota(\aformula')$, 
\item if $\aformula'= \existspath \ \acons$ and $\atree,\anode\models\aformula'$, 
then $\Zed \models \acons(\atree(\anode),\atree(\anode\cdot j))$  with $j=\iota(\aformula')$. 
\end{itemize}
The proof is by induction on the subformula relation. 
\begin{itemize}
\item Suppose $\adks, g(\anode)\models \existspath \ \acons$, and let $j=\iota(\existspath \ \acons)$ for some $1\leq j\leq \degree$. 
By definition of $\atree$, 
$g(\anode\cdot j)=\aworld_1$ and $g(\anode\cdot j\cdot 0^i)=\aworld_{i+1}$, 
where $\aworld_0, \aworld_1, \dots$ is an infinite path in $\adks$ starting from $g(\anode)$ satisfying
$\Zed\models \acons(\avaluation(\aworld_0), \avaluation(\aworld_1))$. We also have $\atree(\anode)=\avaluation(\aworld_0)$ and
$\atree(\anode \cdot j)=\avaluation(\aworld_1)$, hence the result. 
\item Suppose  $\adks,g(\anode)\models \aformula_1\wedge \aformula_2$. Hence  $\adks,g(\anode)\models \aformula_1$ and $\adks,g(\anode)\models\aformula_2$, so that by the induction hypothesis $\atree,\anode\models\aformula_1$ and $\atree, \anode\models\aformula_2$, and thus
  $\atree,\anode\models\aformula_1\wedge\aformula_2$.
  The case with $\aformula_1 \vee \aformula_2$ is similar. 
\item Suppose $\adks,g(\anode)\models\existspath\mynext\aformula'$.
Suppose $j=\iota(\existspath\mynext\aformula')$ for some $1\leq j\leq \degree$. 
We distinguish three cases (coming from the definition of $\atree$). 
\begin{enumerate}
\item $\aformula'$ is of the form $\existspath\aformula_1\until\aformula_2$: 
By definition of $\atree$, we have
$g(\anode \cdot j^i)=\aworld_i$ for all $0\leq i\leq k$, 
where $\aworld_0 \aworld_1 \aworld_2 \dots$ is a path starting from $g(\anode)$
satisfying $\adks,\aworld_k \models \aformula_2$ and $\adks, \aworld_i\models\aformula_1$ for all $1\leq i<k$. 
By the induction hypothesis, we have 
$\atree, \anode \cdot j^k\models\aformula_2$ and $\atree,\anode\cdot j^i\models\aformula_1$ for all $1\leq i<k$. 
Hence $\atree, \anode\models\existspath\mynext\aformula'$. 
\item $\aformula'$ is of the form $\existspath\aformula_1\release\aformula_2$: 
By definition of $\atree$, there are two cases. 
\begin{itemize}
\item Either $g(\anode \cdot j^i)=\aworld_i$ for all $i\geq 0$, where $\aworld_0 \aworld_1 \aworld_2 \dots$ is a path starting from
  $\anode$ such that
  $\adks,\aworld_i\models \aformula_2$
  for all $i\geq 1$. 
  By the induction hypothesis, we have
  $\atree,\anode \cdot j^i\models \aformula_2$ for all $i\geq 1$, and hence
$\atree,\anode\models\existspath\mynext\aformula'$. 
\item Or $g(\anode \cdot j^i)=\aworld_i$ for all $0\leq i\leq k$, where $\aworld_0 \aworld_1 \aworld_2 \dots$ is a path starting from
  $\anode$ such that $\adks,\aworld_k\models\aformula_1\wedge\aformula_2$ and
  $\adks,\aworld_i\models \aformula_2$
  for all $1\leq i <k$. 
  By the induction hypothesis, we have $\atree,\anode\cdot j^k\models \aformula_1\wedge \aformula_2$ and
  $\atree,\anode \cdot j^i\models \aformula_2$
  for all $1\leq i<k$. Hence
$\atree,\anode\models\existspath\mynext\aformula'$. 
\end{itemize}
\item $\aformula'$ is neither an $\existspath\until$-formula nor an $\existspath\release$-formula: 
Then $g(\anode \cdot j)=\aworld_1$ and $g(\anode\cdot j \cdot 0^i)=\aworld_{i+1}$, 
where $\aworld_0 \aworld_1 \aworld_2\dots $ is a path from $g(\anode)$ such that $\adks,\aworld_1\models\aformula'$. 
By the induction hypothesis, we have $\atree,\anode \cdot j\models\aformula'$, and hence $\atree,\anode\models\existspath\mynext\aformula'$.
\end{enumerate}

\item Suppose $\adks, g(\anode)\models\existspath\aformula_1\until\aformula_2$. 
We distinguish two cases. 
\begin{itemize}
\item Suppose $\adks,g(\anode)\models\aformula_2$. By the induction hypothesis,
  $\atree,\anode\models\aformula_2$ and hence $\atree,\anode\models\existspath\aformula_1\until\aformula_2$. 
\item Suppose $\adks,g(\anode)\not\models\aformula_2$. Then  $\adks,g(\anode)\models\aformula_1$ and  $\adks,g(\anode)\models\existspath\mynext
  \existspath\aformula_1\until\aformula_2$. 
By the induction hypothesis, we also have $\atree,\anode\models\aformula_1$.
Suppose $\iota(\existspath\mynext \existspath\aformula_1\until\aformula_2)=j$ for some $1\leq j\leq \degree$. 
By definition of $\atree$, there exists some (minimal) $k\geq 1$ such that  
$g(\anode \cdot j^i)=\aworld_i$ for all $0\leq i\leq k$, 
where $\aworld_0 \aworld_1  \aworld_2 \dots $ is a path starting from $g(\anode)$ satisfying $\adks, \aworld_k\models\aformula_2$ and
$\adks,\aworld_i\models\aformula_1$ for all $1\leq i<k$. 
By the induction hypothesis, 
we have
$\atree, \anode\cdot j^k\models\aformula_2$ and $\atree,\anode\cdot j^i\models\aformula_1$ for all $1\leq i <k$. 
Recall that we also have $\atree,\anode\models\aformula_1$, so that indeed 
$\atree,\anode\models\existspath\aformula_1\until\aformula_2$. 
\end{itemize}
\item Suppose $\adks,g(\anode)\models\existspath\aformula_1\release\aformula_2$. 
We distinguish two cases.
\begin{itemize}
\item Suppose $\adks,g(\anode)\models\aformula_2$ and $\adks,g(\anode)\models\aformula_1$. 
By the induction hypothesis, $\atree,\anode\models\aformula_2$ and $\atree,\anode\models\aformula_1$, so that 
$\atree,\anode\models\existspath\aformula_1\release\aformula_2$. 
\item Suppose $\adks,g(\anode)\models\neg\aformula_1\wedge \aformula_2$ and $\adks,g(\anode)\models\existspath\mynext \existspath\aformula_1\release\aformula_2$. 
  By the induction hypothesis,
  $\atree,\anode\models \aformula_2$. 
Suppose $\iota(\existspath\mynext \existspath\aformula_1\release\aformula_2)=j$ for some $1\leq j\leq \degree$. 
We distinguish two more cases. 
\begin{itemize}
\item There is an infinite path $\aworld_0 \aworld_1 \aworld_2\dots$ from $g(\anode)$ such that 
  $\adks,\aworld_i\models \aformula_2$ for all $i\geq 1$. 
We then have $g(\anode \cdot j^i)=\aworld_i$ for all $i\geq 1$. 
By the induction hypothesis, we have 
 $\atree, \anode\cdot j^i\models \aformula_2$ for all $i\geq 1$. 
Together with 
$\atree,\anode\models \aformula_2$,
we obtain $\atree,\anode\models\existspath\aformula_1\release\aformula_2$. 
\item Otherwise, we have 
$g(\anode\cdot j^k)=\aworld_k$ and $g(\anode\cdot j^i)=\aworld_i$, where 
  $\aworld_0 \aworld_1 \aworld_2\dots$ is a path from $g(\anode)$ such that $\adks,\aworld_k\models\aformula_1\wedge\aformula_2$ and 
  $\adks,\aworld_i\models \aformula_2$ for all $1\leq i<k$. 
By the induction hypothesis, 
we have 
$\atree,\anode\cdot j^k\models\aformula_1\wedge\aformula_2$, and 
 $\atree, \anode\cdot j^i\models \aformula_2$ for all $1\leq i<k$. 
Together with 
$\atree,\anode\models\aformula_2$,
we obtain $\atree,\anode\models\existspath\aformula_1\release\aformula_2$. 
\end{itemize}
\end{itemize}

\item Suppose $\adks,g(\anode)\models \forallpaths \ \acons$. By construction of $\atree$, 
for all infinite paths $\anode \cdot j_1 \cdot j_2 \cdot j_3 \cdots \in \interval{0}{\degree}^{\omega}$,
$g(\anode) \cdot g(\anode \cdot j_1) \cdot g(\anode \cdot j_1 j_2) \cdots$ is an infinite path
from $g(\anode)$. Consequently, for all $\anode \cdot j_1 \cdot j_2 \cdot j_3 \cdots \in \interval{0}{\degree}^{\omega}$,
we have $\Zed \models \acons(\avaluation(g(\anode)),\avaluation(g(\anode \cdot j_1)))$ and therefore
$\Zed \models \acons(\atree(\anode), \atree(\anode \cdot j_1))$ because by definition, we have
$\atree(\anode) = \avaluation(g(\anode))$ and $\atree(\anode \cdot j) = \avaluation(g(\anode \cdot j))$.
As a consequence, $\atree, \anode \models \forallpaths \ \acons$.

\item Suppose $\adks,g(\anode)\models\forallpaths \ \apathformula$, where $\apathformula$ is a path formula
  such that $\apathformula$ is not a constraint.
By construction of $\atree$,
for all infinite paths $\anode \cdot j_1 \cdot j_2 \cdot j_3 \cdots \in \interval{0}{\degree}^{\omega}$,
$g(\anode) \cdot g(\anode \cdot j_1) \cdot g(\anode \cdot j_1 j_2) \cdots$ is an infinite path
from $g(\anode)$. By way of example, assume that $\apathformula = \aformula_1 \until \aformula_2$
(the cases $\apathformula = \aformula_1 \release \aformula_2$ and $\apathformula = \mynext \aformula_1$
are handled in the very same way). 
For all  infinite paths $\anode \cdot j_1 \cdot j_2 \cdot j_3 \cdots \in \interval{0}{\degree}^{\omega}$
with $\apath = g(\anode) \cdot g(\anode \cdot j_1) \cdot g(\anode \cdot j_1 j_2) \cdots$, we have
$\adks, \apath \models  \aformula_1 \until \aformula_2$ and therefore there is $k \in \Nat$
such that $\adks, g(\anode \cdot j_1 \cdots j_k) \models \aformula_2$ and
for all $0 \leq k' < k$, we have $\adks, g(\anode \cdot j_1 \cdots j_{k'}) \models \aformula_1$.
By the induction hypothesis, we have
$\atree, \anode \cdot j_1 \cdots j_k \models \aformula_2$ and
$\atree, \anode \cdot j_1 \cdots j_{k'} \models \aformula_1$ for all $0 \leq k' < k$.
Consequently, $\atree, \anode \cdot j_1 \cdot j_2 \cdot j_3 \cdots \models \aformula_1 \until \aformula_2$
and therefore $\atree, \anode \models \forallpaths \  \aformula_1 \until \aformula_2$ because
the above path from $\anode$ was arbitrary in $\atree$. \qedhere
\end{itemize}
\end{proof}

\subsection{Proof of Lemma~\ref{lemma-correctness-ctlz-aut-to-form}}
\label{appendix-proof-lemma-correctness-ctlz-aut-to-form}
\begin{proof}
``if:'' 
Suppose $\alang(\aautomaton_{\aformula})\neq \emptyset$. 
Then there exists a tree $\atree:\interval{0}{\degree}^*\to \aalphabet\times\Zed^\beta$ in $\alang(\aautomaton_{\aformula})$. 
Let
$\arun:\interval{0}{\degree}^*\to \delta$
be an accepting run of $\aautomaton_{\aformula}$ on $\atree$.  
We prove for all $\aformula'\in\subf{\aformula}$, for all nodes $\anode$ in $\atree$ with
$\atree(\anode)=(\arbitraryletter,\vect{z})$ and  
$\arun(\anode)=(\pair{d_\anode}{\aset_\anode},\arbitraryletter,\pair{\acons_0}{\pair{0}{\aset_{\anode\cdot 0}}}, \ldots,
\pair{\acons_{\degree}}{\pair{\degree}{\aset_{\anode\cdot \degree}}})$,
if $\aformula'\in \aset_{\anode}$, then $\atree, \anode \models \aformula'$
(this is not an equivalence, recall that $\aformula$ is in negation normal form).
Here, $\atree$ is also understood as a tree Kripke model (in which we may ignore the finite alphabet $\aalphabet$). 
Note that this indeed implies $\atree, \varepsilon \models\aformula$, as $\arun(\varepsilon)$ starts with  some initial location
of the form $(0,\aset_{\varepsilon})$ such that $\aformula \in \aset_{\varepsilon}$. Hence $\aformula$ is indeed satisfiable. 

So let $\anode$ be a node in $\atree$ with 
$$\arun(\anode) = (\pair{d_{\anode}}{\aset_{\anode}}, \arbitraryletter, \pair{\acons_0}{\pair{0}{\aset_{\anode \cdot 0}}}, \dots,
\pair{\acons_\degree}{\pair{\degree}{\aset_{\anode \cdot \degree}}})$$
and such that
\begin{enumerate}
  \item for all $\existspath \mynext \aformulabis \in \aset_\anode$, we have $\aformulabis \in \aset_{\anode\cdot j}$ with $j=\iota(\existspath \mynext \aformulabis)$;
  \item for all $\forallpaths \mynext \aformulabis \in \aset_\anode$ and $j \in \interval{0}{\degree}$,
        we have $\aformulabis \in \aset_{\anode \cdot j}$;
  \item
  for all $j \in \interval{0}{\degree}$,
  if there is $\existspath \ \acons \in \aset_{\anode}$ such that $\iota(\existspath \ \acons) =j$, then
    $$
    \acons_j \egdef (\bigwedge_{\forallpaths \acons' \in \aset_{\anode}} \acons') \wedge \acons
    \ \ \ 
    \mbox{otherwise,}
    \ \ \ 
    \acons_j  \egdef \bigwedge_{\forallpaths \acons' \in \aset_\anode} \acons'.
    $$
    \end{enumerate}
The proof of the claim is by the induction on the formula with respect to the subformula ordering. So suppose $\phi'\in\aset_{\anode}$. 
\begin{itemize}
\item $\aformula' = \existspath \ \acons$, where $\acons$ is a Boolean combination
  of terms among $\avariable_1, \dots, \avariable_\beta,\avariable'_1,\dots,\avariable'_\beta$. 
Suppose $\iota(\aformula')=j$ for some $1\leq j\leq \degree$. 
From the condition (3) above, we obtain that $\acons_j \egdef (\bigwedge_{\forallpaths \acons' \in \aset_{\anode}} \acons') \wedge \acons$. 
Since $\arun$ is a run, we have $\Zed \models \acons_j(\vect{z},\vect{z}_j)$, where $\atree(\anode)=\pair{\arbitraryletter}{\vect{z}}$ and 
$\atree(\anode\cdot j) =\pair{\arbitraryletter}{\vect{z}_j}$, 
and hence also $\Zed \models \acons(\vect{z},\vect{z}_j)$. 
This yields $\atree,\anode \models \existspath \ \acons$. 
\item $\aformula' = \forallpaths \ \acons$, where $\acons$ is a Boolean combination  of atomic constraints built over
  $\avariable_1, \dots, \avariable_\beta$ and $\avariable'_1,\dots,\avariable'_\beta$. 
  From the condition (3) above, we know that $\acons$ appears in $\acons_j$ as a conjunct.  Since $\arun$ is a run, we have
  $\Zed \models \acons_j(\vect{z},\vect{z}_j)$ for all $0\leq j\leq \degree$, where $\atree(\anode)=\pair{\arbitraryletter}{\vect{z}}$ and 
$\atree(\anode\cdot j) =\pair{\arbitraryletter}{\vect{z}_j}$, 
and hence also $\Zed \models \acons(\vect{z},\vect{z}_j)$ for all $0\leq j\leq \degree$. 
This yields $\atree,\anode \models \forallpaths \ \acons$. 
\item $\aformula'=\aformula_1\wedge \aformula_2$. Since $\aset_\anode$ is propositionally consistent, we have
  $\aformula_1,\aformula_2\in \aset_\anode$. By the induction hypothesis, $\atree,\anode \models \aformula_1$
  and $\atree,\anode\models \aformula_2$. Hence $\atree,\anode\models\aformula'$. 
\item Similarly for $\aformula'=\aformula_1\vee\aformula_2$. 
\item $\aformula' = \existspath\mynext \aformula_1$. 
From condition (1) above we know that $\aformula_1\in \aset_{\anode\cdot i}$, where $i=\iota(\aformula')$. 
By the induction hypothesis, $\atree,\anode \cdot i\models \aformula_1$. Hence $\atree,\anode\models \aformula'$. 
\item $\aformula' = \forallpaths \mynext \aformula_1$. 
From condition (2) above, we know that $\aformula_1\in \aset_{\anode\cdot i}$ for all $0\leq i\leq \degree$. 
By the induction hypothesis, $\atree,\anode \cdot i \models \aformula_1$ for all $0\leq i\leq \degree$. Hence $\atree,
\anode\models\aformula'$. 
\item $\aformula' = \existspath \aformula_1 \until \aformula_2$. 
Since $\aset_\anode$ is propositionally consistent, 
we have $\aformula_2\in \aset_\anode$, or $\aformula_1, \existspath \mynext \aformula' \in \aset_\anode$. 
In the first case, we have $\atree,\anode \models \aformula_2$ by the induction hypothesis, and hence $\atree, \anode \models \aformula'$. 
In the second case, we have $\atree,\anode\models \aformula_1$ by the induction hypothesis, 
and we have $\aformula' \in \aset_{\anode\cdot j}$, where $j=\iota(\existspath\mynext\aformula')$. 
Consider the infinite path $\anode\cdot j\cdot j \cdot j\dots$ starting from $\anode$.   
Since $\arun$ is an accepting run, 
states from $F_{\aformula'}$ must occur infinitely often along this path; that is, 
there must exist some  $k\geq 0$  
such that $\aformula_2\in \aset_{\anode\cdot j^k}$ 
or $\existspath \aformula_1 \until \aformula_2\not\in \aset_{\anode\cdot j^k}$.
Recall that
 \[
    F_{\existspath \aformulabis_1 \until \aformulabis_2 } \egdef
    \set{\pair{i}{\aset} \in \locations \mid
      i \neq \iota(\existspath \mynext \existspath \aformulabis_1 \until \aformulabis_2 ) \
      \mbox{or} \
      \aformulabis_2 \in \aset \
      \mbox{or} \
      \existspath \aformulabis_1 \until \aformulabis_2 \not \in \aset
    }.
    \]
    Indeed, the first component $i$ in the location $\pair{i}{\aset}$ in $\locations$ records the direction
    $i$ from which the node is reached. 
    Let $k\geq 0$ be minimal. 
We first show by induction that if $\aformula'\in \aset_{\anode\cdot j^{i}}$, then  $\atree, \anode\cdot j^i \models\aformula_1$ and
$\aformula'\in \aset_{\anode\cdot j^{i+1}}$ for all $0\leq i<k$. 
For $i=0$, we have shown this above. 
So let $0< i < k$ and suppose $\aformula'\in \aset_{\anode\cdot j^i}$. 
Since $\arun$ is a run, $\aset_{\anode\cdot j^i}$ is propositionally consistent and hence
$\aformula_2\in \aset_{\anode \cdot j^i}$  or $\aformula_1, \existspath\mynext\aformula'\in \aset_{\anode \cdot j^i}$. 
Note that $\aformula_2\in \aset_{\anode \cdot j^i}$ contradicts that $k$ is minimal; hence
$\aformula_1, \existspath\mynext\aformula'\in \aset_{\anode\cdot j^i}$. 
By the induction hypothesis $\atree,\anode\cdot j^i\models\aformula_1$, and $\aformula'\in \aset_{\anode\cdot j^{i+1}}$. We conclude that 
$\atree,\anode\cdot j^i\models\aformula_1$ for all $0\leq i<k$. 

Next we prove that $\atree,\anode \cdot j^k \models\aformula_2$. 
Recall that  $\aformula_2\in \aset_{\anode\cdot j^k}$ or $\aformula'\not\in \aset_{\anode\cdot j^k}$. 
Note that $\aformula' \in \aset_{\anode\cdot j^k}$ by what we just proved before. 
Hence $\aformula_2\in \aset_{\anode\cdot j^k}$. 
By the induction hypothesis, we obtain $\atree,\anode\cdot j^k\models\aformula_2$. 
Hence, we have  proved the existence of a path starting in $\anode$ and satisfying
$\aformula_1 \until \aformula_2$, which leads to  
$\atree,\anode \models \aformula'$. 
\item $\aformula' =\forallpaths \aformula_1 \until \aformula_2$. 
Since $\aset_\anode$ is propositionally consistent, we have 
 $\aformula_2\in \aset_\anode$, 
or $\aformula_1, \forallpaths \mynext \aformula'\in \aset_\anode$. 
In the first case, we have $\atree,\anode\models\aformula_2$ by the induction hypothesis, and hence $\atree,\anode\models \aformula'$. 
In the second case, we have $\atree,\anode\models \aformula_1$ by the induction hypothesis, and
we have $\aformula' \in \aset_{\anode\cdot i}$ for all $0\leq i\leq \degree$. 
We will prove that every path that starts in $\anode$ satisfies $\aformula_1 \until \aformula_2$. 
So consider an arbitrary infinite path $\anode\cdot j_1 \cdot j_2 \cdot \dots \in \interval{0}{\degree}^\omega$. 
Since $\arun$ is accepting, states from $F_{\aformula'}$ must occur infinitely often in $\apath$; that is 
there exists some $k\geq 0$ such that $\phi_2\in \aset_{\anode\cdot j_1\dots \cdot j_k}$ or $\aformula'
\not\in \aset_{\anode\cdot j_1 \dots \cdot j_k}$. 
Let $k\geq 0$ be minimal. 
We first show by induction that if $\aformula' \in \aset_{\anode\cdot j_1 \dots \cdot j_{m}}$, then $\atree,\anode\cdot j_1 \dots j_m\models\aformula_1$,
and $\aformula' \in \aset_{\anode\cdot j_1 \dots \cdot j_{m+1}}$ for all $0\leq m<k$. 
We have proved this for $m=0$ above. 
So let $0\leq m < k$ and suppose $\aformula'\in \aset_{\anode\cdot j_1\cdot\dots\cdot j_m}$. 
Since $\arun$ is a run, $\aset_{\anode \cdot j_1\dots \cdot j_m}$ is propositionally consistent and hence $\aformula_2\in
\aset_{\anode\cdot j_1\dots j_m}$  or $\aformula_1, \forallpaths\mynext \aformula'\in \aset_{\anode\cdot j_1 \dots j_m}$.  
Note that $\aformula_2\in \aset_{\anode\cdot j_1\dots j_m}$ contradicts that $k$ is minimal; hence
$\aformula_1, \forallpaths\mynext\aformula'\in \aset_{\anode\cdot j_1\dots j_m}$. 
By the induction hypothesis $\atree,\anode\cdot j_1\dots j_m \models \aformula_1$, and since $\arun$ is a run, we also have $\aformula'\in \aset_{\anode\cdot j_1 \dots j_{m+1}}$. 
We conclude that $\atree,\anode\cdot j_1\dots j_m \models \aformula_1$ for all $0\leq m<k$. 

Next we prove that $\atree,\anode\cdot j_1\dots j_k\models\aformula_2$. 
Recall that  $\aformula_2\in \aset_{\anode\cdot j_1\dots\cdot j_k}$ or $\aformula'\not\in \aset_{\anode\cdot j_1 \dots j_k}$. 
But $\aformula' \in \aset_{\anode\cdot j_1 \dots j_k}$ as we proved before. 
Hence $\aformula_2\in \aset_{\anode\cdot j_1 \dots j_k}$. 
By the induction hypothesis, we obtain $\atree,\anode\cdot j_1 \dots \cdot j_k\models\aformula_2$.
Hence, we have proved that an arbitrary chosen path starting in $\anode$
satisfies $\aformula_1 \until \aformula_2$,
leading to $\atree,\anode\models \aformula'$. 
\item $\aformula'=\existspath \aformula_1 \release \aformula_2$.  Since
  $\aset_\anode$ is consistent, we have $\aformula_2\in \aset_\anode$, and $\aformula_1\in \aset_\anode$ or
  $\existspath\mynext\aformula'\in \aset_\anode$. 
  In the first case, by the induction hypothesis we have $\atree,\anode\models\aformula_1$ and $\atree, \anode \models \aformula_2$,
  and hence $\atree,\anode\models \aformula'$. 
  In the second case, we have $\atree,\anode\models\aformula_2$ and $\aformula' \in \aset_{\anode \cdot j}$,
  where $j=\iota(\existspath\mynext\aformula')$. 
Consider the infinite path $\anode \cdot j  \cdot j \cdot j\dots$ starting from $\anode$. 
We prove that this path satisfies  $\aformula_1\release\aformula_2$. 
We distinguish two cases. 
\begin{itemize}
\item Suppose there exists some $k\geq 0$ such that $\aformula_1\in \aset_{\anode \cdot j^k}$. 
Let $k$ be minimal. 
By the induction hypothesis, $\atree,\anode\cdot j^k \models\aformula_1$. 
We prove that if $\aformula'\in \aset_{\anode\cdot j^i}$, then 
$\atree,\anode\cdot j^i\models \aformula_2$ and $\aformula'\in \aset_{\anode\cdot j^{i+1}}$
for all $0\leq i<k$. 
For $i=0$, we have proved this above. So let $0< i< k$ and suppose $\aformula'\in \aset_{\anode\cdot j^i}$. 
Since $\arun$ is a run of $\aautomaton_{\aformula}$, 
$\aset_{\anode\cdot j^i}$ is propositionally consistent, and hence $\aformula_2\in \aset_{\anode\cdot j^i}$, and
$\aformula_1\in \aset_{\anode\cdot j^i}$ or $\existspath\mynext\aformula'\in \aset_{\anode\cdot j^i}$.  
Since $\aformula_1\in \aset_{\anode\cdot j^i}$ contradicts the minimality of $k$, we have $\existspath\mynext\aformula'\in \aset_{\anode\cdot j^i}$. 
By the induction hypothesis, we have $\atree,\anode\cdot j^i\models\aformula_2$, and since $\arun$ is a run,
we have $\aformula'\in \aset_{\anode\cdot j^{i+1}}$. 
We can conclude that 
$\atree,\anode\cdot j^i\models \aformula_2$ for all $0\leq i\leq k$. 
\item Suppose that $\aformula_1\not\in \aset_{\anode\cdot j^k}$ for all $k\geq 0$. 
  We prove that if  $\aformula'\in \aset_{\anode\cdot j^{k}}$, then $\atree,\anode \cdot j^k\models\aformula_2$ and $\aformula'\in
  \aset_{\anode\cdot j^{k+1}}$ for all $k\geq 0$.
  For $k=0$ we have proved this above. So let $k> 0$ and suppose $\aformula'\in \aset_{\anode\cdot j^{k}}$. Since $\arun$ is a run, we know that
  $\aset_{\anode\cdot j^k}$ is
  propositionally consistent. Hence $\aformula_2\in \aset_{\anode\cdot j^k}$ (which, by the induction
  hypothesis, implies $\atree,\anode\cdot j^k\models\aformula_2$) and 
  $\aformula_1\in \aset_{\anode\cdot j^k}$ or $\existspath\mynext\aformula'\in \aset_{\anode\cdot j^k}$. Note that
  $\aformula_1\in \aset_{\anode\cdot j^k}$ cannot be by assumption, hence we have $\existspath\mynext\aformula'\in \aset_{\anode\cdot j^k}$. 
  We conclude that 
$\atree,\anode\cdot j^k\models \aformula_2$ for all $k\geq 0$. 
\end{itemize}

Hence, we have proved that the path above satisfies $\aformula_1\release\aformula_2$, and therefore
$\atree,\anode \models \aformula'$. 
\item $\aformula'= \forallpaths \aformula_1 \release \aformula_2$. The proof is a combination of the proofs
  for $\forallpaths\aformula_1 \until\aformula_2$ and $\existspath\aformula_1\release\aformula_2$. 
   Since
   $\aset_\anode$ is consistent, we have $\aformula_2\in \aset_\anode$, and
   $\aformula_1\in \aset_\anode$ or
  $\forallpaths \mynext\aformula'\in \aset_\anode$. 
   In the first case, by the induction hypothesis we have
   $\atree,\anode\models\aformula_1$ and $\atree, \anode \models \aformula_2$,
   and hence $\atree,\anode\models \aformula'$.
   In the second case, we have $\atree,\anode\models \aformula_1$ by the induction
   hypothesis, and we have $\aformula' \in \aset_{\anode\cdot i}$ for all
   $0\leq i\leq \degree$.  We will prove that every path that starts in $\anode$ satisfies
   $\aformula_1 \release \aformula_2$. 
   So consider an arbitrary infinite path
   $\anode\cdot j_1 \cdot j_2 \cdot \dots \in \interval{0}{\degree}^\omega$.
   We distinguish two cases. 
\begin{itemize}
\item Suppose there exists some $k\geq 1$ such that $\aformula_1\in
  \aset_{\anode \cdot j_1 \cdots j_k}$. Let $k$ be minimal. 
   By the induction hypothesis, $\atree,\anode\cdot j_1 \cdots j_k \models\aformula_1$. 
   We prove that if $\aformula'\in \aset_{\anode\cdot j_1 \cdots j_i}$, then 
   $\atree,\anode\cdot j_1 \cdots j_i \models \aformula_2$ and
   $\aformula'\in \aset_{\anode\cdot j_1 \cdots j_{i+1}}$
   for all $0\leq i<k$. 
   For $i=0$, we have proved this above. So let $0< i< k$ and suppose
   $\aformula'\in \aset_{\anode\cdot j_1 \cdots j_i}$. 
   Since $\arun$ is a run of $\aautomaton_{\aformula}$, 
   $\aset_{\anode\cdot j_1 \cdots j_i}$ is propositionally consistent, and hence
   $\aformula_2\in \aset_{\anode\cdot j_1 \cdots j_i}$, and
   $\aformula_1\in \aset_{\anode\cdot j_1 \cdots j_i}$ or
   $\forallpaths \mynext\aformula'\in \aset_{\anode\cdot j_1 \cdots j_i}$.  
   Since $\aformula_1\in \aset_{\anode\cdot j_1 \cdots j_i}$ contradicts the minimality of $k$,
   we have $\forallpaths \mynext\aformula'\in \aset_{\anode\cdot j_1 \cdots j_i}$. 
   By the induction hypothesis, we have $\atree,\anode\cdot j_1 \cdots j_i \models\aformula_2$,
   and since $\arun$ is a run, we have $\aformula'\in \aset_{\anode\cdot j_1 \cdots j_{i+1}}$. 
   We can conclude that 
   $\atree,\anode\cdot j_1 \cdots j_i \models \aformula_2$ for all $0\leq i\leq k$. 
   \item  Suppose that $\aformula_1\not\in \aset_{\anode\cdot j_1 \cdots j_k}$ for all $k\geq 0$. 
     We prove that if  $\aformula'\in \aset_{\anode\cdot j_1 \cdots j_k}$, then
     $\atree,\anode \cdot j_1 \cdots j_k \models\aformula_2$ and $\aformula'\in
     \aset_{\anode\cdot j_1 \cdots j_{k+1}}$ for all $k\geq 0$.
     For $k=0$ we have proved this above. So let $k> 0$ and suppose
     $\aformula'\in \aset_{\anode\cdot j_1 \cdots j_k}$. Since $\arun$ is a run, we know that
     $\aset_{\anode\cdot j_1 \cdots j_k}$ is
     propositionally consistent. Hence $\aformula_2\in \aset_{\anode\cdot j_1 \cdots j_k}$
     (which, by the induction hypothesis, implies
     $\atree,\anode\cdot j_1 \cdots j_k \models\aformula_2$) and 
     $\aformula_1\in \aset_{\anode\cdot j_1 \cdots j_k}$ or
     $\forallpaths \mynext\aformula'\in \aset_{\anode\cdot j_1 \cdots j_k}$.
     Note that $\aformula_1\in \aset_{\anode\cdot j_1 \cdots j_k}$ cannot be by assumption,
     hence we have $\forallpaths \mynext\aformula'\in \aset_{\anode\cdot j_1 \cdots j_k}$. 
     We conclude that 
     $\atree,\anode\cdot j_1 \cdots j_k \models \aformula_2$ for all $k\geq 0$. 
\end{itemize}
  
\end{itemize}

``only if:''
Suppose $\aformula$ is satisfiable. 
By Proposition~\ref{proposition-tree-model-for-Z}, $\aformula$ has a tree model $\atree$ with
domain $\interval{0}{\degree}^*$ that obeys the direction map $\iota$. 
We prove that $\atree \in \alang(\aautomaton_{\aformula})$, that is,
there exists some accepting run 
$\arun:\interval{0}{\degree}^*\to \delta$
of $\aautomaton_{\aformula}$ on $\atree$.
Note that $\atree$ belongs to $\alang(\aautomaton_{\aformula})$, assuming that each node is labelled with 
the letter $\arbitraryletter$ from the single-letter alphabet $\aalphabet$. Below, we omit to mention the letter
from this singleton alphabet and keep $\atree$ of the form $\atree: \interval{0}{\degree}^* \to \Zed^{\beta}$ to stick
to the Kripke structure.

For every node $\anode$ in $\atree$, define $\aset_\anode$ to be the set of formulas $\aformula'$
in $\subf{\aformula}$ such that $\atree,\anode\models\aformula'$. 
Define $\arun$ inductively as follows.
First set
\[
\arun(\varepsilon)\egdef
\triple{\pair{0}{\aset_{\varepsilon}}}{\arbitraryletter}{\pair{\acons_0}{\pair{0}{\aset_0}},\ldots,\pair{\acons_\degree}{\pair{\degree}{\aset_\degree}}}
\]
where for all $j \in \interval{0}{\degree}$,
if there is $\existspath  \ \acons \in \aset_{\varepsilon}$
such that $\iota(\existspath  \ \acons) =j$, then
    \[
    \acons_j \egdef (\bigwedge_{\forallpaths \ \acons' \in \aset} \acons') \wedge \acons
    \ \ \ 
    \mbox{otherwise,}
    \ \ \ 
    \acons_j  \egdef \bigwedge_{\forallpaths \ \acons' \in \aset} \acons'.
    \]
    More generally, $\arun(\anode \cdot j)$ is defined by
\[
\arun(\anode\cdot j)\egdef
\triple{\pair{j}{\aset_{\anode \cdot j}}}{\arbitraryletter}{\pair{\acons_{0}}{\pair{0}{\aset_{\anode \cdot j \cdot 0}}},\ldots,\pair{\acons_\degree}{\pair{\degree}{\aset_{\anode \cdot j
        \cdot \degree}}}}
\]
where for all $j \in \interval{0}{\degree}$,
if there is $\existspath  \ \acons \in \aset_{\anode \cdot j}$
such that $\iota(\existspath  \ \acons) =j$, then
    \[
    \acons_j \egdef (\bigwedge_{\forallpaths \ \acons' \in \aset} \acons') \wedge \acons
    \ \ \ 
    \mbox{otherwise,}
    \ \ \ 
    \acons_j  \egdef \bigwedge_{\forallpaths \ \acons' \in \aset} \acons'.
    \] 
We prove that $\arun$ is an accepting run of $\aautomaton_{\aformula}$ on $\atree$.

\begin{itemize}
\item Let us first prove a well-known property: 
  $\aset_\anode$ is propositionally consistent for all nodes
  $\anode \in \interval{0}{\degree}^*$.
\begin{itemize}
\item Suppose $\aformula_1 \vee \aformula_2\in \aset_{\anode}$. That is $\atree,\anode\models\aformula_1$ or
  $\atree,\anode\models\aformula_2$. But then also $\aformula_1\in \aset_{\anode}$ or $\aformula_2\in \aset_{\anode}$, and
  hence $\{\aformula_1,\aformula_2\}\cap \aset_{\anode}\neq\emptyset$. 
\item Suppose $\aformula_1\wedge\aformula_2\in \aset_{\anode}$. That is $\atree,\anode\models\aformula_1$ and
  $\atree,\anode\models\aformula_2$. But then also $\aformula_1\in \aset_{\anode}$ and $\aformula_2\in \aset_{\anode}$,
  and hence $\{\aformula_1,\aformula_2\}\subseteq \aset_{\anode}$.
\item Suppose $\existspath \aformula_1\until\aformula_2\in \aset_{\anode}$. We distinguish two cases. 
(i) Suppose $\atree,\anode\models \aformula_2$. Then $\aformula_2\in \aset_{\anode}$. 
  (ii) Suppose $\atree,\anode \not\models \aformula_2$. By $\atree,\anode\models\existspath\aformula_1\until \aformula_2$,
  we conclude that $\atree,\anode\models\aformula_1$ and $\atree,\anode\models \existspath\mynext\existspath \aformula_1\until\aformula_2$.
  But then $\{\aformula_1,\existspath\mynext\existspath\aformula_1\until\aformula_2\}\subseteq \aset_{\anode}$. 
\item Suppose $\forallpaths\aformula_1\until\aformula_2\in \aset_{\anode}$. 
We distinguish two cases. 
(i) Suppose $\atree,\anode\models \aformula_2$. Then $\aformula_2\in \aset_{\anode}$.  
(ii) Suppose $\atree,\anode\not\models\aformula_2$. By $\atree,\anode\models\forallpaths\aformula_1\until\aformula_2$, we conclude that
$\atree,\anode\models\aformula_1$ and $\atree,\anode\models\forallpaths\mynext\forallpaths\aformula_1\until\aformula_2$.
But then $\{\aformula_1,\forallpaths\mynext\forallpaths\aformula_1\until\aformula_2\}\subseteq \aset_{\anode}$. 
\item Suppose $\existspath\aformula_1\release\aformula_2\in \aset_{\anode}$. 
  That is, $\atree,\anode\models\existspath\aformula_1\release\aformula_2\in \aset_{\anode}$ and hence
  $\atree,\anode\models\aformula_2$, so that $\aformula_2\in \aset_{\anode}$, too. 
We distinguish two cases: 
(i) If $\atree,\anode\models\aformula_1$, then $\aformula_1\in \aset_{\anode}$. 
(ii) If $\atree,\anode\not\models\aformula_1$, then $\atree,\anode\models\existspath\mynext\existspath\aformula_1\release\aformula_2$. 
But then $\existspath\mynext\existspath\aformula_1\release\aformula_2\in \aset_{\anode}$. 
Hence we can conclude that $\{\aformula_1,\existspath\mynext\existspath\aformula_1\release\aformula_2\}\cap \aset_{\anode}\neq\emptyset$. 
\item The proof for $\forallpaths \aformula_1 \release \aformula_2$ is analogous using the
  validity of $\forallpaths \aformula_1 \release \aformula_2 \Leftrightarrow
               \big(\aformula_2 \wedge (\aformula_1 \vee \forallpaths \mynext \forallpaths \aformula_1 \release \aformula_2)\big)$. 
\end{itemize}
\item Since $\atree\models\aformula$, we must have $\aformula\in \aset_{\varepsilon}$, so that $(0,\aset_{\varepsilon})$ is an
  initial location of $\aautomaton_\aformula$. 
\item Next we prove that, for all nodes $\anode$ with the source location of $\arun(\anode)$ being $(i,\aset_\anode)$, there exists a transition
  $$((i,\aset_\anode), \arbitraryletter, (\acons_0, (0,\aset_{\anode \cdot 0})), \dots, (\acons_\degree,(\degree,\aset_{\anode \cdot \degree})))\in \delta$$
  satisfying the conditions (1)--(3) below.
\begin{enumerate}
\item Let $\existspath\mynext \aformula_1 \in \aset_\anode$ and suppose $\iota(\existspath\mynext\aformula_1)=j$ for some $1\leq j\leq \degree$. 
By definition of $\aset_\anode$ we have $\atree,\anode \models \existspath\mynext\aformula_1$. 
Since $\atree$ obeys $\iota$, 
we have $\atree, \anode\cdot j\models \aformula_1$. 
Hence indeed $\aformula_1\in \aset_{\anode\cdot j}$ (and $\subf{\aformula}$ is closed under subformulae). 
\item Let $\forallpaths\mynext\aformula_1\in \aset_\anode$. 
  By definition of $\aset_\anode$, we have $\atree,\anode\models \forallpaths\mynext\aformula_1$. By definition of
  the satisfaction relation for
  $\CTL(\Zed)$, we can conclude that $\atree,\anode \cdot j\models \aformula_1$ for all $0\leq j\leq \degree$.
  Hence $\aformula_1\in \aset_{\anode \cdot j}$ for all $0\leq j\leq \degree$. 
\item Let $0\leq j \leq \degree$. 
Let $\forallpaths \ \acons \in \aset_\anode$. 
Hence $\atree,\anode\models \forallpaths \ \acons$. By definition of the satisfaction relation for $\CTL(\Zed)$, 
$\Zed\models \acons(\vect{z},\vect{z}_j)$, where $\vect{z}=\atree(\anode)$ and $\vect{z}_j=\atree(\anode\cdot j)$. 
We can conclude that $\Zed\models \left( \bigwedge_{\forallpaths \acons \in\aset_\anode}\acons \right)\left(\vect{z},\vect{z}_j\right)$. 
If, additionally, there exists $\existspath \acons'\in \aset_\anode$ such that $\iota(\existspath \ \acons')=j$, then, 
since $\atree$ obeys $\iota$, 
we have $\Zed \models\acons'(\vect{z},\vect{z}_j)$. In that case we have
$\Zed\models \left( \bigwedge_{\forallpaths \acons\in\aset_\anode}\acons \wedge \acons' \right)\left(\vect{z},\vect{z}_j\right)$. 
\end{enumerate}
\item Next we prove that $\arun$ is accepting. 

  Suppose $\existspath \aformula_1\until\aformula_2$ is in $\parsubf{\existspath\until}{\aformula}$, and let
  $\iota(\existspath\mynext\existspath\aformula_1\until\aformula_2)=j$.
  {\em This is the place where we use the index recording the direction.}
  We show that for all branches $j_1 j_2 \cdots \in \interval{0}{\degree}^{\omega}$,
  a location in $F_{\existspath\aformula_1\until\aformula_2}$ occurs infinitely
  often in $\arun(j_1) \arun(j_2) \cdots$.
  If $j_1 j_2 \cdots$ is not of the form $\anode \cdot j^{\omega}$,
  $\arun(j_1) \arun(j_2) \cdots$ does not belong to $\locations^+ \set{\pair{j}{\aset} \mid \pair{j}{\aset} \in \locations}^{\omega}$
  and therefore a location $\pair{i}{\asetbis}$ with $j \neq i$ in $F_{\existspath\aformula_1\until\aformula_2}$ occurs infinitely
  often in $\arun(j_1) \arun(j_2) \cdots$.
  Now suppose that $j_1 j_2 \cdots$ is of the form $\anode \cdot j^{\omega}$.
  {\em Ad absurdum}, assume that there exists some $m\geq 0$ such that 
  the source location in $\arun(\anode\cdot j^{m+k})$ does not belong to $F_{\existspath\aformula_1\until\aformula_2}$ for all $k\geq 0$.
  By definition of $F_{\existspath \aformula_1\until\aformula_2}$, we obtain
  the source location of $\arun(\anode\cdot j^{m+k})$ is $(j,\aset_{\anode\cdot j^{m+k}})$, 
  $\aformula_2\not\in\aset_{\anode\cdot  j^{m+k}}$, and $\existspath\aformula_1\until\aformula_2\in\aset_{\anode\cdot j^{m+k}}$,
  for all $k\geq 0$. 
  By definition of $\aset_{\anode\cdot j^{m+k}}$, 
  we have $\atree, \anode \cdot j^{m+k}\models\existspath\aformula_1\until\aformula_2$. 
  Since  $\atree$ obeys $\iota$, 
  there must exist some $p\geq 0$ such that $\atree,\anode\cdot j^{m+p} \models\aformula_2$. 
  But then $\aformula_2\in \aset_{\anode\cdot j^{m+p}}$, which leads to a contradiction.

Suppose $\forallpaths\aformula_1\until\aformula_2$ is in $\parsubf{\forallpaths\until}{\aformula}$. 
We show that for all nodes $\anode$ in $\arun$, 
every path starting in $\anode$ visits $F_{\forallpaths \aformula_1\until\aformula_2}$ infinitely often. 
Towards contradiction, suppose that there exists a node $\anode$ and some infinite path
$\anode\cdot j_1 \cdot j_2 \cdot \dots\in \interval{0}{\degree}^\omega$  that visits
$F_{\forallpaths \aformula_1\until\aformula_2}$ only finitely. 
That is, there exists some $m\geq 0$ such that 
the source location in $\arun(\anode_k)$ does not belong to $F_{\forallpaths\aformula_1\until\aformula_2}$ for all $k\geq m$. 
By definition of $F_{\forallpaths\aformula_1\until\aformula_2}$, 
and assuming $\arun(\anode\cdot j_1 \dots j_k)=(d_{\anode\cdot j_1 \dots j_k},\aset_{\anode\cdot j_1 \dots j_k})$, we have
$\forallpaths \aformula_1\until\aformula_2\in \aset_{\anode\cdot j_1 \dots j_k}$ and $\aformula_2\not\in \aset_{\anode\cdot j_1 \dots j_k}$
for all $k\geq m$. 
By definition of $\aset_{\anode\cdot j_1 \dots j_k}$, we know that $\atree_{\anode\cdot j_1 \dots j_m}\models\forallpaths\aformula_1\until\aformula_2$. 
But then there must exist some $p\geq m$ such that 
$\atree_{\anode\cdot j_1 \dots j_p}\models\aformula_2$. 
By definition, $\aformula_2\in \aset_{\anode\cdot j_1 \dots j_p}$, which leads to a contradiction. \qedhere
\end{itemize}
\end{proof}

\section{Proofs for Section~\ref{section-ctlstarz}} 
\subsection{Proof of Proposition~\ref{proposition-simple-form-ctlstarz}}
\label{appendix-proof-proposition-simple-form-ctlstarz} 
\begin{proof} The proof is a slight variant of the proof of Proposition~\ref{proposition-simple-form},
  we keep the same notations whenever possible.
  First, we can establish that $\CTLStar(\Zed)$ has the tree model
  property, exactly as done in the proof of Proposition~\ref{proposition-simple-form} for
  $\CTL(\Zed)$.
  We use  
  unfoldings for Kripke structures and we omit the details herein. 

  Now, we move to the construction of $\aformula'$ in simple form.
  We use the notion of forward degree introduced in the proof of
  Proposition~\ref{proposition-simple-form} that applies to
   $\CTLStar(\Zed)$ state formulae and to constraints $\acons$. 
  Let $\aformula$ be a state formula in $\CTLStar(\Zed)$ such that $\fd{\aformula} = N$
  (see page~\pageref{page-fd} the definition of $\fd{}$) 
  and the variables occurring in
  $\aformula$ are among $\avariable_1, \ldots, \avariable_{\beta}$.
  Below, we build a formula $\aformula'$  over the variables $\avariable_1^{-N}, \ldots, \avariable_1^{0}, \ldots, \avariable_{\beta}^{-N}, \ldots, \avariable_{\beta}^{0}$
  with $\fd{\aformula'} \leq 1$ such that $\aformula$ is satisfiable in a tree Kripke structure iff
  $\aformula'$ is satisfiable in a tree Kripke structure and $\aformula'$ can be computed in polynomial-time in the size of
  $\aformula$.

  Given  $\acons$ occurring in $\aformula$ with $\fd{\acons} = M$ (and therefore $M \leq N$),
  we write $\mathtt{jump}(\acons,M)$ to denote
  the constraints 
  as done in the proof of
  Proposition~\ref{proposition-simple-form} (page~\pageref{appendix-proof-proposition-simple-form}). 
  
  Let $\atranslation$ be the translation map that is homomorphic for Boolean
  and temporal connectives and path quantifiers such that
  $\atranslation(\acons) \egdef \mynext^M  \ \mathtt{jump}(\acons,M)$, where  $\fd{\acons} = M$ for all
  maximal constraints $\acons$ occurring in $\aformula$. This is the most significant
  change with respect to the proof of Proposition~\ref{proposition-simple-form}
  (we have replaced $(\existspath \mynext)^M$ from the proof of Proposition~\ref{proposition-simple-form}
  by $\mynext^M$).
  As $\acons$ is always in the scope of a path quantifier, the path formula 
  $\mynext^M  \ \mathtt{jump}(\acons,M)$ is well-defined. 
   Let $\aformula'$ be defined as follows:
  \[
  \atranslation(\aformula) \wedge \forallpaths \always \ \forallpaths \big(\bigwedge_{j \in \interval{1}{\beta}, \ k \in \interval{0}{N-1}} \avariable_j^{-k} = \mynext \avariable_j^{-k-1}\big).
  \]
  Observe that the second conjunct of $\aformula'$ is identical to the case
  for $\CTL(\Zed)$ in the proof of Proposition~\ref{proposition-simple-form}. 
  One can show that  $\aformula$ is satisfiable in a tree Kripke structure iff
  $\aformula'$ is satisfiable in a tree Kripke structure.
  The main argument can be provided along the lines of the proof of
  Propositionn~\ref{proposition-simple-form}. Below, we briefly provide the essential steps.

  First, suppose that $\aks, \aworld \models \aformula$, where
  $\aks = \triple{\worlds}{\arelation}{\avaluation}$ is a
  {\em tree Kripke structure} with root $\aworld$.
  Let $\aks' = \triple{\worlds}{\arelation}{\avaluation'}$ be the Kripke structure that differs from
  $\aks$ only in the definition of the valuation.
  Given a node $\aworld' \in \worlds$ reachable from $\aworld$ via the branch $\aworld_0 \cdots \aworld_n$ with $\aworld_0 = \aworld$ and $\aworld_n = \aworld'$ for some $n \geq 0$,
  for all $k \in \interval{0}{N}$ and $j \in \interval{1}{\beta}$, we require
  $\avaluation'(w', \avariable_j^{-k}) \egdef \avaluation(w_{n-k}, \avariable_j)$ if $n-k \geq 0$,
  otherwise $\avaluation'(\aworld', \avariable_j^{-k}) \egdef 0$ (arbitrary value).
  To establish $\aks', \aworld \models \aformula'$, it boils down to check the
  properties below.
  \begin{itemize}
  \item $\aks', \aworld \models \forallpaths \always \ \forallpaths \big(\bigwedge_{j \in \interval{1}{\beta}, \
    k \in \interval{0}{N-1}} \avariable_j^{-k} = \mynext \avariable_j^{-k-1}\big)$.
    This holds thanks to the definition of $\avaluation'$ and the tree structure of $\aks$.
    
  \item For every infinite path $\apath$,
    $\aks, \apath \models \acons$
    implies $\aks', \apath \models \mynext^M \ \mathtt{jump}(\acons,M)$, where  $\fd{\acons} = M$.
    This can be lifted to all path formulae and to all state formulae (in negation normal
    form)
    by structural induction. 
  \end{itemize}
  
  For the other direction, suppose that $\aks, \aworld \models \aformula'$ and $\aks = \triple{\worlds}{\arelation}{\avaluation}$ is a
  {\em tree Kripke structure} with root $\aworld$.
   Let $\aks' = \triple{\worlds}{\arelation}{\avaluation'}$ be the Kripke structure that differs from
   $\aks$ only in the definition of the valuation. More precisely, for all $\aworld' \in \worlds$
   and all $j \in \interval{1}{\beta}$, we have
   $\avaluation'(\aworld',\avariable_j) \egdef \avaluation(\aworld',\avariable_j^0)$.
   By structural induction, one can show that $\aks, \aworld \models \atranslation(\aformula)$
   implies $\aks', \aworld \models \aformula$.
\end{proof}

\subsection{Proof of Proposition~\ref{proposition-ctlstarz-special-form}}
\label{appendix-proof-proposition-ctlstarz-special-form}
\begin{proof}
  By Proposition~\ref{proposition-simple-form-ctlstarz}, we can assume that $\aformula$ is
  in simple form, and $\aformula$ is built over
  the terms  $\avariable_1,\dots,\avariable_\beta$ and
  $\avariable'_1,\dots,\avariable'_\beta$.
  Consequently, we can assume that $\aformula$ is in negation normal
  form.

  We use a standard property that illustrates the renaming technique~\cite{Scott62} used below.
  Let $\aformulabis$ be a $\CTLStar(\Zed)$ state formula with state subformula
  $\aformulabis'$ and $\avariablebis$ be a (fresh) variable {\em not occurring} in $\aformulabis$.
  Then, $\aformulabis$ is satisfiable iff
  $\aformulabis[\aformulabis' \leftarrow \existspath (\avariablebis = 0)]
  \wedge
  \forallpaths \always (\existspath (\avariablebis = 0) \Leftrightarrow \aformulabis')
  $ is satisfiable, where $\aformulabis[\aformulabis' \leftarrow \existspath (\avariablebis = 0)]$
  denotes the state formula obtained from $\aformulabis$ by replacing every occurrence
  of $\aformulabis'$ by $\existspath (\avariablebis = 0)$.
  The idea is to replace $\aformulabis'$ by the atomic constraint
  $(\avariablebis = 0)$, where $\avariablebis$ is a fresh variable not occurring before.
  The new constraint $(\avariablebis = 0)$ is tied to the original subformula $\aformulabis'$
  via the conjunct $\forallpaths \always (\existspath (\avariablebis = 0) \Leftrightarrow \aformulabis')$.
  Note that strictly speaking, $\avariablebis = 0$ is not a state formula
  but morally it is because its satisfaction depends only on the current state.
  That is why we use $\existspath (\avariablebis = 0)$ instead of the more natural
  constraint $\avariablebis = 0$.
  In the transformations below, we need sometimes to remove $\existspath$ in $\existspath  (\avariablebis = 0)$
  when the occurrence of $\avariablebis = 0$ is already  in the scope of a path quantifier.
  Observe that alternatively, we could slightly redefine $\CTLStar(\Zed)$ to accept also as atomic formulae
  constraints $\acons$ in which all terms are some variable $\avariable_i$ (no prefix
  with $\mynext$). In that slight extension, $\avariablebis = 0$ would be authorised
  as a state formula.

  In order to compute $\aformula'$, we perform on $\aformula$ several transformations
  of the above form. We write $\aformulabis$ to denote current state formulae on
  which the transformations are performed and initially $\aformulabis$ takes the value
  $\aformula$, which is a $\CTLStar(\Zed)$ state formula in simple form.
  Through the sequence of transformations, $\aformulabis$ is maintained in the following shape:
  \[
  \aformulater \wedge \bigwedge_i \forallpaths \always (\existspath (\avariablebis_i = 0) \Leftrightarrow
  \pathquantifier_i \ \apathformula_i),
  \]
  where each $\pathquantifier_i \in\{\existspath,\forallpaths\}$, each $\apathformula_i$ is an
  $\LTL(\Zed)$ (path) formula in simple form and $\aformulater$ is a $\CTLStar(\Zed)$ formula in simple form. 
  In order to compute the new value for $\aformulabis$, suppose that $\aformulater$ contains
  a state subformula $\aformulabis'$ of the form $\pathquantifier \ \apathformula$, where
  the only path quantifiers occurring in $\apathformula$ occur in state formulae of the
  form $\pathquantifier' \ \avariableter = 0$ (no need to perform a renaming on $\pathquantifier' \ \avariableter = 0$).
  We write $\apathformula^{\dag}$ to denote
  the path formula (in simple form) obtained from $\apathformula$ by replacing every
  occurrence of $\pathquantifier' \ \avariableter = 0$ by $\avariableter = 0$.
  Obviously, $\pathquantifier \ \apathformula$ is logically equivalent to
  $\pathquantifier \ \apathformula^{\dag}$. 
  The new value for
  $\aformulabis$ is defined below ($\avariablebis$ is a fresh variable):
  \[
  \overbrace{\aformulater[\pathquantifier \ \apathformula \leftarrow \existspath (\avariablebis = 0)]}^{
   \mbox{renaming of $\pathquantifier \ \apathformula$}}
  \wedge
  \underbrace{
  \forallpaths \always (\existspath (\avariablebis = 0) \Leftrightarrow
  \pathquantifier \ \apathformula^{\dag})}_{\mbox{new equivalence}}
  \wedge
  \overbrace{
  \bigwedge_i \forallpaths \always (\existspath (\avariablebis_i = 0) \Leftrightarrow
  \pathquantifier_i \ \apathformula_i)}^{\mbox{conjunction already in $\aformulabis$}},
  \]
  One can show that the transformation preserves satisfiability and moreover, repeating this procedure can be
  done only a 
  polynomial amount of times, guaranteing termination.
  At the end of all these transformations, the resulting
  formula $\aformulabis$ is now of the
  form
  \[
  \aformulater \wedge \bigwedge_i \forallpaths \always (\existspath (\avariablebis_i = 0) \Leftrightarrow
  \pathquantifier_i \ \apathformula_i),
  \]
  where $\aformulater$ is a Boolean combination of state formulae of the form
  $\pathquantifier \ \avariable = 0$ and, the $\avariablebis_i$'s are distinct and new
  variables not occurring in the original formula $\aformula$. 
  We write $\aformulater^{\dag}$ to denote
  the path formula obtained from $\aformulater$ by removing all the path quantifiers.
  Again, $\pathquantifier \ \aformulater$ is logically equivalent to $\pathquantifier \ \aformulater^{\dag}$
  for all $\pathquantifier \in \set{\existspath, \forallpaths}$. 
  The intermediate state formula $\aformula^{\star}$ takes the value below:
  \[
  \existspath \ (\avariableter = 0) \wedge
  \forallpaths \always (\overline{\avariableter = 0 \Leftrightarrow
  \aformulater^{\dag}})
  \wedge
  \bigwedge_i \forallpaths \always (\existspath (\avariable_i = 0) \Leftrightarrow
  \pathquantifier_i \ \apathformula_i),
  \]
  where $\avariableter$ is again a fresh variable
  and $\overline{\avariableter = 0 \Leftrightarrow
    \aformulater^{\dag}}$ is in negation normal form
  and logically equivalent to $\avariableter = 0 \Leftrightarrow
    \aformulater^{\dag}$. 
  The formulae $\aformula$ and $\aformula^{\star}$ are equi-satisfiable.

  It remains to explain how to transform each
  $\forallpaths \always (\existspath (\avariable_i = 0) \Leftrightarrow
  \pathquantifier_i \ \apathformula_i)$ so that we get the final formula $\aformula'$
  in special form from $\aformula^{\star}$.
  For instance, $\forallpaths \always (\existspath (\avariablebis_i = 0) \Leftrightarrow
  \forallpaths \ \apathformula_i)$ shall be replaced by
  \[
  \forallpaths \always (\neg (\avariablebis_i = 0) \vee 
  \apathformula_i)
  \ \ \wedge \ \
  \forallpaths \always \existspath \
    (\avariablebis_i = 0 \vee
    \overline{\neg \apathformula_i}),
   \]
     where $\overline{\neg \apathformula_i}$ is logically equivalent
  to $\neg \apathformula_i$ but in negation normal form.
  Note that $\neg (\avariablebis_i = 0) \vee 
  \apathformula_i$ and
  $\avariablebis_i = 0 \vee
    \overline{\neg \apathformula_i}$
  are  $\LTL(\Zed)$ formulae in simple form, exactly what is needed
  for the final $\CTLStar(\Zed)$ state formula $\aformula'$ in special form.

  Below, we list the  logical equivalences (hinted above) we take advantage of and 
  that are slight variants of equivalences used for \CTLStar in the
  proof of~\cite[Theorem 3.1]{Emerson&Sistla84}.
  \begin{itemize}
  \item $\forallpaths \always (\existspath (\avariablebis_i = 0) \Rightarrow
  \forallpaths \ \apathformula_i) \ \Leftrightarrow \forallpaths \always (\neg (\avariablebis_i = 0) \vee 
  \apathformula_i)$. 
  \item $\forallpaths \always (\existspath (\avariablebis_i = 0) \Rightarrow
    \existspath \ \apathformula_i) \ \Leftrightarrow \forallpaths \always \existspath (\neg (\avariablebis_i = 0)
    \vee \apathformula_i)$.
  \item $\forallpaths \always (\neg \existspath (\avariablebis_i = 0) \Rightarrow
    \neg \forallpaths \ \apathformula_i) \ \Leftrightarrow \forallpaths \always \existspath \
    (\avariablebis_i = 0 \vee
    \overline{\neg \apathformula_i})$.
  \item $\forallpaths \always (\neg \existspath (\avariablebis_i = 0) \Rightarrow
    \neg \existspath \ \apathformula_i) \ \Leftrightarrow \forallpaths
    \always  ((\avariablebis_i = 0) \vee
  \overline{\neg \apathformula_i})$.
  \end{itemize}
  The formula $\aformula'$ is obtained from $\aformula^{\star}$ by replacing 
  each element of the generalised conjunction in $\aformula^{\star}$ by two formulae based
  on these equivalences.
\end{proof}

\begin{exa}
  To illustrate the construction from the above proof of Proposition~\ref{proposition-ctlstarz-special-form},
  we consider the formula $\aformula$ below. 
  $$\aformula = \existspath ((\avariable_1'<\avariable_1) \until \ \forallpaths \mynext (\avariable_2 =\avariable'_2)) \,
  \wedge \, \existspath \always (\avariable_1 < \avariable_2).$$
  Below, we present the formulae $\aformulabis_0 = \aformula$, $\aformulabis_1$, $\aformulabis_2$,
  $\aformulabis_3$, $\aformula^{\star}$ obtained by application of the different renaming steps.
\[
  \aformulabis_1 =
  \overbrace{\forallpaths \always \big(\existspath(\avariablebis_1 = 0) \Leftrightarrow
  (\forallpaths \mynext (\avariable_2 =\avariable'_2))\big)}^{=  \ \aformulabis'_1} \, \wedge \,
  \existspath ((\avariable_1'<\avariable_1) \until \ \existspath(\avariablebis_1 = 0)) \,
  \wedge \, \existspath \always (\avariable_1 < \avariable_2).
\]
\[
  \aformulabis_2 =
  \overbrace{\forallpaths \always \big(\existspath(\avariablebis_2 = 0) \Leftrightarrow
  (\existspath \ (\avariable_1'<\avariable_1) \until \ \avariablebis_1 = 0)\big)}^{=  \ \aformulabis'_2}
  \, \wedge \, \aformulabis'_1 \, \wedge \,
  \existspath(\avariablebis_2 = 0) \,
  \wedge \, \existspath \always (\avariable_1 < \avariable_2).
\]
  Note that `$\existspath$' is removed from $\existspath(\avariablebis_1 = 0)$
  in $\aformulabis'_2$. 
\[
  \aformulabis_3 =
  \overbrace{\forallpaths \always \big(\existspath(\avariablebis_3 = 0) \Leftrightarrow
  \existspath \always (\avariable_1 < \avariable_2)\big)}^{=  \ \aformulabis'_3} \, \wedge \,
  \aformulabis'_2 \, \wedge  \, \aformulabis'_1 \, \wedge \,
  \existspath(\avariablebis_2 = 0) \,
  \wedge \, \existspath(\avariablebis_3 = 0).
\]
\[
  \aformula^{\star} = \forallpaths \always \big((\avariablebis_4 = 0) \Leftrightarrow
  (\avariablebis_2 = 0 \wedge \avariablebis_3 = 0)\big) \, \wedge \,
  \aformulabis'_3 \, \wedge \, \aformulabis'_2 \, \wedge  \, \aformulabis'_1 \, \wedge \,
  \existspath(\avariablebis_4 = 0).
\]
  Each state formula $\aformulabis_i'$ in $\aformula^{\star}$ is then also replaced by a conjunction of two state
  formulae in order to compute $\aformula'$. By way of example, we present below the conjunction replacing
  $\aformulabis'_2$.
  \[
  \forallpaths \always \existspath \big( \neg (\avariablebis_2 = 0) \vee 
  (\avariable_1'<\avariable_1) \until \ (\avariablebis_1 = 0)\big)
  \, \wedge \,
  \forallpaths \always \big( (\avariablebis_2 = 0) \vee
  (\neg (\avariable_1'<\avariable_1)) \release \ \neg(\avariablebis_1 = 0)\big)
  \]
\end{exa}

\subsection{Proof of Proposition~\ref{proposition-ltlz}}
\label{appendix-proof-proposition-ltlz}
\begin{proof} We use the standard automata-based approach for \LTL~\cite{Vardi&Wolper94}, except that we have
  to deal with constraints. Let $\apathformula$ be an $\LTL(\Zed)$ in simple form. We provide below
  usual notations to define the automaton $\aautomaton_{\apathformula}$. 
We write $\subf{\apathformula}$ to denote the smallest set
such that
\begin{itemize}
\item $\apathformula \in \subf{\apathformula}$; $\subf{\apathformula}$ is closed under subformulae,
\item for all $\mathsf{Op} \in \set{\until, \release}$, if
  $\apathformula_1 \ \mathsf{Op} \ \apathformula_2 \in \subf{\apathformula}$, then
   $ \mynext (\apathformula_1  \ \mathsf{Op} \ \apathformula_2) \in \subf{\apathformula}$.
  \end{itemize}
The cardinality of $\subf{\apathformula}$ is at most twice the
number of subformulae of $\apathformula$.
Given $\aset \subseteq \subf{\apathformula}$, $\aset$
is \defstyle{propositionally consistent}  $\equivdef$ the conditions below hold.
\begin{itemize}
\item If $\apathformula_1 \vee \apathformula_2 \in \aset$, then
      $\set{\apathformula_1, \apathformula_2} \cap \aset \neq \emptyset$; 
      if $\apathformula_1 \wedge \apathformula_2 \in \aset$, then
  $\set{\apathformula_1, \apathformula_2} \subseteq \aset$.
\item If $\apathformula_1 \until \apathformula_2 \in \aset$, then
  $\apathformula_2 \in \aset$ or
  $\set{\apathformula_1, \mynext (\apathformula_1 \until \apathformula_2)} \subseteq \aset$.
  
  \item If $\apathformula_1 \release \apathformula_2 \in \aset$, then
  $\apathformula_2 \in \aset$ and
    $\set{\apathformula_1,  \mynext (\apathformula_1 \release \apathformula_2)} \cap \aset \neq \emptyset$.    
  \end{itemize}

We write $\parsubf{\mynext}{\apathformula}$ to denote the set of formulae in $\subf{\apathformula}$
of the form $\mynext \apathformula'$. Similarly, we write
$\parsubf{\until}{\apathformula}$ 
to denote the set of formulae in $\subf{\apathformula}$
of the form $\apathformula_1 \until \apathformula_2$.
Finally, we  write  $\parsubf{{\rm cons}}{\apathformula}$ to denote the set of formulae of
the form $\acons$ in $\subf{\apathformula}$. 

We build a generalised 
word constraint automaton
$\aautomatonbis_{\apathformula} = \triple{\locations,\aalphabet,\beta}{\locations_{\init},\delta}{\rabinacc}$ 
such that
$
  \set{\aword: \Nat \to \Zed^{\beta} \mid \aword \models \apathformula}
  = \alang(\aautomaton_\apathformula)
$.
The automaton $\aautomatonbis_{\apathformula}$ accepts infinite
words $\aword: \Nat \rightarrow \aalphabet \times \Zed^{\beta}$ with $\aalphabet = \set{\dag}$. 
Let us define $\aautomatonbis_{\apathformula}$ formally.
\begin{itemize}
\item $\aalphabet \egdef \set{\arbitraryletter}$; 
  $\locations \subseteq\powerset{\subf{\apathformula}}$ contains all the
  propositionally consistent sets by definition.
\item $\locations_{\init} \egdef \set{\aset \in \locations \mid \apathformula \in \aset}$.
\item The transition relation $\delta$ is made of tuples of the form
  $
  \triple{\aset}{\dag}{\acons,\aset'}
  $, 
  verifying the conditions below.
  \begin{enumerate}
  \item For all $\mynext \apathformula' \in \aset$, we have $\apathformula' \in \aset'$.
  \item $\acons$ is equal to $(\bigwedge_{\acons' \in \parsubf{{\rm cons}}{\apathformula} \cap \aset} \acons')$. 
    \end{enumerate}
 
\item $\rabinacc$ is made of the sets  $F_{\apathformulabis_1 \until \apathformulabis_2}$ with
  $\apathformulabis_1 \until \apathformulabis_2 \in \parsubf{\until}{\apathformula}$
  with $
    F_{\apathformulabis_1 \until \apathformulabis_2 } \egdef
    \set{\aset \mid
      \apathformulabis_2 \in \aset \
      \mbox{or} \
      \apathformulabis_1 \until \apathformulabis_2 \not \in \aset
    }$. 
\end{itemize}
Transforming generalised B\"uchi conditions to standard B\"uchi conditions leads to
a set of locations multiplied by the factor $\card{\rabinacc} + 1$ (which is bounded
by $\size{\apathformula}$) and to the word constraint automaton
$\aautomaton_{\apathformula}$. Moreover, one can check that the number of transitions
in $\aautomaton_{\apathformula}$ is exponential in $\size{\apathformula}$.
We omit the standard proof for correctness (similar to the proof of Lemma~\ref{lemma-correctness-ctlz-aut-to-form}).
\end{proof}

\subsection{Proof of Lemma~\ref{lemma-automaton-AGEPhi}}
\label{appendix-proof-lemma-automaton-AGEPhi} 
\begin{proof}
  Let us start checking that (I) and (II) hold true.
  We have $\locations' = \interval{0}{\degree-1} \times (\locations\cup\{\bot\})$
  with  $\card{\locations}$ exponential in $\size{\apathformula_i}$ (and therefore
  exponential in $\size{\aformula}$) by Lemma~\ref{proposition-ltlz}(I).
  Moreover, $\degree$  is linear in $\size{\aformula}$, whence
  $\card{\locations'}$ is exponential in $\size{\aformula}$.
  By Lemma~\ref{proposition-ltlz}(II),
  $\maxconstraintsize{\aautomaton}$ is quadratic in
  $\size{\apathformula_i}$ (and therefore
  quadratic in $\size{\aformula}$).
  The constraints in $\aautomaton_i$ are those from $\aautomaton$ (maybe except $\top$).
  Consequently, $\maxconstraintsize{\aautomaton_i}$ is quadratic in
  $\size{\aformula}$. This concludes the proof of (I) and (II).

  Let $\aformula = \forallpaths \always \existspath \ \apathformula_i$
and $\atree:\interval{0}{\degree-1}^*\to \aalphabet\times\Zed^\beta$  be a tree
   such that $\atree \in \alang(\aautomaton_i)$.
   Below, we prove that for all nodes 
   $\anode\in \interval{0}{\degree-1}^*$, 
   the word $\atree(\anode) \ \atree(\anode \cdot i) \ \atree(\anode\cdot i \cdot 0) \ 
   \atree(\anode \cdot i \cdot 0^2) \dots$ satisfies $\apathformula_i$. 
   This implies $\atree \models \aformula$ and
   $\atree$ satisfies $\aformula$ via $i$. 

   So let $\anode\in \interval{0}{\degree-1}^*$ be an arbitrary node in $\atree$. 
   Consider the infinite path $\anode \cdot i \cdot 0^\omega \in \interval{0}{\degree-1}^\omega$. 
   Let $\atree(\anode)=\pair{\aletter}{\vect{z}}$ and
   $\atree(\anode \cdot i \cdot 0^k) = \pair{\aletter_k}{\vect{z}_k}$ for all $k\geq 0$. 
   Let $\arun: \interval{0}{\degree-1}^*\to \delta'$
   be an accepting run of $\aautomaton_i$ on $\atree$. 
   By definition of $\delta'$, 
   $\arun(\anode \cdot i)$ is of the form $(\pair{i}{\alocation_0},\dag, \cdots)$
   for some transition $(\alocation_\init,\aletter, \acons,\alocation_0)\in\delta$
   with $\alocation_\init\in\locations_\init$ and $\Zed\models\acons(\vect{z},\vect{z}_0)$. 
   Again by definition of $\delta'$, 
   we have, for all $k\geq 1$, 
   $\arun(\anode\cdot i \cdot 0^k)$ is of the form $(\pair{0}{\alocation_k},\dag,\cdots)$ for some location
   $\alocation_k\in\locations$ such that there exists some transition
   $(\alocation_{k-1},\aletter_{k-1},\acons_{k-1},\alocation_k)\in\delta$, and
   we have $\Zed\models\acons_{k-1}(\vect{z}_{k-1},\vect{z}_k)$. 
   Since $\arun$ is accepting, 
we know that there are infinitely many positions $\ell\geq 1$ such that $\alocation_{\ell}\in F$. 
Hence the run $\alocation \step{\pair{\aletter}{\vect{z}}} \alocation_0
\step{\pair{\aletter_0}{\vect{z}_0}} \alocation_1 \dots$ is an accepting run of $\aautomaton$ on
$(\aletter,\vect{z})(\aletter_0,\vect{z}_0) \dots$. 
But then also $(\aletter,\vect{z})(\aletter_0,\vect{z}_0) \dots$ satisfies $\apathformula_i$. 
Since $\anode$ is arbitrary, we have
$\atree \models \aformula$
and $\atree$ satisfies $\aformula$ via $i$.

For the other direction, suppose that
$\atree \models \aformula$
and $\atree$ satisfies $\aformula$ via $i$. 
We prove that $\atree \in \alang(\aautomaton_i)$, that is, there exists
some accepting run $\arun:\interval{0}{\degree-1}^* \to \delta'$ of $\aautomaton_i$ on $\atree$.
We prove that we can construct a run $\arun$ such that, for every
node $\anode\in\interval{0}{\degree-1}^*$, 
the path $\arun(\anode)\arun(\anode\cdot i) \arun(\anode\cdot i \cdot 0)
\arun(\anode\cdot i \cdot 0^2) \dots$ corresponds to some accepting run of
$\aautomaton_i$. This implies that $\arun$ is an accepting run of $\aautomaton$ on $\atree$
(all the other paths are non-critical). 

Let $\anode\in\interval{0}{\degree-1}^*$ be an arbitrary node in $\atree$. 
Let us assume $\atree(\anode)=\pair{\aletter}{\vect{z}}$ and 
$\atree(\anode\cdot i \cdot 0^k)=\pair{\aletter_k}{\vect{z}_k}$ for all $k\geq 0$, 
and we use $\aword_\anode$ to denote the corresponding infinite word 
$\pair{\aletter}{\vect{z}} \pair{\aletter_0}{\vect{z}_0} \pair{\aletter_1}{\vect{z}_1}\dots$. 
Since $\atree$ satisfies $\aformula$ via $i$, we know that 
$\aword_\anode \models \apathformula_i$. 
But then also $\aword_\anode \in \alang(\aautomaton)$. 
So there must exist some accepting run of $\aautomaton$ on $\aword_\anode$,
say with projection on the set of locations equal to $\alocation_\init^{\anode},\alocation_0^{\anode},\alocation_1^{\anode},\alocation_2^{\anode}\dots$
such that
\begin{itemize}
\item $\alocation_\init^{\anode} \in \locations_\init$,
\item there exist transitions
$(\alocation_\init^{\anode}, \aletter,\acons,\alocation_0^{\anode})\in\delta$ and
$(\alocation_k^{\anode},\aletter_k,\acons_k,\alocation_{k+1}^{\anode})\in\delta$ for all $k\geq 0$, 
with  $\Zed\models\acons(\vect{z},\vect{z}_0)$ and $\Zed\models\acons_k(\vect{z}_k,\vect{z}_{k+1})$
for all $k\geq 0$.
\end{itemize}
The definition of $\delta'$ now allows us to define 
$\arun(\anode\cdot i)= (\pair{i}{\alocation_0^{\anode}}, \dag, \cdots, \pair{i}{\alocation_0^{\anode \cdot i}},
\cdots)$, 
and $\arun(\anode\cdot i \cdot 0^k)= (\pair{0}{\alocation_k^{\anode}},\dag, \cdots,
\pair{i}{\alocation_0^{\anode \cdot i \cdot 0^k}},\cdots)$ for every $k\geq 1$.
\end{proof}

\section{Proofs for Section~\ref{section-ctlstarz-determinisation-safra}} 
\subsection{Proof of Lemma~\ref{lemma-safra-node-number}}
\label{appendix-proof-lemma-safra-node-number} 
\begin{proof}
Let $\safra$ be a Safra tree over $\locations$. 
If $\safra$ has no nodes, the claim is of course true. 
So let us assume that $\safra$ contains at least one node. 
We prove the claim by induction on the height $H$ of $\safra$. 
For the induction base, let $H=1$. The tree then has exactly one node, namely the root node, and the claim is trivially true. 
So suppose the claim holds for $H\geq 1$. We prove the claim for $H+1$. 
Suppose the root node of $\safra$ has $k$ children, denoted by $\anode_1, \dots, \anode_k$. 
For every $1\leq i \leq k$, let $\locations_i\subseteq \locations$ denote the label of $\anode_i$. 
For every $1\leq i \leq k$, 
the subtree of $\anode_i$ is a Safra tree over $\locations_i$ 
with depth at most $H$. 
By the induction hypothesis, such a subtree has at most $\card{\locations_i}$ nodes. 
By condition 5 of Safra trees, 
the sets $\locations_1,\dots, \locations_k$ are pairwise disjoint. 
By condition 4, the union $\bigcup_{1\leq i\leq k} \locations_i$ is a proper subset of
$\locations$. 
Hence $\sum_{i=1}^k \card{\locations_i} < \card{\locations}$. 
Altogether, the number of nodes in $\safra$ is at most
$1+\sum_{i=1}^k \card{\locations_i} < 1 + \card{\locations} \leq \card{\locations}$. 
\end{proof}

\subsection{Proof of Lemma~\ref{lemma_safra_dir_one_rabin}}
\label{appendix-proof-lemma_safra_dir_one_rabin}
\begin{proof}
Suppose $i\geq 1$, $1\leq J_i \leq 2 \cdot \card{\locations}$, and for all $k\geq i$, 
$\safra_k$ contains a node with name $J_i$ and $\alocation_k\in \saflab(\safra_k, J_i)$. 
If the first property holds, we are done. 
Otherwise, there exists some position $m\geq i$ such that 
\begin{itemize}
\item for all $k\geq m$, $\safra_k$ contains the node with name $J_i$ unmarked, and
\item $\alocation_{m}=\alocation_\acc$, and hence $\alocation_\acc\in \saflab(\safra_{m}, J_i)$. 
\end{itemize}
Using the definition of $\delta'$, it is easy to prove that in $\safra_{m+1}$, 
the node with name $J_i$ has a child node with name $J_m\neq J_i$ such that 
$\alocation_{m+1}\in \saflab(\safra_{m+1}, J_m)$. 
If for all $k\geq m+1$, $\safra_k$ contains a node with name $J_m$ and $\alocation_{k}\in \saflab(\safra_k,J_m)$, we are done. 
Otherwise, there must exist some position $p> m+1$ such that 
\begin{itemize}
\item $\safra_k$ contains a node with name $J_m$ and $\alocation_k\in \saflab(\safra_k, J_m)$ for all $m+1\leq k < p$, and 
\item $\safra_p$ does not contain a node with name $J_m$, or $\alocation_p\not\in \saflab(\safra_p, J_m)$. 
\end{itemize}
By definition of $\delta'$, there are three cases: during the construction of $\safra_p$ out of $\safra_{p-1}$. 
\begin{enumerate}
\item[(a)] The location $\alocation_{p}$ is removed from the label of the node with name $J_m$, because there exists some younger sibling of the node named $J_m$ (that is, an older child of $J_i$) whose label set contains $\alocation_p$. 
\item[(b)] The node with name $J_m$ has been removed from the Safra tree during step (5). But for this, the node with the name $J_m$ must
  have an empty label set, contradicting  $\alocation_{p-1}\in \saflab(\safra_{p-1}, J_m)$ and
  there is $(\alocation_{p-1}, \aletter_{p-1}, \acons'_{p-1}, \alocation_p) \in \delta$
  such that $\Zed \models \acons'_{p-1}(\vect{z}_{p-1},\vect{z}_{p})$-- so this case cannot occur.
\item[(c)] The node with name $J_m$ has been removed from the Safra tree during step (6). 
But for this the parent node with name $J_i$ must be marked, contradiction   -- so this case cannot occur. 
\end{enumerate}
Note that  case (a) can only occur at most $\card{\locations}-1$ times, as by
Lemma~\ref{lemma-safra-node-number}, 
the node with name $J_i$ can have at most $\card{\locations}-1$ children nodes. We can conclude that there must exist some position $i'\geq i$ and some name $1\leq J_{i'} \leq 2 \cdot \card{\locations}$
with $J_i\neq J_{i'}$ such that for all $k\geq i'$, $\safra_k$ contains a node with name $J_{i'}$ and $\alocation_k\in\saflab(\safra_k,J_{i'})$. 
\end{proof}

\subsection{Proof of Nonemptiness of $\text{Acc}(\alocation,j)$ and $\text{Pre}(\alocation, j,k)$}
\label{appendix-proof-lemma-safra-correctness}
In this section, we prove that the sets 
$\text{Acc}(\alocation,j)$ and $\text{Pre}(\alocation, j,k)$ are nonempty as stated in forthcoming Lemma~\ref{proof-lemma-safra-correctness-two-nonemptysets}.
Before, we prove one helpful lemma, also related to the correctness of the Safra construction. 
\begin{lem}
\label{lemma_safra_powerset_path}
For all infinite runs $\arun'$ of the form
$$(\safra_1,\aletter_1,\acons_1,\safra_2)(\safra_2,\aletter_2,\acons_2,\safra_3)(\safra_3,\aletter_3,\acons_3,\safra_4)\dots$$
of $\aautomaton'$ on
$(\aletter_{1}, \vect{z}_1)  (\aletter_{2},\vect{z}_{2})(\aletter_{3},\vect{z}_{3}) \dots$, 
for every $1\leq J\leq 2 \cdot \card{\locations}$ and for every $1\leq j\leq j'$, 
if $\safra_k$ contains a node with name $J$ for all $j\leq k\leq j'$, 
then for every $\alocation'\in \saflab(\safra_{j'}, J)$ there exist some $q\in \saflab(\safra_j,J)$ and some finite run 
$$(\alocation_j, \aletter_j,\acons'_j,\alocation_{j+1}) \dots (\alocation_{j'-1},\aletter_{j'-1},\acons'_{j'-1}, \alocation_{j'})$$ of $\aautomaton$ on $(\aletter_{j}, \vect{z}_j)(\aletter_{j+1}, \vect{z}_{j+1}) \dots (\aletter_{j'},\vect{z}_{j'})$ with 
$\alocation_j=\alocation$ and $\alocation_{j'}=\alocation'$. 
\end{lem}
\begin{proof}
The proof is by induction on the difference $\Delta=j'-j$. 
For the induction base, let $\Delta=0$. 
By convention, $\alocation'$ is a finite run on $(\aletter_{j'},\vect{z}_{j'})$ for every $\alocation'\in \saflab(\safra_{j'}, J)=\saflab(\safra_j,J)$. 
Suppose  the claim holds for $\Delta\geq 0$; we prove it for $\Delta+1$. 
So suppose $\alocation'\in \saflab(\safra_{j'},J)$. 
Consider the transition step $(\safra_{j'-1},\aletter_{j'-1}, \acons_{j'-1}, \safra_{j'})\in \delta'$ used in $\arun'$. 
By step (3) of the definition of $\delta'$, 
there exists some $\alocation''\in \saflab(\safra_{j'-1},J)$ 
and some transition $(\alocation'',\aletter_{j'-1},\acons_{j'-1}',\alocation')\in \delta$ such that
$\acons_{j'-1} \models \acons_{j'-1}'$. 
From $\arun'$ being a run, we obtain that
$\Zed\models \acons_{j'-1}(\vect{z}_{j'-1},\vect{z}_{j'})$, hence also 
$\Zed\models\acons_{j'-1}'(\vect{z}_{j'-1},\vect{z}_{j'})$. 
Hence $\arun_2=(\alocation'',\aletter_{j'-1},\acons_{j'-1}',\alocation')$ is a finite run of $\aautomaton$ on $(\aletter_{j'-1},\vect{z}_{j'-1})(\aletter_{j'},\vect{z}_{j'})$. 
By the induction hypothesis, 
there exist $\alocation\in \saflab(\safra_j,J)$ and 
some finite run 
$$\arun_1=(\alocation_j, \aletter_j,\acons'_j,\alocation_{j+1}) \dots (\alocation_{j'-2},\aletter_{j'-2},\acons'_{j'-2}, \alocation_{j'-1})$$ of $\aautomaton$ on $(\aletter_{j}, \vect{z}_j)(\aletter_{j+1}, \vect{z}_{j+1}) \dots (\aletter_{j'-1},\vect{z}_{j'-1})$ with 
$\alocation_j=\alocation$ and $\alocation_{j'}=\alocation''$. 
The final run is obtained by composing $\arun_1$ and $\arun_2$.
\end{proof}

\begin{lem}
\label{proof-lemma-safra-correctness-two-nonemptysets}
For every  $j\geq 1$ and for every $\alocation\in\saflab(\safra_{i_j}, J)$, $\textup{Acc}(\alocation,j)\neq \emptyset$, 
and  for every $j\geq 2$, every $i_{j-1} \leq k < i_j $ and every $\alocation \in \saflab(\safra_{k},J)\cap F$, 
$\text{Pre}(\alocation,j,k)\neq\emptyset$. 
\end{lem}
\begin{proof}
Let $j\geq 1$ and $\alocation\in\saflab(\safra_{i_j}, J)$. 
We prove that $\textup{Acc}(\alocation,j)\neq \emptyset$. 
By definition of $\delta'$, $\saflab(\safra_{i_j},J) = \saflab((\safra_{i_j-1})^{(5)},J)$, hence $\alocation\in\saflab((\safra_{i_j-1})^{(5)},J)$. 
Since the node with name $J$ is marked in $\safra_{i_j}$ (step (6)), 
it  must have some child node in $(\safra_{i_j-1})^{(5)}$, say with name $K\neq J$, such that $\alocation\in \saflab((\safra_{i_j-1})^{(5)}, K)$.
Then there exists $\alocation' \in \saflab((\safra_{i_j-1})^{(2)}, K)$ such that  
$(\alocation',\aletter_{i_j-1},\acons'_{i_j-1},\alocation)$ is a finite run of $\aautomaton$ on $(\aletter_{i_j-1},\vect{z}_{i_j-1})(\aletter_{i_j},\vect{z}_{i_j})$. 
By property (4) of Safra trees, we know that $\alocation'\in \saflab((\safra_{i_j-1})^{(2)}, J)$. 
If $\alocation'\in F$, then we are done. 
Otherwise, $K$ is a node in $\safra_{i_j-1}$ (it was not created as a youngest child of $J$ by step (2) of $\delta'$). 
Let $i_{j-1}<k<i_j-1$ be the minimal position such that node named $J$ has no child node with name $K$ in $\safra_{k-1}$, and node named $J$ has a child node
with name $K$ in $\safra_{k'}$ for all $k\leq k'\leq i_j-1$.
Such a position necessarily exists; indeed, for all $m\geq 0$, 
node $J$ in $\safra_{i_m}$ has no children nodes: for $\safra_{i_0}$, this is because
$J$ is freshly introduced by step (2),
and for all $m>1$, this is because the node with name $J$ in $\safra_{i_m}$ is marked and  marked nodes (step (6)) do not have children nodes. 
By Lemma~\ref{lemma_safra_powerset_path}, 
there exists $\alocation_k\in\saflab(\safra_k,K)$ and some finite run 
$(\alocation_k,\aletter_k,\acons'_k,\alocation_{k+1})\dots (\alocation_{i_j-2},\aletter_{i_j-2},\acons'_{i_j-2},\alocation_{i_j-1})$ of $\aautomaton$ on $(\aletter_k,\vect{z}_k)\dots (\aletter_{i_j-1},\vect{z}_{i_j-1})$, where $\alocation_{i_j-1}=\alocation'$. 
Finally, using the fact that $\safra_k$ contains the node with name $K$, whereas $\safra_{i_{j-1}}$ does not, it is not hard to prove that there exists $\alocation_{k-1}\in\saflab(\safra_{k-1}, J)\cap F$ such that
$(\alocation_{k-1},\aletter_{k-1},\acons'_{k-1},\alocation_k)$ is a finite run of $\aautomaton$ on
$(\aletter_{k-1},\vect{z}_{k-1})(\aletter_k,\vect{z}_k)$.
 Hence $\alocation_{k-1}\in\textup{Acc}(\alocation,j)$, which finishes the proof. 

 Let $j\geq 2$,  $i_{j-1} \leq k < i_j $ and  $\alocation \in \saflab(\safra_{k},J)\cap F$. For
 proving that $\text{Pre}(\alocation,j,k)\neq\emptyset$, we can apply Lemma~\ref{lemma_safra_powerset_path}. 
\end{proof}

\section{Proofs for Section~\ref{section-starproperty}}
\label{appendix-last} 
\subsection{Proof of Lemma~\ref{lemma-shortcircuit-equality-lessthan}}
\label{appendix-proof-lemma-shortcircuit-equality-lessthan} 
\begin{proof} Let $\apath = \pair{\anode_0}{\advar_0} \step{\sim_1} \cdots \step{\sim_n} \pair{\anode_n}{\advar_n}$ be a path
  in $\newGt$ such that $\anode_0$ and $\anode_n$ are neighbours.

(Property 1)
Firstly, if $n \geq 2$, then one can show that there is  $0 < h < n$ such that $\anode_0$, $\anode_h$ and $\anode_n$ are pairwise neighbours.
Let us explain briefly how $h$ is computed.
\begin{itemize}
\item Case $\length{\anode_0} = \length{\anode_n}$ ($\anode_0 = \anode_n$). $h \egdef n-1$.
\item Case $\length{\anode_0} = \length{\anode_n} + 1$ ($\anode_0$ is a child of $\anode_n$)
  and for all $i \in \interval{1}{n-1}$, $\length{\anode_i} \geq \length{\anode_0}$.
      $h \egdef n-1$.
\item Case $\length{\anode_0} = \length{\anode_n} + 1$ and there is $i \in \interval{1}{n-1}$, $\length{\anode_i} < \length{\anode_0}$.
  $h \egdef \min \set{i \in \interval{1}{n-1} \mid \length{\anode_i} = \length{\anode_n}}$.
  Note that actually $\anode_h = \anode_n$. 
\item Case $\length{\anode_n} = \length{\anode_0} + 1$
      ($\anode_n$ is a child of $\anode_0$)
      and for all $i \in \interval{1}{n-1}$, $\length{\anode_i} \geq \length{\anode_n}$.
      $h \egdef 1$.
      Note that actually $\anode_h = \anode_n$. 
\item Case $\length{\anode_n} = \length{\anode_0} + 1$ and there is $i \in \interval{1}{n-1}$, $\length{\anode_i} < \length{\anode_n}$.
      $h \egdef \min \set{i \in \interval{1}{n-1} \mid \length{\anode_i} = \length{\anode_0}}$. 
\end{itemize}
Observe that $h \leq n-1$, $n-h \leq n-1$, and that $\anode_0$ and $\anode_h$ are neighbours, and 
$\anode_n$ and $\anode_h$ are neighbours.

(Property 2) Second, by construction of $\newGt$ from $\asymtree$ built over satisfiable constraints in $\sattypes{\beta}$, we can show the
property below (this requires a lengthy case analysis).
Let $\pair{\anodebis_1}{\advar_1}$, $\pair{\anodebis_2}{\advar_2}$ and $\pair{\anodebis_3}{\advar_3}$
be nodes in the graph $\newGt$ that are pairwise neighbours such that
$\pair{\anodebis_1}{\advar_1} \step{\sim_1} \pair{\anodebis_2}{\advar_2}$ and $\pair{\anodebis_2}{\advar_2} \step{\sim_2}
\pair{\anodebis_3}{\advar_3}$ with $\sim_1, \sim_2 \in \set{<,=}$. If $< \in \set{\sim_1, \sim_2}$, then
$\pair{\anodebis_1}{\advar_1} \step{<} \pair{\anodebis_3}{\advar_3}$ otherwise
$\pair{\anodebis_1}{\advar_1} \step{=} \pair{\anodebis_3}{\advar_3}$
(this uses the local consistency of $\asymtree$).

Now, we can prove the lemma. If $n = 0$ or $n =1$, we are done. Otherwise, the induction hypothesis
assumes that the property holds for $n \leq K$ and
let  $\apath = \pair{\anode_0}{\advar_0} \step{\sim_1} \cdots \step{\sim_n} \pair{\anode_n}{\advar_n}$ be a path
  in $\newGt$ such that $\anode_0$ and $\anode_n$ are neighbours with $n = K+1$.
By (Property 1), there  is $0 < h < n$ such that $\anode_0$, $\anode_h$ and $\anode_n$ are pairwise neighbours,
$h \leq K$ and $n-h \leq K$. By the induction hypothesis, we have
$\pair{\anode_0}{\advar_0} \step{<} \pair{\anode_h}{\advar_h}$ if $< \in \set{\sim_1, \ldots, \sim_h}$,
otherwise $\pair{\anode_0}{\advar_0} \step{=} \pair{\anode_h}{\advar_h}$. Similarly,
we have
$\pair{\anode_h}{\advar_h} \step{<} \pair{\anode_n}{\advar_n}$ if $< \in \set{\sim_{h+1}, \ldots, \sim_n}$,
otherwise $\pair{\anode_h}{\advar_h} \step{=} \pair{\anode_n}{\advar_n}$.
By (Property 2), we get $\pair{\anode_0}{\advar_0} \step{<} \pair{\anode_n}{\advar_n}$ if
$< \in \set{\sim_{1}, \ldots, \sim_n}$, otherwise $\pair{\anode_0}{\advar_0} \step{=} \pair{\anode_n}{\advar_n}$. 
\end{proof}

\subsection{Proof of  Lemma~\ref{lemma-characterisation-satisfiability-goplus}}
\label{appendix-proof-lemma-characterisation-satisfiability-goplus}
\begin{proof} \fbox{(I) $\Rightarrow$ (II)}
Suppose $\asymtree$ is satisfiable.
Then there exists a tree $\atree:\interval{0}{\degree-1}^*\to \aalphabet \times \Zed^\beta$ such that
for all $\anode \cdot i \in \interval{0}{\degree-1}^+$ with
$\asymtree(\anode \cdot i) = \pair{\aletter}{\acons}$,
we have $\Zed \models \acons(\atree(\anode), \atree(\anode \cdot i))$.
Moreover, if $\asymtree(\varepsilon) = \pair{\aletter}{\acons}$
and $\atree(\varepsilon) = \pair{\aletter}{\vect{z}}$, then
$\Zed \models \acons(\vect{0},\vect{z})$.

Given $\anode\in \interval{0}{\degree-1}^*$ and $\avariable_i\in \{\avariable_1,\dots,\avariable_\beta\}$, 
in the following, we write $\atree(\anode)(\avariable_i)$ to denote the data value $z_i$ if $\atree(\anode)=(z_1,\dots,z_i,\dots,z_\beta)$. 
Similarly, we write $\atree(\anode)(\adatum_1)$ to denote $\adatum_1$ and $\atree(\anode)(\adatum_\alpha)$ to denote $\adatum_\alpha$. 

{\em Ad absurdum}, suppose there exists $(\anode,\avariable_i)\in U_{< \adatum_1} \cup U_{> \adatum_{\alpha}}$  in $\newGt$ such that 
$\slen{\anode,\avariable_i} = \omega$. 
We prove the claim for the case $(\anode,\avariable_i)\in U_{< \adatum_1}$; the proof for
$(\anode,\avariable)\in U_{> \adatum_{\alpha}}$ is analogous.
Recall that, by definition, $\slen{\anode,\avariable_i} =\slen{\pair{\anode}{\avariable_i},\pair{\anode}{\adatum_1}}$.
Define $\Delta = \adatum_1 - \atree(\anode)(\avariable_i)$. 
Let $\apath$ be a path 
$(\anode_0,\advar_0)\step{\sim_1} (\anode_1,\advar_1) \step{\sim_2} \dots \step{\sim_k}(\anode_k,\advar_k)$ such that 
$\pair{\anode_0}{\advar_0}=\pair{\anode}{\avariable_i}$,
$\pair{\anode_k}{\advar_k}=\pair{\anode}{\adatum_1}$, 
and 
$\slen{\apath}>\Delta$. Such a path must exist by assumption because $\slen{\anode,\avariable_i} = \omega$.
By Lemma~\ref{lemma-correctness-newgt},
we have $\atree(\anode_{i-1})(\advar_{i-1}) \sim_i \atree(\anode_i)(\advar_i)$ for all $1\leq i\leq k$. 
But this implies that there are more than $\Delta$ different data values in the interval
$\interval{\atree(\anode_0)(\advar_0)}{\atree(\anode_k)(\advar_k)} = 
\interval{\atree(\anode)(\avariable_i)}{\atree(\anode)(\adatum_1)}$, which leads to a contradiction.

\fbox{(II) $\Rightarrow$ (I)} Suppose that for all $\pair{\anode}{\avariable}$ in $(U_{< \adatum_1} \cup U_{> \adatum_{\alpha}})$ in
$\newGt$ we have $\slen{\anode,\avariable} < \omega$.             
We define the mapping $g:\interval{0}{\degree-1}^*  \times \{\avariable_1,\dots,\avariable_\beta\}      \to \Zed$ as follows:
\begin{itemize}
\item $g(\anode,\avariable) \egdef \adatum$ if $\pair{\anode}{\avariable} \in U_{\adatum}$ for
some $\adatum \in\interval{\adatum_1}{\adatum_{\alpha}}$, 
\item $g(\anode,\avariable) \egdef \adatum_1 - \slen{\anode,\avariable}$ if $\pair{\anode}{\avariable} \in U_{< \adatum_1}$, and
\item $g(\anode,\avariable) \egdef \adatum_\alpha + \slen{\anode,\avariable}$ if
$\pair{\anode}{\avariable} \in U_{> \adatum_{\alpha}}$. 
\end{itemize}
Recall that
$\set{U_{\adatum} \mid \adatum \in \interval{\adatum_1}{\adatum_{\alpha}}} \cup \set{U_{< \adatum_1}, U_{> \adatum_{\alpha}}}$
is a partition of
$\interval{0}{\degree-1}^* \times \DVAR{\beta}{\adatum_1}{\adatum_{\alpha}}$ so that $g$ is indeed well-defined. 
Now 
define $\atree':\interval{0}{\degree-1}^* \to \aalphabet \times \Zed^\beta$ by $\atree'(\anode) \egdef
\pair{\aletter}{\tuple{g(\anode,\avariable_1)}{
g(\anode,\avariable_\beta)}}$ for all $\anode\in\interval{0}{\degree-1}^*$ with  $\asymtree(\anode)=(\aletter,\cdot)$.
We prove that $\atree'$ witnesses the satisfaction of $\asymtree$. 
For this,
we verify that
for all
$\anode \cdot j \in \interval{0}{\degree-1}^+$ with $\asymtree(\anode \cdot j)=(\aletter,\acons)$,
we have $\Zed \models\acons(\atree'(\anode), \atree'(\anode \cdot j))$.
Moreover, concerning the case with the root $\varepsilon$, one can show that $\Zed \models\acons(\vect{0}, \atree'(\varepsilon))$
with $\asymtree(\varepsilon)=(\aletter,\acons)$ but we omit it below as it is very similar to the general case. 

\begin{itemize}
\item Suppose $\avariable'=\adatum \in \acons$ for some $\adatum \in\interval{\adatum_1}{\adatum_{\alpha}}$.
By definition of $\newGt$, 
we have $(\anode \cdot j,\avariable)\in U_{\adatum}$. By definition, $g(\anode \cdot j,\avariable)=\adatum$
and therefore $\Zed \models (\avariable'=\adatum)(\atree'(\anode), \atree'(\anode \cdot j))$.
\item Suppose $\avariable=\adatum \in \acons$ for some $\adatum \in\interval{\adatum_1}{\adatum_{\alpha}}$.
By definition of $\newGt$, and as $\asymtree$ is locally consistent, 
we have $(\anode,\avariable)\in U_{\adatum}$. By definition, $g(\anode,\avariable)=\adatum$
and therefore $\Zed \models (\avariable=\adatum)(\atree'(\anode), \atree'(\anode \cdot j))$. 
\item Suppose $\avariable'<\adatum_1\in \acons$. By definition of $\newGt$, 
we have $(\anode \cdot j,\avariable)\in U_{< \adatum_1}$. By definition, $g(\anode \cdot j,\avariable)=\adatum_1-\slen{\anode \cdot j,\avariable}$.  
By definition of $\newGt$, we also have $\pair{\anode \cdot j}{\adatum_1}\in U_{\adatum_1}$ and hence
$\pair{\anode \cdot j}{\avariable} \step{<} \pair{\anode \cdot j}{\adatum_1}$. Hence $\slen{\anode \cdot j,\avariable}\geq 1$,
so that indeed $g(\anode \cdot j,\avariable)<\adatum_1$ and  $\Zed \models (\avariable'<\adatum_1)(\atree'(\anode), \atree'(\anode \cdot j))$.
\item Suppose $\avariable<\adatum_1\in \acons$. By definition of $\newGt$, and as $\asymtree$ is locally consistent, 
we have $(\anode,\avariable)\in U_{< \adatum_1}$. By definition, $g(\anode,\avariable)=\adatum_1-\slen{\anode,\avariable}$.  
By definition of $\newGt$,
we also have $\pair{\anode}{\adatum_1}\in U_{\adatum_1}$ and hence
$\pair{\anode}{\avariable} \step{<} \pair{\anode}{\adatum_1}$. Hence $\slen{\anode,\avariable}\geq 1$,
so that indeed $g(\anode,\avariable)<\adatum_1$ and  $\Zed \models (\avariable<\adatum_1)(\atree'(\anode), \atree'(\anode \cdot j))$.
\item Suppose $\avariable'>\adatum_\alpha\in \acons$. By definition of $\newGt$, 
we have $(\anode \cdot j,\avariable)\in U_{> \adatum_{\alpha}}$. 
By definition, $g(\anode \cdot j,\avariable)=\adatum_\alpha+\slen{\anode \cdot j,\avariable}$. 
We also have $\pair{\anode \cdot j}{\adatum_\alpha} \in U_{\adatum_{\alpha}}$ and hence
$\pair{\anode \cdot j}{\adatum_\alpha} \step{<} \pair{\anode \cdot j}{\avariable}$. Hence $\slen{\anode \cdot j,\avariable}\geq 1$, so that
indeed $g(\anode \cdot j,\avariable)>\adatum_\alpha$ and therefore $\Zed \models (\avariable'> \adatum_{\alpha})(\atree'(\anode), \atree'(\anode \cdot j))$.
\item Suppose $\avariable>\adatum_\alpha\in \acons$. By definition of $\newGt$ and as $\asymtree$ is locally consistent, 
we have $(\anode,\avariable)\in U_{> \adatum_{\alpha}}$. 
By definition, $g(\anode,\avariable)=\adatum_\alpha+\slen{\anode,\avariable}$. 
We also have $\pair{\anode}{\adatum_\alpha} \in U_{\adatum_{\alpha}}$ and hence
$\pair{\anode}{\adatum_\alpha} \step{<} \pair{\anode}{\avariable}$. Hence $\slen{\anode,\avariable}\geq 1$, so that
indeed $g(\anode,\avariable)>\adatum_\alpha$ and therefore $\Zed \models (\avariable> \adatum_{\alpha})(\atree'(\anode), \atree'(\anode \cdot j))$.
\item Suppose $\avariable' < \avariablebis'\in\acons$. 
We distinguish the following cases.

\begin{itemize}
\item Suppose $\pair{\anode \cdot j}{\avariable}\in U_{\adatum}$ and $\pair{\anode \cdot j}{\avariablebis}\in U_{\adatum'}$ for some
$\adatum, \adatum' \in \interval{\adatum_1}{\adatum_{\alpha}}$. 
This also implies $\avariable'=\adatum, \avariablebis'=\adatum'\in\acons$. 
Recall that $\acons$ is satisfiable, hence $\adatum<\adatum'$ must hold. 
By definition, $g(\anode \cdot j,\avariable)=\adatum$ and $g(\anode \cdot j,\avariablebis)=\adatum'$, and hence clearly
$g(\anode \cdot j,\avariable)<g(\anode \cdot j,\avariablebis)$ and therefore
$\Zed \models (\avariable' < \avariablebis')(\atree'(\anode), \atree'(\anode \cdot j))$. 
\item  Suppose $\pair{\anode \cdot j}{\avariable}\in U_{< \adatum_1}$ and $\pair{\anode \cdot j}{\avariablebis}\in U_{\adatum}$ for some
$\adatum \in \interval{\adatum_1}{\adatum_{\alpha}}$.
This also implies $\avariable'<\adatum_1\in\acons$. We have proved above that $g(\anode \cdot j,\avariable)<\adatum_1$. 
By definition, $g(\anode \cdot j,\avariablebis)=\adatum \geq \adatum_1$.  
Hence $g(\anode \cdot j,\avariable)<g(\anode \cdot j,\avariablebis)$
and therefore
$\Zed \models (\avariable' < \avariablebis')(\atree'(\anode), \atree'(\anode \cdot j))$. 
\item 
Suppose $\pair{\anode \cdot j}{\avariable}\in U_{< \adatum_1}$ and $\pair{\anode \cdot j}{\avariablebis}\in U_{< \adatum_1}$. 
By definition, 
$g(\anode \cdot j,\avariable)=\adatum_1 - \slen{\anode \cdot j,\avariable}$ and 
$g(\anode \cdot j,\avariablebis)=\adatum_1 - \slen{\anode \cdot j,\avariablebis}$. 
By assumption and definition of $\newGt$, 
we have 
$\pair{\anode \cdot j}{\avariable} \step{<} \pair{\anode \cdot j}{\avariablebis}$. 
Recall that $\slen{\anode \cdot j,\avariable} = \slen{\pair{\anode \cdot j}{\avariable},\pair{\anode \cdot j}{\adatum_1}}$
and  $\slen{\anode \cdot j,\avariablebis} = \slen{\pair{\anode \cdot j}{\avariablebis},\pair{\anode \cdot j}{\adatum_1}}$. 
By construction of $\newGt$, we have $\pair{\anode \cdot j}{\adatum_1}\in U_{\adatum_1}$, and hence
$\pair{\anode \cdot j}{\avariable} \step{<} \pair{\anode \cdot j}{\adatum_1}$ and $\pair{\anode \cdot j}{\avariablebis} \step{<} \pair{\anode \cdot j}{\adatum_1}$.
This clearly yields $\slen{\anode \cdot j ,\avariable}
\geq \slen{\anode \cdot ,\avariablebis}+1$. Hence $\slen{\anode \cdot j,\avariable}>\slen{\anode \cdot j,\avariablebis}$, so that indeed 
$g(\anode \cdot j,\avariable)<g(\anode \cdot j,\avariablebis)$ and $\Zed \models (\avariable' < \avariablebis')(\atree'(\anode), \atree'(\anode \cdot j))$. 
\item Suppose $\pair{\anode \cdot j}{\avariable}\in U_{< \adatum_1}$ and $\pair{\anode \cdot j}{\avariablebis}\in U_{> \adatum_{\alpha}}$. 
This implies $\avariable'<\adatum_1,\avariablebis'>\adatum_\alpha\in \acons$. 
We have proved above that 
$g(\anode \cdot j,\avariable)<\adatum_1$ and $g(\anode \cdot j,\avariablebis)>\adatum_\alpha$. Hence 
$g(\anode \cdot j,\avariable)<g(\anode \cdot j,\avariablebis)$ and $\Zed \models (\avariable' < \avariablebis')(\atree'(\anode), \atree'(\anode \cdot j))$.
\item Suppose $\pair{\anode \cdot j}{\avariable}\in U_{\adatum}$ for some
$\adatum \in \interval{\adatum_1}{\adatum_{\alpha}}$ and $\pair{\anode \cdot j}{\avariablebis}\in U_{> \adatum_{\alpha}}$. 
This also implies $\avariablebis'>\adatum_\alpha\in\acons$. We have proved above that $g(\anode \cdot j,\avariablebis)>\adatum_\alpha$.
By definition, $g(\anode \cdot j,\avariable)=\adatum\leq \adatum_\alpha$. 
Hence 
$g(\anode \cdot j,\avariable)<g(\anode \cdot j,\avariablebis)$ and
$\Zed \models (\avariable' < \avariablebis')(\atree'(\anode), \atree'(\anode \cdot j))$.
\item Suppose $\pair{\anode \cdot j}{\avariable}, \pair{\anode \cdot j}{\avariablebis}\in U_{> \adatum_{\alpha}}$.
By definition, 
$g(\anode \cdot j,\avariable)=\adatum_\alpha + \slen{\anode \cdot j,\avariable}$ and 
$g(\anode \cdot j,\avariablebis)=\adatum_\alpha + \slen{\anode \cdot j,\avariablebis}$. 
By assumption and definition of $\newGt$, 
we have 
$\pair{\anode \cdot j}{\avariable} \step{<} \pair{\anode \cdot j}{\avariablebis}$. 
Recall that $\slen{\anode \cdot j,\avariable} = \slen{\pair{\anode \cdot j}{\adatum_\alpha},\pair{\anode \cdot j}{\avariable}}$
and  $\slen{\anode \cdot j,\avariablebis} = \slen{\pair{\anode \cdot j}{\adatum_\alpha},\pair{\anode \cdot j}{\avariablebis}}$. 
By construction of $\newGt$, we have $\pair{\anode \cdot j}{\adatum_\alpha}\in U_{\adatum_{\alpha}}$, and hence
$\pair{\anode \cdot j}{\adatum_\alpha} \step{<} \pair{\anode \cdot j}{\avariable}$ and $\pair{\anode \cdot j}{\adatum_\alpha} \step{<}
\pair{\anode \cdot j}{\avariablebis}$.
This clearly yields $\slen{\anode \cdot j,\avariablebis}\geq \slen{\anode \cdot j,\avariable}+1$
because $\pair{\anode \cdot j}{\avariable} \step{<} \pair{\anode \cdot j}{\avariablebis}$.
Hence $\slen{\anode \cdot j,\avariable}<\slen{\anode \cdot j,\avariablebis}$, so that indeed 
$g(\anode \cdot j,\avariable)<g(\anode \cdot j,\avariablebis)$ and $\Zed \models (\avariable' < \avariablebis')(\atree'(\anode), \atree'(\anode \cdot j))$.
\end{itemize}
The other cases cannot happen thanks to local consistency. For instance,
$\avariablebis' < \adatum_{1}, \avariable' > \adatum_{\alpha}$ and
$\avariable' < \avariablebis'$ in $\acons$ cannot happen due to consistency
for elements in $\sattypes{\beta}$. 

\item The cases $\avariable' < \avariablebis \in \acons$, $\avariable < \avariablebis' \in \acons$
      and $\avariable < \avariablebis \in \acons$ are similar to the previous case and are omitted herein.

\item Suppose $\avariable' = \avariablebis'\in\acons$ with $\asymtree(\anode \cdot j) = \pair{\aletter}{\acons}$.
  Since $\acons$ is satisfiable, for some $\adatum^{\dag}$
  in `$< \adatum_1$', $\adatum_{1}$, \ldots,
  $\adatum_{\alpha}$, `$> \adatum_{\alpha}$', 
   we have
   $\pair{\anode \cdot j}{\avariable}, \pair{\anode \cdot j}{\avariablebis} \in U_{\adatum^{\dag}}$. If $\adatum^{\dag}$ is different from `$> \adatum_1$' and
   `$> \adatum_{\alpha}$',
necessarily $g(\anode \cdot j,\avariable) = g(\anode \cdot j,\avariablebis)$. Otherwise, since
$\pair{\anode \cdot j}{\avariable} \step{=} \pair{\anode \cdot j}{\avariablebis}$ in $\newGt$, we have
$\slen{\anode \cdot j, \avariable} = \slen{\anode \cdot j, \avariablebis}$. Consequently, $g(\anode \cdot j,\avariable) =
g(\anode \cdot j,\avariablebis)$ too and $\Zed \models (\avariable' = \avariablebis')(\atree'(\anode), \atree'(\anode \cdot j))$. 
\item Suppose $\avariable = \avariablebis\in\acons$ with $\asymtree(\anode \cdot j) = \pair{\aletter}{\acons}$.
  Since $\acons$ is satisfiable and $\asymtree$ is locally consistent, for some $\adatum^{\dag}$
  in `$< \adatum_1$', $\adatum_{1}$, \ldots,
  $\adatum_{\alpha}$, `$> \adatum_{\alpha}$', 
   we have
   $\pair{\anode}{\avariable}, \pair{\anode}{\avariablebis} \in U_{\adatum^{\dag}}$. If $\adatum^{\dag}$ is different from `$> \adatum_1$' and
   `$> \adatum_{\alpha}$',
necessarily $g(\anode,\avariable) = g(\anode,\avariablebis)$. Otherwise, since
$\pair{\anode}{\avariable} \step{=} \pair{\anode}{\avariablebis}$ in $\newGt$, we have
$\slen{\anode, \avariable} = \slen{\anode, \avariablebis}$. Consequently, $g(\anode,\avariable) =
g(\anode,\avariablebis)$ too and $\Zed \models (\avariable = \avariablebis)(\atree'(\anode), \atree'(\anode \cdot j))$.
\item The case  $\avariable' = \avariablebis \in\acons$ with $\asymtree(\anode \cdot j) = \pair{\aletter}{\acons}$ is similar and it is omitted
  below. \qedhere
\end{itemize}
\end{proof}

\end{document}